\documentclass[12pt,a4paper]{article}

\usepackage{amscd,amsfonts,amsmath,amssymb,amsthm,bbm,bm,latexsym,mathrsfs}
\usepackage{epsfig,graphics,graphicx,natbib,subfigure,longtable}
\usepackage[english]{babel}
\usepackage[usenames]{color}
\usepackage{bookmark}
\usepackage{booktabs}
\usepackage{sgame}
\usepackage[top=1.2in, bottom=1.2in, left=0.5in, right=0.5in]{geometry}
\allowdisplaybreaks
\makeatother

\def\SC { \mathscr{C}}

\def\SD { \mathscr{D}}

\def\CL { \mathcal{L}}
\def\CM { \mathcal{M}}

\def\CB {\mathcal{B}}
\def\CP { \mathcal{P}}

\def\CL { \mathcal{L}}
\def\C { \mathcal{C}}

\def\CD {\mathcal{D}}

\def\CS {\mathcal{S}}

\def\CP {\mathcal{P}}

\def \CJ{\mathcal{J}}
\def \CO{\mathcal{O}}
\def \CC{\mathcal{C}}

\def \FP{\mathfrak{q}}  

\def \EPk{\mathfrak{q}} 

\def\th {\theta}
\def\s {\sigma}

\def\n {\nu}

\def\z {\zeta}
\def\lm {\lambda}

\def\J {\mathbf{1}}

\def\N {\mathbb{N}}

\def\E {\mathbb{E}}

\def\RE {\mathbb{R}}

\def\N {\mathbb{N}}
\def\P {\mathbb{P}}

\def \XX{\mathbb{X}}

\def \tm {\tilde \mu}

\usepackage{color,comment}

\newtheorem{definition}{Definition}

\newtheorem{lemma}[definition]{Lemma}
\newtheorem{proposition}{Proposition}
\newtheorem{corollary}[definition]{Corollary}
\newtheorem{remark}{Remark}
\newtheorem{example}{Example}

\begin{document}

\title{\vspace{-50pt} Hierarchical Species Sampling Models\thanks{We thank  for their useful comments on a previous version Federico Camerlenghi and Antonio Lijoi and the participants of the Italian-French Statistics Workshop 2017, Venice.}}

\author{
\hspace{-20pt} 
Federico Bassetti\textsuperscript{a} \hspace{15pt} Roberto Casarin\textsuperscript{b} \hspace{15pt} Luca Rossini\textsuperscript{c} 
        \\ 
        \vspace{5pt}
        \\
        {\centering {\small \textsuperscript{a}Polytechnic University of Milan, Italy \hspace{5pt} \textsuperscript{b}Ca' Foscari University of Venice, Italy }} \\
         {\centering {\small 
        \textsuperscript{c}Free University of Bozen-Bolzano, Italy}}}
\date{\today}
\maketitle
\abstract{
This paper introduces a general class of hierarchical nonparametric prior distributions. The random probability measures are constructed by a hierarchy of generalized species sampling processes with possibly non-diffuse base measures. The proposed framework provides a general probabilistic foundation for hierarchical random measures with either atomic or mixed base measures and allows for studying their properties, such as the distribution of the marginal and total number of clusters. We show that hierarchical species sampling models have a Chinese Restaurants Franchise representation and can be used as prior distributions to undertake Bayesian nonparametric inference. We provide a method to sample from the posterior distribution together with some numerical illustrations. Our class of priors includes some new hierarchical mixture priors such as the hierarchical Gnedin measures, and other well-known prior distributions such as the hierarchical Pitman-Yor and the hierarchical normalized random measures.\\

}

\section{Introduction}
Cluster structures in multiple groups of observations can be modelled by means of {\it hierarchical random probability measures} or {\it hierarchical processes} that allow for heterogenous clustering effects across groups and for sharing clusters among the groups. As an effect of the heterogeneity, in these models the number of clusters in each group (marginal number of clusters) can differ, and due to cluster sharing, the number of clusters in the entire sample (total number of clusters) can be smaller than the sum of the marginal number of clusters. An important  example of hierarchical random measure is the  Hierarchical Dirichlet Process (HDP), introduced in  the seminal paper of \cite{TehJordan2006}. The HDP involves a simple Bayesian hierarchy where the common base measure for a set of Dirichlet processes is itself distributed according to a Dirichlet process. This means that the joint law of a vector of random probability measures 
$(p_1,\dots,p_I)$ is
\begin{equation}\label{Init_hdp0}
\begin{split}
 p_{i}    | p_0  &\stackrel{iid}{\sim} DP (\theta_1,p_0), \quad i=1,\dots,I,  \\
  p_{0}    | H_0 &{\sim} DP  (\theta_0,H_0),  \\
 \end{split}
\end{equation}
where $DP (\theta, p)$ denotes the Dirichlet process with base measure 
$p$ and concentration parameter $\theta \in \RE$. Once the joint law of $(p_1,\dots,p_I)$ has been specified,  observations $[\xi_{i,j}]_{i=1,\dots,I; j \geq 1}$ 
are assumed to be conditionally independent given $(p_1,\dots,p_I)$ with 
\[
\xi_{i,j} | (p_1,\dots,p_I)  \stackrel{ind}{\sim} p_i,  \,\,i=1,\dots,I \mbox{ and } j \geq 1.
\]
If the observations take values in a Polish space, then this is equivalent to  partial exchangeability of the array $[\xi_{i,j}]_{i=1,\dots,I; j \geq 1}$  (de Finetti's representation Theorem). Hierarchical processes are widely used as prior distributions in Bayesian nonparametric inference (see \cite{Teh10} and reference therein), by assuming $\xi_{i,j}$ are hidden variables describing the clustering structures of the data and the observations in the $i$-th group, $Y_{i,j}$, are conditionally independent given $\xi_{i,j}$ with 
\[ 
Y_{i,j}| \xi_{i,j}   \stackrel{ind}{\sim} f(\cdot|\xi_{i,j}),
\]
where $f$ is a suitable kernel density.

In this paper we provide a new class of hierarchical random probability measures constructed by a hierarchy of generalized species sampling sequences, and call them Hierarchical Species Sampling Models (HSSM). A species sampling sequence is an exchangeable sequence whose directing measure is a discrete random probability 
\begin{equation}
p=\sum_{j \geq 1} \delta_{Z_j} q_j, \label{Def_P}
\end{equation} 
where $(Z_j)_{j \geq 1}$ and $(q_j)_{j \geq 1}$ are stochastically independent, $Z_j$ are i.i.d. with common distribution $H_0$ and the non-negative weights $q_j\geq 0$ sum to one almost surely. By Kingman's theory on exchangeable partitions,  any random sequence of positive weights such that $\sum_{j\geq 1}  q_{j} \leq 1$ can be associated to an exchangeable random partition of integers $(\Pi_n)_{n \geq 1}$.  Moreover, the law of an exchangeable random partition  $(\Pi_n)_{n \geq 1}$ is completely described by the so-called exchangeable probability partition function (EPPF). Hence, the law of the above defined random probability measure $p$,  turns out to be parametrized by  an EPPF  $\FP$  and  by a base measure $H_0$, which
 is diffuse in species sampling sequences, and possibly non-diffuse in our generalized species sampling construction.

A vector of random measures in the HSSM class can be described as follows. Denote by $SSrp(\FP,H_0)$ the law of the random probability measure  $p$ defined in \eqref{Def_P} and let $(p_0,p_1,\dots,p_I)$ be a vector of random probably measures such that
\begin{equation}\label{HSSMintro}
\begin{split}
 p_{i}    | p_0  &\stackrel{iid}{\sim} SSrp( \FP, p_0), \qquad i=1,\dots,I,  \\
  p_{0}     &{\sim} SSrp(\FP_0,H_0), \\
 \end{split}
\end{equation}
where $H_0$ is a base measure and $\FP_0$ and $\FP$ are two EPPFs.  A HSSM is an array of observations 
 $[\xi_{i,j}]_{i=1,\dots,I,j\geq 1 }$ conditionally independent given $(p_1,\dots,p_I)$ with 
$\xi_{i,j} | (p_1,\dots,p_I)  \stackrel{ind}{\sim} p_i$, $ i=1,\dots,I$  and  $j \geq 1$. 

The proposed framework is general enough to  provide a probabilistic foundation of both existing and novel hierarchical random measures, also allowing for non diffuse base measures. Our HSSM class includes the HDP, its generalization given by the Hierarchical Pitman--Yor process (HPYP), see \cite{Teh2006,Buntine10,Buntine16} and the hierarchical normalized  random measures with independent increments (HNRMI), recently studied  in \cite{Cam18}. Among the novel measures, we study the hierarchical generalization of the Gnedin process (\cite{Gnedin10}) and of finite mixtures (e.g., \cite{MillerHarrison}) and the asymmetric hierarchical constructions with $p_0$ and $p_i$ of different type (\cite{Buntine10,Buntine14}). 
Also, we consider new HSSM with non-diffuse base measure in the spike-and-slab class of prior introduced by \cite{George1993} and now widely studied in Bayesian parametric (e.g., \cite{Castillo2015}, \cite{George2017} and \cite{Rockova2018}) and nonparametric (e.g., \cite{Canale2017}) inference. Finally, note that non-diffuse base measures are also used in other random measures models (e.g., see \cite{Prunster13}), although these models are not in our class, the hierarchical species sampling specification can be used to generalize them.

By exploiting the properties of hierarchical species sampling sequences, we are able to provide the finite sample distribution of the number of clusters for each group of observations and the total number of clusters. Moreover, we provide some new asymptotic results when the number of observations goes to infinity, thus extending the asymptotic approximations for species sampling given in \cite{Pit06}) and for hierarchical normalized random measures given in \cite{Cam18}.

We show that the measures in the proposed class have a Chinese Restaurant Franchise representation, that is appealing for the applications to Bayesian nonparametrics, since it sheds light on the clustering mechanism of the processes and suggests a simple sampling algorithm for posterior computations whenever the EPPFs  $\FP_0$ and $\FP$ are known explicitly. In the Chinese Restaurant Franchise metaphor, observations are attributed to ``customers'', identified by the indexes $(i,j)$,  $i=1,\dots,I$ denote the  restaurants (groups) and the customers are clustered according to ``tables''. Hence, the first step of the clustering process  (from now on, bottom level) is the {\it restaurant-specific  sitting plan}. 
Tables are  then clustered, in an higher level of the hierarchy (top level), by means of ``dishes'' served to the tables. 
In a nutshell, observations driven by HSSM can be described as follows:
\begin{equation}\label{h0bis}
\begin{split}
\xi_{i,j}=\phi_{d_{i,c_{i,j}}},\quad \phi_n   \stackrel{iid}{\sim} H_0, \quad [d_{i,c},c_{i,j}] \sim Q,
\end{split}
\end{equation}
where
$c_{i,j}$ is the (random) table at which the $j$-th ``customer'' of ``restaurant'' $i$ sits,  $d_{i,c}$ is the  (random) index of the ``dish'' served at table $c$ in restaurant $i$ and the $\phi_n$ are the "dishes" drawn from the base probability measure $H_0$. The distribution of the observation process is completely described once the law $Q$ of the process $[d_{i,c},c_{i,j}]_{i,j}$ is specified. 
We will see that the process $[d_{i,c},c_{i,j}]_{i,j}$ plays a role similar to the one of the  random partition  
associated with  exchangeable species sampling sequences.

The paper is organized as follows. Section \ref{Sec_Back} reviews exchangeable random partitions, generalized species sampling sequences and species sampling random probability measures. Some examples are discussed and new results obtained under the assumption of non-diffuse base measure. Section \ref{Sec_HSSM} introduces hierarchical species sampling models, presents some special cases and shows some properties such as the Chinese restaurant franchise representation, which are useful for applications to Bayesian nonparametric inference. Section \ref{Sec_asympt} provides the finite-sample and the asymptotic distribution of the marginal and total number of clusters under both assumptions of diffuse and non-diffuse base measure. A Gibbs sampler for hierarchical species sampling mixtures is established in Section \ref{Sec_Gibbs}. Section \ref{Sec_Emp} presents some simulation studies and a real data application.


\section{Background material}
\label{Sec_Back}
Exchangeable random partitions  provide an important probabilistic tool 
for a wide range of theoretical and applied problems. They have been used in various fields such as population genetics \cite{Ewens,Kingman1980,Donnelly1986,Hoppe1984},  combinatorics, algebra and number theory \cite{Donnelly1993,Diaconis2012,Arratia2003}, machine learning  
\cite{Teh2006,WooArcGas2009a}, psychology \cite{NAVARRO2006101}, and model-based clustering \cite{pet07,pet10}.
 In Bayesian nonparametrics  they are used to describe the latent 
clustering structure  of infinite mixture models, see e.g. \cite{BNP2010} and the references therein. 
For a comprehensive review on exhangeable random partitions from a probabilistic perspective see \cite{Pit06}.

Our HSSM build on exchangeable random partitions and related processes, such as 
 species sampling sequences and species sampling random probability measures. We present their 
 definitions and some properties which will be used in this paper.

\subsection{Exchangeable partitions}\label{Sec:1-1}

A (set) partition $\pi_n$ of $[n]:=\{1,\dots,n\}$  is an unordered
collection $(\pi_{1,n},\dots,\pi_{k,n})$  of disjoint non-empty subsets (blocks) {of $\{1,\dots,n\}$}
 such that $\cup_{j=1}^k \pi_{j,n}=[n]$. 
 A partition $\pi_n=[\pi_{1,n},\pi_{2,n},\dots,\pi_{k,n}]$   has  $|\pi_n|:=k$ blocks (with $1 \leq |\pi_n|  \leq n$)  and we denote
by $|\pi_{c,n}|$, the number of  elements of the block 
$c=1,\dots,k$. 
We denote by
 $\CP_n$  the collection of all  partitions of $[n]$ and, 
given a partition,  
we list its blocks in ascending order of their smallest element. In other words,
 a partition $\pi_n \in \CP_n$ is
coded with elements {\it in order of appearance}. 

A  sequence of random partitions, $\Pi=(\Pi_n)_n$, is called 
 {\it random partition of $\N$} if 
for each $n$ the random variable  $\Pi_n$ takes  values in $\CP_n$ 
and, for $m < n$, the restriction of $\Pi_n$  to $\CP_m$ is $\Pi_m$ ({\it consistency property}).
A random partition of $\N$  is said to be {\it exchangeable} if  for every $n$ the distribution of $\Pi_n$ is invariant
under the action of all  permutations  (acting on $\Pi_n$ in the natural way). 

Exchangeable random partitions are characterized by the fact that 
their distribution depends on $\Pi_n$ only through its block sizes. In point of fact, a random partition on $\N$  is exchangeable if and only  if its distribution can be expressed  by an  {\it exchangeable partition  probability function} (EPPF).
An EPPF  is a family\footnote{
 An EPPF can be seen as a family of 
 symmetric functions $\EPk_{k}^n(\cdot)$ defined on 
 $\CC_{n,k}=\{ (n_1,\dots,n_k) \in \N^k: \sum_{i=1}^k n_{i}=n\}$. 
 To lighten the notation we simply write $\EPk$ in place of $\EPk_{k}^n$.
 Alternatively, one can think that $\EPk$ is a function on 
 $\cup_{n \in \N} \cup_{k=1}^n \CC_{n,k}$.
 }
 of 
 symmetric functions
$\EPk$
defined on 
 the integers  $(n_1,\dots,n_k)$, with $\sum_{i=1}^k n_{i}=n$,  
that satisfy the additions rule
\[
\EPk(n_1,\dots,n_k)=\sum_{j=1}^k \EPk(n_1,\dots,n_j+1,\dots,n_k)+\EPk(n_1,\dots,n_k,1),
\] 
(see \cite{Pit06}). 
In particular, if $(\Pi_n)_n$ is an exchangeable random partition of $\N$, there exists an EPPF such that for every $n$ and $\pi_n \in \CP_n$
\begin{equation}\label{EPPF1}
\P\{ \Pi_n=\pi_n\}=\EPk\left(|\pi_{1,n}|,\dots,|\pi_{k,n}|\right)
\end{equation}
 where $k=|\pi_n|$. %
 In other words, $\EPk(n_1,\dots,n_k)$
corresponds to the probability that $\Pi_n$ is equal to 
any particular partition of $[n]$  having $k$ distinct blocks with
frequencies $(n_1,\dots,n_k)$.

Given an EPPF  
 $\EPk$, one  deduces the corresponding sequence of predictive distributions. Starting with $\Pi_1=\{1\}$, given $\Pi_{n}=\pi_n$ (with $|\pi_n|=k$), 
the conditional  probability of adding a new block (containing $n+1$)  to $\Pi_n$  is
\begin{equation}\label{predEPPF1}
\nu_{n}(|\pi_{1,n}|,\dots,|\pi_{k,n}|) :=  \frac{\EPk(|\pi_{1,n}|,\dots,|\pi_{k,n}|,1)}{\EPk(|\pi_{1,n}|,\dots,|\pi_{k,n}|)};
\end{equation}
while the conditional 
 probability of adding $n+1$ to the $c$-th block of  $\Pi_n$ (for $c=1,\dots,k$) is 
\begin{equation}\label{predEPPF2}
\omega_{n,c}(|\pi_{1,n}|,\dots,|\pi_{k,n}|):= 
\frac{\EPk(|\pi_{1,n}|,\dots,|\pi_{c,n}|+1,\dots,|\pi_{k,n}|)}{\EPk(|\pi_{1,n}|,\dots,|\pi_{k,n}|)}. 
\end{equation}



An important class of exchangeable random partitions is the  Gibbs-type partitions,  
 introduced  in 
 \cite{GP2006}
  and  characterized by an EPPF
with a product form, that is
\[
\EPk(n_1,\dots,n_k):=V_{n,k} \prod_{c=1}^k (1-\sigma)_{n_c-1},
\]
where $(x)_n=x(x+1)\dots (x+n-1)$ is the rising factorial (or Pochhammer polynomial), $\sigma<1$ and  $V_{n,k}$ are positive real numbers such that  $V_{1,1}=1$ and
\begin{equation}
  (n-\sigma k) V_{n+1,k} + V_{n+1,k+1}=V_{n,k} \label{Rec_V}
\end{equation}
for every $n\geq 1$ and $1 \leq k \leq n$. Hereafter, we report some important examples of Gibbs-type random partitions.

\begin{example}[Pitman-Yor two-parameter  distribution]\label{PYEPPFex}
A noteworthy example
of Gibbs-type  EPPF is the so-called 
Pitman-Yor two-parameters  family, $PY(\sigma,\theta)$.  It is defined by  
\[
\EPk(n_1,\dots,n_k):=\frac{\prod_{i=1}^{k-1} (\theta+i\sigma)}{(\theta+1)_{n-1}} \prod_{c=1}^k (1-\sigma)_{n_c-1},
\]
where $0 \leq \sigma < 1$ and $\theta >-\sigma$; or $\sigma<0$ and $\theta=|\sigma|m$ for some integer $m$, see
\cite{Pit95,Pit97}.
This leads to the following predictive rules 
\[
\nu_n(|\pi_{1,n}|,\dots,|\pi_{k,n}|)=\frac{\theta+\sigma k}{\theta+n}  \quad \text{and} \quad \omega_{n,c} (|\pi_{1,n}|,\dots,|\pi_{k,n}|)=\frac{|\pi_{c,n}|  -\sigma}{\theta+n}.
\]
The Pitman-Yor two-parameters  family generalizes the Ewens  distribution \cite{Ewens}, 
which  is obtained for $\sigma=0$
\[
\EPk(n_1,\dots,n_k):=\frac{\theta^k}{(\theta)_n} \prod_{i=1}^k (n_i-1)!
\]
If  $\sigma<0$ and $\theta=|\sigma|m$,  then $V_{n,k}=0$ for $k>\min\{n,m\}$, which means
that the maximum number of blocks in a random partition of length $n$ is $\min\{n,m\}$ with probability one. 
It is possible to show that these random partitions can be obtained by sampling $n$ individuals from   a population composed by $m$
different species with proportions distributed according to a symmetric Dirichlet distribution of parameter $|\sigma|$, see \cite{GP2006}. 
\end{example}

\begin{example}[Partitions induced by Mixtures of Finite Mixtures]\label{Ex:MfM}
 In  \cite{GP2006}, it has been proved that any Gibbs-type
EPPF with $\sigma < 0$ is a mixture of $PY(\sigma, m|\sigma|)$ partitions with respect to $m$, with mixing probability measure $\rho=(\rho_m)_{m\geq 1}$ on the positive integers.
In this case 
$\EPk(n_1,\dots,n_k):=V_{n,k} \prod_{i=1}^k (1-\sigma)_{n_i-1}$,
where $\sigma < 0$ and 
\begin{equation}\label{MFM-EPPF-Vkn}
V_{n,k} 
=|\sigma|^{k-1} \sum_{m \geq k} \frac{\Gamma(m)\Gamma(|\sigma|m+1)}{\Gamma(m-k+1)\Gamma(|\sigma|m+n)} \rho_m.
\end{equation}
These Gibbs type EPPFs can also be  obtained by considering the random partitions induced by the so-called Mixture of Finite Mixtures (MFM), 
see  \cite{MillerHarrison}.
When $|\sigma| = 1$,  \cite{Gnedin10} shows 
a distribution on  $m$  for which $V_{n,k}$ has closed-form and this special case will be described in the following example. 
\end{example}

\begin{example}[Gnedin model]\label{ex:gnedinpart}
\cite{Gnedin10} introduced a sequence of exchangeable partitions with  explicit predictive weights
\begin{align}
\omega_{n,c}(|\pi_{1,n}|,\dots,|\pi_{k,n}|)&=\frac{(|\pi_c|+1)(n-k+\gamma)}{n^2+\gamma n+ \zeta}, \label{oldpb}\\
\nu_n(|\pi_{1,n}|,\dots,|\pi_{k,n}|)&=\frac{k^2-\gamma k+\zeta}{n^2+\gamma n+ \zeta}, \label{newpb}
\end{align}
where the parameter $(\gamma,\zeta)$ must be chosen such that $\gamma \ge 0$ and $k^2-\gamma k+\zeta$ is (i) either (strictly) positive for all $k \in \mathbb{N}$ or (ii) positive for $k \in \{1,\dots,k_0-1\}$ and has a root at $k_0$. 
In point of fact, the Gnedin model, denoted with GN($\gamma,\zeta)$, can be deduced  as special case of  Gibbs-type EPPF  with negative $\sigma=-1$
described in the previous example. As shown in Theorem 1 by \cite{Gnedin10}, these random partitions have representation 
\eqref{MFM-EPPF-Vkn} with  
\begin{equation}
\rho_m=\frac{\Gamma(z_1+1)\Gamma(z_2+1)}{\Gamma(\gamma)} \frac{\prod_{l=1}^{m-1} (l^2-\gamma l+\zeta)}{m!(m-1)!}, \label{DistXbis}
\end{equation}
where $z_1,z_2$ are complex root for the equation $(x^2+\gamma x+\zeta)=(x+z_1)(x+z_2)$, 
that is
$
z_{1,2}=(\gamma\pm \sqrt{\gamma^2-4\zeta})/{2}$. 
See also \cite{Cerquetti}.
\end{example}


\begin{example}[Poisson-Kingman partitions]  Using the ranked random discrete distribution derived from an
inhomogeneous Poisson point process,  \cite{Pitman2003}
introduced a very broad class of EPPF, the so-called \emph{Poisson-Kingman exchangeable partition probability functions}, 
\begin{equation}\label{PKEPPF}
\EPk(n_1,\dots,n_k)= \frac{\theta^{k}}{\Gamma(n)}
\int_{\RE^+} \lm^{n-1} L(\lm) \prod_{j=1}^{k} \int_{\RE^+}
y^{n_j}e^{-\lm y} \eta(y) dy d \lm,
\end{equation}
with L\'evy density $\theta \eta$, where $\theta>0$ and
$\eta$ is a measure on $\RE^+$ (absolutely continuous with respect to the Lebesgue measure) 
such that
\begin{equation}\label{PKregularity}
\int_{\RE^+}(1-e^{-\lm y}) \eta(y)dy< +\infty, \quad \text{for all} \,\, \lm>0 \quad \text{and} \quad \eta(\RE^+)=+\infty, 
\end{equation}
 and \( L(\lm)=\exp\{-\theta \int_{\RE^+}(1-e^{-\lm v}) \eta(v)dv\} \).
 This EPPF is  related to normalized homogeneous  completely random measures of 
 \cite{Lijoi2009}
 (see Example \ref{Ex:NRM} and Appendix  \ref{App:NRM} for details).
 
\end{example}

\subsection{Species Sampling Models}%
Kingman's theory of random partitions  sets up a one-one correspondence
(Kingman's correspondence) 
between EPPFs  and distributions for  decreasing sequences of random
variables $(\tilde q_k^{\downarrow })_k$ 
with $\tilde  q_i^{\downarrow }  \geq 0$ and $\sum_i \tilde  q_i^{\downarrow }   \leq 1$
a.s., by using the notion 
of random partition induced by a sequence of random variables. 

A sequence of random variables $(\xi_n)_n$    induces a 
random  partition on  $\N$
by  equivalence classes 
$i \sim j $ if and only if $\xi_i = \xi_j$.
Note that if  $(\xi_n)_n$ is exchangeable then the induced random partition is also exchangeable.  

Kingman's correspondence theorem states that 
for any exchangeable random partition $\Pi$ with EPPF $ \FP$, 
it exists a random decreasing sequence $\tilde q_1^{\downarrow }  \geq \tilde q_2^{\downarrow } \geq \dots$ with $\tilde  q_i^{\downarrow }  \geq 0$ and $\sum_i \tilde  q_i^{\downarrow }   \leq 1$, 
such that 
if $I_1,\dots,I_n,\dots$ are conditionally independent allocation variables with 
\begin{equation}\label{allocationv}
\P\{ I_i= j |  (\tilde q_k^{\downarrow })_k \}= 
\left \{
\begin{array}{cc}
  \tilde q_j^{\downarrow }&   j \geq 1,   \\
  1-\sum_k \tilde q_k^{\downarrow} &  j =-i,    \\
 0    & \text{otherwise}    \\
\end{array}
\right .
\end{equation}
the partition induced by  $(I_n)_n$ has  the same law of $\Pi$.
See \cite{Kingman78} and \cite{Aldous85}.
 When  $(\tilde q_j^{\downarrow })_j$ is such that $\sum_i \tilde  q_i^{\downarrow }   = 1$ a.s., 
 Kingman's correspondence can be made more explicit by the following result: let $\FP$ be the EPPF  of a random partition $\Pi$  built following the above construction 
and let $(\tilde q_j)_j$ be any (possibly random) permutation of $(\tilde q_j^{\downarrow })_j$, then 
 \begin{equation}\label{EPPFkingman}
 \EPk(n_1,\dots,n_k)=\sum_{j_1,\dots,j_k} \E \left[ \prod_{i=1}^{k} \tilde q_{j_i}^{n_i}  \right],
 \end{equation}
where $j_1,\dots,j_k$ ranges over all ordered $k$-tuples of distinct positive integers, see equation (2.14) in \cite{Pit06}.

We call   {\it Species Sampling random probability} of parameter $\FP$ and $H$, 
$p \sim SSrp(\FP,H)$, a  
   random distribution  $p= \sum_{j \geq 1} \delta_{Z_j} \tilde q_j$, where
 $(Z_j)_j$ are i.i.d. with common distribution $H$ (not necessarily diffuse) 
 and  the EPPF defined by  $(\tilde q_j)_{j}$ via  \eqref{EPPFkingman} is  $\FP$. 
 Such random probability  measures  are sometimes called {\it species sampling models}.

Given   an EPPF $\FP$   and a diffuse probability measure  $H$  (i.e.,  $H(\{x\})=0$ for every $x$)
on a Polish space $\XX$, 
an exchangeable sequence $(\xi_n)_n$ taking values on $\XX$ is 
 a {\it Species Sampling Sequence}, $SSS(\FP,H)$, 
if the law of  $(\xi_n)_n$ is characterized by 
the predictive system: 
 \begin{itemize}
\item[\textit{i})]  $\P\{\xi_1 \in dx\} = H(dx) $;
\item[\textit{ii})] the conditional distribution of $\xi_{n+1}$ given $(\xi_1,\dots,\xi_n)$ is 
\begin{equation}\label{predictiveSSS}
\P\left\{\xi_{n+1}  \in dx | \xi_1,\dots,\xi_n \right\}= \sum_{c=1}^{k} \omega_{n,c} \delta_{\xi^*_c}(dx) +  
\nu_{n}  H(dx), 
\end{equation}
where $(\xi_1^*,\dots,\xi_k^*)$ is the vector of distinct observations in order of appearance,
$\omega_{n,c} =\omega_{n,c}(|\Pi_{1,n}|,\dots,|\Pi_{k,n}|)$, $\nu_{n} =\nu_{n}(|\Pi_{1,n}|,\dots,|\Pi_{k,n}|)$,
$k=|\Pi_n|$, 
 $\Pi_n$ is the random partition induced by $(\xi_1,\dots,\xi_n)$ and 
 $\omega_{n,c}$ and $\nu_n$ are related with the  $\FP$ by \eqref{predEPPF1}-\eqref{predEPPF2}.
\end{itemize}
See \cite{Pitman96}. 

In point of fact, as shown in   Proposition 11 of \cite{Pitman96}, 
 $(\xi_n)_n$ is a $SSS(\FP,H)$ if and only  if the $\xi_n$ are conditionally i.i.d. given $p$, with common distribution
\begin{equation}\label{spec}
p= \sum_{j \geq 1} \delta_{Z_j} \tilde q_j +\left(1-\sum_{j \geq 1} \tilde q_j\right) H,
\end{equation}
 where $\sum_j \tilde q_j \leq 1$ a.s., $(Z_j)_j$ and $(\tilde q_j)_j$ are stochastically independent and
 $(Z_j)_j$ are i.i.d. with common diffuse distribution $H$. 
 The random probability measure $p$ in \eqref{spec}
 is said to be proper if
 $\sum_j \tilde q_j=1$ a.s.
  In this case,   $p$ is a $SSrp$ and
the EPPF of the exchangeable partition  $(\Pi_n)_n$  induced by $(\xi_n)_n$  is  $\FP$, 
where  $\FP$  and $(\tilde q_j)_j$ are related by \eqref{EPPFkingman}. For more details see  \cite{Pitman96}.

The name species sampling sequences is due to the following interpretation:
think to $(\xi_n)_n$ as an infinite population of individuals belonging to possibly infinite  different species. The number of partition blocks $|\Pi_n|$  takes on the interpretation of the number of distinct species in the
sample $(\xi_1,\dots, \xi_n)$,   the $\xi^*_c$ are the observed distinct  species types and the $|\Pi_{c,n}|$ are the corresponding species frequencies.  
In this species metaphor, $\nu_n$ is the probability of observing
a new species  at the $(n+1)$-th sample, while $\omega_{n,c}$ is the probability of observing an already sampled species of type $\xi^*_c$.

\subsection{Generalized species sampling sequences}
Usually a species sampling sequence  $(\xi_n)_n$ is defined  by the predictive system \eqref{predictiveSSS} assuming that
the measure $H$ is diffuse (i.e. non-atomic). While this assumption is  essential for the results recalled above to be valid, 
an exchangeable sequence can be defined by sampling from a $SSrp(\FP,H)$ for any measure $H$. 
More precisely,  we say that  a sequence $(\xi_n)_n$ is a 
{\it generalized species sampling sequence}, 
 $gSSS(\FP,H)$, if the variables 
$\xi_n$ are conditionally i.i.d. (given  $p$) with law $p \sim SSrp(\FP,H)$ or equivalently if 
the directing measure of  $(\xi_n)_n$ is $p$.
From the previous discussion, it should be clear that if  $(\xi_n)_n$ is  a $gSSS(\FP,H)$ with $H$ diffuse,  then it is a  $SSS(\FP,H)$, see Proposition 13 in \cite{Pitman96}.
When $H$ is not diffuse,
 the relationship between the random partition induced by the sequence $(\xi_n)_n$ and the EPPF $\FP$ is not as simple as in the non-atomic case. Understanding this relation and the  partition structure of the $gSSS$  is of paramount importance in order to define and study hierarchical models of type \eqref{HSSMintro}, since  
 the random measure $p_0$ in  the hierarchical specification \eqref{HSSMintro} is almost surely discrete (i.e. atomic).
Moreover, properties of these sequences are also relevant for studying non-parametric prior distributions with mixed base measure, such as the Pitman-Yor process with spike-and-slab base measure introduced in \cite{Canale2017}.

Given a random partition $\Pi$, let $\SC_j(\Pi)$ be the random index of the block containing $j$, that is 
\begin{equation}\label{partitionindex}
\SC_j(\Pi)=c  \quad \text{ if $j \in \Pi_{c,j} $}
\end{equation}
or equivalently  if $j \in \Pi_{c,n} $ for some (and hence all)  $n \geq j$.  In the rest of the paper,  if 
$\pi_n=[\pi_{1,n},\dots,\pi_{k,n}]$ is a partition of $[n]$ and $\EPk$ is any EPPF, we 
will write $\EPk(\pi_n)$ in place of $\EPk(|\pi_{1,n}|,\dots,|\pi_{k,n}|)$.
 
\begin{proposition}\label{prop_gsss}
Let $p \sim SSrp(\FP,H)$, with $p$ proper and $H$ not necessarily diffuse. 
Let $(I_n)_n$ be the allocation sequence defined in \eqref{allocationv} by
taking the weights of $p$ in decreasing order. Assume  that 
$(I_n)_n$ and  $(Z_j)_j$ (in the definition 
of $p$) are independent. Finally, let $\Pi$ be a random partition, with EPPF equal to $\FP$, also 
independent of  $(Z_j)_j$. 
We define, for every $n \geq 1$, 
\begin{equation*}\label{allocationxi}
\xi_n:=Z_{I_n}   \qquad \text{and} \qquad \xi_n':=Z_{\SC_n(\Pi)}. 
\end{equation*}
Then: 
\begin{itemize}
\item[(i)] the sequence $(\xi_n)_n$  is a  $gSSS(\FP,H)$;
\item[(ii)] the sequences $(\xi_n)_n$  and  $(\xi_n')_n$ have the same law;
\item[(iii)] for any $A_1,\dots,A_n$ in the Borel $\sigma$-field of $\XX$
\[
\P\{ \xi_1 \in A_1, \cdots, \xi_n \in A_n\}=  
\sum_{\pi_n \in \CP_n} \FP(\pi_n) \prod_{c=1}^{|\pi_n|} H( \cap_{j \in \pi_{c,n}} A_j ),
\]
\end{itemize}
\end{proposition}

 \begin{remark}   
 If  $(\xi_n)_n$ is  a $gSSS(\FP,H)$ 
and $H$ is not diffuse, then $\FP$ is
not necessarily the EPPF induced by $(\xi_n)_n$.
To see this, take $\XX=\{0,1\}$, $H\{0\}=p$ and $H\{1\}=1-p$. Let
 $\tilde \Pi$ be the random partition induced by $(\xi_n)_n$
 and $\Pi$ a random partition with EPPF  $\FP$.
Using {\rm Proposition \ref{prop_gsss}}, 
one gets $\P\{ \tilde \Pi_2=[(1,2)]\} =\P\{\xi_1=\xi_2\}=\P\{\xi_1'=\xi_2'\}=$ $\P\{\Pi_2=[(1,2)]\}+\P\{\Pi_2=[(1),(2)]\} \P\{Z_1=Z_2\}$ $=\P\{\Pi_2=[(1,2)]\}+\P\{\Pi_2=[(1),(2)]  \} [p^2+(1-p)^2] \not= \P\{\Pi_2=[(1,2)]\}$, if $\P\{\Pi_2=[(1),(2)]  \} >0$, which shows that the EPPF of $\tilde \Pi$ is not $\FP$. 
\end{remark}   

When the base measure is not diffuse, the representation in   {\rm Proposition \ref{prop_gsss}} can be used to derive the EPPF of the partition induced by any  $gSSS(\FP,H)$. Since this property is not used in the rest of the paper we leave it for further research.  Here we focus on the distribution of the number of distinct observations in $\xi_1,\dots,\xi_n$, i.e.  $|\tilde \Pi_n|$, for any base measure $H$. We specialize the result for the spike-and-slab type of base measures, which have been used by \cite{suarez2016} in DP and by \cite{Canale2017} in PY processes. 

\begin{corollary}\label{prop_gsssTRIS} 
Let the assumptions of {\rm Proposition \ref{prop_gsss}} hold. 
\begin{itemize}  
\item[(i)]  If   $H^*(d|k)$ (for $1\leq d\leq k$) is the probability of observing exactly $d$ distinct values 
in the vector $(Z_1,\dots,Z_k)$ and let  $\tilde \Pi$ be the random partition induced by $(\xi_n)_n$
then, 
\[
\P\{ |\tilde \Pi_n|=d\}= \sum_{k =d}^{n} H^*(d|k) \P\{|\Pi_n| =k\}.    
\]
\item[(ii)]  If the base measure is in the spike-and-slab class, i.e.
$H(dx)=a \delta_{x_0}(dx)+ (1-a) \tilde H(dx)$ where $a \in (0,1)$, $x_0$ is a point of $\XX$ and $\tilde H$ is a diffuse measure on $\XX$, then 
\[
\P\{ |\tilde \Pi_n|=d\}=(1-a)^d \P\{|\Pi_n| =d\}  +  \sum_{k =d}^{n} {k \choose d-1} a^{k+1-d}(1-a)^{d-1}  
\P\{|\Pi_n| =k\} . 
\]
\end{itemize}
\end{corollary}

\begin{corollary}\label{prop_gsssBIS} Under the assumptions of
{\rm Proposition \ref{prop_gsss}},
 for every $n \geq 1$ 
and $c=1,\dots,|\Pi_n|$ one has  $Z_c=\xi'_{R(n,c)}$ 
with $R(n,c)=\min\{ j : j \in \Pi_{c,n}\}$ and 
\begin{equation}
\P\left\{\xi_{n+1}'  \in dx | \xi_1',\dots,\xi_n', \Pi_n \right\}= 
\sum_{c=1}^{|\Pi_n|} \omega_{n,c}(\Pi_n) \delta_{Z_c}(dx) +  
\nu_{n}(\Pi_n)   H(dx) \label{Eq_Corr2}
\end{equation}
 where $\omega_{n,c}$ and $\nu_{n}$
are related to  $\FP$ by \eqref{predEPPF1}-\eqref{predEPPF2},
\end{corollary}

\begin{remark}   Equation \eqref{Eq_Corr2} in Corollary \ref{prop_gsssBIS}   differs substantially from 
   the predictive system in \eqref{predictiveSSS} due to the
 conditioning on the latent partition $\Pi_n$. Nevertheless, 
 if 
$H$ is diffuse then $\Pi_n=\tilde \Pi_n$ a.s., conditioning on $(\xi_1,\dots,\xi_n,\Pi_n)$ 
is the same as   conditioning on $(\xi_1,\dots,\xi_n)$
and $Z_c$ is equal  to the $c$-th distinct observation in order of appearance.
Hence, in this case, \eqref{Eq_Corr2} reduces to \eqref{predictiveSSS}.
\end{remark}

Hereafter, we discuss some examples of $SSrp$ and $gSSS$  which  will be used in our hierarchical constructions. 
  
\begin{example}[Pitman-Yor and Dirichlet processes]\label{Ex:PYPDP}
If $\FP$ is the two-parameter Pitman-Yor distribution $PY(\sigma,\theta)$ and $p \sim SSrp(\FP,H)$, then $p$ is a Pitman-Yor process (PYP), denoted with $p \sim PYP (\sigma, \theta,H)$ where $\sigma$ and $\theta$ are the discount and concentration parameters, respectively,  and $H$ is the  base measure (see \cite{Pitman96}). To see this equivalence, recall the description of the PYP in terms of its stick-breaking construction
\begin{equation}\label{PY1}
p=\sum_{k \geq 1} q_k \delta_{Z_k},
\end{equation}
where $(Z_k)_k$ is an i.i.d. sequence of random variables with law $H$ and 
$q_{k}=V_{k} \prod_{l=1}^{k-1} (1-V_{l})$,
with $V_{k}$ independent  $Beta(1-\sigma,\theta+k\sigma)$ random variables. From \eqref{PY1}, it is clear that a $PYP$ is a $SSrp$. Moreover, it is well-known that the EPPF associated to the weights defined above is the Pitman-Yor EPPF  of Example \ref{PYEPPFex} (see \cite{Pit95,Pitman96,Pit97}). As a special case, if $\FP$ is the Ewens distribution, $PY(0,\theta)$, and $p \sim SSrp(\FP,H)$, then $p$ is a Dirichlet process (DP) denoted with $DP(\sigma,H)$. 
Note that this is true even if $H$ have atoms. 
\end{example}

\begin{example}[Mixture of finite mixture processes]\label{Ex:GPFM} 
If $\FP$ is the distribution described in Example \ref{Ex:MfM} and $p \sim SSrp(\FP,H)$, then $p$ is a mixture of finite mixture process denoted with $MFMP(\sigma,\rho,H)$ and $p$ can be written as 
\begin{equation}\label{MFM1}
p=\sum_{k=1}^{K} q_{k} \delta_{Z_k},
\end{equation}
where $K\sim\rho$, $\rho=(\rho_m)_{m \geq 1}$ is a p.m. on $\{1,2,\dots\}$, $(Z_k)_{k\geq 1}   \stackrel{i.i.d.}{\sim} H$, and 
\begin{equation}\label{MFM2}
(q_{1},\dots,q_{K}) \mid K {\sim}\,\, \text{Dirichlet}_K(\sigma,\dots,\sigma),
\end{equation}
see   \cite{MillerHarrison}.
For $|\sigma|=1$ and $\rho$ given by \eqref{DistXbis}, an analytical expression for $\FP$ is available (see Example \ref{ex:gnedinpart})
and the process $p$ is called Gnedin process  (GP) denoted with $GP(\gamma,\zeta,H)$.
\end{example}

\begin{example}[Normalized completely random measures]\label{Ex:NRM}   Assume $\theta>0$ and let $\eta$ be a measure satisfying \eqref{PKregularity}, $\FP$ a Poisson-Kingman EPPF defined in \eqref{PKEPPF} and $H$ a probability measure (possibly not diffuse) on $\XX$. If $p  \sim SSrp(\FP,H)$,  then $p$ is a normalized homogeneous completely random measure, $NRMI(\theta,\eta,H)$, of parameters $(\theta,\eta,H)$. See Appendix  \ref{App:NRM} for the definition.  
The sequence $(\xi_n)_n$ obtained by sampling from $p$, i.e. a $gSSS(\FP,H)$, 
is a sequence from a $NRMI(\theta,\eta,H)$. 
All these facts are well known  when $H$ is a non-atomic measure (see \cite{Lijoi2009}). 
The results for general  measures and $\XX=\RE$ 
are implicitly contained, although not explicitly stated, in 
 \cite{Sangalli2007}. A detailed proof of the general case 
 is given in Proposition  \ref{prop:hNRM} of Appendix  \ref{App:NRM}. 
\end{example}

\section{Hierarchical Species Sampling Models}
\label{Sec_HSSM}
In this section we introduce hierarchical species sampling models (HSSMs) and study the relationship between HSSMs and a general class of hierarchical random measures which contains some well-known random measures  
(e.g., the HDP of \cite{TehJordan2006}). Some examples of HSSM are provided and some relevant properties of the HSSMs are given, such as the clustering structure and the induced random partitions.

\subsection{HSSM definition, properties and examples}
In the following definition a hierarchy of exchangeable random partitions is used to build hierarchical species sampling models. 
\begin{definition}\label{def_hssm} Let $\FP$ and $\FP_0$ be two EPPFs and $H_0$ a probability distribution on $\XX$. 
 We define an array $[\xi_{i,j}]_{i=1,\dots,I, j \geq 1}$ as a {\it Hierarchical  Species Sampling model}, $HSSM(\FP,\FP_0,H_0)$, of parameters
$(\FP,\FP_0,H_0)$, if for every vector of integer numbers 
$(n_1,\dots,n_I)$ and every collection of Borel sets  $\{ A_{i,j}: i=1,\dots,I, j=1,\dots,n_{i} \}$ 
 it holds 
\begin{equation}\label{defHSSM}
\begin{split}
\P\{ & \xi_{i,j} \in A_{i,j}   \,\,  i=1,\dots,I, j=1,\dots,n_{i } \} \\
& =\sum_{\pi^{(1)}_{n_1}\in \CP_{n_1}  ,\dots,\pi^{(I)}_{n_I} \in \CP_{n_{I}} } \prod_{i=1}^I 
\EPk\left( \pi^{(i)}_{n_i}\right) 
\E \left[  \prod_{i=1}^I  \prod_{c=1}^{|\pi^{(i)}_{n_i}|}  \tilde p\left( \cap_{j \in \pi^{(i)}_{c,n_i}} A_{i,j} \right)  \right],\\
\end{split}
\end{equation}
with  $\tilde p \sim SSrp(\FP_0,H_0)$. Moreover, 
the directing random measures $(p_1,\dots,p_I)$ of the array  $[\xi_{i,j}]_{i=1,\dots,I, j \geq 1}$
are called Hierarchical Species Sampling random measures, $HSSrp$.
\end{definition}
The next result states a relationship between hierarchies of $SSS$ and hierarchies of random probabilities, which are widely used in Bayesian nonparametrics, thus motivating the choice of name Hierarchical Species Sampling Model (HSSM) for the stochastic representation in Definition \ref{def_hssm}.
\begin{proposition}\label{prop_hssm}
Let $(p_0,p_1,\dots,p_I)$ be a vector of random probably measures such that
\[
\begin{split}
 p_{i}   | p_0  &\stackrel{iid}{\sim} SSrp( \FP, p_0), \qquad i=1,\dots,I,  \\
  p_{0}  &   \sim SSrp(\FP_0,H_0), \\
 \end{split}
\]
where $H_0$ is a base measure. 
Let $[\xi_{i,j}]_{i=1,\dots,I,j\geq 1 }$ be conditionally independent given $(p_1,\dots,p_I)$ with $\xi_{i,j} | (p_1,\dots,p_I)  \stackrel{ind}{\sim} p_i$, where $i=1,\dots,I$and $j \geq 1$. Then, $[\xi_{i,j}]_{i=1,\dots,I,j\geq 1}$ is a  $HSSM(\FP,\FP_0,H_0)$.
\end{proposition}

Proposition \ref{prop_hssm} provides a probabilistic foundation to a wide class of hierarchical random measures. It is worth noticing that the base measure $H_0$ is not necessarily diffuse and, thanks to the properties of the SSrp and of the gSSS (see Proposition \ref{prop_gsss}), the hierarchical random measures in Proposition \ref{prop_hssm} are well defined also for non-diffuse (e.g., atomic or mixed) probability measures $H_0$. Our result is general enough to be valid for many of the existing hierarchical random measures (e.g., \cite{TehJordan2006}, \cite{Teh2006},  \cite{Bac17}). As with species sampling sequences, HSSMs enjoy some exchangeability and clustering properties stated in the following proposition, where, recalling \eqref{partitionindex},  $\SC_j(\Pi)$ denotes the random index of the block 
of the random partition $\Pi$ that contains $j$.

\begin{proposition}\label{prop0} Let $ \Pi^{(1)},\dots,\Pi^{(I)} $ be i.i.d. exchangeable partitions
with EPPF $\FP$ and $\tilde p$ a random probability measure independent of $ \Pi^{(1)},\dots,\Pi^{(I)} $. If  $[\zeta_{i,j}]_{i=1,\dots,I, j \geq 1}$  are conditionally i.i.d. with law $\tilde p$ given  $\tilde p$, then the random variables 
\begin{equation}\label{def00HSSM}
\xi_{i,j}:= \zeta_{i,\SC_{j}(\Pi^{(i)})}, \quad \mbox{ where } i=1,\dots,I \mbox{ and }  j \geq 1
\end{equation}
are partially exchangeable and satisfy Eq.  \eqref{defHSSM}.
Furthermore, if $\tilde p$ is a $SSrp(\FP_0,H_0)$ then the sequence $[\xi_{i,j}]_{ij}$ is a 
$HSSM(\FP,\FP_0,H_0)$.  
\end{proposition} 

The stochastic representation given in Proposition \ref{prop0} allows us to find a simple representation of the HSSM clustering structure (see Section \ref{S:CRF}). In Bayesian nonparametric inference, such representation turns out to be very useful because it leads to a generative interpretation of the nonparametric-priors in the HSSM class, and also makes possible to design general procedures for approximated posterior inference (see Section \ref{Sec_Gibbs}).

The definition of Hierarchical Species Sampling models includes some well known hierarchical processes and allows for the definition of new processes, as showed in the following set of examples.
\begin{example}[Hierarchical Pitman-Yor process]\label{HPY}
Let $PYP (\sigma, \theta,H)$ and $DP (\theta,H)$ denote PY and DP processes, respectively, given in Example \ref{Ex:PYPDP}. A vector of dependent random measures $(p_1,\dots,p_I)$, with law characterized by the following hierarchical structure
\begin{equation}\label{hdp0}
\begin{split}
 p_{i}  &  | p_0  \stackrel{iid}{\sim} PYP (\sigma_1, \theta_1,p_0), \quad i=1,\dots,I,  \\
  p_{0}  &  | H_0 {\sim} PYP (\sigma_0, \theta_0,H_0)  \\
 \end{split}
\end{equation}
is called \emph{Hierarchical Pitman-Yor Process}, $HPYP (\sigma_0, \theta_0,\sigma_1,\theta_1,H_0)$, of parameters $0 \leq \sigma_i \leq 1$, $-\sigma_i<\theta_i$, $i=0,1$ and $H_0$ 
(see \cite{Teh2006,Buntine10,Buntine16}).  Combining Example \ref{Ex:PYPDP}  with Proposition \ref{prop_hssm}, it is apparent that samples from a  $HPYP (\sigma_0, \theta_0,\sigma_1,\theta_1,H_0)$ form a $HSSM$ of parameters $(\FP,\FP_0,H_0)$ where $\FP$ and $\FP_0$ are Pitman-Yor EPPFs of parameters $(\sigma_1,\theta_1)$ and $(\sigma_0,\theta_0)$, respectively. If $\sigma_0=\sigma_1=0$, then 
one recovers the  $HDP(\theta_0,\theta_1,H_0)$. It is also possible to define some \emph{mixed cases}, where one considers a DP in one of the two stages of the hierarchy and a PYP with strictly positive discount parameter in the other, that are: $HDPYP(\theta_0,\sigma_1,\theta_1,H_0)=HPYP (0, \theta_0,\sigma_1,\theta_1,H_0)$ and $HPYDP(\sigma_0,\theta_0,\theta_1,H_0)=HPYP (\sigma_0, \theta_0,0,\theta_1,H_0)$. 
For an example of $HDPYP$ see  \cite{Dubey14}.
\end{example}

\begin{example}[Hierarchical homogeneous normalized random measures]
Hierarchical homogeneous Normalized Random Measures (HNRMI) introduced in \cite{Cam18} are defined by 
\[
\begin{split}
 p_{i}  &  | p_0  \stackrel{iid}{\sim} NRMI(\theta_1,\eta_1,p_0), \quad i=1,\dots,I,  \\
  p_{0}  &  | H {\sim}  NRMI(\theta_0,\eta_0,H_0),  \\
 \end{split}
\]
where $NRMI(\theta,\eta,H)$ denotes a normalized homogeneous random measure
with parameters $(\theta,\eta,H)$. For the definition of a $NRMI(\theta,\eta, H)$ see 
Appendix   \ref{App:NRM}. Since $p_0$ is almost surely a discrete measure, it is essential  in the definition 
of NRMI  to allow for a non diffuse measure $H$. Since as observed in Example \ref{Ex:NRM}
a NRMI is a SSrp, then  a HNRMI is a HSSrp and  observations sampled from a HNRMI are a HSSM.
\end{example}

Our class of HSSM includes also  new hierarchical mixture of finite mixture processes as detailed in the following examples.

\begin{example}[Hierarchical mixture of finite mixture processes]\label{EX:HMFM}  
Let $MFMP(\sigma,\rho,H)$ be the mixture of finite mixture process defined in Example \ref{Ex:GPFM}, then one can define the following hierarchical structure
\begin{equation}\label{hgp0}
\begin{split}
 p_{i}  &  | p_0  \stackrel{iid}{\sim} MFMP (\sigma_1,\rho^{(1)},p_0), \quad i=1,\dots,I,  \\
  p_{0}  &  | H_0 {\sim} MFMP (\sigma_0,\rho^{(0)},H_0),  \\
 \end{split}
\end{equation}
which is a Hierarchical MFMP with parameters $\sigma_i$, $\rho^{(i)}=(\rho^{(i)}_k)_{k \geq 1}$, $i=0,1$ and base measure $H_0$, and is denoted with $HMFMP (\gamma_0,\rho^{(0)},\rho^{(1)},H_0)$. This process extends, to the hierarchical case, the finite mixture model of \cite{MillerHarrison}.  As a special case, when $|\sigma_i|=1$ and $\rho^{(i)}$ is given by \eqref{DistXbis} $i=0,1,\dots$, one obtains the Hierarchical Gnedin Process with parameters $(\gamma_0,\zeta_0,\gamma_1,\zeta_1,H_0)$, denoted with $HGP (\gamma_0,\zeta_0,\gamma_1,\zeta_1,H_0)$. 
\end{example}
Finally, new hierarchical processes can also be obtained by assuming a finite mixture process at some level of the hierarchy and a $PYP$ at the other level, 
as showed in the following.

\begin{example}[Mixed Cases]\label{EX:MixCases} The following hierarchical structure
\begin{equation}\label{gppy0}
\begin{split}
 p_{i}  &  | p_0  \stackrel{iid}{\sim} PYP (\sigma_1,\theta_1,p_0), \quad i=1,\dots,I,  \\
  p_{0}  &  | H_0 {\sim} GP (\gamma_0, \zeta_0,H_0)  \\
 \end{split}
\end{equation}
is a hierarchical Gnedin-Pitman-Yor process, denoted with $HGPYP(\gamma_0,\zeta_0,\sigma_1,\theta_1,H_0)$. The hierarchical Gnedin-Dirichlet process is then obtained as special case for $\sigma_1=0$ and denoted with $HGDP(\gamma_0,\zeta_0,\theta_1,H_0)$. Exchanging the role of $PYP$ and $GP$ in the above construction, 
one gets the $HPYGP(\sigma_0,\theta_0,\gamma_1,\zeta_1,H_0)$.
\end{example}

\subsection{HSSM Clustering Structure:  Chinese Restaurant  Franchising}\label{S:CRF}
In this section, based on Proposition \ref{prop0}, we prove that the clustering structure of a HSSM can be described with the same metaphor used to describe 
the HDP: the Chinese Restaurant  Franchise. In this metaphor,  observations are attributed to ``customers'', identified by the indices $(i,j)$,  and groups are described as ``restaurants'' ($i=1,\dots,I$).  In each ``restaurant'',  ``customers'' are clustered  according to ``tables'', which are then clustered in an higher hierarchy by means of ``dishes''. Observations  are ``attached''  at the second level of the clustering process, when dishes are associated to tables. 
One can think that the first customer sitting at each table chooses a dish from a common menu and this dish is
shared by all other customers who join the same table afterwards.

The  first  step of the clustering process, acting 
within each group, will be driven  by independent random partitions $\Pi^{(1)},\dots,\Pi^{(I)}$
with EPPF $\FP$. 
The second level, acting between groups,  will be driven by a random partition $\Pi^{(0)}$ with EPPF $\FP_0$. 

Given  $n_1, \dots, n_I$ integer numbers,  
we introduce the following set of observations:
\[
\CO:=\{ \xi_{i,j}: j=1, \dots, n_{i};  i=1, \dots I\}.
\]

\begin{proposition}\label{prop1} 
 If  $[\xi_{i,j}]_{i=1,\dots,I, j \geq 1}$ is a  $HSSM(\FP,\FP_0, H_0)$, 
then
the law of $\CO$ is the same as the law of  $[\phi_{ d^*_{i,j} }: j=1, \dots, n_{i};  i=1, \dots I],$ where
\[
d_{i,j}^*:=\SC_{\SD(i,c_{i,j}) } \left(\Pi^{(0)}\right),
\quad
\SD(i,c):=\sum_{i'=1}^{i-1} |\Pi^{(i')}_{n_{i'}}| + c,
\quad
 c_{i,j}:=\SC_j\left(\Pi^{(i)}\right),
\]
$(\phi_n)_n$ is a sequence of i.i.d. random variables with distribution $H_0$, 
   $\Pi^{(1)},\dots,\Pi^{(I)} $ are i.i.d. exchangeable partitions
with EPPF $\FP$ and $\Pi^{(0)}$ is an exchangeable partition
with EPPF $\FP_0$. All the previous random elements are independent. 
\end{proposition}

Since  $d_{i,j}^*=d_{i,c_{i,j}}$ for $d_{i,c}:=\SC_{\SD(i,c) } \left(\Pi^{(0)}\right)$, 
then the construction in Proposition \ref{prop1} can be summarized by the following  hierarchical  structure
\begin{equation}\label{h0}
\begin{split}
& \xi_{i,j}=\phi_{d_{i,c_{i,j}}}, \,\, d_{i,c}=\SC_{\SD(i,c) } \left(\Pi^{(0)}\right), \,\, c_{i,j}=\SC_j\left(\Pi^{(i)}\right),\,\, \phi_n \stackrel{i.i.d.}{\sim} H_0, \\
& \left(\Pi^{(0)},\Pi^{(1)},\dots,\Pi^{(I)}\right) {\sim} \FP_0 \otimes  \FP \otimes \cdots \otimes  \FP,
 \end{split}
\end{equation}
where, following the Chinese Restaurant Franchise metaphor, $c_{i,j}$ is the table at which the $j$-th ``customer'' of the ``restaurant'' $i$ sits,  $d_{i,c}$ is the
index of the ``dish'' served at table $c$ in restaurant $i$ and $d^*_{i,j}$  is the index of the ``dish''   served  to the $j$-th customer of the $i$-th restaurant. 

Proposition \ref{prop1} can also be used to describe in a "generative" way the array $\CO$. Having in mind the Chinese Restaurant Franchising, we shall denote by $n_{icd}$  the number of customers in restaurant $i$ seated at table $c$ and being served dish $d$ and $m_{id}$ the number of tables in restaurant $i$ serving dish $d$. We denote with dots the marginal counts. Thus, $n_{i\cdot d}$ is the number of customers in restaurant $i$ being served dish $d$, $m_{i \cdot}$ is the number of tables in restaurant $i$, $n_{i\cdot\cdot}$ is the number of customers in restaurant $i$ (i.e. the $n_i$ observations), and $m_{\cdot \cdot}$ is the number of tables.
\par Finally, let $\omega_{n,k}$ and $\nu_n$ be the weights of the predictive distribution of the random partitions $\Pi^{(i)}$ ($i=1,\dots,I$)  with EPPF $\FP$ (see Section \ref{Sec:1-1}). Also, let $\tilde \omega_{n,k}$ and $\tilde \nu_n$ be the weights of the predictive distribution of the random partitions $\Pi^{(0)}$ with EPPF $\FP_0$ defined in an analogous way by using 
$\FP_0$ in place of $\EPk$.
\par We can  sample $\{ \xi_{i,j}; i=1, \dots I, j=1, \dots, n_{i} \}$  starting with $i=1$, $m_{1 \cdot}=1$, $n_{11\cdot}=1$, $D=1$, $m_{\cdot 1}=1$ and $\xi_{1,1}=\xi^*_{1,1}=\phi_1 \sim H_0$ and then iterating for $i=1,\ldots,I$ the following steps:
\begin{itemize}
\item[(S1)]  for $t=2,\dots,n_{i \cdot \cdot}$,  sample $\xi_{i,t}$ given $\xi_{i,1},\dots, \xi_{i,t-1}$ and $k:=m_{i\cdot}$ from $G_{it}^{*}+\nu_{t}(n_{i 1\cdot}, \dots,n_{ik \cdot} ) G_{it}( \cdot )$ where $G_{it}( \cdot )=\tilde{G}_{it}+\tilde \nu_{m_{\cdot\cdot}}(m_{\cdot 1}, \dots,m_{ \cdot D} )  H_0( \cdot )$ and
\begin{align*}
G_{it}^{*}=\sum_{c=1}^{k} \omega_{t,c} (n_{i 1\cdot}, \dots,n_{ik \cdot }) \delta_{\xi^*_{i,c}}( \cdot ), \, \tilde{G}_{it}=\sum_{d=1}^D   \tilde \omega_{m_{\cdot \cdot} c} (m_{\cdot 1}, \dots,m_{\cdot D}) \delta_{\phi_d}( \cdot )
\end{align*}
\item[(S2)]
If $\xi_{i,t}$ is sampled from $G_{it}^{*}$, then we set $\xi_{i,t}=\xi^*_{i,c}$ and let $c_{it}=c$ for the chosen $c$, we leave $m_{i\cdot}$ the same and set $n_{ic\cdot}=n_{ic\cdot}+1$, while, if $\xi_{i,t}$ is sampled from $G_{it}$, then we set $m_{i\cdot}=m_{i\cdot}+1$, $\xi_{i,t}=\xi^*_{i,m_{i\cdot}}$ and $c_{it}=m_{i\cdot}$. If $\xi_{i,t}$ is sampled from $\tilde{G}_{it}$, we set $\xi^*_{i,c}=\phi_d$, let $d_{ic}=d$ for chosen $d$ and increment $m_{\cdot c}$ by one. If $\xi_{i,t}$ is sampled from $H_0$, then we increment $D$ by one and set $\phi_{D}=\xi_{it}$, $\xi^*_{i,c}=\xi_{i,t}$ and $d_{ic}=D$. In both cases, we increment $m_{\cdot\cdot}$ by one. 
\item[(S3)] Having sampled $\xi_{i,t}$ with $t=n_{i \cdot \cdot}$ in the previous Step,  set $i=i+1$,  $m_{i \cdot}=1$, $n_{i 1\cdot}=1$ take 
$\xi_{i,1}=\xi_{i,1}^*$ where 
 $\xi_{i,1}^*$ is sampled from $G_{it}$.  Update D,$m_{\cdot c}$, $d_{ic}$ and $m_{\cdot\cdot}$ as described in the previous step. 
\end{itemize}

\subsection{Random partitions induced by a HSSM}

\newcommand{\bn}{\boldsymbol{n}}
\newcommand{\bmm}{\boldsymbol{m}}
\newcommand{\bl}{\boldsymbol{ \ell}}
\newcommand{\bq}{\boldsymbol{ q}}
\newcommand{\blam}{\boldsymbol{ \lambda}}

The Chinese Restaurant Franchising representation in Section \ref{S:CRF} provides 
a description of the HSSM clustering structure 
which is satisfactory for both theoretical and computational aspects. Nevertheless, 
the theoretical description can be completed  by deriving explicit expression for   
the law of the random partition induced by 
an HSSM. 

A partially exchangeable random array induces a random partition in this way: two couples of indices $(i,j)$ and $(i',j')$
are in the same block of the partition if and only if   $\xi_{i,j}$  and  $\xi_{i','j'}$ take on the same value.  A possible way of coding this random partition is the following. 
Let $D$ be the (random) number of distinct values, say $\phi_1,\dots,\phi_D$,  in the set of observations $\CO$. 
While in the case of a single sequence of observations there is a natural order  to enumerate distinct observations (the order of appearance), 
here one needs to choose a suitable order to list the different values $\phi_i$s.  In what follows, we 
choose the lexicographical order, which depends on the fact that we fix the numbers of observations $n_1,\dots,n_I$ for each group. 

 Given $\phi_1,\dots,\phi_D$, for $i=1,\dots,I$ and $d=1,\dots,D$, 
 let $\Pi^*_{i,d}$  the set of indices of the observations of the $i$-th group that belongs to the $d$-th cluster, that is
\[
\Pi^*_{i,d}=\left \{   j \in \{1,\dots,n_{i..}\}:   \xi_{i,j}=\phi_d  \right \}=\{    j \ \in \{1,\dots,n_{i}\}: d_{i,c_{ij}}=d  \} .
\]
Note that
for some $i$ it is possible that $\Pi^*_{i,d}=\emptyset$ but clearly, identifying  any element $j$ in $\Pi^*_{i,d}$ with 
the couple $(i,j)$, one has
\[
\cup_i \Pi^*_{i,d}  = \{ (i,j): \xi_{i,j}=\phi_d\}  \not = \emptyset.
\]
We shall say that  
\[
\Pi^*=\Pi^*(\CO):=[\Pi^*_{i,d}: i=1,\dots,I, d=1,\dots,D]
\]
is the random partition induced by the set of observations  $\CO$. Of course others coding are possible. 

In order to express  the law of $\Pi^*$ we need some more notation. 
Given integer numbers
 $n_{i  d}$ for $ i=1,\dots,I$ and  $d=1,\dots,D$,
 set   $\bn_i:=(n_{i  1},\dots, n_{i  D})$,
  $\bn=[\bn_1,\dots,\bn_I]$, 
and define 
  \[
 \CM[{\bn_i}]:=\Big \{  \bmm_i=[m_{i 1},\dots,m_{i D}] :  1 \leq m_{i d} \leq n_{i d} \,\, \text{if $n_{i d}>0$} ,   m_{i d} =0 \text{ if $n_{i d}=0$} \Big \}.
\]
For $\bmm_i$ in  $\CM[{\bn_i}]$ define 
\[
 \Lambda(\bmm_i)
 :=\left \{ 
\begin{array}{lll}
&  
\text{$ \blam_i=[\blam_{i1},\dots,\blam_{iD}]$ where $\blam_{id}=(\lambda_{id1},\dots,\lambda_{idn_{id}}) \in \N^{n_{id}} $}:  \\
%
& \quad  \sum_{j=1}^{n_{id}} j \lambda_{idj}=n_{id},  \sum_{j=1}^{n_{id}}  \lambda_{idj}=  m_{id} \,\, \text{ for $d=1,\dots,D$ } \\
\end{array}
\right\}.
\]
Set also 
\[
\CM[{\bn}]:= \CM[{\bn_1}] \times  \CM[{\bn_2}] \cdots \times  \CM[{\bn_I}] \quad \text{and} \quad   
\Lambda(\bmm) : = \Lambda(\bmm_1) \times \cdots \times  \Lambda(\bmm_I).
\] 
For  $\bmm$ in  $\CM[\bn]$ and 
$\blam=[\blam_1,\dots,\blam_I]$   in $\Lambda(\bmm)  $, following the 
convention of the previous section, set 
 $m_{i \cdot}=\sum_{d=1}^D m_{ i d}$,  $m_{\cdot d}=\sum_{i=1}^I m_{ i d}$ and  $m_{\cdot \cdot}=\sum_{d=1}^D \sum_{i=1}^I m_{ i d}$
and define
\[
 \FP[[\blam_i]] :=  \FP (\ell_{i11},\dots,\ell_{im_{i1}1},\dots,\ell_{i,m_{i,D},D} ),  
 \]
where $\ell_{i11},\dots,\ell_{im_{i1}1},\dots,\ell_{i,m_{i,D},D} $ are any integer numbers such that 
 $ \sum_{c=1}^{m_{id}} \ell_{i c d} = n_{i  d}$ for every $d$ and 
$\#\{ c: \ell_{icd} =j \}= \lambda_{idj} $ for every $d$ and $j$.
Note that in the previous definition  if $n_{i d}=0$, and hence  $m_{id}=0$,  
then $\ell_{i c d}=0$ and the expression $\FP_{m_{i \cdot}} (\ell_{i11},\dots,\ell_{im_{i1}1},\dots,\ell_{i,m_{i,D},D} )$ 
must be read as the EPPF obtained erasing all the zeros. For instance, if $n_{1,d}>0$ for all $d=1,\dots,D-1$ and
 $n_{1,D}=0$, then  $\FP (\ell_{111},\dots,\ell_{1 m_{11}1},\dots,\ell_{1 ,m_{1D},D} )$
is simply $\FP (\ell_{i11},\dots,\ell_{im_{i1}1},\dots,\ell_{1,m_{1,D-1},D-1} )$. Note also 
that  $\FP [[\blam_i]]$ is well defined since the value of $\FP(\ell_{i11},\dots,\ell_{im_{i1}1},\dots,\ell_{i,m_{i,D},D} ) $
depends only on the statistics $\blam_i$. See e.g. \cite{Pit06}.

For expository purposes and differently from all the other results in this paper, the law of the random partition is given under the assumption $H_0$ is non-atomic. 

\begin{proposition}\label{prop_part}
Under the same assumptions of Proposition \ref{prop1}, assume further that $H_0$ is non-atomic. 
Let $\pi^*$ be a possible  realisation of the random partition  $\Pi^*$ induced by $\CO$. If $\pi^*$ 
has 
$D$ distinct block and the cardinality of $\pi^*_{i,d}$ is $n_{i d}$ for $i=1,\dots,I$ and $d=1,\dots,D$, then 
 \[
\P\{ \Pi^*=\pi^*\}= \sum_{ \bmm \in \CM[{\bn}] } 
\FP_0 (m_{\cdot 1},\dots,m_{\cdot D}) 
\sum_{  \blam  \in \Lambda(\bmm) }
 \prod_{i=1}^I 
  \FP [[ \blam_i ]] \prod_{d=1}^D  \frac{n_{i d}!}{\prod_{j=1}^{n_{i d}} \lambda_{idj}! (j!)^{\lambda_{idj}}    }.
\]
\end{proposition}

\begin{remark}
For a suitable choice of $\FP$ and $\FP_0$, one obtains as special cases of Proposition \ref{prop_part} the results in Theorem 3 and 4 of \cite{Cam18}. The proof of Proposition \ref{prop_part} relies on combinatorical arguments combined with the results in Proposition \ref{prop0} and \ref{prop1}, whereas \cite{Cam18} use mainly specific properties of normalized random measures.
\end{remark}

\section{Cluster sizes distributions}
\label{Sec_asympt}
We study the distribution of the number of clusters in each group of observations 
(i.e.,  the number of distinct dishes served in the restaurant $i$), as well as the global number of clusters 
(i.e. the total number of distinct dishes in the restaurant franchise). We introduce a sequence of observation sets, $\CO_t$, $t=1,2,\dots$, each containing $n_i(t)$ elements in the $i$-th group and $n(t)=\sum_{i=1}^I n_i(t)$ total elements. We show the exact finite sample distribution of the number of clusters for given $n(t)$ and $n_{i}(t)$ when $t<\infty$. Some properties, such as the prior mean and variance, are discussed in order to provide some guidelines to setting HSSM parameters in applications. Moreover, we give some new asymptotic results when the number of observations goes to infinity, such that $n(t)$ diverges to $+\infty$ as $t$ goes to $+\infty$. The results extend existing asymptotic approximations for species sampling (\cite{Pit06}) and for hierarchical normalized random measures (\cite{Cam18}) to our HSSM. Finally, we provide a numerical study of the approximation accuracy.

\subsection{Distribution of the cluster size under the prior}
For every $i=1,\dots,I$,   we define
\[
K_{i,t}:=| \Pi^{(i)}_{n_i(t)}|,\quad K_t:=\sum_{i=1}^I | \Pi^{(i)}_{n_i(t)}|,\quad D_{i,t}=|\Pi^{(0)}_{K_{i,t}}|,\quad D_t=|\Pi^{(0)}_{K_t}|.
\]
By Proposition \ref{prop1},  for every fixed $t$, the laws of $K_{i,t}$ and $K_t$ are the same as the ones of the number of "active tables" in "restaurant" $i$ and 
of the total number of "active tables" in the whole franchise, respectively. Analogously, the laws of $D_t$ and $D_{i,t}$ are the same as the laws of the number of dishes served in the restaurant $i$ and in the whole franchise, respectively. If $H_0$ is diffuse, then $D_t$ and the number of distinct clusters in $\CO_t$ have the same law and also $D_{i,t}$ and the number of clusters in the group $i$ follow the same law.

The distribution of both  $D_t$ and $D_{i,t}$ can be derived as follows.

\begin{proposition}\label{Prop:clusterdistribution} For every $n\geq 1$ and $k=1,\dots,n$,
we define
$q_n(k):=\P\left\{ |\Pi^{(i)}_{n}|=k\right\}$  and $q^{(0)}_n(k):=\P\left\{ |\Pi^{(0)}_{n}|=k\right\}$.
Then, for every $i=1,\dots,I$, 
\[
\P\{ D_{i,t}=k\}= \sum_{m =k}^{n_i(t)}  q_{n_i(t)} (m) q^{(0)}_m(k), \quad \text{$k=1,\dots,n_i(t)$}, 
\]
and
\[
\P\{ D_{t}=k\} 
=   \sum_{m =\max(I,k)}^{n(t)} \left(  \sum_{ \substack{
m_1+\dots+m_I=m,\\ 1\leq m_i \leq n_i(t)}} \prod_{i=1}^I q_{n_i(t)} (m_i)   \right)  q^{(0)}_m(k),
  \quad \text{$k=1,\dots,n(t)$}.
\]
Moreover,  for every $r>0$ 
\[
\E\left[ D_{i,t}^r\right]=   
 \sum_{m =1}^{n_i(t)}     \E \left[|\Pi^{(0)}_{m}|^r \right]  q_{n_i(t)} (m)
 \]
and
\[
\quad \E\left[ D_{t}^r\right]=    \sum_{ \substack{ m_1,\dots,m_I: \\ 
1\leq m_i \leq n_i(t)}}   \E \left[  \Big  |\Pi^{(0)}_{\sum_{i=1}^I m_i} \Big  |^r \right] \prod_{i=1}^I q_{n_i(t)} (m_i) .
\]
\end{proposition}

In particular, we note that, for every $i=1,\dots,I$, $n\geq 1$ and $k=1,\dots,n$,
\[
q_n(k)
=\sum_{(\lambda_1,\dots,\lambda_n) \in \Lambda_{n,k} } 
\frac{n!}{\prod_{j=1}^n ( \lambda_j ! )(j!)^{\lambda_j} }
  \EPk \left[\left[\lambda_1,\dots,\lambda_n\right]\right],
\]
where $ \Lambda_{n,k}$ is the set of integers $(\lambda_1,\dots,\lambda_n)$ such that $\sum_j \lambda_j=k$ and $\sum_j j \lambda_j=n$, $\EPk [[\lambda_1,\dots,\lambda_n]]$
is the common value of the symmetric function $\EPk$ for all $n_1,\dots,n_k$ with
$|\{ i: n_i =j \}|= \lambda_j $ for $j=1,\dots,n$ and 
\[
\frac{n!}{\prod_{j=1}^n ( \lambda_j ! )(j!)^{\lambda_j} }
\]
is the number  of partitions of $[n]$ with 
$\lambda_j$ blocks of cardinality $j=1,\dots,n$
(see Eq. (11) in \cite{Pit95}).
Similar expressions can be obtained  for  $q^{(0)}_n(k)$. 

The results in Proposition \ref{Prop:clusterdistribution} generalize to HSSM those given in Theorem 5 of \cite{Cam18} for HNRMI. Our proof relies on the hierarchical SSS construction (see Proposition \ref{prop1}), whereas the proof in \cite{Cam18} builds on the partial exchangeable partition function given in Proposition \ref{prop_part}.

Also, the generalized SSS construction allows us to derive the distribution of the number of clusters when $H_0$ is not diffuse. Indeed, it can be deduced by considering  possible coalescences  of  latent clusters (due to  ties in the i.i.d. sequence $(\phi_n)_n$ of Proposition \ref{prop1}) forming a true cluster. Let us denote with
$\tilde D_t$ and $\tilde D_{i,t}$  the number of distinct clusters in $\CO_t$ and in the group $i$, respectively, at ``time'' $t$.

\begin{proposition}\label{Prop:clusterdistribution00}
Let  $H^*_0(d|k)$ (for $1\leq d\leq k$) be the probability of observing exactly $d$ distinct values 
in the vector $(\phi_1,\dots,\phi_k)$ where the  $\phi_n$s are i.i.d. $H_0$. 
Then, 
\[
\P\{ \tilde D_{i,t}=d\}= \sum_{k =d}^{n_i(t)} H^*_0(d|k) \P\{ D_{i,t}=k\}    
\]
for $d=1,\dots,n_i(t)$. The probability of $\tilde D_t$ has the same 
expression as above with $D_t$ in place of $D_{i,t}$ and $n_t$ in place of $n_{i,t}$. 
If $H_0$ is diffuse, then
$\P\{ \tilde D_{i,t}=d\}=\P\{ D_{i,t}=d\}$ and $\P\{ \tilde D_{t}=d\}=\P\{ D_{t}=d\}$,
 for every $d \geq 1$.
\end{proposition}

The assumption of atomic base measures behind HDP and HPYP has been used in many studies, and some of its theoretical and computational implications have been investigated (e.g., see \cite{nguyen2016}) and \cite{sohn2009}), whereas mixed base measures appeared only recently in Bayesian nonparametrics and their implications are not yet well studied, especially in hierarchical constructions. In the following we state some new results for the case of spike-and-slab base measures.

\begin{proposition}\label{HierclustrH0spikeandslub} Assume that $H_0(dx)=a \delta_{x_0}(dx)+ (1-a) \tilde H(dx)$, 
where $a \in (0,1)$, $x_0$ is a point of $\XX$ and $\tilde H$ is a diffuse measure on $\XX$, then 
\[
\P\{ \tilde D_{i,t}=d\}=(1-a)^d \P\{ D_{i,t}=d\}+  \sum_{k =d}^{n_i(t)} {k \choose d-1} a^{k+1-d}(1-a)^{d-1}  
\P\{ D_{i,t}=k\}, 
\]
for $d=1,\dots,n_i(t)$. The probability of $\tilde D_t$ has the same 
expression as above with $D_t$ in place of $D_{i,t}$ and $n_t$ in place of $n_{i,t}$.  
Moreover, 
\[
\E[\tilde D_{i,t}]=1-\E[(1-a)^{D_{i,t}}]+(1-a)\E[{D_{i,t}}] \leq \E[D_{i,t}]
\]
and  $\E[\tilde D_t]$ has an analogous expression   with $D_{i,t}$ replaced 
by $D_{t}$. 
\end{proposition}

For a Gibbs type EPPF with $\sigma>0$, using results in 
\cite{GP2006}, we get
\[
q_n(k)= 
V_{n,k}  S_{\sigma}(n,k) 
\]
where $V_{n,k}$ satisfies the partial difference equation in \eqref{Rec_V} and $S_{\sigma}(n,k)$ is a generalized Stirling number of the first kind, defined as
\[
S_{\sigma}(n,k)=\frac{1}{\sigma^k k!}  \sum_{i=1}^k (-1)^i { k \choose i} (-i\sigma)_n,
\]
for $\sigma\not= 0$ and $S_{0}(n,k)=|s(n,k)|$ for $\sigma=0$, where $|s(n,k)|$ is the unsigned Stirling number of the first kind, see \cite{Pit06}.  See \cite{Deb2017} for an up-to-date review of Gibbs-type prior processes.

For the hierarchical PY process $\Pi^{(i)} \sim PY(\sigma,\th)$, the distribution $q_{n}(k)$ has closed-form expression
 \[
  q_n(k) =\frac{\prod_{i=1}^{k-1} (\theta+i\sigma)}{(\theta+1)_{n-1}}   \frac{1}{\sigma^k k!}  \sum_{i=1}^k (-1)^i { k \choose i} (-i\sigma)_n, 
 \]
 when $0 < \sigma<1$ and $\th>-\sigma$, whilst
\[
  q_n(k) =\frac{\theta^k \Gamma(\theta)}{\Gamma(\theta+n)} |s(n,k)|,
\]
when $\sigma=0$.


For the Gnedin model (\cite{Gnedin10}), $\Pi^{(i)} \sim GN(\gamma,\zeta)$, 
of Example \ref{ex:gnedinpart} it is possible to derive explicit expression for $q_n$
\begin{equation}
q_n(k)  =\binom{n-1}{k-1} \frac{n!}{k!}  \nu_{n,k},\,\, \hbox{with}\,\, \nu_{n,k}=\frac{(\gamma)_{n-k} \prod_{i=1}^{k-1} (i^2-\gamma i+\zeta)}{\prod_{m=1}^{n-1} (m^2+\gamma m +\zeta)}. \label{PrKn}
\end{equation}

Fig. \ref{all} in Appendix \ref{App:NumRes} shows the prior distribution 
of the number of clusters for each restaurant (group) $\P\{ D_{i,t}=k\}$  (left panel) and for the restaurant franchise $\P\{ D_{t}=k\}$ (right panel), when $I=2$ and $n_1(t)=n_2(t)=50$, for the processes $HDP(\theta_0,\theta_1,H_0)$ (black solid), $HPYP(\sigma_0,\theta_0,\sigma_1,\theta_1,H_0)$ (blue) and $HGP(\gamma_0,\zeta_0,\gamma_1,\zeta_1,H_0)$ (red). 

The values of the parameters are chosen in such a way that $E[D_{i,t}]$ is approximately equal to 25 when the number of customers is $n_i=50$ in each restaurant $i=1,2$. The HSSM parameter settings and the resulting prior means and variances of the number of clusters are given in Tab. \ref{Tab:priorEV}. Both $HPYP(H_0,\sigma_0,\theta_0;\sigma_1,\theta_1)$ and $HGP(H_0,\gamma_0,\zeta_0;\gamma_1,\zeta_1)$ offer more flexibility than the hierarchical Dirichlet process, in setting the prior number of clusters, since their parameters can be chosen to allow for different variance levels while keeping fix the expected number of clusters. An example is given by the blue and red dotted lines in Fig. \ref{all} for the two settings reported in the last two rows of Tab. \ref{Tab:priorEV} in Appendix \ref{App:NumRes}.

Further results on the sensitivity of the prior distribution to changes in the HSSM parameters are reported in Appendix \ref{App:NumRes}. In Fig. \ref{EffectsSymHierarchy}-\ref{EffectsSymHierarchy_scale}  in Appendix \ref{App:NumRes}, the analysis is done when the parameters of the top-hierarchy and of the bottom-hierarchy random measures are all equals (homogeneous case). For the easy of comparison, in each plot the blue lines represent the baseline cases corresponding to the black lines of Fig. \ref{all}. Changes in the concentration parameters, $\theta_i$ and $\zeta_i$, have effects on the mean and dispersion of the distribution as in the non-hierarchical case, see Fig. \ref{EffectsSymHierarchy}. Changes in the discount parameters $\sigma_i$ and $\gamma_i$ affect mainly the dispersion. In Fig. \ref{EffectsTopHierarchy}, we study the sensitivity in the non-homogeneous case, assuming the measure at the top of the hierarchy has different parameter values with respect to the measures at the bottom of the hierarchy which are assumed identical.

The results of the sensitivity analysis discussed in this section show that the expected global number of clusters is smaller than the sum of the expected marginal number of clusters, indicating that all the hierarchical prior measures considered allow for various degrees of information pooling across the different groups (restaurants). This information sharing effect is illustrated in Fig. \ref{Exp_N_Clusters} where the expected value (top panels) and variance (bottom panels) of the number of clusters are exhibited increasing $n(t)$ from $2$ to $1000$, assuming $I=2$ and $n_1(t)=n_2(t)=t$, with $t=1,\dots,500$. The expected global number of clusters, $E[D_t]$ (blue solid lines), is larger than the sum of the expected marginal number of clusters, $\sum_{i=1,2} E[D_{i,t}]$, for HSSM (red dashed lines) and independent SSM (black dash-dot lines). Qualitatively, the variance of the global number of clusters $V(D_t)$ (blue solid lines) is closed to the sum of the marginal variances $\sum_{i=1,2} V[D_{i,t}]$ for HSSM and much smaller than in the independent SSM (black dash-dot lines).
\begin{figure}[t]
  \caption[]{Expected value (top) and variance (bottom) of the number of clusters when $I=2$, $n=2,\ldots, 1000$ for: i) the HDP with $\theta_0 = \theta_1 = 43.3$ (first column); ii) the HPYP with $(\theta_0,\sigma_0) = (\theta_1,\sigma_1) = (29.9,0.25)$ (second column); iii) the HGP with $(\gamma_0,\zeta_0) = (\gamma_1,\zeta_1) = (15,1450)$ (third column). In each plot of the first (second) row, the blue solid lines refer to $E[D_t]$ ($V[D_t]$), the red dashed lines to $\sum_{i=1,2} E[D_{i,t}]$ ($\sum_{i=1,2} V[D_{i,t}]$) for HSSM and the black dash-dot lines to $\sum_{i=1,2} E[D_{i,t}]$ ($\sum_{i=1,2} V[D_{i,t}]$) for SSM.}
  \label{Exp_N_Clusters}
\begin{center}
  \begin{tabular}{ccc}
   \includegraphics[width=3.5cm]{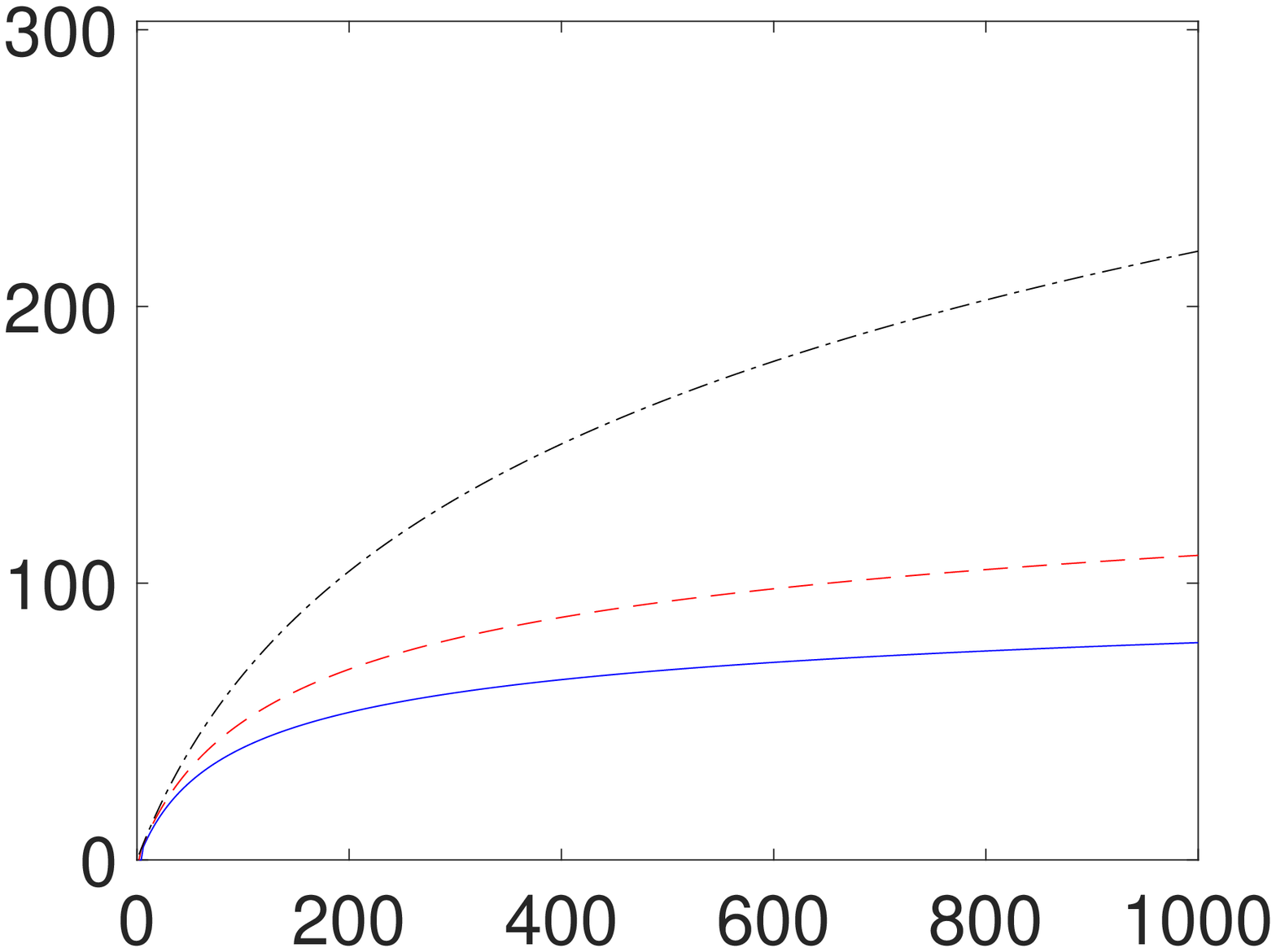} \hspace{-10pt} &
   \includegraphics[width=3.5cm]{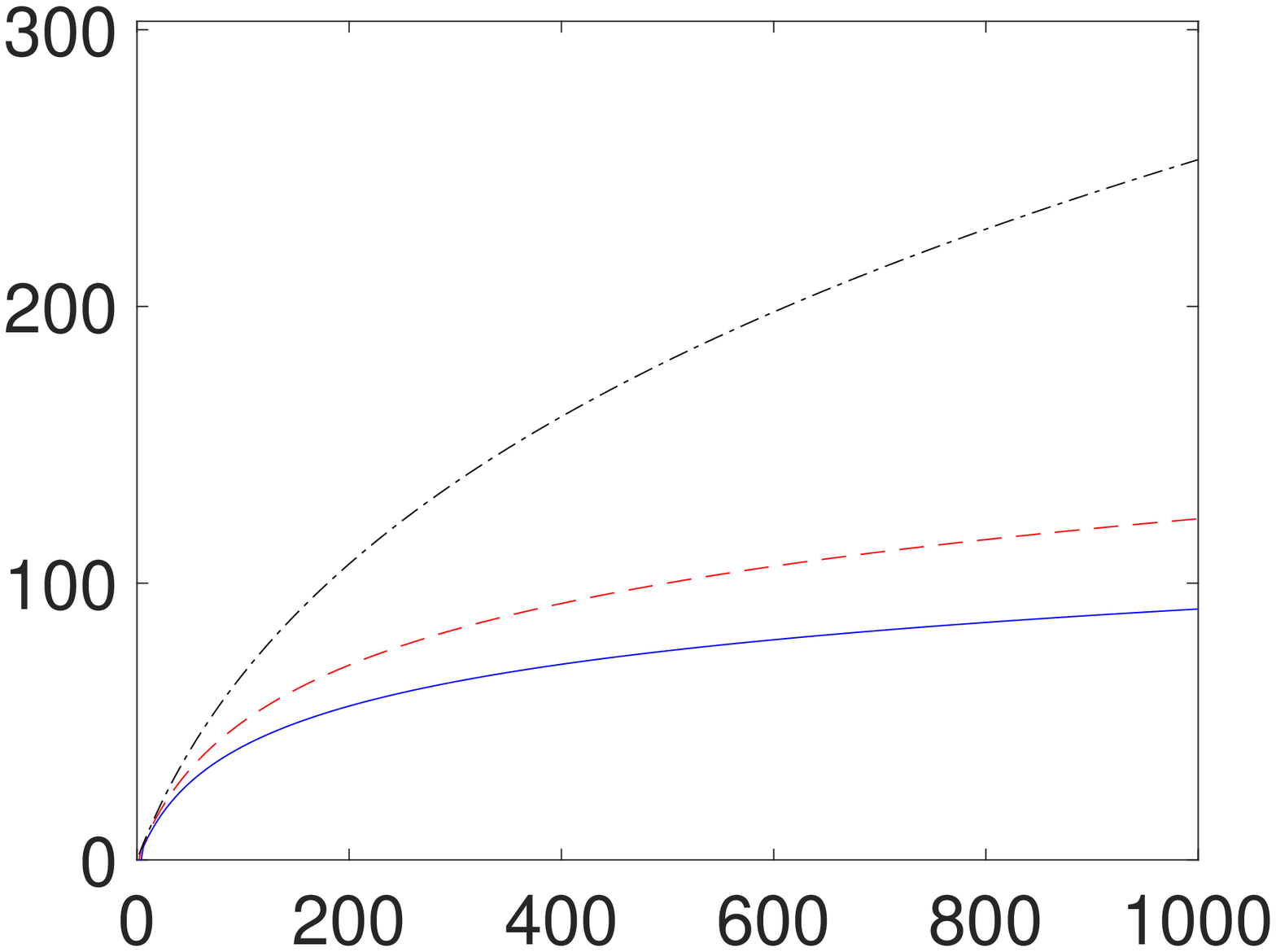} \hspace{-10pt} &
   \includegraphics[width=3.5cm]{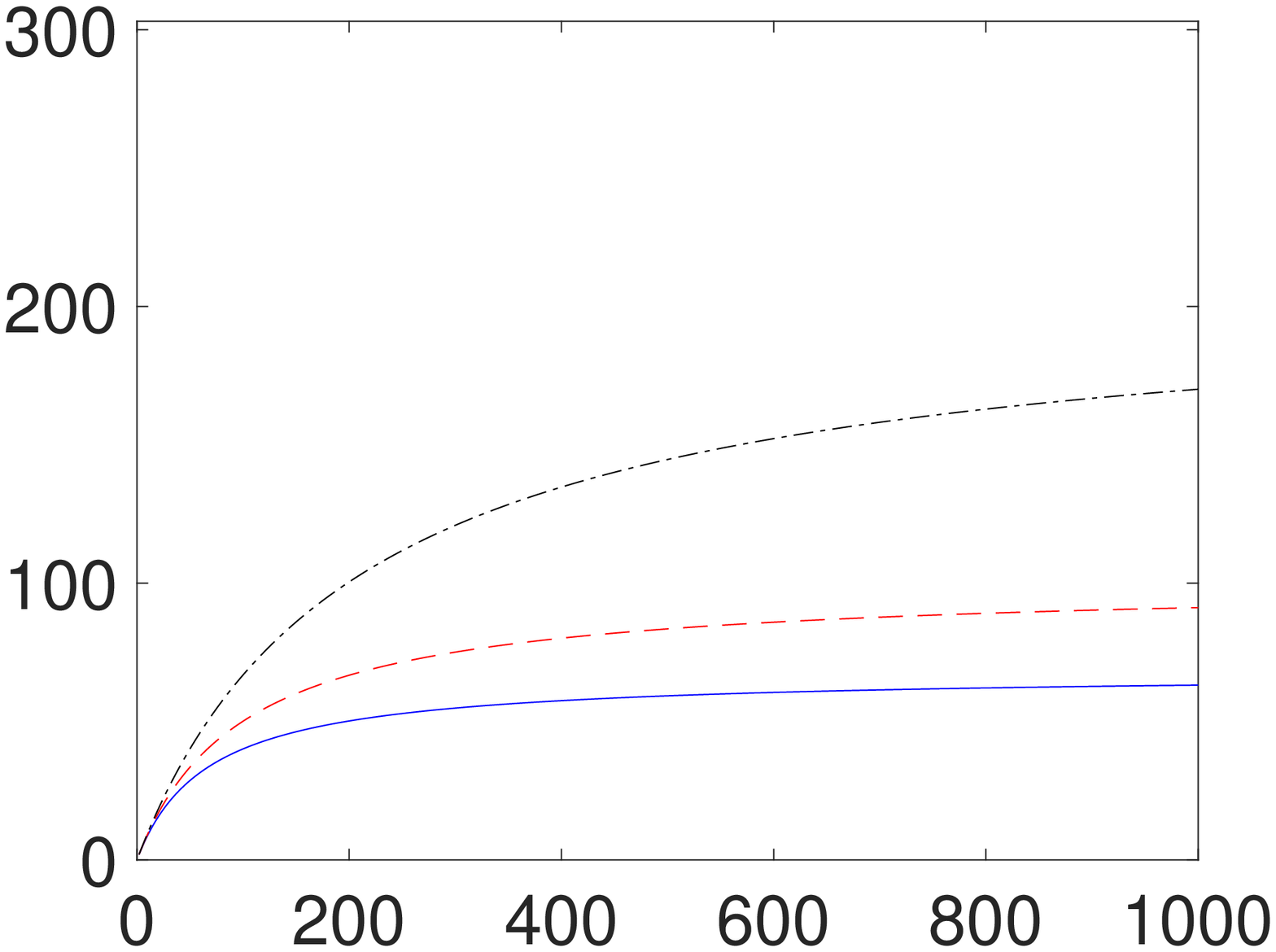} \hspace{-10pt} \\
   \includegraphics[width=3.5cm]{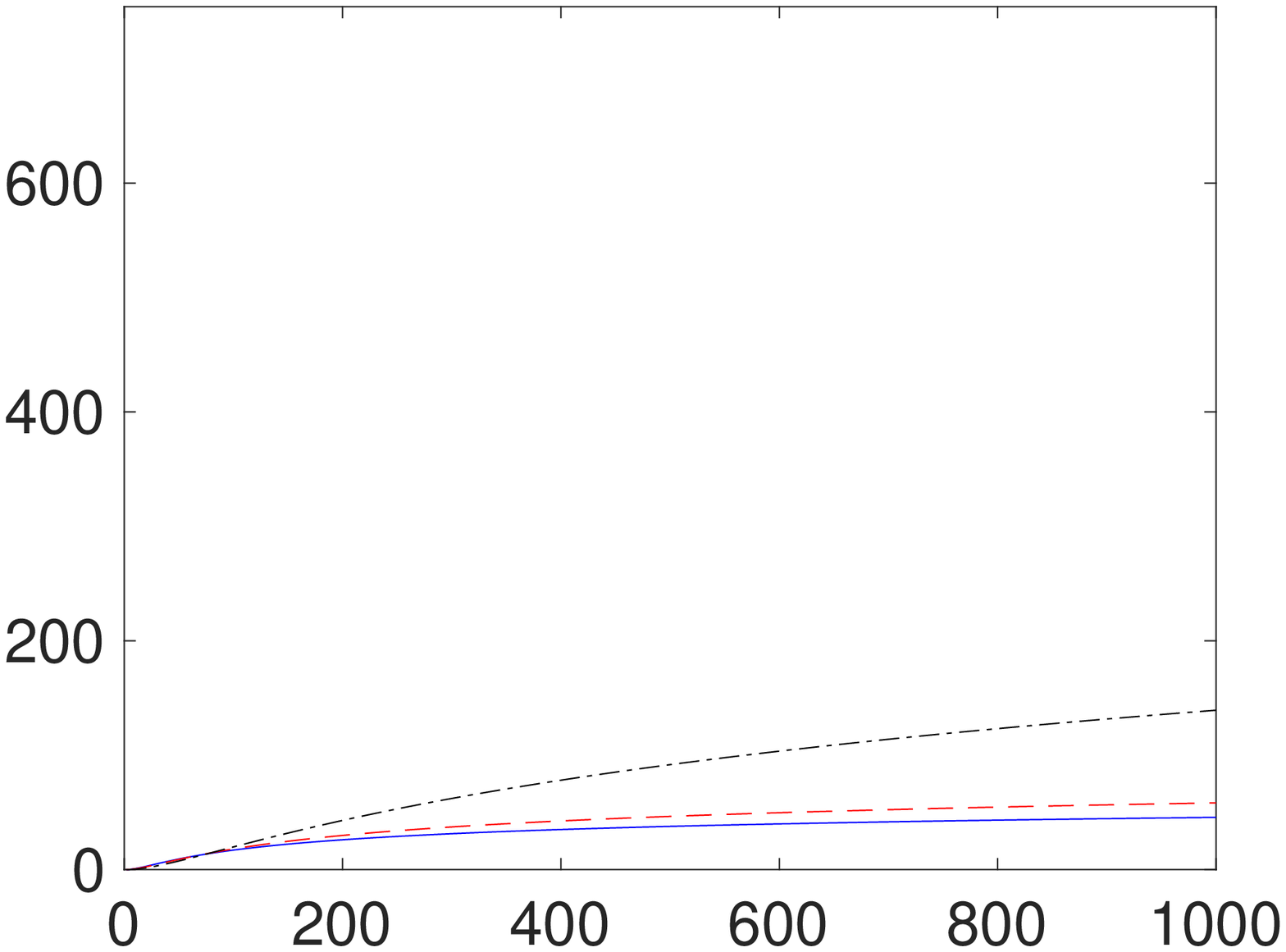} \hspace{-10pt} &
   \includegraphics[width=3.5cm]{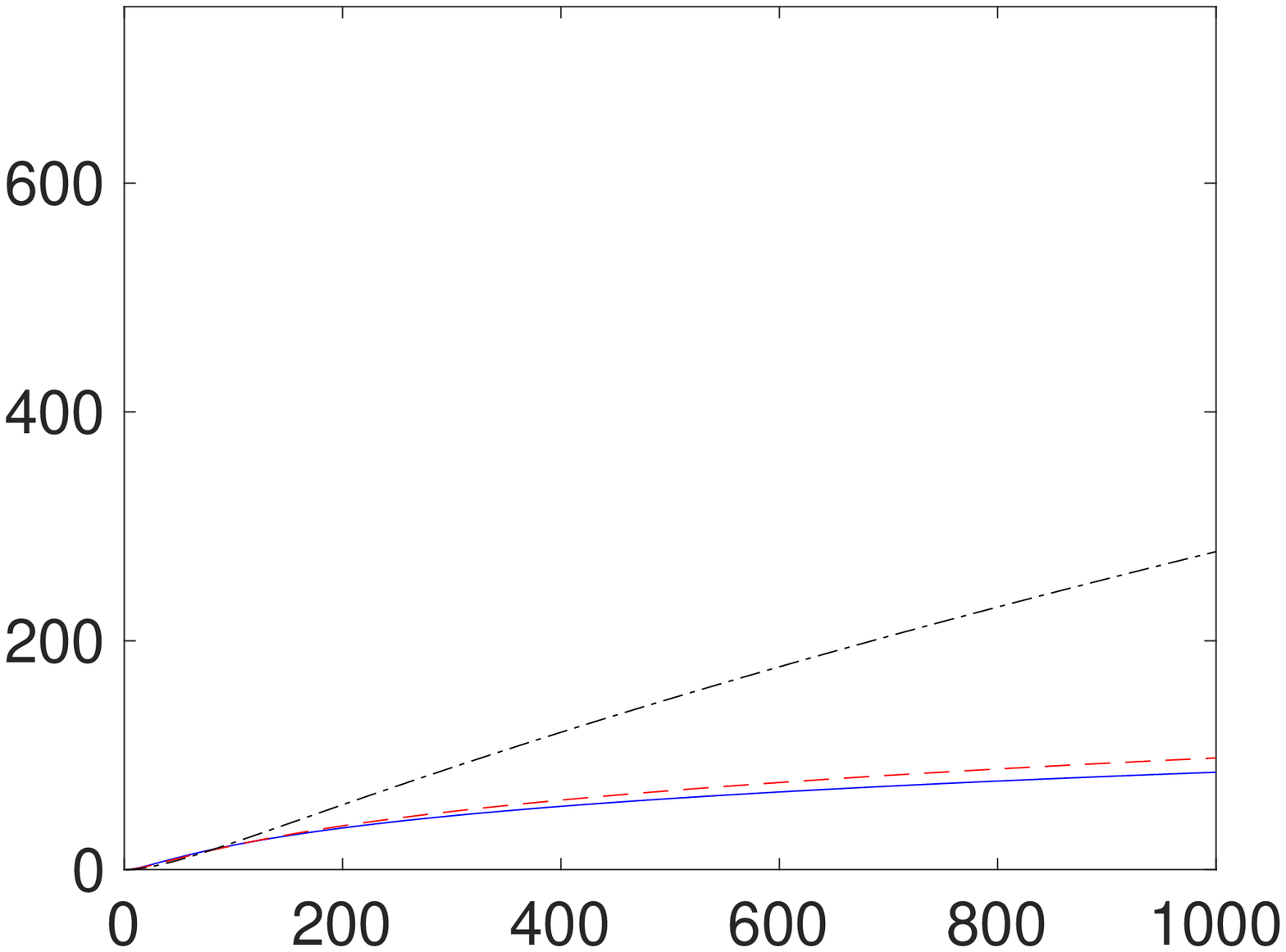} \hspace{-10pt} &
   \includegraphics[width=3.5cm]{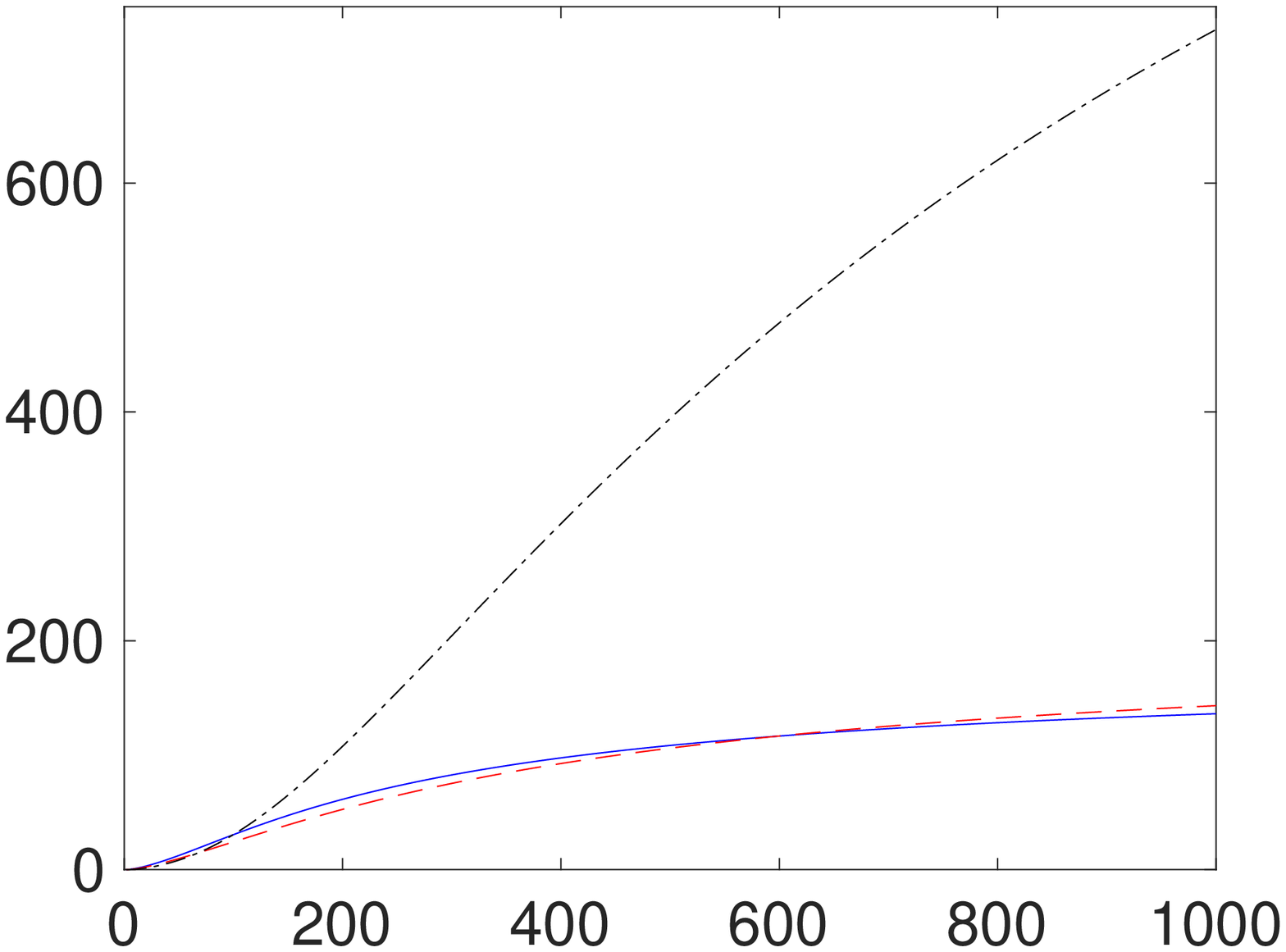} \hspace{-10pt} \\
  \end{tabular}
\end{center}
\end{figure}

\subsection{Asymptotic distribution of the cluster size} 
An exchangeable random partition $(\Pi_n)_{n \geq 1}$ has asymptotic 
diversity $S$ if  
\begin{equation}\label{sigmadiversity}
   |\Pi_n|/c_n \to S \quad a.s. 
\end{equation}
for  a positive random variable $S$ and a suitable normalizing  sequence $(c_n)_{n\geq 1}$. Asymptotic diversity generalizes the notion of $\sigma$-diversity, see 
Definition 3.10 in  \cite{Pit06}. An exchangeable random partition $(\Pi_n)_{n \geq 1}$ 
has $\sigma$-diversity $S$ if  \eqref{sigmadiversity} holds with $c_n=n^{\sigma}$. 
For any Gibbs-type partition 
  $(\Pi_n)_{n \geq 1}$, \eqref{sigmadiversity}
holds with 
 \[
c_n:=\left \{ 
 \begin{array}{ll}
1 & \text{if $\sigma<0$}, \\
\log(n) & \text{if $\sigma=0$}, \\
n^\sigma & \text{if $\sigma>0$}, \\
 \end{array}
 \right . 
 \]
 see Section 6.1 of \cite{Pitman2003}. Other characterizations of $\sigma$-diversity can be found in Lemma 3.1 in \cite{Pit06}. 


In the following propositions,  we use the (marginal) limiting behaviour  \eqref{sigmadiversity} of the random partitions $\Pi_n^{(i)}$ ($i=0,\dots,I$), to obtain the asymptotic distribution of $D_{i,t}$ and $D_t$ assuming  $c_n=n^\sigma L(n)$, with $L$ slowly-varying. 

The first general result deals with HSSM where $\Pi_{n}=\Pi_{n}^{(i)}$ satisfies  \eqref{sigmadiversity} for every $i=1,\dots,I$ and  $c_n \to +\infty$ and hence the cluster size $|\Pi_{n}^{(i)}|$ diverges to $+\infty$.

\begin{proposition}\label{PropAsym1} Assume that $\Pi^{(0)}$ and $\Pi^{(i)}$ {\rm (}for $i=1,\dots,I${\rm)} are independent exchangeable 
random partitions such that
$|\Pi^{(0)}_n|/a_n$  {\rm (}$|\Pi^{(i)}_n|/b_n$ for $i=1,\dots,I$, respectively{\rm)} converges almost surely 
to a strictly positive random variable $D^{(0)}_\infty$  {\rm (}$D^{(i)}_\infty$, respectively{\rm)} for 
suitable diverging sequences $a_n$ and $b_n$.
Moreover assume that 
 $a_n=n^{\sigma_0} L_0(n)$ and $b_n=n^{\sigma_1} L_1(n)$, with $\sigma_i \geq 0$ and $L_i$ slowly varying function, $i=0,1$, and set $d_n:=a_{b_{n}}=n^{\sigma_0\sigma_1} L_0(n^{\sigma_1}L_1(n))$. 
\begin{itemize}
\item[(i)]
If $\lim_{t \to +\infty} n_i(t)=+\infty$ for some $i$, 
then for $t \to +\infty$
\[
\frac{D_{i,t}}{d_{n_i(t)}} \to D^{(0)}_\infty  \left(D^{(i)}_\infty\right)^{\sigma_0} \quad \text{a.s.} 
\]  
\item[(ii)] If $\lim_{t \to +\infty} n_i(t) =+\infty$ and ${n_i(t)}/n(t) \to \ w_i>0$  for every $i=1,\dots,I$ then for $t \to +\infty$
\[
\frac{D_{t}}{d_{n(t)}} \to D^{(0)}_\infty  \left(\sum_{i=1}^I w_i^{\sigma_1} D^{(i)}_\infty \right)^{\sigma_0} \quad \text{a.s.} 
\]

\end{itemize}
\end{proposition}


\begin{remark}
Part (ii) extends to HSSM with different group sizes, $n_i(t)$, the results given in Theorem 7 of \cite{Cam18} for HNRMI with groups of equal size. Both part (i) and (ii) provide  deterministic scaling of  diversities, in the spirit of \cite{Pit06}, and differently from \cite{Cam18} where a random scaling is obtained. 
\end{remark}

The second general result describes  the asymptotic behaviour of $D_{i,t}$ and $D_t$ 
in presence of  random partitions for which $c_n=1$ for every $n$.

\begin{proposition}\label{PropAsym2} Assume that $\Pi^{(0)}$ and $\Pi^{(i)}$, $i=1,\dots,I$ are independent exchangeable 
random partitions and that $\lim_{t \to \infty} n_i(t)=+\infty$ for every $i=1,\dots,I$.
\begin{itemize}
\item[(i)] If $|\Pi^{(i)}_n|$ converges a.s. to a positive random variable  $K_i$ as $n \to +\infty$, then 
for every $k\geq 1$
\[
\lim_{t \to +\infty}\P\left\{ D_{i,t}=k\right\}= \sum_{m  \geq k}  \P\left\{ K_i= m\right\}  q^{(0)}_m(k),
\]
and 
\[
\lim_{t \to +\infty} \P\{ D_{t}=k\} =   \sum_{m \geq \max(I,k)}\sum_{ \substack{
m_1+\dots+m_I=m,\\ 1\leq m_i }} q^{(0)}_m(k)\prod_{i=1}^I\P\{ K_i= m_i\}.
\]
\item[(ii)]  If $|\Pi^{(i)}_n|/b_n$ converges a.s. 
to a strictly positive random variable $D^{(i)}_\infty$ for a suitable diverging sequences $b_n$
and $|\Pi^{(0)}_n|$ converges a.s. to a positive random variable  $K_0$ as $n \to +\infty$,
then, for every $k \geq 1$,
\[
\lim_{t \to +\infty}\P\{ D_{t}=k\}=\lim_{t \to +\infty}\P\{ D_{i,t}=k\}=  \P\{ K_0= k\} . 
\]
\end{itemize}
\end{proposition}

Starting from Propositions \ref{PropAsym1} and \ref{PropAsym2},  analytic  expressions for the asymptotic  distributions of $D_{i,t}$ and $D_{t}$
can be deduced for some special HSSMs.

As an example, consider the HGP and the mixed hierarchical models 
described in Examples \ref{EX:HMFM} and \ref{EX:MixCases}. 
If  $(\Pi_n)_n$ is the Gnedin's partition  of parameters $(\gamma,\zeta)$, then $| \Pi_{n}|$ converges almost surely to a random variable $K$ with distribution \eqref{DistXbis},  see \cite{Gnedin10}.  
Hence, the asymptotic behaviour of the number of clusters in a HGP 
and in HPYGP  can be derived from Proposition \ref{PropAsym2} as stated here below. 

\begin{proposition}\label{Prop:AsymGnedin}
In  a $HGP (\gamma_0,\zeta_0,\gamma_1,\zeta_1,H_0)$, one has
\[
\lim_{t \to +\infty}\P\left\{ D_{i,t}=k\right\}=\frac{c_{\gamma_1,\zeta_1}}{k!}
\left (
\prod_{i=1}^{k-1}(i^2-\gamma_0i+\zeta_0)  \right )
\sum_{m \geq k} \frac{(\gamma_0)_{m-k}}{(k-1)!(m-k)!}  
\prod_{j=1}^{m-1} \frac{ (j^2-\gamma_1 j+\zeta_1)  }{ (j^2+\gamma_0 j+\zeta_0)}
\]
 with 
  \[
c_{\gamma_1,\zeta_1}= \frac{\Gamma(1 + (\gamma_1 + \sqrt{\gamma_1^2-4\zeta_1})/{2})\Gamma(1 + (\gamma_1 - \sqrt{\gamma_1^2-4\zeta_1})/{2})}{\Gamma(\gamma_1)}.
 \]
In contrast, in a $HPYGP(\sigma_0,\theta_0,\gamma_1,\zeta_1,H_0)$, 
\[
\lim_{t \to +\infty}\P\{ D_{t}=m\}=\lim_{t \to +\infty}\P\{ D_{i,t}=m\}= c_{\gamma_1,\zeta_1}
 \frac{\prod_{l=1}^{m-1} (l^2-\gamma l+\zeta)}{m!(m-1)!}.
\]
 \end{proposition}

%

Also for HPYPs  one can derive  explicit asymptotic distributions using the previous general results. Indeed, if $(\Pi_n)_n \sim PY(\sigma,\theta)$ with $0<\sigma<1$ and $\theta>-\sigma$, then $ | \Pi_{n}|/n^{\sigma}$ converges almost surely and in $L^p$ (for every $p>0$) to a strictly positive random variable $S_{\sigma,\theta}$ with density 
\begin{equation}\label{mittlefftilted}
g_{\sigma,\theta}(s):=\frac{\Gamma(\theta+1)}{\Gamma(\frac{\theta}{\sigma}+1)} s^{\theta/\sigma} g_\sigma(s), \quad s>0,
\end{equation}
where $g_\sigma$ is the  type-2 Mittag-Leffler density, i.e.  the unique density such that
\begin{equation}\label{momentML}
\int_0^{+\infty} x^p g_\sigma(x)dx=\frac{\Gamma(p+1)}{\Gamma(p\sigma+1)}.
\end{equation}
See Theorem 3.8 in \cite{Pit06}.   Moreover, if $\sigma=0$,  we have that 
$ | \Pi_{n}|/\log(n)$ converges almost surely  
and in $L^p$ for every $p>0$
to $\theta>0$. 

On the basis of these results, Proposition \ref{PropAsym1} can be specialized for the case of HPYPs and convergence in $L^{p}$ obtained.

\begin{proposition}\label{asymptoticPY} Assume that $\Pi^{(0)} \sim PY(\sigma_0,\th_0)$ and $\Pi^{(i)} \sim PY(\sigma_1,\th_1)$ (for $i=1,\dots,I$)
with $\sigma_0,\sigma_1\geq 0$.
Then (i) and (ii) of {\rm Proposition \ref{PropAsym1}} hold a.s. and in $L^{p}$, $p>0$, with 
the following specifications:
\begin{itemize}
\item[(i)] for $HPYP(\sigma_0,\theta_0;\sigma_1,\theta_1)$ with $\sigma_0,\sigma_1>0$, $d_n=n^{\sigma_0 \sigma_1}$ and 
\[
D_{i,\infty} \stackrel{\CL}{=} S_{\sigma_0,\th_0} \left(S_{\sigma_1,\th_1}^{(i)}\right)^{\sigma_0}, \quad 
D_\infty \stackrel{\CL}{=} S_{\sigma_0,\th_0} \left(    \sum_{i=1}^I w_i^{\sigma_1} S_{\sigma_1,\th_1}^{(i)}  \right)^{\sigma_0}, 
\]
with $S_{\sigma_0,\th_0},S_{\sigma_1,\th_1}^{(1)},\dots,S_{\sigma_1,\th_1}^{(I)}$ independent random variables with densities $g_{\sigma_0,\th_0}$ and $g_{\sigma_1,\th_1}$, respectively;
\item[(ii)] for $HPYDP(\sigma_0,\theta_0;\theta_1)$ with $\sigma_0>0$, $d_n=\log(n)^{\sigma_0}$ and 
\[
D_{i,\infty} \stackrel{\CL}{=} D_\infty \stackrel{\CL}{=} S_{\sigma_0,\th_0}  \th_1^{\sigma_0},
\]
 with $S_{\sigma_0,\th_0}$ random variable with density $g_{\sigma_0,\th_0}$;
 \item[(iii)] for $HDPYP(\theta_0;\sigma_1,\theta_1)$ with $\sigma_1>0$, $d_n=\sigma_1 \log(n)$ and $D_{i,\infty}=D_\infty= \th_0$;
 \item[(iv)] for $HDP(\theta_0;\theta_1)$,  $d_n=\log( \log(n))$ and 
$D_{i,\infty}=D_\infty= \th_0$.
\end{itemize}
\end{proposition}

Proposition \ref{asymptoticPY} can be used for approximating the moments (e.g., expectation and variance) of the number of clusters as stated in the following
\begin{corollary}\label{corolASY} 
Let $x_n \simeq y_n$ if and only if $\lim_{n \to +\infty} x_n/y_n=1$, then under the same assumptions of Proposition \ref{asymptoticPY},
for every $r>0$:
\begin{itemize}
\item[(i)] for $HPYP(\sigma_0,\theta_0,\sigma_1,\theta_1)$ with $\sigma_0,\sigma_1>0$:
\[
\E\left[D_{i,t}^r\right] \simeq 
n_i(t)^{r \sigma_0 \sigma_1}
\frac{\Gamma(\theta_0+1)}{\Gamma\left(\frac{\theta_0}{\sigma_0}+1\right)} 
\frac{\Gamma\left(r+\frac{\theta_0}{\sigma_0}+1\right)}{\Gamma(\theta_0 + r \sigma_0+1)}
\frac{\Gamma(\theta_1+1)}{\Gamma\left(\frac{\theta_1}{\sigma_1}+1\right)} 
\frac{\Gamma\left(r\sigma_0+\frac{\theta_1}{\sigma_1}+1\right)}{\Gamma(\theta_1 + r \sigma_0 \sigma_1+1)};
\]
\item[(ii)] for $HPYDP(\theta_0,\sigma_0;\theta_1)$ with $\sigma_0>0$: 
\[
\E\left[D_{i,t}^r\right] \simeq  (\log(n_i(t)))^{r\sigma_0}    \th_1^{r \sigma_0} \frac{\Gamma(\theta_0+1)}{\Gamma\left(\frac{\theta_0}{\sigma_0}+1\right)} 
\frac{\Gamma\left(r+\frac{\theta_0}{\sigma_0}+1\right)}{\Gamma(\theta_0 + r \sigma_0+1)};
\]
  \item[(iii)] for $HDPYP(\theta_0,\sigma_1,\theta_1)$ with $\sigma_1>0$:
  $
\E\left[D_{i,t}^r\right] \simeq     \log(n_i(t))^r  {(\sigma_1 \th_0)^r}   $ ;

 \item[(iv)] for $HDP(\theta_0,\theta_1)$:
  $
\E\left[D_{i,t}^r\right] \simeq   {\th_0^r}  \log(\log(n_i(t)))^r$.
\end{itemize}
\end{corollary}
\begin{figure}[p]
  \caption[]{Exact (dashed lines) and asymptotic (solid lines) expected marginal number of clusters $E(D_{i,t})$ when $n_i(t)=1, \dots, 500$ for different HSSMs.}
  \label{Asym_Exp_N_Clusters}
\begin{center}
  \setlength{\tabcolsep}{10pt}
  \begin{tabular}{cc}  
   \includegraphics[width=4cm]{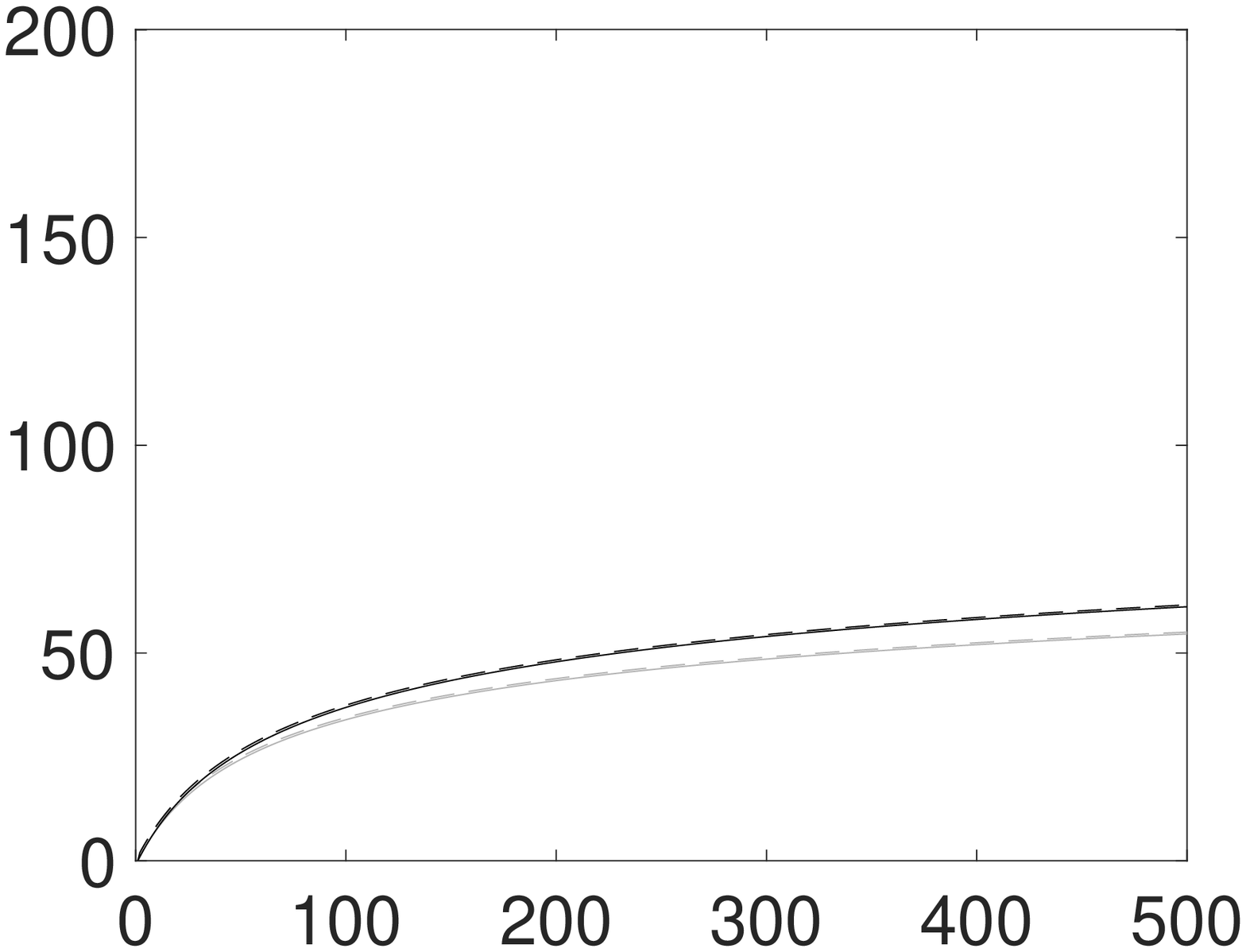} &
   \includegraphics[width=4cm]{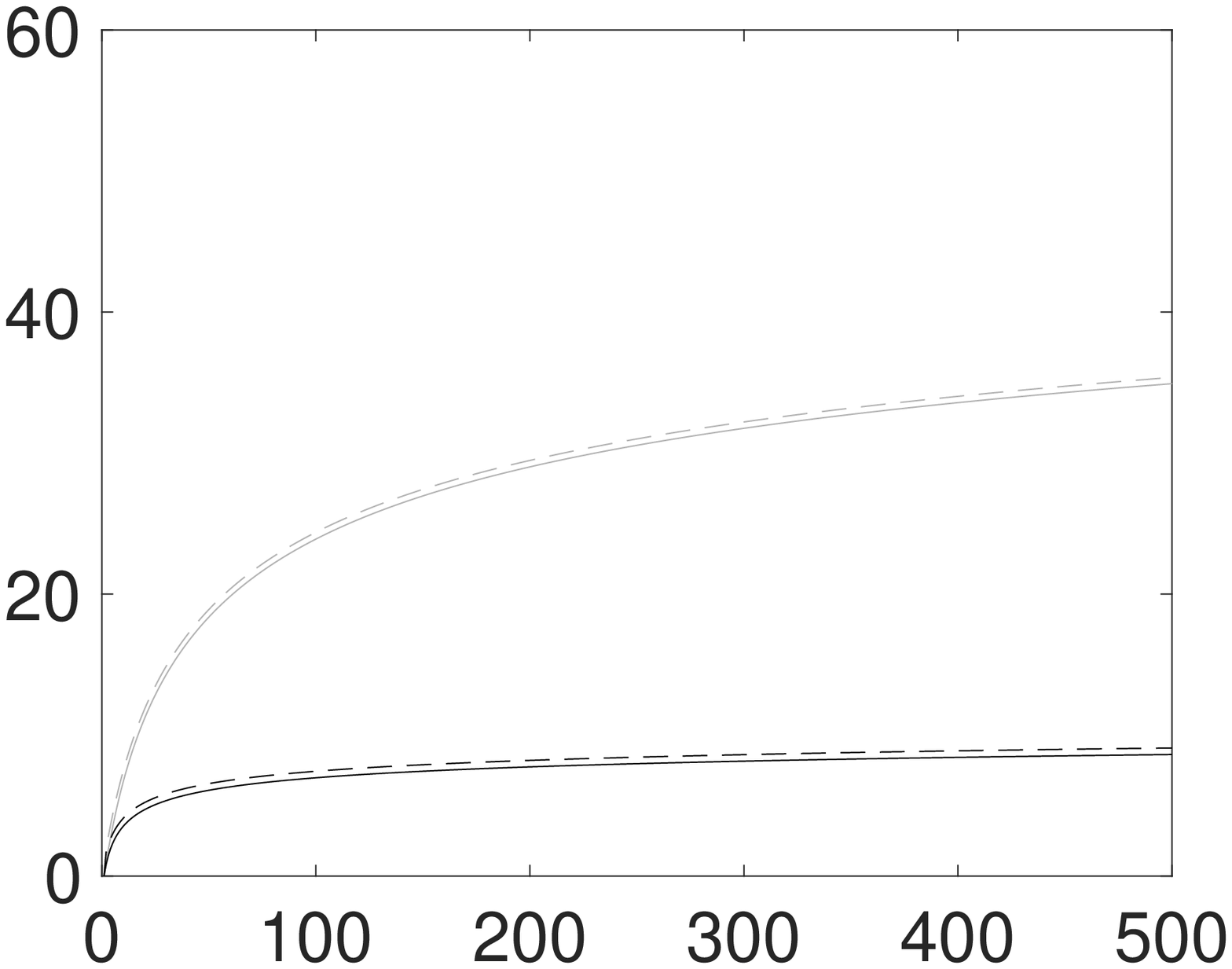} \\
\end{tabular}
\end{center}
\begin{minipage}[]{0.84\textwidth}
\begin{small}
(i) HDP with $\theta_0 = \theta_1 = 43.3$ (left, gray), $\theta_0 = \theta_1 = 50$ (left, black), $\theta_0 = \theta_1 = 25$ (right, gray) and $\theta_0 = \theta_1 = 5$ (right, black).
\end{small}
\end{minipage}
\begin{center}
  \setlength{\tabcolsep}{10pt}
  \begin{tabular}{cc}
   \includegraphics[width=4cm]{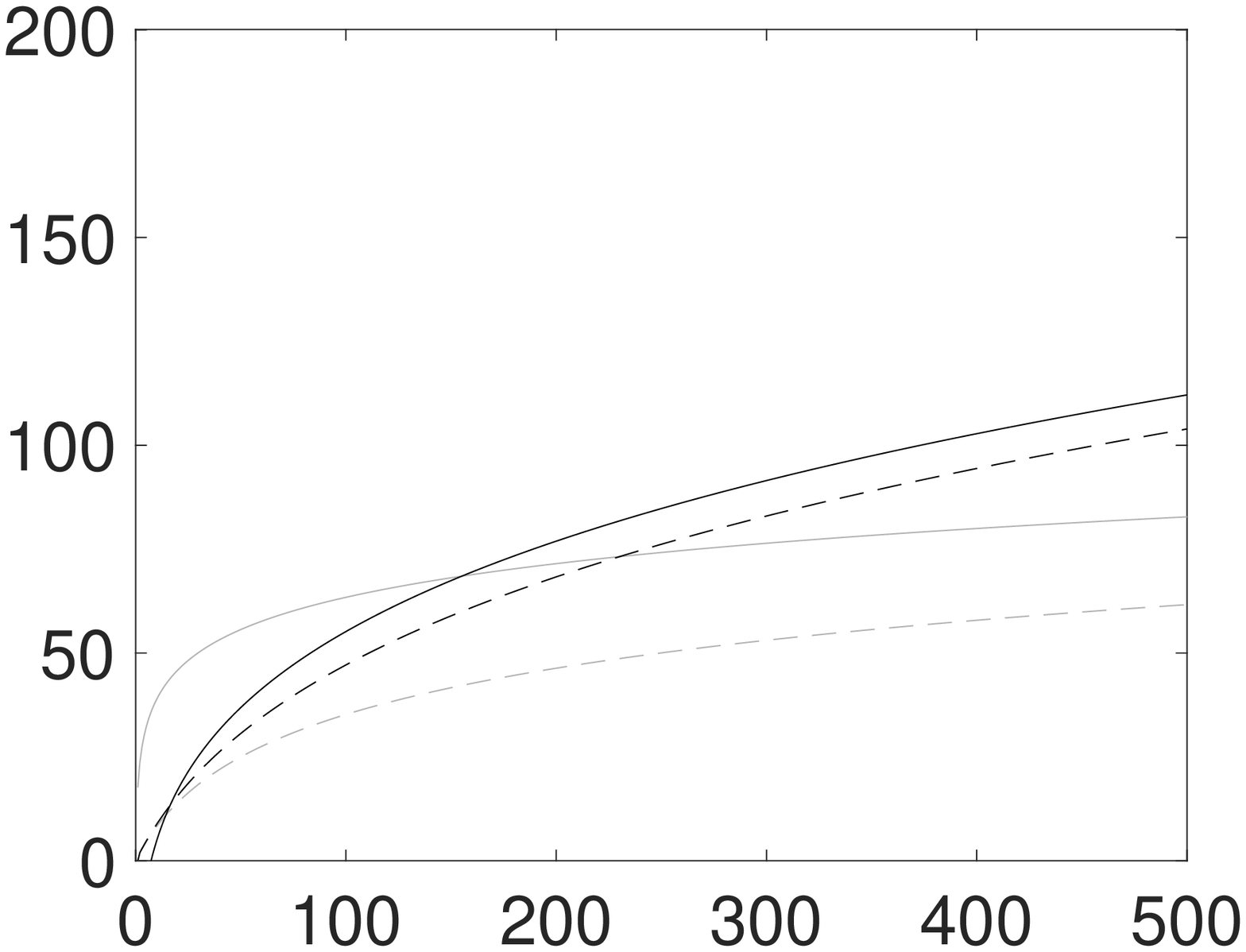}&
   \includegraphics[width=4cm]{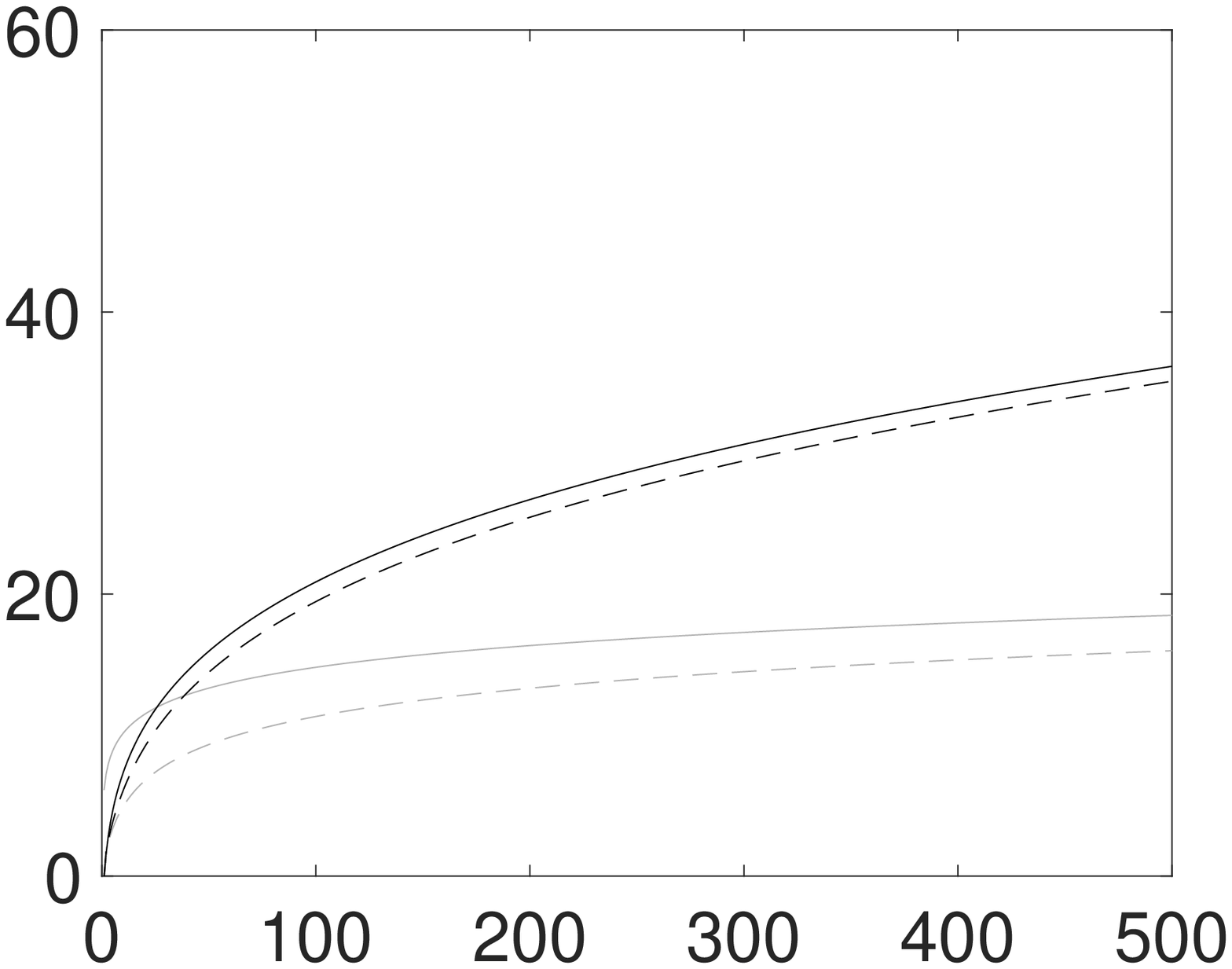}\\
\end{tabular}
\end{center}
\begin{minipage}[]{0.84\textwidth}
\begin{small}
(ii) HPYP with $(\theta_0,\sigma_0) = (\theta_1,\sigma_1) = (29.9,0.25)$ (left, gray), $(\theta_0,\sigma_0) = (\theta_1,\sigma_1) = (29.9,0.5)$ (left, black), $(\theta_0,\sigma_0) = (\theta_1,\sigma_1) = (5,0.25)$ (right, gray) and $(\theta_0,\sigma_0) = (\theta_1,\sigma_1) = (5,0.5)$ (right, black).
\end{small}
\end{minipage}
\begin{center}
  \setlength{\tabcolsep}{10pt}
  \begin{tabular}{cc}
   \includegraphics[width=4cm]{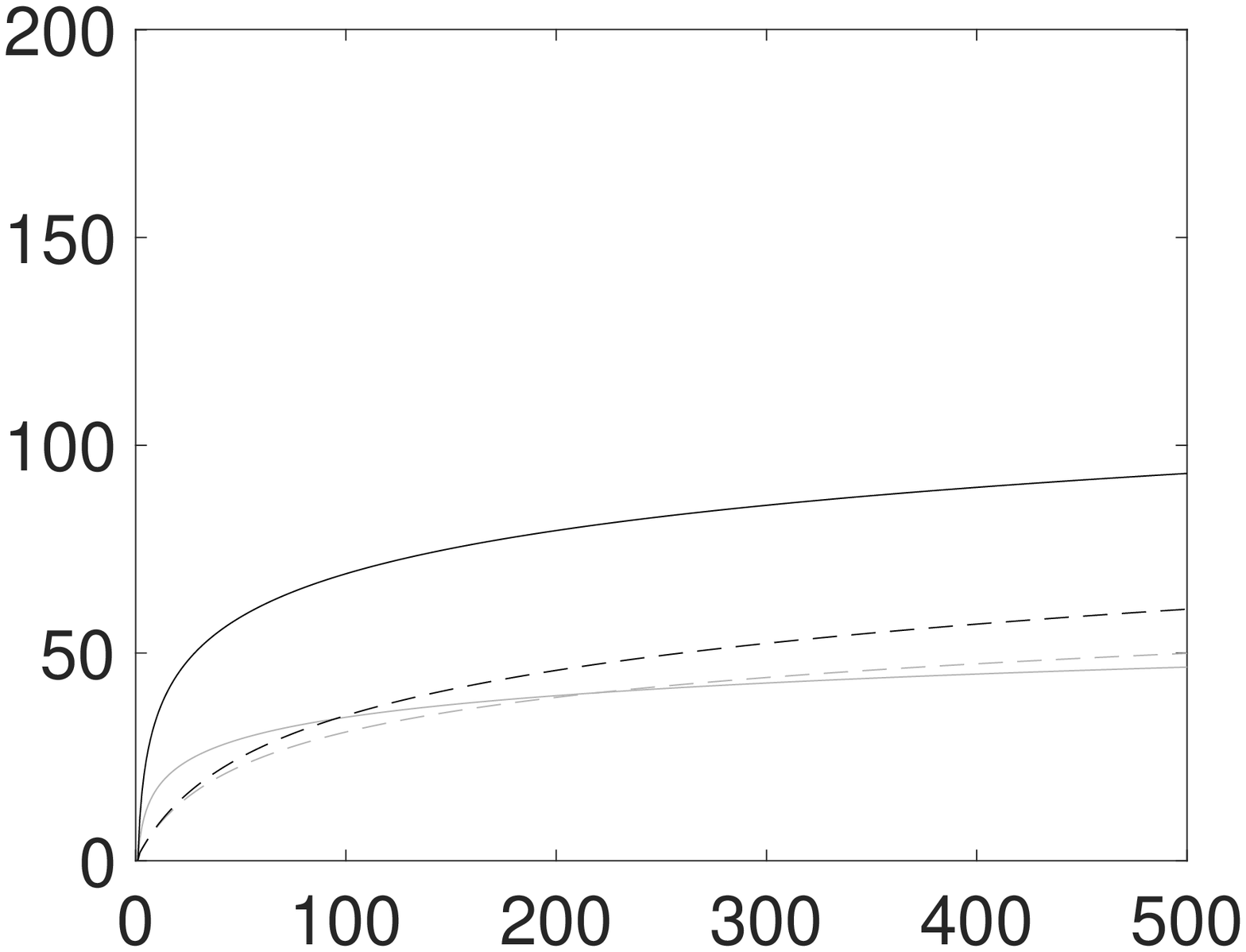}&
   \includegraphics[width=4cm]{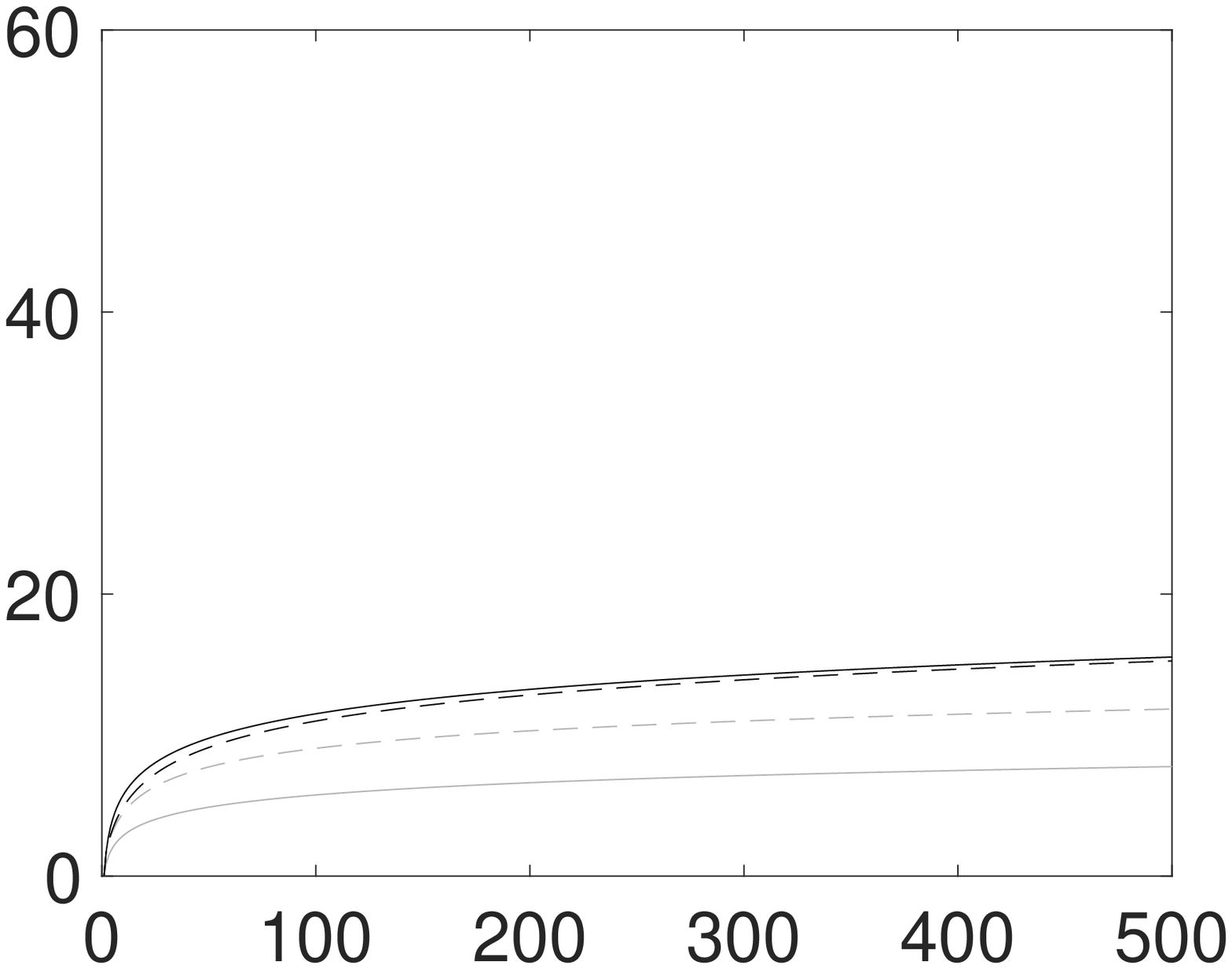} \\
\end{tabular}
\end{center}
\begin{minipage}[]{0.84\textwidth}
\begin{small}
(iii) HDPYP with $\theta_0 = 30$, $(\theta_1,\sigma_1) = (30,0.25)$ (left, gray), $\theta_0 = 30$, $(\theta_1,\sigma_1) = (30,0.5)$ (left, black), $\theta_0 = 5$, $(\theta_1,\sigma_1) = (5,0.25)$ (right, gray) and $\theta_0 = 5$, $(\theta_1,\sigma_1) = (5,0.5)$ (right, black).
\end{small}
\end{minipage}
\begin{center}
  \setlength{\tabcolsep}{10pt}
  \begin{tabular}{cc}
   \includegraphics[width=4cm]{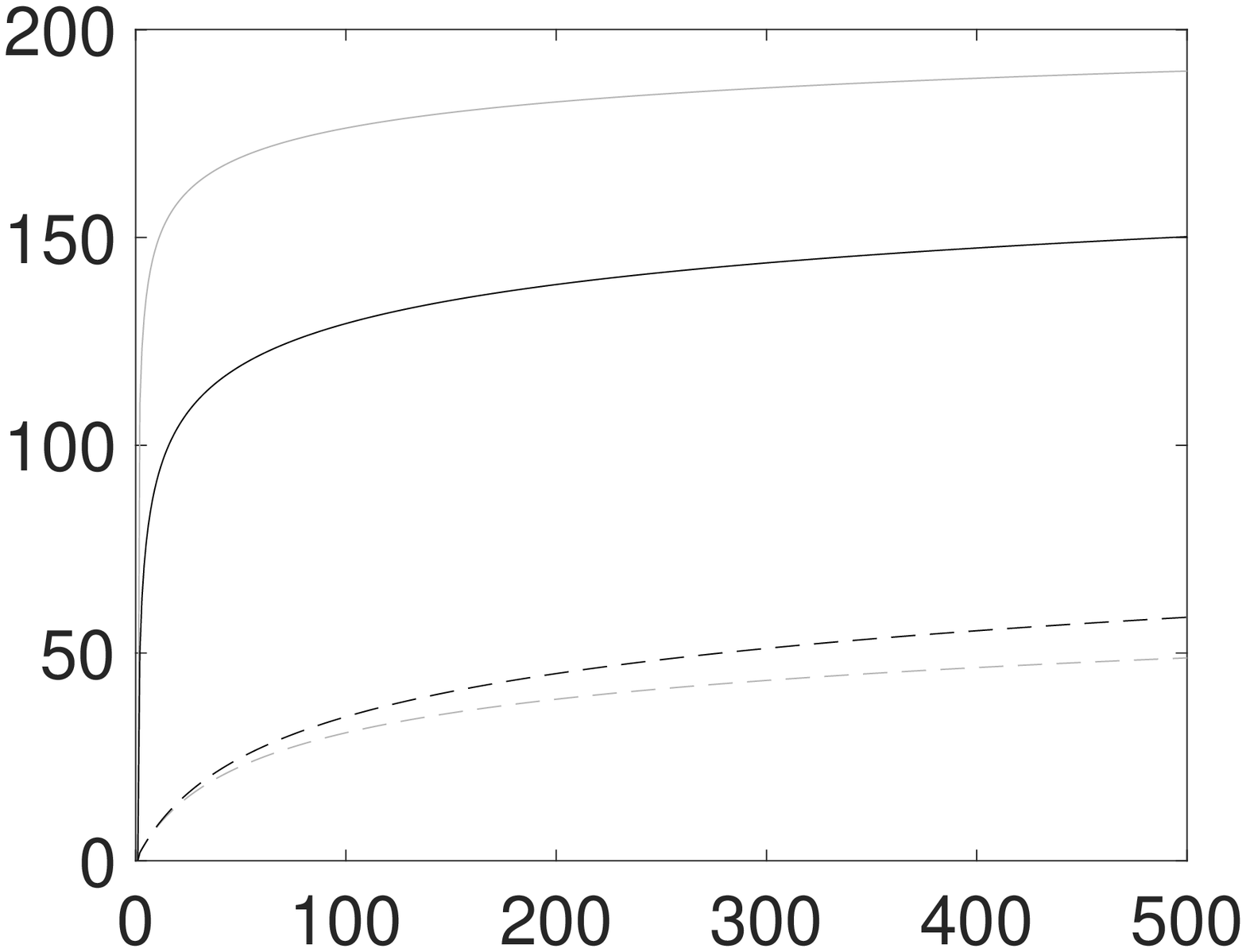}&
   \includegraphics[width=4cm]{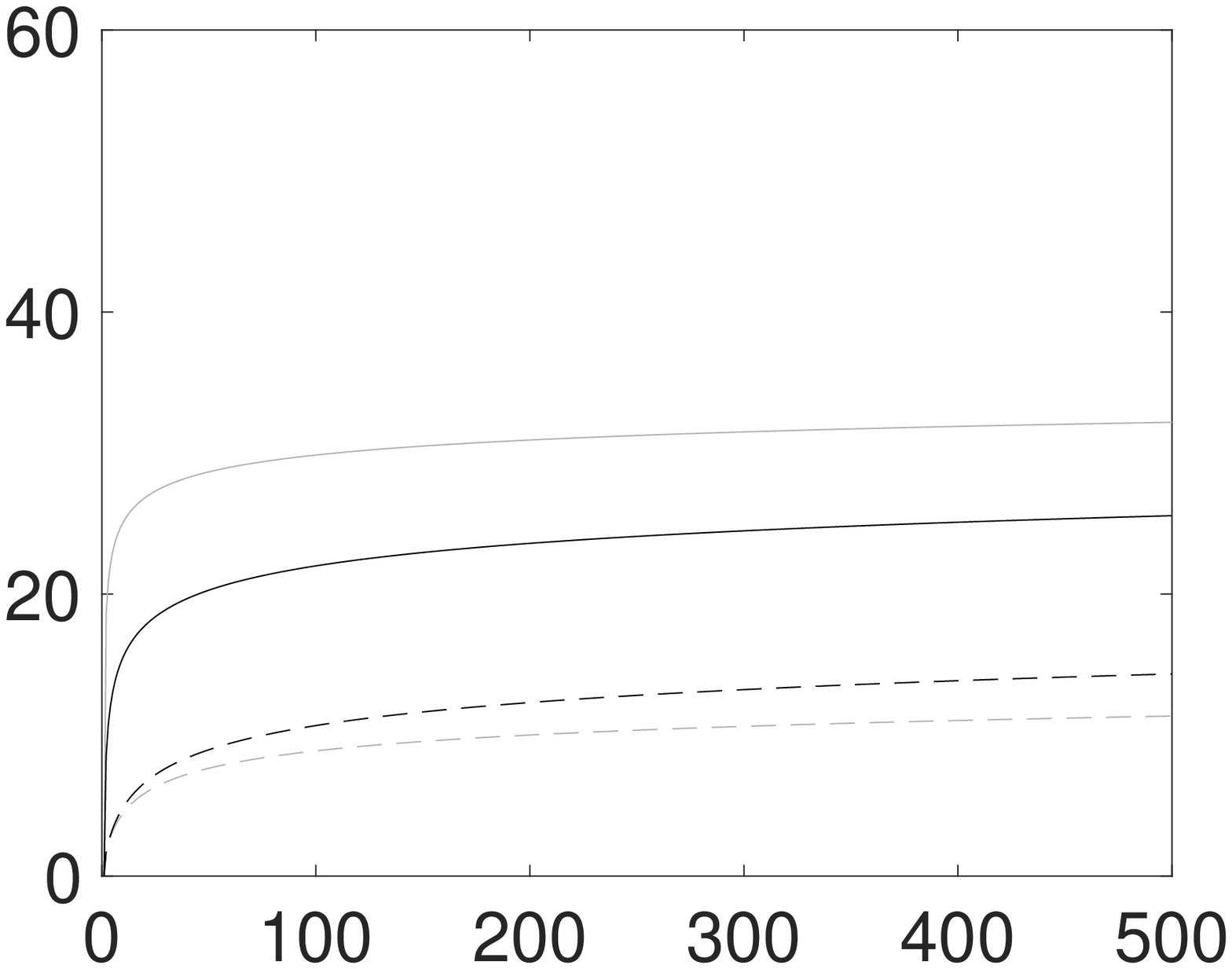}\\
\end{tabular}
\end{center}
\begin{minipage}[]{0.84\textwidth}
\begin{small}
(iv) HPYDP with $(\theta_0,\sigma_0) = (30,0.25)$, $\theta_1 = 30$ (left, gray), $(\theta_0,\sigma_0) = (30,0.5)$, $\theta_1 = 30$ (left, black), $(\theta_0,\sigma_0) = (5,0.25)$, $\theta_1 = 5$ (right, gray) and $(\theta_0,\sigma_0) = (5,0.5)$, $\theta_1 = 5$ (right, black).
\end{small}
\end{minipage}
\end{figure}

In Fig. \ref{Asym_Exp_N_Clusters}, we compare exact and asymptotic values (see Proposition \ref{Prop:clusterdistribution} and Corollary \ref{corolASY}, respectively) of the expected marginal number of clusters for the HSSMs in the PY family: $HDP(\theta_0;\theta_1)$, $HDPYP(\theta_0;\sigma_1,\theta_1)$, $HPYP(\sigma_0,\theta_0;\sigma_1,\theta_1)$ and $HPYDP(\theta_0,\sigma_0;\theta_1)$ (different rows of Fig. \ref{Asym_Exp_N_Clusters}). For each HSSM we consider $n_i(t)$ increasing from 1 to 500 and different parameter settings (different columns and lines). For the HDP the exact value (dashed lines) is well approximated by the asymptotic one (solid line) for all sample sizes $n_{i}(t)$, and different values of $\theta_i$ (gray and blacks lines in the left and right plots of panel (i)). For the HPYP, the results in panel (ii) show that there are larger differences when $\theta_i$, $i=0,1$ are large and $\sigma_0$ and $\sigma_1$ are close to zero (left plot). The approximation is good for small $\theta_i$ (right plot) and improves slowly with increasing $n_{i}(t)$ for smaller $\sigma_i$ (gray lines in the right plot). In the panels (iii) and (iv) for HDPYP and HPYDP, there exist parameter settings where the asymptotic approximation is not satisfactory and is not improving when $n_{i}(t)$ increases. 
Our numerical results point out that the asymptotic approximation for both PY and HPY may lack of accuracy. Thus, the exact formula for the number of clusters should be used in the applications when calibrating the parameters of the process.

\section{Samplers for Hierarchical Species Sampling Mixtures}
\label{Sec_Gibbs}
\newcommand{\bd}{\boldsymbol{d}}
\newcommand{\bphi}{\boldsymbol{\phi}}
\newcommand{\bc}{\boldsymbol{c}}
\newcommand{\bx}{\boldsymbol{Y}}
Random measures and hierarchical random measures are widely used in Bayesian nonparametric inference (see \cite{BNP2010} for an introduction) as prior distributions for the parameters of a given density function. In this context a further stage is added to the hierarchical structure of Eq. \eqref{h0} involving an observation model
\[
Y_{i,j}| \xi_{i,j}   \stackrel{ind}{\sim} f(\cdot|\xi_{i,j}),
\]
where $f$ is a suitable kernel density (e.g., with respect to the Lebesgue measure). The resulting model is an infinite mixture, which is then the object of Bayesian inference. In this framework, the posterior distribution is usually not tractable and Gibbs sampling has been proposed to approximate the posterior quantities of interest. There are two main classes of samplers for posterior approximation in infinite mixture models: the marginal (see \cite{Escobar1994} and \cite{EscobarWest1995}) and the conditional (\cite{Wal07}, \cite{Pap08}, \cite{walker2011}) samplers. See also \cite{favaro2013} for an up-to-date review. In this section, we extend the marginal sampler for HDP mixture (see \cite{TehJordan2006}, \cite{Teh2006} and \cite{Teh10}), to our general class of HSSM. We present the sampler for the conjugate case, where the atoms can be integrated out analytically, nevertheless the sampling method can be modified following the auxiliary variable sampler of \cite{Neal2000} and \cite{favaro2013}.

Following the notation in Section \ref{S:CRF}, we consider the data structure
\[
\begin{split}
 Y_{i,j},c_{i,j} &:   i \in \CJ, \mbox{ and } j=1,\dots, n_{i\cdot\cdot} \\
 d_{i,c} &:   i \in \CJ, \mbox{ and } c=1,\dots, m_{i \cdot} \\
 \phi_d &: d \in \CD,
\end{split}
\]
where $Y_{i,j}$ is the $j$-th observation in the $i$-th group, 
$n_{i\cdot\cdot}=n_i$ is the total number of observations in the $i$-th group,
and $\CJ=\{1,\dots,I\}$ is the set of group indexes. 
The latent variable $c_{i,j}$ denotes the table at which the $j$-th ''customer`` of ''restaurant" $i$ sits and $d_{i,c}$ the index of the ''dish`` served at table $c$ in restaurant $i$. The random variables $\phi_d$ are the ''dishes`` and $
\CD=\left \{d:  d= d_{i,c} \,\, \text{for some} i \in \CJ \text{and} c \in \{1,\dots, m_{i\cdot}\} \right \}$ is the set of indexes of the served dishes. 

Let us assume that the distribution $H$ of the atoms $\phi_d$s has  density $h$  (with respect to the Lebesgue measure, or any other reference measure) and the observations $Y_{i,j}$ have a kernel density $f(\cdot|\cdot)$, then our hierarchical infinite mixture model is
$$
Y_{i,j} |\bphi, \bc, \bd  \stackrel{ind}{\sim} f\left(\cdot|\phi_{d_{i,c_{i,j}}}\right),\quad \bphi | \bc, \bd \stackrel{i.i.d.}{\sim} h(\cdot),\quad [\bc, \bd] \sim HSSM
$$
where
\[
\begin{split}
& \bc=[\bc_{i}: i \in \CJ ], \,\, \mbox{ and } \bc_i=[c_{i,j}: j=1,\dots,n_{i\cdot\cdot}], \\
& \bd=[d_{i,c}: i \in \CJ \mbox{ and } c=1,\dots, m_{i \cdot}], \\
& \bphi=[\phi_{d}: d \in \CD], \\
\end{split}
\]
and, with a slight abuse of notation, we write $[\bc, \bd] \sim HSSM$ in order to 
denote the distribution of the labels  $[\bc, \bd]$  obtained from a $HSSM$ as in
\eqref{h0}. If we define 
\[
d^*_{i,j}=d_{i,c_{i,j}} \quad \text{ and} \quad  \bd^*=[d^*_{i,j}:  i \in \CJ, j=1,\dots, n_{i\cdot\cdot}]
\]
then $[\bc,\bd]$ and $[\bc,\bd^*]$ contain the same amount of information, indeed 
$\bd^*$ is a function of $\bd$ and $\bc$, while $\bd$ is a function of $\bd^*$ and $\bc$. 
From now on, we denote with $\bx=[Y_{i,j}: i \in \CJ, j=1,\dots, n_{i\cdot\cdot}]$ the set of observations.

\subsection{Chinese Restaurant Franchise Sampler} 
If $f$ and $H$ are conjugate,   the {\it Chinese Restaurant Franchise Sampler} of \cite{TehJordan2006} can be generalized and a new sampler can be obtained for our  class of models. 

Denotes with the superscript $^{\neg ij}$ the counts and sets in which customer $j$ in restaurant $i$ is removed and, 
analogously, with $^{ \urcorner ic}$ the counts and sets in which all the customers in table $c$ of restaurant $i$ are removed. 
We denote with $p(X)$ the density of the random variable $X$. It should be clear from the context if this density is with respect to the counting measure (i.e., a discrete distribution), or with respect to the Lebesgue measure. In order to avoid the proliferation of the symbols, we shall use the same letter 
(e.g., $\bx, \bphi, \bc, \bd$) to denote both the random variables and their realizations. 

The proposed Gibbs sampler simulates iteratively the elements of 
$\bc$ and $\bd$ from their full conditional distributions, where 
the latent variables $\phi_d$  are integrated out analytically. In sampling 
 the latent variable $\bc$, we need to sample jointly $[\bc,\bd^*]$ and, since $\bd$ is a 
 function of $[\bc,\bd^*]$, this also gives a sample for $\bd$. In order to improve the mixing we re-sample $\bd$ given $\bc$ in a second step. In summary, the sampler iterates for $i=1,\dots,I$ according to the following steps: 
\begin{itemize}
\item[(i)] sample  $[c_{i,j},d^*_{i,j}]$ from  $p(c_{i,j},d^*_{i,j}| \bx, \bc^{ \urcorner ij}, {\bd^*}^{ \urcorner ij})$ (see Eq. \eqref{fullcd-final}), for $j=1,\dots,n_{i\cdot\cdot}$;
\item[(ii)] (re)-sample $d_{i,c}$ from $p(d_{i,c}|\bx, \bc, {\bd}^{ \urcorner ic})$ (see Eq. \eqref{full-dishes}), for $c=1,\dots,m_{i \cdot}$.
\end{itemize}

In what follows, $\omega_{n,k}$ and $\nu_n$ indicate the weights of the predictive distribution 
  of the random partitions $\Pi^{(i)}$ ($i=1,\dots,I$)  with EPPF $\FP$ (see Section \ref{Sec:1-1}) and
   $\tilde \omega_{n,k}$ and $\tilde \nu_n$ the weights of the predictive distribution 
  of the random partitions $\Pi^{(0)}$ with EPPF $\FP_0$. 
Moreover,  we set
\[
\CC_i=\{c: c=c_{i,j} \,  \text{for some some $j=1,\dots,n_{i \cdot \cdot}$} \}.
\]

Finally, for  an arbitrary index set  $\CS$ and a "dish" label $d$, the marginal conditional density  of $\{Y_{i,t}\}_{it \in \CS}$ given all the other observations  assigned to  the cluster $d$ is 
\begin{equation}\label{Ml}
p(\{Y_{i,t}\}_{it \in \CS} | Y_{i',t'} :  (i't') \in \CS_d \setminus \CS,  \bc, \bd  )
= \frac{\int h(\phi) \prod_{i't' \in \CS_d \cup \CS} f(Y_{i',t'}|\phi)d\phi}{\int h(\phi) 
\prod_{i't' \in \CS_d \setminus \CS} f(Y_{i',t'}|\phi) d\phi},
\end{equation}
where  $\CS_d=\{(i',t') : d_{i',t'}^*=d\}$, and, following \cite{TehJordan2006}, will be denoted by 
$f_{d}(\{Y_{i,t}\}_{it \in \CS})$.

As regards the full conditional $p(c_{i,j},d^*_{i,j}| \bx, \bc^{ \urcorner ij}, {\bd^*}^{ \urcorner ij})$, the outcomes of the sampling are of three types. The customer $j$ can sit ''alone`` at a table, 
$c_{ij}=c^{new}$, or he/she can sit at a table with other customers, $c_{i,j}=c^{old}$ where $c^{old}$ is a table index already present in  $\CC_i^{ \urcorner ij}$. If $c_{i,j}=c^{old}$, then $d^*_{i,j}=d_{i,c^{old}}\in{\CD}^{ \urcorner ij}$, whereas if $c_{i,j}=c^{new}$, then we can have two disjoint events: either 
$d^*_{i,j}=d$ for some $d\in{\CD}^{ \urcorner ij}$ or $d^*_{i,j}=d$ with $d \not \in {\CD}^{ \urcorner ij}$, say $d^*_{i,j}=d^{new}$. In formula, the full conditional for $[c_{i,j},d^*_{i,j}]$ is  
\begin{equation}\label{fullcd-final}
\begin{split}
& p(c_{i,j}=c^{old},d^*_{i,j}=d_{i,c^{old}}| \bx, \bc^{ \urcorner ij}, {\bd^*}^{ \urcorner ij}) \propto
 \omega_{n_{i\cdot\cdot}-1,c^{old}}(\bc_i^{\urcorner ij})  f_{d_{i,c^{old}}}(\{Y_{i,j}\}),  \\
&p(c_{i,j}=c^{new},d^*_{i,j}=d^{old}|\bx, \bc^{ \urcorner ij}, {\bd^*}^{ \urcorner ij}) 
\propto  \ \nu_{n_{i\cdot\cdot}-1}(\bc_i^{\urcorner ij}) \tilde
 \omega_{m_{\cdot \cdot}^{ \urcorner ij},d^{old}}(\bd^{\urcorner ij}) f_{d^{old}}(\{Y_{i,j}\}), \\
& p(c_{i,j}=c^{new},d^*_{i,j}=d^{new}|\bx, \bc^{ \urcorner ij}, {\bd^*}^{ \urcorner ij})\propto
 \nu_{n_{i\cdot\cdot}-1}(\bc_i^{\urcorner ij}) 
\tilde \nu_{m_{\cdot \cdot}^{ \urcorner ij}}(\bd^{\urcorner ij}) 
f_{d^{new}}(\{Y_{i,j}\}), \\
\end{split}
\end{equation}
where 
\[ 
\omega_{n_{i\cdot\cdot}-1,c^{old}}(\bc_i^{\urcorner ij})= \omega_{n_{i\cdot\cdot}-1,c^{old}}(n_{i1\cdot}^{\urcorner ij},\dots,n_{i m_{i \cdot}^{\urcorner ij} \cdot}^{\urcorner ij}) ,
\quad \nu_{n_{i\cdot\cdot}-1}(\bc_i^{\urcorner ij})=\nu_{n_{i\cdot\cdot}-1}(n_{i1\cdot}^{\urcorner ij},\dots,n_{i m_{i \cdot}^{\urcorner ij} \cdot}^{\urcorner ij}),
\]
and
\[ \tilde
 \omega_{m_{\cdot \cdot}^{ \urcorner ij},d^{old}}(\bd^{\urcorner ij})= \tilde
 \omega_{m_{\cdot \cdot}^{ \urcorner ij},d^{old}} (m_{\cdot 1}^{\urcorner ij},\dots,m_{\cdot |\CD^{\urcorner ij}|}^{\urcorner ij}),  \quad 
 \tilde \nu_{m_{\cdot \cdot}^{ \urcorner ij}} (\bd^{\urcorner ij})=
\tilde \nu_{m_{\cdot \cdot}^{ \urcorner ij}} (m_{\cdot 1}^{\urcorner ij},\dots,m_{\cdot |\CD^{\urcorner ij}|}^{\urcorner ij}).
\]

The second step of the sampler is related to the full conditional $p(d_{i,c}|\bx, \bc, {\bd}^{ \urcorner ic})$. Denote with $^{ \urcorner ic}$  the counts and sets in which all the customer in table $c$ of restaurant $i$ are removed. Then, the full conditional for $d_{i,c}$ is
\begin{equation}\label{full-dishes}
\begin{split}
& p(d_{i,c}=d^{new} |\bx, \bc, {\bd}^{ \urcorner ic}) \propto \tilde \nu_{m_{\cdot\cdot}^{ \urcorner ic} }(\bd^{ \urcorner ic}) f_{d^{new}}( \{Y_{i,j}: (i,j) \in \CS_{ic} \}  ), \\
&   p(d_{i,c}=d^{old} |\bx, \bc, {\bd}^{ \urcorner ic}) \propto \tilde \omega_{m_{\cdot\cdot}^{ \urcorner ic}, d^{old} }(\bd^{ \urcorner ic}) f_{d^{old}}( \{Y_{i,j}: (i,j) \in \CS_{ic} \}  ), \\
\end{split}
\end{equation}
where $d^{old}$ runs in $\CD^{ \urcorner ic}$.

If needed one can always sample  $\bphi$ given $[\bx,\bc,\bd]$,  
from the conditional distribution 
\begin{equation}\label{phifull_out}
p(\bphi |\bx,\bc,\bd) \propto \prod_{d \in \CD} h(\phi_d) \prod_{(i,j): d^*_{i,j}=d} f(Y_{i,j}|\phi_d).
\end{equation}
The detailed  derivations of the full conditional distributions are given in Appendix \ref{App:proofs}.

\subsection{Approximating  predictive distributions}

The posterior predictive distribution $p(Y_{i,n_{i}+1}|\bx)$ can be approximated by 
\[
 \frac{1}{M}\sum_{m=1}^{M} p(Y_{i,n_{i}+1}|\bx, \bc^{(m)}, \bd^{*(m)}), 
\]
where $(\bc^{(m)}, \bd^{*(m)})_{m=1,\dots,M}$ is the output of $M$ iterations of Gibbs sampler and 
\begin{equation}\label{pred1bis}
p(Y_{i,n_{i}+1}|\bx, \bc, \bd^*) = \!\!\!\! \! \sum_{c_{i,n_i+1},d_{i,n_i+1}^*} p(Y_{i,n_{i}+1}|\bx, \bc, \bd^*,c_{i,n_i+1},d_{i,n_i+1}^*) p(c_{i,n_i+1},d_{i,n_i+1}^*|\bx, \bc, \bd^*)
\end{equation}
is the conditional predictive distribution. The first term appearing in the conditional predictive can be written as
\begin{equation}\label{pred2bis}
p(Y_{i,n_{i}+1}|\bx, \bc, \bd^*,c_{i,n_i+1},d_{i,n_i+1}^*)  \propto {\int
f(Y_{i,n_{i}+1}|\phi)  \prod_{i't' :d_{i't'}^*=d_{i,n_i+1}^*  } 
f(Y_{i',t'}|\phi)  h(\phi)  d\phi}. 
\end{equation}
As for  the second term is concerned,  one gets
\begin{equation}\label{cdpredictive00}
p(c_{i,n_i+1},d_{i,n_i+1}^*|\bx, \bc, \bd^*)  = p\left(c_{i,n_i+1},d_{i,n_i+1}^*| \bc, \bd \right) 
\end{equation}
with 
\begin{equation}\label{cdpredictive}
p\left(c_{i,n_i+1}=c,d_{i,n_i+1}^*=d | \bc, \bd \right) = 
\left \{
\begin{array}{ll}
 \omega_{n_{i\cdot\cdot},c^{old}}(\bc_i),
 & \text{if  $c=c^{old}$, $d=d_{i,c^{old}}$},    \\
  \nu_{n_{i\cdot\cdot}}(\bc_i) \tilde
 \omega_{m_{\cdot \cdot},d^{old}}(\bd), &    \text{if $c=c^{new}$, $d=d^{old}$},   \\
 \nu_{n_{i\cdot\cdot}}(\bc_i) \tilde
 \nu_{m_{\cdot \cdot}}(\bd) &  \text{if $c=c^{new}$, $d=d^{new}$.}  
\end{array}
\right .
\end{equation}
See  Appendix \ref{App:proofs} for details.

Summation and integration in  \eqref{pred1bis}-\eqref{pred2bis}  can be avoided by 
using the following approximation of  the predictive $p\left(Y_{i,n_{i}+1}|\bx\right)$
\[
 \frac{1}{M}\sum_{m=1}^{M} f\left( Y_{i,n_{i}+1}|  \phi_{d_{i,n_i+1}^*}^{(m)}\right),
 \]
where  $(c_{i,n_i+1}^{(m)},$ ${d_{i,n_i+1}^*}^{(m)})$ is sampled  from
\eqref{cdpredictive}, 
and $\phi_{d_{i,n_i+1}^*}^{(m)}$
 given   $(c_{i,n_i+1}^{(m)},{d_{i,n_i+1}^*}^{(m)})$ from 
\[
p\left(\phi_{d_{i,n_i+1}^*}^{(m)}| c_{i,n_i+1}^{(m)}, {d_{i,n_i+1}^*}^{(m)}, 
\bx,\bc,\bd\right) 
\propto h\left(\phi_{d_{i,n_i+1}^*}^{(m)}\right) \prod_{(i,j): d^*_{i,j}=d_{i,n_i+1}^*} f\left(Y_{i,j}|\phi_{d_{i,n_i+1}^*}\right)
\]
at each Gibbs iteration.


\section{Illustrations}\label{Sec_Emp}

\subsection{Simulation Experiments}
This section illustrates the performance of the Gibbs sampler described in the previous section when applied to different processes (HDP, HPYP, HGP, HDPYP, HPYDP, HGDP and HGPYP) and sets of synthetic data. In the first experimental setting, we consider three groups of observations from three-component normal mixtures with common mixture components, but different mixture probabilities:
\begin{eqnarray*}
Y_{1j}&\overset{iid}{\sim}&0.3\mathcal{N}(-5,1)+0.3\mathcal{N}(0,1)+0.4\mathcal{N}(5,1),\,\, j=1,\ldots,100,\\
Y_{2j}&\overset{iid}{\sim}&0.3\mathcal{N}(-5,1)+0.7\mathcal{N}(0,1),\,\, j=1,\ldots,50,\\
Y_{3j}&\overset{iid}{\sim}&0.8\mathcal{N}(-5,1)+0.1\mathcal{N}(0,1)+0.1\mathcal{N}(5,1),\,\, j=1,\ldots,50.
\end{eqnarray*}
The parameters of the different prior processes are chosen such that the marginal expected value is $\mathbb{E}(D_{i,t}) = 5$ and the variance is between $1.97$ and $3.53$ (see panel (a) in Tab. \ref{Tab:priorSIM} and left panel in Fig. \ref{all_Experm_1}, Appendix \ref{App:EmpRes}).

In the second experimental setting, we consider ten groups of observations from two-component normal mixtures with one common component and different mixing probabilities:
\begin{eqnarray*}
Y_{ij}&\overset{iid}{\sim}&0.7\mathcal{N}(-5,1)+0.3\mathcal{N}(-4+i,1),\,\,j=1,\ldots,50.
\end{eqnarray*}
The parameters of the different prior processes are chosen such that the marginal expected value is $\mathbb{E}(D_{i,t}) = 10$ and the variance is between $4.37$ and $6.53$ (see panel (b) in Tab. \ref{Tab:priorSIM} and right panel in Fig. \ref{all_Experm_1}, Appendix \ref{App:EmpRes}).

For each settings we generate 50 independent dataset and run the marginal sampler described in the previous section with $6000$ iterations to approximate the posterior predictive distribution and the posterior distribution of the clustering variables $\bc$ and $\bd$. We discard the first $1.000$ iterations of each run. All inferences are averaged over the 50 independent runs. 

We study the goodness of fit of each model by evaluating its co-clustering errors and predictive abilities (see \cite{favaro2013} and \cite{Dahl}). We put $\bd^{(m)}$ in the vector form $
\tilde{\mathbf{d}}^{(m)} = (d_{1,c_{11}}^{(m)},\dots, d_{1,c_{1n_1}}^{(m)}, \dots,d_{I,c_{I1}}^{(m)},\dots, d_{I,c_{In_I}}^{(m)})$, $m=1,\ldots,M$, where $M$ is the number of Gibbs iterations. The co-clustering matrix of 
posterior pairwise probabilities of joint classification is estimated by:
\begin{align*}
P_{lk} = \frac{1}{M}\sum_{m=1}^M \delta_{\left\{\tilde{d}_l^{(m)}\right\}}\left(\tilde{d}_k^{(m)}\right)\quad \,l,k=1,\ldots,n_{...}
\end{align*}

Let $\tilde{\bd}_{0}$ be the vector of the true values of the clustering variable $\tilde{\mathbf{d}}$. The co-clustering error can be measured as the average $L_1$ distance between the true pairwise co-clustering matrix, $\delta_{\left\{d_{0l}\right\}}\left(d_{0k}\right)$ and the estimated co-clustering probability matrix, $P_{lk}$, i.e.:
\begin{align}
CN = \frac{1}{n_{...}^2} \sum_{l=1}^{n_{...}} \sum_{k=1}^{n_{...}} |\delta_{\left\{d_{0l}\right\}}\left(d_{0k}\right) - P_{lk}|. \label{Cocl}
\end{align}

An alternative measure can be defined by using the Hamming norm and the estimated co-clustering matrix, $\mathbb{I}(P_{lk}>0.5)$, i.e. 
\begin{align}
CN^{\ast} = \frac{1}{n_{...}^2} \sum_{l=1}^{n_{...}} \sum_{k=1}^{n_{...}} |\delta_{\left\{d_{0l}\right\}}\left(d_{0k}\right) - \mathbb{I}(P_{lk}>0.5)|. \label{Cocl_thr}
\end{align}
Both accuracy measures $CN$ and $CN^{\ast}$ attain $0$ in absence of co-clustering error and $1$ when co-clustering is mispredicted.

The $L_1$ distance between the true group-specific densities, $f(Y_{i,n_i+1})$ and the corresponding posterior predictive densities, $p(Y_{i,n_i+1}|\mathbf{Y})$, can be used to define the predictive score, i.e.:
\begin{align*}
SC = \frac{1}{I}\sum_{i=1}^{I}\int \left| f(Y_{i,n_i+1})- p(Y_{i,n_i+1}|\mathbf{Y}) \right| dY_{i,n_i+1}.
\end{align*}

Finally, we consider the posterior median ($\widehat{q_{0.5}(D)}$) and variance ($\widehat{V(D)}$) of the total number of clusters $D$.

\begin{table}[t]                                
\caption{Model accuracy for seven HSSMs in two experimental settings (panel (a) and (b)) using different measures: co-clustering norm ($CN$), threshold co-clustering norm ($CN^{\ast}$), predictive score ($SC$), posterior median ($\widehat{q_{0.5}(D)}$) and variance ($\widehat{V(D)}$) of the number of clusters. The accuracy and its standard deviation (in parenthesis) have been estimated with 50 independent MCMC experiments. Each experiment consists in 6000 MCMC iterations.}                     
\label{table:1st_Experim}                   
\centering   
\begin{scriptsize}
\setlength{\tabcolsep}{2pt}                           
\begin{tabular}{lccccccc}   
\hline              
            &HDP     & HPYP      & HGP   & HDPYP   & HPYDP & HGDP & HGPYP\\      
\hline
&\multicolumn{7}{c}{(a) Three-component normal mixtures}\\
\hline
$CN$         &  0.0975 & 0.0829  & 0.1220 & 0.0668  & 0.0888 & 0.1018 &0.0982 \\
           & (0.0032)& (0.0035)&(0.0040)& (0.0026)&(0.0037)&(0.0032)&(0.0033) \\
$CN^{\ast}$ & 0.0073  & 0.0056  & 0.0311 & 0.0053 & 0.0057 & 0.0079 &0.0070\\
           & (0.0031)& (0.0014)&(0.0018)& (0.0024)&(0.0020)&(0.0028)&(0.0023) \\
$SC$          & 0.5732  & 0.5556  & 0.6121 & 0.5368  &0.5651  &0.5872  &0.5917\\
           & (0.0187)& (0.0186)&(0.0193)& (0.0197)&(0.0199)&(0.0200)&(0.0205)\\
$\widehat{q_{0.5}(D)}$ &  7 & 7 & 6.7 & 5 & 6.96 & 6.04 & 6 \\
& (0) & (0) & (0.4629) & (0) & (0.1979) & (0.1979) & (0) \\

$\widehat{V(D)}$ & 3.3365 & 4.5520 & 2.5166 & 2.1800 & 4.2211 & 2.3580 & 2.3509 \\
& (0.1195) & (0.2014) & (0.0994) & (0.0882) & (0.1563) & (0.1144) & (0.0958) \\
\hline
&\multicolumn{7}{c}{(b) Two-component normal mixtures}\\
\hline
$CN$       &     0.3120      &   0.2570      &  0.4115      & 0.2674        &  0.2825      &  0.4003     &  0.3967     \\
            &   (0.0078)     &   (0.0061)      &  (0.0062)      &(0.0076)          &  (0.0073)       &  (0.0066)      &  (0.0057)      \\
$CN^{\ast}$ &  0.1870      &   0.1598      &  0.5558      & 0.1568        & 0.1617       &   0.5508    &   0.5462    \\
            &  (0.0091)      &   (0.0011)      &  (0.0016)       &  (0.0012)        &   (0.0035)      &   (0.0047)     &  (0.0053)      \\
$SC$           &   2.2666     &   2.2217      &  2.2657      &        2.1186 &  2.1612      &   2.2855    &   2.3054    \\
            &  (0.0239)      &   (0.0300)      &  (0.0222)       &   (0.0239)       & (0.0222)        & (0.0254)       &   (0.0238)     \\
$\widehat{q_{0.5}(D)}$ & 19.14 & 19 & 14.98 & 15.34 & 22 & 14.14 & 14.04 \\
 & (0.3505) & (0) & (0.1414) & (0.4785) & (0) & (0.3505) & (0.1979) \\
 $\widehat{V(D)}$ & 9.3477 & 11.8293 & 6.4858 & 6.5009 &  13.2239 & 6.0336 & 5.8835 \\
 & (0.2949) & (0.4377) & (0.3676) & (0.2179) & (0.4273) & (0.3515) & (0.3315) \\
\hline                    
\end{tabular}                                                                  
\end{scriptsize}
\end{table}           

The results in Tab. \ref{table:1st_Experim} point out similar co-clustering accuracy across HSSMs and experiments. For the exemplification purposes we report in Fig. \ref{Fig:coclust} and \ref{Fig:coclust_r} in Appendix \ref{App:EmpRes} $P_{lk}$ and $\mathbb{I}(P_{lk}> 0.5)$ for one of the $50$ experiments. In comparison to the other HSSMs, HPYP and HDPYP have significantly small co-clustering errors, $CN$ and $CN^{\ast}$. As regard the predictive score $SC$, the seven HSSMs behave similarly in the three-component mixture experiment, whereas in the two-component experiment the HDPYP performs slightly better with respect to the other HSSMs (see also predictive densities Fig. \ref{Fig:simpredictive} in Appendix \ref{App:EmpRes}).

The posterior number of clusters for the HDPYP and the Hierarchical Gnedin processes, HGP, HGPDP and HGPYP (see Tab. \ref{table:1st_Experim}), is significantly closer to the true value ($3$ and $11$ for the first and second experiment, respectively). The HDP, HPYP and HPYDP processes tend to have extra clusters causes the posterior number of clusters to be inflated (see $\widehat{V(D)}$ in Tab \ref{table:1st_Experim} and Fig. \ref{Fig:SimNumberClust} in Appendix \ref{App:EmpRes}); conversely the HDPPY and the Hierarchical Gnedin processes have a smaller dispersion of the number of clusters.

In our set of experiments, we can conclude that using the Pitman-Yor process at some stage of the hierarchy may lead to a better accurancy. The HDPYP did reasonably well in all our experiments in line with previous findings on hierarchical Dirichlet and Pitman-Yor process for topic models (see \cite{Buntine10, Buntine14}). Also Hierarchical Gnedin processes tend to over-perform other HSSMs in estimating the posterior number of clusters.

Finally, we would not say that either model is uniformly better than the other, rather, following \cite{MillerHarrison} one should use the model that is best suited to the application. A Bayesian nonparametric analysis should include robustness checks of the results not only with respect to the hyperparameters of the chosen prior, but also with respect to the prior in the general HSSM class. Alternatively, averaging of models with different HSSM priors could be considered (see \cite{HoeMadRafVol99})
 
%

\subsection{Real Data Application}

\begin{figure}[t]
  \caption[]{(a) Co-clustering matrix for the US (bottom left block) and EU (top right block) business cycles and cross-co-clustering (main diagonal blocks) between US and EU. (b) Posterior number of clusters. Total (b.1), marginal for US (b.2) and EU (b.3) and common (b.4).}\label{Fig:Real Data}
 \centering
\setlength{\tabcolsep}{1pt} 
\begin{tabular}{cc}
\scriptsize{(a)} & \\ 
\includegraphics[trim=0cm 1cm 1.9cm 0.9cm, clip=true, width=6.8cm]{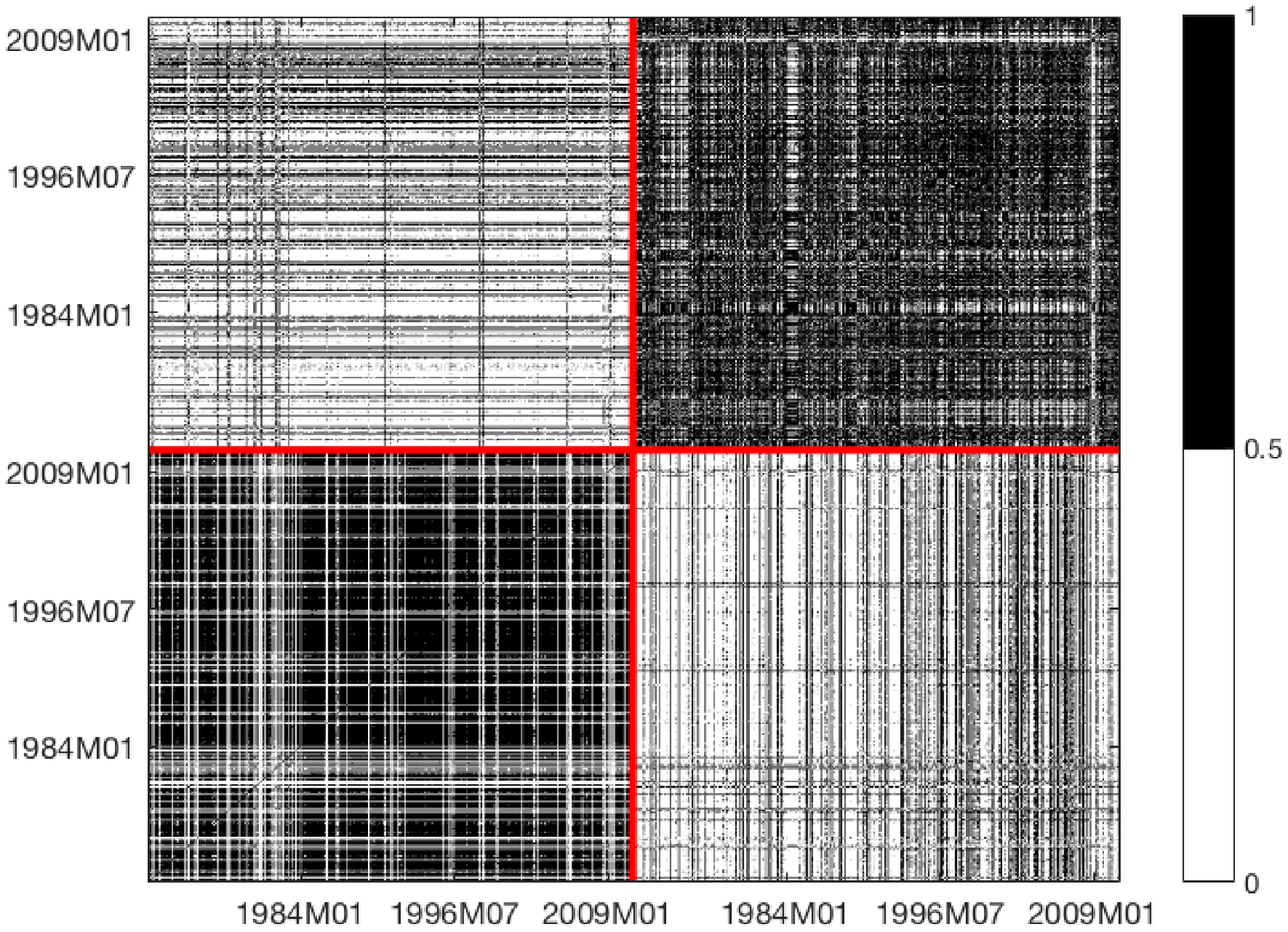}&
\begin{minipage}[]{0.5\textwidth}
\vspace{-147pt}
  \begin{tabular}{cc}
      \scriptsize{(b.1)} & \scriptsize{(b.2)}\\ 
 \includegraphics[width=2.7cm]{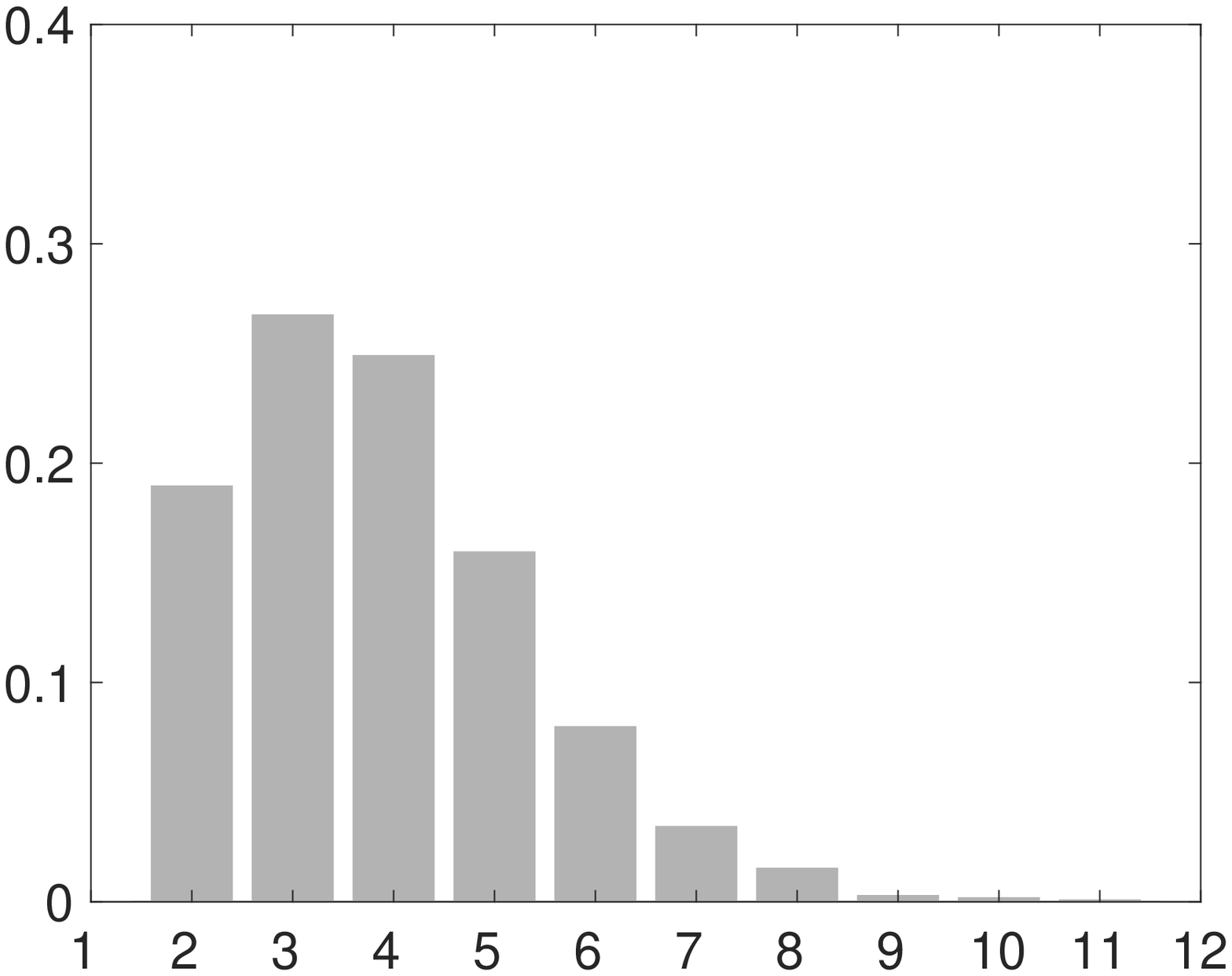}&
     \includegraphics[width=2.7cm]{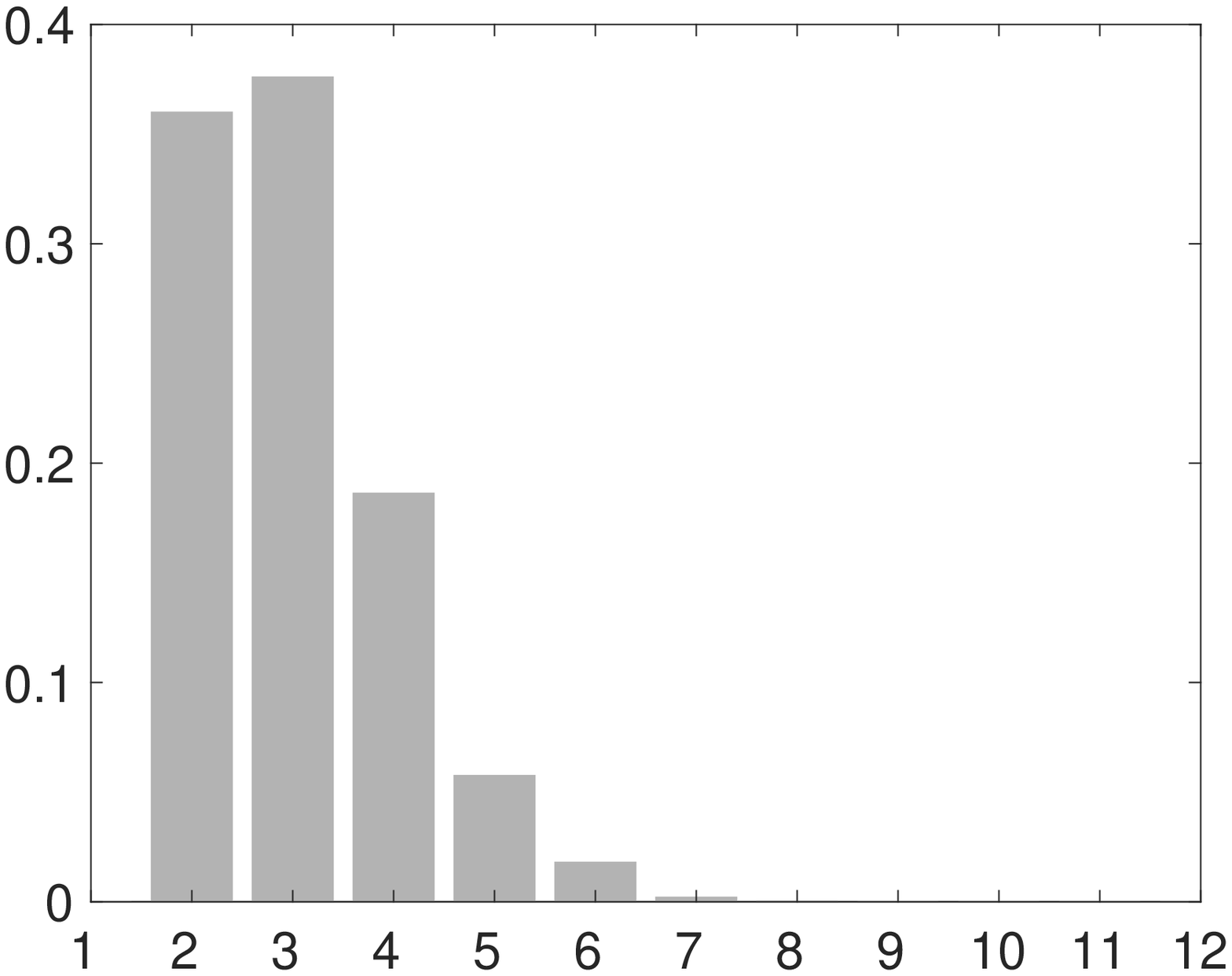}\\ \scriptsize{(b.3)}  & \scriptsize{(b.4)}\\
     \includegraphics[width=2.7cm]{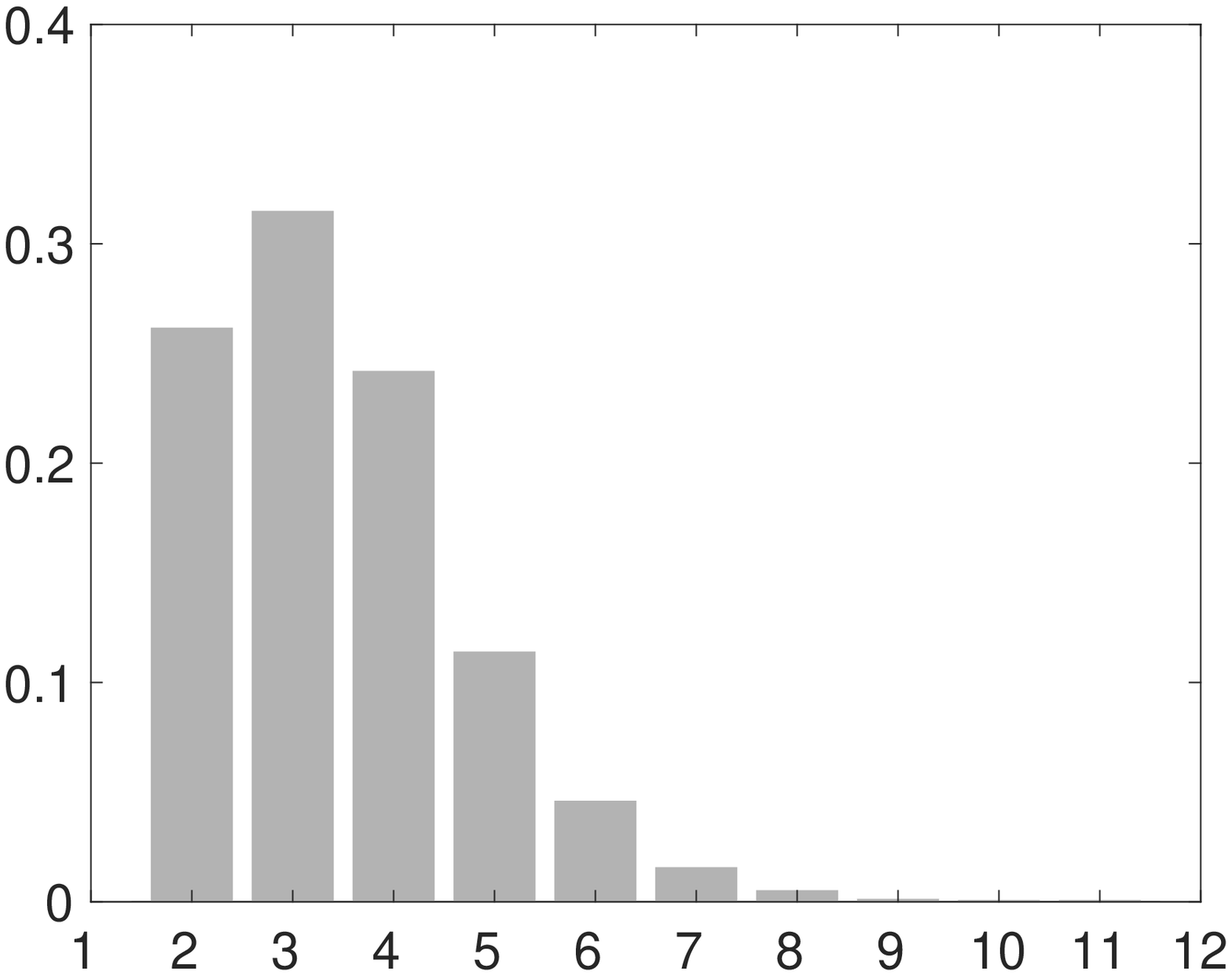}&
\includegraphics[width=2.7cm]{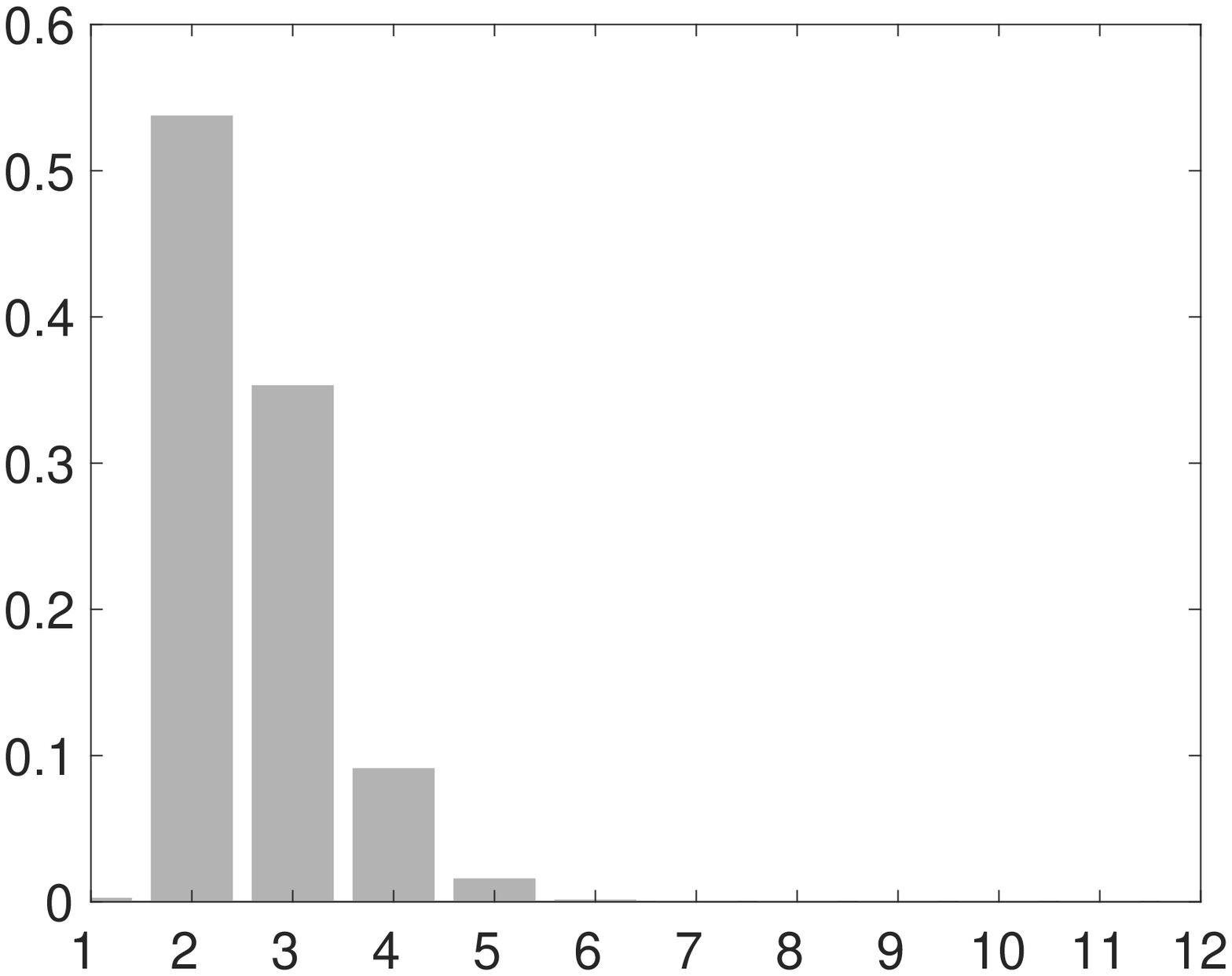}
     \\     
  \end{tabular}
\end{minipage}
\end{tabular} 
\end{figure}

The data contains the seasonally and working day adjusted industrial production indexes (IPI) at a monthly frequency from April 1971 to January 2011 for both United States (US) and European Union (EU) and has been previously analysed by \cite{Bas13}. We generate  autoregressive-filtered IPI quarterly growth rates by calculating the residuals of Vector autoregressive model of order $4$. 

We follow a Bayesian nonparametric approach based on HSSM prior for the estimation of the number of regimes or structural breaks. Based on the simulation results, we focus on the HPYP with hyperparameters, $(\theta_0, \sigma_0) = (1.2,0.2)$ and $(\theta_1,\sigma_1) = (2,0.2)$ such that the prior mean and variance of the number of clusters are $5.48$ and $23.71$, respectively. The main results of the nonparametric inference can be summarized through the implied data clustering (panel (a) of Fig. \ref{Fig:Real Data}) and the marginal, total and common posterior number of clusters (panel (b)).

One of the most striking feature of the co-clustering in Fig. \ref{Fig:Real Data} is that in the first and second block of the minor diagonal there are vertical and horizontal white lines. They correspond to observations of the two series, which belong to the same cluster and are associated with crisis periods.

Another feature that motivates the use of HSSMs are the black horizontal and vertical lines in the two main diagonal blocks. They correspond to observation from two different groups allocated to common clusters.

The appearance of the posterior total number of clusters (see panel b.1) suggests that at least three clusters should be used in a joint modelling of the US and EU business cycle. The larger dispersion of the marginal number of cluster for EU (b.3) with respect to US (b.2) confirms the evidence in \cite{Bas13} of a larger heterogeneity in the EU cycle.  Finally, we found evidence (panel b.4) two common clusters of observations between the EU and the US business cycles.

\section{Conclusions}
We propose generalized species sampling sequences as a general unified framework for constructing hierarchical random probability measures. The new class of hierarchical species sampling models (HSSM) includes some existing nonparametric priors, such as the hierarchical Dirichlet process, and other new measures such as the hierarchical Gnedin and the hierarchical mixtures of finite mixtures. In the proposed framework we derive the distribution of the marginal and total number of clusters under general assumptions for the base measure, which is useful for setting prior distribution in the applications to Bayesian nonparametric inference. Also, our assumptions allow for non-diffuse base measures, such as the spike-and-slab prior, used in sparse Bayesian nonparametric modeling. We show that HSSMs allow for the franchise Chinese restaurant representation and provide a general Gibbs sampler, which is appealing for posterior approximation in Bayesian inference.

\bibliographystyle{imsart-nameyear}
\bibliography{HDPBiblio}

\begin{thebibliography}{}

\bibitem[Aldous, 1985]{Aldous85}
Aldous, D.~J. (1985).
\newblock Exchangeability and related topics.
\newblock In {\em \'Ecole d'\'et\'e de probabilit\'es de {S}aint-{F}lour,
  {XIII}---1983}, volume 1117 of {\em Lecture Notes in Math.}, pages 1--198.
  Springer, Berlin.

\bibitem[Arratia et~al., 2003]{Arratia2003}
Arratia, R., Barbour, A.~D., and S., T. (2003).
\newblock {\em Logarithmic combinatorial structures: a probabilistic approach}.
\newblock European Mathematical Society.

\bibitem[Bacallado et~al., 2017]{Bac17}
Bacallado, S., Battiston, M., Favaro, S., and Trippa, L. (2017).
\newblock Sufficientness postulates for {G}ibbs-type priors and hierarchical
  generalizations.
\newblock {\em Statistical Science}, 32(4):487--500.

\bibitem[Bassetti et~al., 2014]{Bas13}
Bassetti, F., Casarin, R., and Leisen, F. (2014).
\newblock Beta-product dependent pitman-{Y}or processes for {B}ayesian
  inference.
\newblock {\em Journal of Econometrics}, 180(1):49--72.

\bibitem[Buntine and Mishra, 2014]{Buntine14}
Buntine, W.~L. and Mishra, S. (2014).
\newblock Experiments with non-parametric topic models.
\newblock In {\em Proceedings of the 20th ACM SIGKDD International Conference
  on Knowledge Discovery and Data Mining, KDD 2014, New York, NY, USA. ACM.},
  pages 881--890.

\bibitem[Camerlenghi et~al., 2018]{Cam18}
Camerlenghi, F., Lijoi, A., Orbanz, P., and Pruenster, I. (2018).
\newblock Distribution theory for hierarchical processes.
\newblock {\em Annals of Statistics}.

\bibitem[Canale et~al., 2017]{Canale2017}
Canale, A., Lijoi, A., Nipoti, B., and Pr{\"u}nster, I. (2017).
\newblock On the pitman--yor process with spike and slab base measure.
\newblock {\em Biometrika}, 104(3):681--697.

\bibitem[Castillo et~al., 2015]{Castillo2015}
Castillo, I., Schmidt-Hieber, J., and van~der Vaart, A. (2015).
\newblock Bayesian linear regression with sparse priors.
\newblock {\em Ann. Statist.}, 43(5):1986--2018.

\bibitem[Cerquetti, 2013]{Cerquetti}
Cerquetti, A. (2013).
\newblock Marginals of multivariate gibbs distributions with applications in
  bayesian species sampling.
\newblock {\em Electron. J. Statist.}, 7:697--716.

\bibitem[Dahl, 2006]{Dahl}
Dahl, D.~B. (2006).
\newblock Model-based clustering for expression data via a {D}irichlet process
  mixture model.
\newblock In Do, K.-A., M\"{u}ller, P.~P., and Vannucci, M., editors, {\em
  {B}ayesian {I}nference for {G}ene {E}xpression and {P}roteomics}, pages
  201--218. Cambridge University Press.

\bibitem[Daley and Vere-Jones, 2008]{DaleyVere-Jones}
Daley, D.~J. and Vere-Jones, D. (2008).
\newblock {\em An introduction to the theory of point processes. {V}ol. {II}}.
\newblock Probability and its Applications (New York). Springer, New York,
  second edition.
\newblock General theory and structure.

\bibitem[De~Blasi et~al., 2015]{Deb2017}
De~Blasi, P., Favaro, S., Lijoi, A., Mena, R.~H., Prunster, I., and Ruggiero,
  M. (2015).
\newblock Are {G}ibbs-type priors the most natural generalization of the
  {D}irichlet process?
\newblock {\em IEEE Transactions on Pattern Analysis \& Machine Intelligence},
  37(2):212--229.

\bibitem[de~Haan and Ferreira, 2006]{deHaanFerreira}
de~Haan, L. and Ferreira, A. (2006).
\newblock {\em Extreme value theory}.
\newblock Springer Series in Operations Research and Financial Engineering.
  Springer, New York.
\newblock An introduction.

\bibitem[Diaconis and Ram, 2012]{Diaconis2012}
Diaconis, P. and Ram, A. (2012).
\newblock A probabilistic interpretation of the {M}acdonald polynomials.
\newblock {\em Annals of Probability}, 40(5):1861--1896.

\bibitem[Donnelly, 1986]{Donnelly1986}
Donnelly, P. (1986).
\newblock Partition structures, {P}\'olya urns, the {E}wens sampling formula,
  and the ages of alleles.
\newblock {\em Theoret. Population Biol.}, 30(2):271--288.

\bibitem[Donnelly and Grimmett, 1993]{Donnelly1993}
Donnelly, P. and Grimmett, G. (1993).
\newblock On the asymptotic distribution of large prime factors.
\newblock {\em J. London Math. Soc. (2)}, 47(3):395--404.

\bibitem[Du et~al., 2010]{Buntine10}
Du, L., Buntine, W., and Jin, H. (2010).
\newblock A segmented topic model based on the two-parameter
  {P}oisson-{D}irichlet process.
\newblock {\em Mach. Learn.}, 81(1):5--19.

\bibitem[Dubey et~al., 2014]{Dubey14}
Dubey, A., Williamson, S., and Xing, E. (2014).
\newblock Parallel markov chain monte carlo for pitman-yor mixture models.
\newblock In {\em Uncertainty in Artificial Intelligence - Proceedings of the
  30th Conference, UAI 2014}, pages 142--151.

\bibitem[Escobar, 1994]{Escobar1994}
Escobar, M. (1994).
\newblock Estimating normal means with a {D}irichlet process prior.
\newblock {\em Journal of the American Statistical Association},
  89(425):268--277.

\bibitem[Escobar and West, 1995]{EscobarWest1995}
Escobar, M. and West, M. (1995).
\newblock Bayesian density estimation and inference using mixtures.
\newblock {\em Journal of the American Statistical Association},
  90(430):577--588.

\bibitem[Ewens, 1972]{Ewens}
Ewens, W.~J. (1972).
\newblock The sampling theory of selectively neutral alleles.
\newblock {\em Theoret. Population Biology}, 3:87--112; erratum, ibid. 3
  (1972), 240; erratum, ibid. 3 (1972), 376.

\bibitem[Favaro and Teh, 2013]{favaro2013}
Favaro, S. and Teh, Y.~W. (2013).
\newblock Mcmc for normalized random measure mixture models.
\newblock {\em Statistical Science}, 28(3):335--359.

\bibitem[George and McCulloch, 1993]{George1993}
George, E.~I. and McCulloch, R.~E. (1993).
\newblock Variable selection via gibbs sampling.
\newblock {\em Journal of the American Statistical Association},
  88(423):881--889.

\bibitem[Gnedin, 2010]{Gnedin10}
Gnedin, A. (2010).
\newblock A species sampling model with finitely many types.
\newblock {\em Electronic Communications in Probability}, 15(8):79--88.

\bibitem[Gnedin and Pitman, 2005]{GP2006}
Gnedin, A. and Pitman, J. (2005).
\newblock Exchangeable {G}ibbs partitions and {S}tirling triangles.
\newblock {\em Zap. Nauchn. Sem. S.-Peterburg. Otdel. Mat. Inst. Steklov.
  (POMI)}, 325(Teor. Predst. Din. Sist. Komb. i Algoritm. Metody. 12):83--102,
  244--245.

\bibitem[Hjort et~al., 2010]{BNP2010}
Hjort, N.~L., Homes, C., M\"uller, P., and Walker, S.~G. (2010).
\newblock {\em {B}ayesian Nonparametrics}.
\newblock Cambridge University Press.

\bibitem[Hoeting et~al., 1999]{HoeMadRafVol99}
Hoeting, J.~A., Madigan, D., Raftery, A.~E., and Volinsky, C.~T. (1999).
\newblock Bayesian {M}odel {A}veraging: {A} {T}utorial.
\newblock {\em Statistical Science}, 14(4):382--417.

\bibitem[Hoppe, 1984]{Hoppe1984}
Hoppe, F.~M. (1984).
\newblock P\'olya-like urns and the {E}wens' sampling formula.
\newblock {\em J. Math. Biol.}, 20(1):91--94.

\bibitem[James et~al., 2009]{Lijoi2009}
James, L.~F., Lijoi, A., and Pr\"unster, I. (2009).
\newblock Posterior analysis for normalized random measures with independent
  increments.
\newblock {\em Scand. J. Stat.}, 36(1):76--97.

\bibitem[Kalli et~al., 2011]{walker2011}
Kalli, M., Griffin, J.~E., and Walker, S. (2011).
\newblock Slice sampling mixture models.
\newblock {\em Statistics and Computing}, 21(1):93--105.

\bibitem[Kingman, 1967]{Kingman67}
Kingman, J. F.~C. (1967).
\newblock Completely random measures.
\newblock {\em Pacific Journal of Mathematics}, 21(1):59--78.

\bibitem[Kingman, 1978]{Kingman78}
Kingman, J. F.~C. (1978).
\newblock The representation of partition structures.
\newblock {\em J. London Math. Soc. (2)}, 18(2):374--380.

\bibitem[Kingman, 1980]{Kingman1980}
Kingman, J. F.~C. (1980).
\newblock {\em Mathematics of genetic diversity}, volume~34 of {\em CBMS-NSF
  Regional Conference Series in Applied Mathematics}.
\newblock Society for Industrial and Applied Mathematics (SIAM), Philadelphia,
  Pa.

\bibitem[Lau and Green, 2007]{pet07}
Lau, J.~W. and Green, P.~J. (2007).
\newblock Bayesian model-based clustering procedures.
\newblock {\em Journal of Computational and Graphical Statistics},
  16(3):526--558.

\bibitem[Lim et~al., 2016]{Buntine16}
Lim, K.~W., Buntine, W., Chen, C., and Du, L. (2016).
\newblock Nonparametric {B}ayesian topic modelling with the hierarchical
  {P}itman-{Y}or processes.
\newblock {\em Internat. J. Approx. Reason.}, 78(C):172--191.

\bibitem[Miller and Harrison, 2017]{MillerHarrison}
Miller, J. and Harrison, M. (2017).
\newblock Mixture models with a prior on the number of components.
\newblock {\em Journal of the American Statistical Association}.

\bibitem[M\"uller and Quintana, 2010]{pet10}
M\"uller, P. and Quintana, F. (2010).
\newblock Random partition models with regression on covariates.
\newblock {\em Journal of Statistical Planning and Inference},
  140(10):2801--2808.

\bibitem[Navarro et~al., 2006]{NAVARRO2006101}
Navarro, D.~J., Griffiths, T.~L., Steyvers, M., and Lee, M.~D. (2006).
\newblock Modeling individual differences using dirichlet processes.
\newblock {\em Journal of Mathematical Psychology}, 50(2):101 -- 122.

\bibitem[Neal, 2000]{Neal2000}
Neal, R. (2000).
\newblock Markov {C}hain sampling methods for {D}irichlet process mixture
  models.
\newblock {\em Journal of Computational and Graphical Statistics},
  9(2):249--265.

\bibitem[Nguyen, 2016]{nguyen2016}
Nguyen, X. (2016).
\newblock Borrowing strengh in hierarchical bayes: Posterior concentration of
  the dirichlet base measure.
\newblock {\em Bernoulli}, 22(3):1535--1571.

\bibitem[Papaspiliopoulos and Roberts, 2008]{Pap08}
Papaspiliopoulos, O. and Roberts, G.~O. (2008).
\newblock Retrospective markov chain monte carlo methods for dirichlet process
  hierarchical models.
\newblock {\em Biometrika}, 95(1):169--186.

\bibitem[Pitman, 1995]{Pit95}
Pitman, J. (1995).
\newblock Exchangeable and partially exchangeable random partitions.
\newblock {\em Probab. Th. Rel. Fields}, 102(2):145--158.

\bibitem[Pitman, 1996]{Pitman96}
Pitman, J. (1996).
\newblock Some developments of the {B}lackwell-{M}ac{Q}ueen urn scheme.
\newblock In {\em Statistics, probability and game theory}, volume~30 of {\em
  IMS Lecture Notes Monogr. Ser.}, pages 245--267. Inst. Math. Statist.,
  Hayward, CA.

\bibitem[Pitman, 2003]{Pitman2003}
Pitman, J. (2003).
\newblock Poisson-{K}ingman partitions.
\newblock In {\em Statistics and science: a {F}estschrift for {T}erry {S}peed},
  volume~40 of {\em IMS Lecture Notes Monogr. Ser.}, pages 1--34. Inst. Math.
  Statist., Beachwood, OH.

\bibitem[Pitman, 2006]{Pit06}
Pitman, J. (2006).
\newblock {\em Combinatorial Stochastic Processes}, volume 1875.
\newblock Springer-Verlag.

\bibitem[Pitman and Yor, 1997]{Pit97}
Pitman, J. and Yor, M. (1997).
\newblock The two-parameter {P}oisson-{D}irichlet distribution derived from a
  stable subordinator.
\newblock {\em The Annals of Probability}, 25(2):855--900.

\bibitem[Prunster and Ruggiero, 2013]{Prunster13}
Prunster, I. and Ruggiero, M. (2013).
\newblock A bayesian nonparametric approach to modeling market share dynamics.
\newblock {\em Bernoulli}, 19(1):64--92.

\bibitem[Regazzini et~al., 2003]{RLP}
Regazzini, E., Lijoi, A., and Pr\"unster, I. (2003).
\newblock Distributional results for means of normalized random measures with
  independent increments.
\newblock {\em Annals of Statistics}, 31(2):560--585.

\bibitem[Rockova, 2018]{Rockova2018}
Rockova, V. (2018).
\newblock Bayesian estimation of sparse signals with a continuous
  spike-and-slab prior.
\newblock {\em Annals of Statistics}.

\bibitem[Rockova and George, 2017]{George2017}
Rockova, V. and George, E.~I. (2017).
\newblock The {S}pike-and-{S}lab {LASSO}.
\newblock {\em Journal of the American Statistical Association}.

\bibitem[Sangalli, 2006]{Sangalli2007}
Sangalli, L.~M. (2006).
\newblock Some developments of the normalized random measures with independent
  increments.
\newblock {\em Sankhy\=a}, 68(3):461--487.

\bibitem[Sohn and Xing, 2009]{sohn2009}
Sohn, K.-A. and Xing, E.~P. (2009).
\newblock A hierarchical dirichlet process mixture model for haplotype
  reconstruction from multi-population data.
\newblock {\em The Annals of Applied Statistics}, 3(2):791--821.

\bibitem[Suarez and Ghosal, 2016]{suarez2016}
Suarez, A.~J. and Ghosal, S. (2016).
\newblock Bayesian clustering of functional data using local features.
\newblock {\em Bayesian Analysis}, 11(1):71--98.

\bibitem[Teh and Jordan, 2010]{Teh10}
Teh, Y. and Jordan, M.~I. (2010).
\newblock Hierarchical {B}ayesian nonparametric models with applications.
\newblock In Hjort, N.~L., Holmes, C., M\"uller, P., and Walker, S., editors,
  {\em {B}ayesian Nonparametrics}. Cambridge University Press.

\bibitem[Teh, 2006]{Teh2006}
Teh, Y.~W. (2006).
\newblock A hierarchical bayesian language model based on pitman-yor processes.
\newblock In {\em Proceedings of the 21st International Conference on
  Computational Linguistics and the 44th Annual Meeting of the Association for
  Computational Linguistics}, ACL-44, pages 985--992, Stroudsburg, PA, USA.
  Association for Computational Linguistics.

\bibitem[Teh et~al., 2006]{TehJordan2006}
Teh, Y.~W., Jordan, M.~I., Beal, M.~J., and Blei, D.~M. (2006).
\newblock Hierarchical {D}irichlet processes.
\newblock {\em Journal of the American Statistical Association},
  101(476):1566--1581.

\bibitem[Walker, 2007]{Wal07}
Walker, S.~G. (2007).
\newblock Sampling the dirichlet mixture model with slices.
\newblock {\em Communications in Statistics - Simulation and Computation},
  36(1):45--54.

\bibitem[Wood et~al., 2009]{WooArcGas2009a}
Wood, F., Archambeau, C., Gasthaus, J., James, L.~F., and Teh, Y.~W. (2009).
\newblock A stochastic memoizer for sequence data.
\newblock In {\em International Conference on Machine Learning (ICML)},
  volume~26, pages 1129--1136.

\end{thebibliography}

\vfill

\pagebreak

\newpage

\appendix

\section{Homogeneous Normalized Random Measures}\label{App:NRM}

A {\it completely additive random measure} on a Polish space $\XX$ 
is  a random measure  $\tm$  such that, for any
measurable collection $\{A_1,\dots,A_k \}$ ($k \geq 1$) of pairwise
disjoint measurable subsets of $\XX$, the random variables
$\tm(A_1),\dots, \tm(A_k)$ are stochastically independent. 
Under very general assumption ($\Sigma$-boundeness), 
such random measures can be written as the sum of three 
independent random measures: a deterministic measure, an atomic random measure 
with fixed atoms $\sum_{i \geq 1} U_i \delta_{x_i}$  (where 
the points $x_1,x_2,\dots$ are fixed in $\XX$ and $U_i$ are independent positive random variables) 
and the  ordinary component, i.e. a discrete random measure $\tm_O$ without fixed atoms. This last measure 
can be express as an integral of a Poisson random measure $N$ on $\XX \times \RE^+$, 
more precisely as  $\tm_O(A)=\int_{A \times \RE^+} y N(dx dy)$ (see \cite{Kingman67}).

To define the Normalized Random Measures, we
consider the sub-class of completely random measures, without deterministic component, characterized 
by the Laplace functional 
\begin{equation}\label{zero}
\E(e^{-\lm \tm(A)})=\exp \left \{-\int_{A \times \RE^+}(1-e^ {-\lm
y}) \nu(dx dy) \right \} \qquad (\lm > 0, \,\, A \in \CB(\XX)),
\end{equation}
where $\CB(\XX)$ is the Borel $\s$--field on $\XX$ and 
$\nu$ a $\sigma$-finte measure   on $\XX
\times \RE^+$. In particular, if $\nu(\{x\} \times \RE^+)=0$ for every $x \in \XX$, then 
$\tm$ is the most general form of the ordinary component of a completely random measure,
see Chapter 10.1 in \cite{DaleyVere-Jones}. If we allow $\nu$ to have atomic components in $\XX$, i.e. if $\nu(\{x\} \times \RE^+)>0$
for some $x \in D \subset \XX$ ($D$ countable), it is not difficult to prove that  $\tm=\tm_{FA}+\tm_O$
for two independent completely random measures, $\tm_{FA}$ with fixed atoms on $D$ and $\tm_O$ ordinary
with Levy measures $\nu(dxdy)-\sum_{x \in D} \nu(\{x\} dy)$.
Note that it is not in contradiction with  Kingman's result as stated for instance in Theorem 10.1.III of  \cite{DaleyVere-Jones}, 
since in that statement  the nonatomicity of $\nu$ is assumed only in order 
to ensure  uniqueness in the representation. 
With few exceptions,  fixed
atoms are in general  ignored in the Bayesian nonparametrics literature, but for our proposes 
it is fundamental to assume possible atoms  (at least of the previous particular type).  

To go further, we require the following 
two regularity conditions:
\begin{equation}\label{CMRr1}
\int_{\XX \times \RE^+} (1-e^ {-\lm y}) \nu(dx dy) < +\infty, \quad \forall \lm>0,
\end{equation}
which entails that $\tm(\XX)<+\infty$ a.s., while the second condition
\begin{equation}\label{CMRr2}
 \nu(\XX \times \RE^+) = +\infty
 \end{equation}
entails that $\tm(\XX)>0$ a.s. (see \cite{RLP}).
Under these regularity conditions, following \cite{RLP}, 
 one  can define a {\it normalized completely random measure} (or  equivalently a {\it normalized random measure with independent increments},  NRMI) setting
\[
p(\cdot):=\frac{\tm(\cdot)}{\tm(\XX)}.
\]
The so--called {\it normalized homogeneous random measure}  of parameter 
$(\theta,\rho,H)$, $NRMI(\theta,\eta,H)$,
 is obtained for the special case 
\begin{equation}\label{fact}
\n(dx dy)=\theta H(dx) \eta(dy),
\end{equation}
 $\theta$ being a positive number, $H$ a probability measure on $\XX$ and
$ \eta$ a measure on $\RE^+$ such that $\eta(\RE^+)=+\infty$ and $\int_{\RE^+} (1-e^ {-\lm y}) 
\eta(dy) < +\infty $ for every positive $\lm$. 

The most classical example of NRMI is the Dirichlet
process, characterized by $\eta(dv)=v^{-1} e^{-v} dv$ and the connection between NRMI and $SSrp$ is clarified in the following proposition. 

\begin{proposition}\label{prop:hNRM}
Let $p$ be a normalized homogeneous random measure of parameter $(\theta,\eta,H)$,
where $\eta$ is absolutely continuous with respect to the Lebesgue measure
and $H$ any probability measures on $\XX$, 
 then 
$p$ is a $SSrp(\FP,H)$ for  $\FP$ specified by 
\eqref{PKEPPF}-\eqref{PKregularity}.
\end{proposition}

\begin{proof} We start by proving that 
$p$ can be represented as a $SSrp$. 
Let $\tilde N$ be a inhomogeneous Poisson point process on $\RE^+$ 
with $\sigma$-finte L\'evy measure $\tilde \eta(y) dy=\theta \eta(y)dy$.  
Let $(Z_j)_j$ be a sequence of i.i.d. random variables on $\XX$ with distribution $H$, independent of $\tilde N$. 
Denote by  $(J_1,J_2,\dots)$ the jumps of the Poisson process $\tilde N$, i.e. 
$\tilde N(dy)=\sum_j \delta_{J_j}(dy)$, and set  
\[
\mu(dx)= \sum_j \delta_{Z_j} J_j.
\]
Using independence of $(Z_j)_j$ and $\tilde N$, one easily shows that
\[
\E[e^{-\lambda \mu(A)}]
=\exp \left \{-\int_{A \times \RE^+}(1-e^ {-\lm
y}) \tilde \eta(y)dy H(dx) \right \}, \qquad \lm > 0, \,\, A \in \CB(\XX).
\]
Hence 
\[
p(dx)= \frac{\mu(dx)}{\mu(\XX)}=\sum_j \delta_{Z_j}  q_j,
\]
with $q_j=J_j/(\sum_k J_k)$,
is a $NRMI(\theta,\eta,H)$. The thesis follows since, by results in \cite{Pitman2003},
the EPPF of the random partition derived by sampling from  $(q_j)_j$ is specified by 
\eqref{PKEPPF}-\eqref{PKregularity}.
\end{proof}

\begin{remark}
If $H$ has atoms, then 
$\FP$ is not the EPPF induced by a sequence of exchangeable random variables sampled from $p$. Moreover, $p$ can not be derived by normalization of an ordinary completely random measure, since it 
has fixed atoms. 
\end{remark}

\section{Proofs of the results in the paper}\label{App:proofs}

\subsection{Proofs of the results in Section \ref{Sec_Back}}

\begin{proof}[Proof of Proposition  \ref{prop_gsss}]
We assume without loss of generality that 
$ \tilde q= \sum_{j \geq 1} \delta_{Z_j} \tilde q_j^{\downarrow}$. 
Given the Borel sets, $A_1,\dots,A_n$, and the integers numbers $i_1,\dots,i_n$, then we have 
\[
\P\left\{\xi_1 \in A_1, \dots, \xi_n \in A_n, I_1=i_1,\dots,I_n =i_n \middle \lvert  \tilde q,  \left(\tilde q_j^{\downarrow}\right)_n, \left(Z_n\right)_n \right\}=
\prod_{j=1}^n \delta_{Z_{i_j}} (A_j) \tilde q_{i_j}^{\downarrow},
\]
and by marginalising,  
\[
\P\left\{\xi_1 \in A_1, \dots, \xi_n \in A_n \middle \lvert  \tilde q,  \left(\tilde q_j^{\downarrow}\right)_n, (Z_n)_n \right\}=
\sum_{i_1 \geq 1, \dots, i_n \geq 1} \prod_{j=1}^n \delta_{Z_{i_j}} (A_j) \tilde q_{i_j}^{\downarrow}
=\prod_{j=1}^n \tilde q (A_j).
\]
Hence, given $A_1,\dots,A_n$, 
\[
\P\left\{\xi_1 \in A_1, \dots, \xi_n \in A_n |  \tilde q\right\}=\prod_{j=1}^n \tilde q (A_j)
\]
almost surely. Since $\XX$ is Polish, we prove (i). 
Let us denote by $\pi(I_1,\dots,I_n)$ the partition induced by $I_1,\dots,I_n$ and by  Kingman's correspondence
its law is characterised by the EPPF $\FP$. Hence the law of  $\pi(I_1,\dots,I_n)$ is the same as the law of $\Pi_n$. 
If $i$ and $j$ belong to the same block of  $\pi(I_1,\dots,I_n)$, it follows that $I_i=I_j=k$ for some $k$ and 
  then $\xi_i=\xi_j=Z_k$. In particular, using the independence of the $Z_k$s, we can write 
\[
\P\left\{\xi_1 \in A_1, \dots, \xi_n \in A_n \right\}= \sum_{\pi_n \in \CP_n} \P\left\{ \pi(I_1,\dots,I_n)=\pi_n\right\} \prod_{c=1}^{|\pi_n|} H( \cap_{j \in \pi_{c,n}} A_j ). 
\]
Since $ \P\{ \pi(I_1,\dots,I_n)=\pi_n\} = \P\{ \Pi_n=\pi_n\}=\EPk(\pi_n)$, then (iii) follows immediately. 
Finally, we immediately check that 
$\P\left\{ \xi_1' \in A_1, \cdots, \xi_n' \in A_n\right\}= \sum_{\pi_n \in \CP_n} \P\{ \Pi_n=\pi_n\} \prod_{c=1}^{|\pi_n|} H( \cap_{j \in \pi_{c,n}} A_j )$, and hence, by (iii), we obtain
(ii).
\end{proof}

\begin{proof}[Proof of Corollary \ref{prop_gsssTRIS}]
By Proposition \ref{prop_gsss},  the probability of the event $\{ |\tilde \Pi_n|=d\}$
is equal to the probability of observing $d$ distinct values in $(\xi_1',\dots,\xi_n')$.
Since $(\xi_1',\dots,\xi_n')=(Z_{\SC_1(\Pi)},\dots,Z_{\SC_n(\Pi)})$, the conditional 
probability of observing $d$ distinct values  given the event $ |\Pi_n|=k$ is
zero if $k<d$ and equals to the the probability of observing exactly $d$ distinct values 
in the vector $(Z_1,\dots,Z_k)$  if $d\leq k \leq n$. This shows that 
$\P\{ |\tilde \Pi_n|=d| |\Pi_n|=k \}=H^*(d | k )$ and proves point (i). 
Point (ii) follows by specialising point (i) for $H(dx)=a \delta_{x_0}(dx)+ (1-a) \tilde H(dx)$.
In this case, a simple computation shows that  
\[
H^*(d|k)={k \choose d-1} a^{k+1-d}(1-a)^{d-1}+\J\{d=k\}(1-a)^d.
\]
\end{proof}

\begin{proof}[Proof of Corollary \ref{prop_gsssBIS}]
For   $c=1,\dots,|\Pi_n|$ by construction   $Z_c=\xi'_{R(n,c)}$ 
with $R(n,c)=\min\{ j : j \in \Pi_{c,n}$. Then 
\[
\begin{split}
\P\left\{\xi_{n+1}'  \in dx | \xi_1',\dots,\xi_n', \Pi_n \right\} & = 
\sum_{c=1}^{|\Pi_n|}P\{  (n+1) \in \Pi_{n+1,c} | \Pi_n \} 
\delta_{Z_c}(dx)  \\
& \quad+ \P\{(n+1) \in  \Pi_{n+1,|\Pi|+1} | \Pi_n \} 
\P\{ Z_{|\Pi_n|+1} \in dx\}
\\
& = \sum_{c=1}^{|\Pi_n|}
\omega_{n,c}(\Pi_n) \delta_{Z_c}(dx) +  
\nu_{n}(\Pi_n)   H(dx).
\\
\end{split}
\]
\end{proof}

\subsection{Proofs of the results in Section \ref{Sec_HSSM}}

\begin{proof}[Proof of Proposition \ref{prop_hssm}]
Since $p_1,\dots,p_I$ are conditionally independent given $p_0$, then we can write
\[
\P\left\{ \xi_{i,j} \in A_{i,j}  \,\,  \text{for  $i=1,\dots,I, j=1,\dots,n_{i }$} \right\}=\E \left[ 
 \prod_{i=1}^I   \E  \left[   \prod_{j=1}^{n_i} p_i(A_{i,j})  \middle \lvert p_0 \right]  \right].
\]
Given $p_0$,  then $\E  \left[   \prod_{j=1}^{n_i} p_i(A_{i,j}) \middle \lvert p_0 \right]$ 
is the probability that the first $n_i$ observations of  a $gSSM(\FP,p_0)$ take values in 
$A_{i1} \times \dots \times A_{in_i}$, hence by point (iii) of  Proposition \ref{prop_gsss}, we can write 
\[
\E  \left[   \prod_{j=1}^{n_i} p_i(A_{i,j})  \middle \lvert p_0 \right]=\sum_{\pi^{(i)}\in \CP_{n_i} }
 \P\left\{ \Pi^{(i)}_{n_{i}}=\pi^{(i)} \right\}    \prod_{c=1}^{|\pi^{(i)}|}  p_0 \left( \cap_{j \in \pi^{(i)}_c} A_{i,j} \right), 
\]
where $\Pi^{(i)}$ has EPPF $\FP$ and we can assume that the $\Pi^{(i)}$  are independent on all the other random elements. 
Taking the product and then the expectation,  we get 
\[
\begin{split}
\P\{ \xi_{i,j} \in A_{i,j}  & \,\,  \text{for  $i=1,\dots,I, j=1,\dots,n_{i }$} \} 
\\
&=  \E \left[  \prod_{i=1}^I 
\sum_{\pi^{(i)}\in \CP_{n_i} } 
 \P\left\{ \Pi^{(i)}_{n_{i}}=\pi^{(i)}\right \}    \prod_{c=1}^{|\pi^{(i))}|}  p_0 	\left( \cap_{j \in \pi^{(i)}_c} A_{i,j} \right)  \right] \\
 &=  \sum_{\pi^{(1)}\in \CP_{n_1}  ,\dots,\pi^{(I)} \in \CP_{n_{I}}} 
 \prod_{i=1}^I  \P\left\{ \Pi^{(i)}_{n_{i}}=\pi^{(i)} \right\}     \E \left[ 
   \prod_{i=1}^I \prod_{c=1}^{|\pi^{(i))}|}  p_0 \left( \cap_{j \in \pi^{(i)}_c} A_{i,j} \right)  \right] \\
   &=  \sum_{\pi^{(1)}\in \CP_{n_1}  ,\dots,\pi^{(I)} \in \CP_{n_{I}}} 
 \prod_{i=1}^I   \FP  \left(\pi^{(i)} \right)     \E \left[ 
   \prod_{i=1}^I \prod_{c=1}^{|\pi^{(i))}|}  p_0 \left( \cap_{j \in \pi^{(i)}_c} A_{i,j} \right)  \right] \\
 \end{split}
\]
with $p_0  \sim SSrp(\FP_0,H_0)$, that concludes the proof. 
\end{proof}

\begin{proof}[Proof of Proposition \ref{prop0}]
 Consider  an array of i.i.d. random variables 
 $[\zeta_{i,j}]_{i=1,\dots,I, j \geq 1}$ with common distribution $q$. It follows immediately that 
 \[
 \xi_{i,j}= \zeta_{i,\SC_{j}(\Pi^{(i)})}
 \]
is a partially exchangeable array. In this case, the row $[\xi_{i,j}]_{j \geq 1}$ turns out to be independent
and each row is an exchangeable sequence. 
Since mixtures of partially exchangeable random variables are still partially exchangeable, we get the first part of the proof. The second part follows easily. 
\end{proof}

The proof of Proposition \ref{prop1} 
is a consequence of  the next simple result. 
Given $\pi^{(i)}_{n_i} \in \CP_{n_{i}}$ for $i=1,\ldots,I$, let
\[
C\left(\pi^{(1)}_{n_1},\dots,\pi^{(I)}_{n_I}\right) :=\left\{ (i,c) : i =1,\dots,I; c=1, \dots |\pi^{(i)}_{n_{i}}|\right\}
\]
and fix a bijection $\SD:C\left(\pi^{(1)}_{n_1},\dots,\pi^{(I)}_{n_I}\right) \to \left\{1,\dots,\sum_i |\pi^{(i)}_{n_i}|\right\}$, 
e.g., $\SD(i,c)=\sum_{i'=1}^{i-1} |\pi^{(i')}_{n_{i'}}| + c$. Note that clearly $\SD$ depends on $C\left(\pi^{(1)}_{n_1},\dots,\pi^{(I)}_{n_I}\right)$ although 
we do not write it explicitly. 
 
\begin{lemma}\label{prop2} Under the same assumptions of Proposition \ref{prop1} , let $(\zeta_n)_n$ be a sequence of exchangeable random variables with directing random measure 
$\tilde p \sim SSrp(\FP_0,H_0)$ independent of all the others random elements. 
Then 
the law of $\CO$ is the same as the law of 
\[
\left\{ \zeta_{ \SD(i,\SC_j(\Pi^{(i)})) }; i=1, \dots I, j=1, \dots, n_{i}\right\}.
\]
\end{lemma}

\begin{proof}
Given 
 $\pi_1,\dots,\pi_I$ and $\SD$ as above, we have
\[
\begin{split}
\E \Big [  \prod_{i=1}^I  & \prod_{c=1}^{|\pi^{(i))}|} \tilde q\left( \cap_{j \in \pi^{(i)}_{c}} A_{i,j} \right)  \Big ] \\
& =\P\left\{ \z_{\SD(i,c)}  \in \cap_{j \in \pi^{(i)}_{c} } A_{i,j}: \,\,  i=1,\dots,I,   c=1, \dots, |\pi^{(i)}|  \right\}  \\
& =\P\left\{ \z_{\SD(i,\C_j(\pi^{(i)}))}  \in A_{i,j}: \,\,  i=1,\dots,I,   j=1, \dots, n_{i}  \right\}.
\end{split}
\]
Hence
\[
\begin{split}
\P\{   \z_{\SD(i,\C_j(\Pi^{(i)}))} &  \in A_{i,j}: \,\,  i=1,\dots,I,   j=1, \dots, n_{i}  \}  \\
 & =
 \sum_{\pi^{(1)}\in \CP_{n_1}  ,\dots,\pi^{(I)} \in \CP_{n_{I}}} 
 \prod_{i=1}^I   \FP  \left(\pi^{(i)} \right)     \E \left[ 
   \prod_{i=1}^I \prod_{c=1}^{|\pi^{(i))}|}  \tilde p \left( \cap_{j \in \pi^{(i)}_c} A_{i,j} \right)  \right]. 
   \\
   \end{split}
\]
\end{proof}

\begin{proof}[Proof of Proposition \ref{prop1}]
The thesis follows by combining Lemma \ref{prop2} with 
Proposition \ref{prop_gsss}.  
Indeed, by (ii) of Proposition  \ref{prop1}, one can take in 
 Lemma \ref{prop2} 
\[
\zeta_n=\phi_{\SC_n \left(\Pi^{(0)}\right)},
\]
getting 
\[
\zeta_{\SD(i,\SC_j(\Pi^{(i)}))}=\phi_{\SC_{\SD(i,\SC_j(\Pi^{(i)}))} \left(\Pi^{(0)}\right)}, 
\]
which proves the thesis. 
%

\end{proof}

\begin{proof}[Proof of Proposition \ref{prop_part}]
Following Proposition \ref{prop1}, we can assume that $\CO$ is described by \eqref{h0}.
Hence, using
 \eqref{h0} and the fact that $H$ is non-atomic, we can express the event 
$\Pi^*=\pi^*$  as the union of  disjoint  events of the form 
$\{\Pi^{(0)}=\pi^{(0)},\Pi^{(1)}=\pi^{(1)},\dots,\Pi^{(I)}=\pi^{(I)}\}$, 
where $(\pi^{(1)},\dots,\pi^{(I)})$ run over all 
the possible partitions compatible with $\pi^*$. 

Note that, given $\Pi^{(1)}=\pi^{(1)},\dots,\Pi^{(I)}=\pi^{(I)}$
with $(\pi^{(1)},\dots,\pi^{(I)})$ compatible with $\pi^*$, we have that necessarily 
(on the event $\Pi^*=\pi^*$) $\Pi^{(0)}=\pi^{(0)}$ for a partition $\pi^{(0)}$ as function of
$\pi^*$ and  $(\pi^{(1)},\dots,\pi^{(I)})$. In what follow we shall write this partition as
$\pi^{(0)}(\pi^*,\pi^{(1)},\dots,\pi^{(I)})$.

For example, if $I=2$, $n_1=4$, $n_2=3$ and $D=3$ with
$\phi_1=\xi_{1,1}=\xi_{1,4}=\xi_{2,2}$, $\phi_2=\xi_{1,2}=\xi_{1,3}$
and $\phi_3=\xi_{2,1}=\xi_{2,3}$, we have 
$\pi^*_{1,1}=[1,4]$, $\pi^*_{1,2}=[2,3]$, $\pi^*_{1,3}=\emptyset$,
$\pi^*_{2,1}=[2]$, $\pi^*_{2,2}=\emptyset$, $\pi^*_{2,3}=[1,3]$. In this case
$\pi^{(1)}$ can be one of the partitions: $[(1),(2),(3),(4)]$, 
$[(1,4),(2),(3)]$, $[(1,4),(2,3)]$, $[(1),(2,3),(4)]$ and, analogously, 
$\pi^{(2)}$ can be $[(1),(2),(3)]$ or $[(1,3),(2)]$.  Finally, 
assuming for instance that 
$\pi^{(1)}=[(1,4),(2),(3)]$  and that $\pi^{(2)}=[(1,3),(2)]$ we have
that necessarily $\pi^{(0)}=[(1,5),(2,3),(4)]$. 

In general, given $i=1,\dots,I$, the subset  of partitions in $\CP_{n_{i}}$ that are compatible 
with $[\pi^*_{i,1},\dots,\pi^*_{i,D}]$ is 
\[
\cup_{ \bmm_i \in \CM[{\bn_i}] } \cup_{  \blam_i  \in \Lambda(\bmm_i) } \CP(i,\pi^*;\blam_i),
\]
 where $ \CP(i,\pi^*;\blam_i)$ are partitions $\pi$ in $ \CP_{n_{i}}$ with $\sum_{d=1}^D m_{i d}$ blocks
such that for every $d=1,\dots,D$: (i)  there are $m_{id}$ blocks in $\pi$ with cardinality $\ell_{i c d}$ ($c=1,\dots,m_{i,d}$), such that
 $ \sum_{c=1}^{m_{id}} \ell_{i c d} = n_{i  d}$ for every $d$ and 
$\#\{ c: \ell_{i c d} =j \}= \lambda_{idj} $ for every $d$ and $j$;
(ii) the union of these blocks coincides with  $\pi^*_{i,d}$. 
Using the well-known fact that 
the number $N(\ell_1,\dots,\ell_n)$ of partitions of $[n]$ with 
$\ell_j$ blocks of cardinality $j=1,\dots,n$ 
can be written as
\[
N(\ell_1,\dots,\ell_n)=\frac{n!}{\prod_{j=1}^n ( \ell_j ! )(j!)^{\ell_j} },
\]
see equation  (11) in \cite{Pit95}, 
it is easy to see that 
\begin{equation}\label{cardPi}
| \CP(i,\pi^*;\blam_i) |=  \prod_{d=1}^D  \frac{n_{i d}!}{\prod_{j=1}^{n_{i d }} \lambda_{idj}! (j!)^{\lambda_{idj}}    }.
\end{equation} 
Setting
\[
A_{\pi^*}:=\cup_{ \bmm \in \CM[{\bn}] } \cup_{  \blam  \in \Lambda(\bmm) }   \,\, \CP(1,\pi^*;\blam_1) \times \cdots \times  \CP(I,\pi^*;\blam_I)
\]
we can write
\[
\left\{ \Pi^*=\pi^*\right\} = \cup_{(\pi^{(1)},\dots,\pi^{(I)} )\in A_{\pi^*} } 
\left\{ \Pi^{(1)}= \pi^{(1)},\dots,\Pi^{(I)}= \pi^{(I)}, \Pi^{(0)}=   
  \pi^{(0)}\left(\pi^*,\pi^{(1)},\dots,\pi^{(I)}\right)\right\}.
\]
Now,  given $\bmm \in \CM[{\bn}]$, $ \blam  \in \Lambda(\bmm)$  and
$(\pi^{(1)},\dots,\pi^{(I)} ) \in \CP(1,\pi^*;\blam_1) \times \cdots \times  \CP(I,\pi^*;\blam_I) $,
since $\Pi^{(0)},\Pi^{(1)},\dots,\Pi^{(I)}$ are independent exchangeable random partitions, we have that  
\[
\begin{split}
\P\{ \Pi^{(1)} & = \pi^{(1)},\dots,\Pi^{(I)}= \pi^{(I)}, \Pi^{(0)} =   
  \pi^{(0)}(\pi^*,\pi^{(1)},\dots,\pi^{(I)} )\} \\
&   = 
  \FP_0(m_{\cdot 1},\dots,m_{\cdot D}) 
 \prod_{i=1}^I 
  \FP [[ \blam_i ]] .
\\
  \end{split}
\]
Combining it with \eqref{cardPi}, 
we have finally that
\[
\begin{split}
 \P\{ \Pi^*=\pi^*\} 
& =\sum_{ \substack{  \bmm \in \CM[{\bn}] \\ \blam  \in \Lambda(\bmm)} } 
\sum_{*}
\P\left\{ \Pi^{(1)}  = \pi^{(1)},\dots,\Pi^{(I)}= \pi^{(I)}, \Pi^{(0)} =   
  \pi^{(0)}\left(\pi^*,\pi^{(1)},\dots,\pi^{(I)} \right)\right\} \\
  & = \sum_{ \bmm \in \CM[{\bn}] } 
 \FP_0 (m_{\cdot 1},\dots,m_{\cdot D}) 
\sum_{  \blam  \in \Lambda(\bmm) }
 \prod_{i=1}^I 
  \FP  [[ \blam_i ]] \prod_{d=1}^D  \frac{n_{i d}!}{\prod_{j=1}^{n_{i d}} \lambda_{idj}! (j!)^{\lambda_{idj}}    }
  \\
  \end{split}
\]
where $*= \left\{ (\pi^{(1)},\dots,\pi^{(I)} )  \in \CP(1,\pi^*;\blam_1) \times \cdots \times  \CP(I,\pi^*;\blam_I)\right\}$.
\end{proof}

\subsection{Proofs of the results in Section \ref{Sec_asympt}}

\begin{proof}[Proof of Proposition \ref{Prop:clusterdistribution}]
It is prompt to see that 
\[
\begin{split}
\P\left\{ D_{i,t}=k\right\} & =\sum_{m=k}^{n_i(t)} \P\left\{   |\Pi_m^{(0)}| =k  \middle \lvert  
|\Pi^{(i)}_{n_i(t)}|  =m \right\} \P\left\{  |\Pi^{(i)}_{n_i(t)}|=m \right\}  \\
& =\sum_{m=k}^{n_i(t)} \P\left\{   |\Pi_m^{(0)}| =k \right\} \P\left\{  |\Pi^{(i)}_{n_i(t)}|=m\right \},  \\
\end{split}
\]
where we use the independence of $\Pi^{(i)}$ and $\Pi^{(0)}$. Moreover,
\[
\E\left[ D_{i,t}^r\right]=  
 \sum_{m =1}^{n_i(t)} \E\left[   \left \lvert\Pi_m^{(0)}\right \lvert^r \middle \lvert     \left \lvert \Pi^{(i)}_{n_i(t)}\right \lvert
 =m\right] P\left\{ \left \lvert\ \Pi^{(i)}_{n_i(t)}\right \lvert=m \right\} 
 = \sum_{m =1}^{n_i(t)}\E\left[   \left\lvert\Pi_m^{(0)}\right\lvert^r \right] P\left\{  \left \lvert\Pi^{(i)}_{n_i(t)}\right \lvert=m \right\}.
\]
 The second part of the proposition can be proved in 
an analogous way.  
\end{proof}

\begin{proof}[Proof of Proposition \ref{Prop:clusterdistribution00}]
The proof is analogous to the one of Proposition \ref{prop_gsssTRIS}. 
By Proposition \ref{prop1} the probability of the event $\{ D_{i,t}=d\}$
is equal to the probability of observing $d$ distinct values in $(\phi_1,\dots,\phi_{|\Pi^{(0)}_{K_t^i}|})$.
Hence, conditionally on $\{ |\Pi^{(0)}_{K_t^i}|=D_{i,t}=k\}$ this probability is
$H^*_0(d|k)$ and the thesis follows. Analogously one proves the statement for $D_t$, since in this case 
the probability of the event $\{ D_{t}=d\}$
is equal to the probability of observing $d$ distinct values in $(\phi_1,\dots,\phi_{|\Pi^{(0)}_{K_t}|})$.
\end{proof}

\begin{proof}[Proof of Proposition \ref{HierclustrH0spikeandslub}]
As already observed in the proof of Proposition \ref{prop_gsssTRIS},   if $H_0$ is a spike-and-slab prior, then 
\[
H^*_0(d|k)={k \choose d-1} a^{k+1-d}(1-a)^{d-1}  +\J\{k=d\} (1-a)^{d}.
\]
Hence, the first part of the thesis follows immediately from Proposition \ref{Prop:clusterdistribution00}.
Now 
\[
\E[\tilde D_{i,t}]=\sum_{d=0}^{n_{i}(t)} d  (1-a)^d \P\{ D_{i,t}=d\}+
\sum_{d=0}^{n_{i}(t)} d   \sum_{k =d}^{n_i(t)} {k \choose d-1} a^{k+1-d}(1-a)^{d-1}  \P\{ D_{i,t}=k\},
\]
where
\[
\begin{split}
\sum_{d=0}^{n_{i}(t)} d & \sum_{k =d}^{n_i(t)} {k \choose d-1} a^{k+1-d}(1-a)^{d-1}\P\{ D_{i,t}=k\} \\
& 
= \sum_{k =0}^{n_{i}(t)}  \sum_{d=1}^k   d {k \choose d-1} a^{k+1-d}(1-a)^{d-1}\P\{ D_{i,t}=k\}\\
& =\sum_{k =0}^{n_i(t)}  \sum_{l=0}^{k-1}   (l+1) {k \choose l} a^{k-l}(1-a)^{l}\P\{ D_{i,t}=k\} \\
& =\sum_{k =0}^{n_i(t)}[1+(1-a)k-(k+1)(1-a)^k] P\{ D_{i,t}=k\}. \\
\end{split}
\]
Which gives 
\[
\begin{split}
\E[\tilde D_{i,t}]& =\sum_{k=0}^{n_{i,t}} k  (1-a)^k \P\{ D_{i,t}=k\}
+\sum_{k =0}^{n_i(t)}[1+(1-a)k-(k+1)(1-a)^k] P\{ D_{i,t}=k\} \\
& =1+\E[(1-a)D_{i,t}-(1-a)^{D_{i,t}}].\\
\end{split}
\]

\end{proof}

\begin{proof}[Proof of Proposition \ref{PropAsym1}] Since $\lim_{t \to \infty} n_i(t) =+\infty$ and $\lim_n b_n = +\infty$, we get that   $b_{n_i(t)} \to \infty$. 
By assumption,  $\left\lvert\Pi^{(i)}_n\right\lvert/b_n \to D^{(i)}_\infty$ a.s. 
and then  $\left\lvert\Pi^{(i)}_{n_i(t)}\right\lvert/b_{n_i(t)} =K_{i,t}/b_{n_i(t)} \to D^{(i)}_\infty$ a.s. and  $K_{i,t} \to +\infty$ a.s.. 
Then, we can write
\[
\frac{D_{i,t}}{d_{n_i(t)}}=\frac{\Pi^{(0)}_{K_{i,t}}}{a_{K_{i,t}}}  \frac{{a_{K_{i,t}}}}{a_{b_{n_i(t)}}}.
\]
Combining the fact that  $\left\lvert\Pi^{(0)}_n\right\lvert/a_n \to D^{(0)}_\infty$  a.s. and  $K_{i,t} \to +\infty$ a.s., 
we have that ${\Pi^{(0)}_{K_{i,t}}}/{a_{K_{i,t}}}  \to D^{(0)}_\infty$  a.s..
Now  we can write
\[
  \frac{{a_{K_{i,t}}}}{a_{b_{n_i(t)}}}= \left( \frac{K_{i,t}}{b_{n_i(t)}} \right)^{\sigma_0} \frac{L_0 \left(  \frac{K_{i,t}}{b_{n_i(t)}} b_{n_i(t)}  \right)}{L_0( b_{n_i(t)})}.
\]
Recalling that for any slowly varying function 
$L_0(x_n y_n)/L_0(y_n) \to 1$ whenever $y_n \to +\infty$ and $x_n \to x>0$ (see
Theorem B.1.4 in \cite{deHaanFerreira}) and that 
 $K_{i,n}/b_{n_i} \to D^{(i)}_\infty$ a.s., it is easy to see that 
\[
  \frac{{a_{K_{i,n}}}}{a_{b_{n_i}}} \to  \left(D^{(i)}_\infty\right)^{\sigma_0} \quad a.s.
\]
In order to proof (ii), we note that 
\[
\frac{D_{t}}{d_{n(t)}}=\frac{\Pi^{(0)}_{K_{t}}}{a_{K_{t}}}  \frac{{a_{K_{t}}}}{a_{b_{n(t)}}}.
\]
Now, using again that
 $L_1(x_n y_n)/L_1(y_n) \to 1$ whenever $y_n \to +\infty$ and $x_n \to x>0$, we get
\[
\frac{b_{n_i(t)}}{b_{n(t)}}= \left( \frac{{n_i(t)}}{{n(t)}}\right)^{\sigma_1} \frac{L_1\left( \frac{n_i(t)}{n(t)}  n(t) \right)  }{L_1 (n(t)) }
\to w_i^{\sigma_1}.
\]
Hence
\[
\frac{K_{t}}{b_{n(t)}}=\sum_{i=1}^{I} \frac{K_{i,t}}{b_{n_i(t)}} \frac{b_{n_i(t)}}{b_{n(t)}} 
 \to \sum_{i=1}^I D^{(i)}_\infty    w_i^{\sigma_1}  \quad a.s.
\]
Hence, it follows that $K_{t} \to +\infty$ a.s..
To conclude we can follow the same line of the first part of the proof. 
\end{proof}


\begin{proof}[Proof of Proposition \ref{PropAsym2}] 
Part (i) follows immediately by taking the limit for $t \to + \infty$ 
in the expression of $\P\left\{ D_{i,t}=k\right\}$ and  $\P\left\{ D_{i}=k\right\}$ given in 
Proposition \ref{Prop:clusterdistribution}.
For part (ii), note that since $|\Pi^{(i)}_n|/b_n$ converges a.s. 
to a strictly positive random variable $D^{(i)}_\infty$  with $b_n \to +\infty$, then 
$|\Pi^{(i)}_n|$ diverges to $+\infty$ a.s. Hence, $D_{i,t}$ converges 
a.s. to the same limit of $|\Pi^{(0)}_n|$, i.e. to $K_0$. Since, 
both $D_{i,t}$ and $K_0$ are integer valued random variables the thesis follows. 
The proof for $D_t$ is similar. 
\end{proof}

\begin{proof}[Proof of Proposition \ref{Prop:AsymGnedin}] 
It is enough to apply Proposition \ref{PropAsym2}
and the fact that 
if  $(\Pi_n)_n$ is the Gnedin's partition  of parameters $(\gamma,\zeta)$,  $| \Pi_{n}|$ converges almost surely to a random variable $K$ with distribution \eqref{DistXbis}. Algebraic manipulations  give the thesis. 
\end{proof}

\begin{proof}[Proof of Proposition \ref{asymptoticPY}]
The statement is essentially a  corollary  of Proposition \ref{PropAsym1}, except for 
the fact that we now want to show that all the convergences are in $L^p$. We shall use  a classical dominated convergence argument. 
Let us  give the  details for the proof of  statement (iii). The other cases are similar. 
Assume that we are dealing with a   $HPYDP(\sigma_0;\sigma_1,\theta_1)$. Recall that in this case 
$ |\Pi_{n}^{(0)}|/\log(n)$ converges almost surely and in $L^p$ (for every $p>0$) to  $\theta_0$, 
while  $ |\Pi_{n}^{(i)}|/n^{\sigma_1}$ converges almost surely and in $L^p$ (for every $p>0$) to 
$S_{\sigma_1,\theta_1}^{(i)}$ for $i=1,\dots,I$,  where $S_{\sigma_1,\theta_1}^{(i)}$
are independent and identically distributed random variables with density  $g_{\sigma_1,\th_1}$.
 Proposition   \ref{PropAsym1} applies  with $a_n=\log(n)$, $\sigma_0=0$, $L_0(x)=\log(n)$, 
 $b_n=n^{\sigma_1}$, $\sigma_1=\sigma_1$, $L_1(x)=1$, 
$D_{\infty}^{(i)}=S_{\sigma_1,\theta_1}^{(i)}$, $D_{\infty}^{(0)}=\theta_0$.
 Hence, it remains to prove that all  the  convergences hold also in $L^p$. 
  We already know that $K_{i,t}=\left\lvert\Pi^{(i)}_{n_i(t)}\right\lvert \to +\infty$ a.s., and,  
 since $\left\lvert\Pi^{(0)}_{n_i(t)}\right\lvert/a_{n_i(t)} \to \theta_0$ in $L^p$ 
 for every $p>0$, it is easy to check that 
  \[
\frac {\Pi^{(0)}_{K_{i,t}}}{a_{K_{i,t}}}  \to D^{(0)}_\infty \quad \text{in $L^p$ for every $p>0$}.
  \]
Using $\log(x) \leq x$ for every $x>0$, write
  \[
0 \leq  \frac{{a_{K_{i,t}}}}{a_{b_{n_i(t)}}}= \frac{\log \left(  \frac{K_{i,t}}{n_i(t)^{\sigma_1}}{n_i(t)^{\sigma_1}}  \right)}{\log({n_i(t)}^{\sigma_1})}
  = \frac{\log \left(  \frac{K_{i,t}}{n_i(t)^{\sigma_1}}  \right)}{\log({n_i(t)}^{\sigma_1})}+1
   \leq  \frac{ \frac{K_{i,t}}{n_i(t)^{\sigma_1}}  }{\log({n_i(t)}^{\sigma_1})}+1.
\]
Hence, for every $p>0$, 
\[
 \left |  \frac{{a_{K_{i,t}}}}{a_{b_{n_i(t)}}} \right|^p \leq C_p  \left[
  \left| \frac{K_{i,t}}{n_i(t)^{\sigma_1}} \frac{1}{\log({n_i(t)}^{\sigma_1})}  \right|^p+1  \right].
\]
Since we already know from the proof of Proposition  \ref{PropAsym1}  that  ${{a_{K_{i,t}}}}/{a_{b_{n_i(t)}}}$ converges
a.s. to $1$ and 
$K_{i,t}/{n_i(t)^{\sigma_1}}$ 
converges in $L^p$ for every $p>0$, by dominated convergence theorem  it follows that  
 ${{a_{K_{i,t}}}}/{a_{b_{n_i(t)}}}$ converges il $L^p$ 
to $1$  for every $p>0$. Since $p$ is arbitrary, combining all these results one gets
 \[
\frac{D_{i,t}}{d_{n_i(t)}}=\frac{\Pi^{(0)}_{K_{i,t}}}{a_{K_{i,t}}}  \frac{{a_{K_{i,t}}}}{a_{b_{n_i(t)}}} \to D^{(0)}_\infty \quad \text{in $L^p$ for every $p>0$}.
\]
Arguing in a similar way, one proves that  also
${D_{t}}/{d_{n(t)}} \to D^{(0)}_\infty$ in $L^p$ for every $p>0$.
\end{proof}

\begin{proof}[Proof of Corollary \ref{corolASY}]
Let $S_{\sigma,\theta}$ be a random variable with density
\eqref{mittlefftilted}. 
Then 
\[
\E\left[S_{\sigma,\theta}^p\right]=
\frac{\Gamma(\theta+1)}{\Gamma\left(\frac{\theta}{\sigma}+1\right)} 
\int_0^{+\infty} s^{\theta/\sigma+p} g_\sigma(s)ds
=\frac{\Gamma(\theta+1)}{\Gamma(\frac{\theta}{\sigma}+1)} 
\frac{\Gamma(p+\theta/\sigma+1)}{\Gamma(\theta + p\sigma+1)},
\]
where in the second part we use \eqref{momentML}. 
Using the previous expression with $p=r$, we get
\[
\E\left[ S_{\sigma_0,\th_0}^r \right]=
\frac{\Gamma(\theta_0+1)\Gamma(\theta_0/\sigma_0 + r+1)    }{\Gamma({\theta_0}/{\sigma_0}+1)\Gamma(\theta_0+r\sigma_0 +1)} 
 \]
 and for  $p=r\sigma_0$
  \[
 \E\left[  (S_{\sigma_1,\th_1}^{(i)})^{r\sigma_0}\right]=
\frac{\Gamma(\theta_1+1)\Gamma(\theta_1/\sigma_1 + r\sigma_0+1) }{\Gamma({\theta_1}/{\sigma_1}+1)\Gamma(\theta_1+r\sigma_1\sigma_0 +1)}.
 \]
 Now the thesis follows easily from Proposition \ref{asymptoticPY}. For example, in case (i), we have that 
 $D_{i,t}/n^{\sigma_0 \sigma_1}$ converges in $L^r$ for every $r>0$ to  $S_{\sigma_0,\th_0} \left(S_{\sigma_1,\th_1}^{(i)}\right)^{\sigma_0}$, 
 hence  $\E[D_{i,t}^r]  \simeq n^{\sigma_0 \sigma_1} \E\left[ S_{\sigma_0,\th_0}^r (S_{\sigma_1,\th_1}^{(i)})^{\sigma_0r}\right]$ and the thesis follows. 
 The other cases can be obtained in a similar way. 
 \end{proof}

\subsection{Proofs of the results in Section \ref{Sec_Gibbs}}

In order to derive the  full conditionals  of Section \ref{Sec_Gibbs}, we start 
from the joint distribution of $[\bx, \bphi, \bc, \bd]$, that is
\begin{equation}\label{joint1}
p(\bx, \bphi, \bc, \bd)=p(\bx| \bphi, \bc, \bd) p(\bphi | \CD)  p(\bc, \bd),
\end{equation}
where 
\begin{equation}\label{joint2}
 p(\bx| \bphi, \bc, \bd) =\prod_{i \in \CJ} \prod_{j=1}^{n_{i\cdot\cdot}} f \left(Y_{i,j}|\phi_{d_{i,c_{i,j}}}\right)
 \quad   \text{and} \quad 
  p(\bphi| \CD)=\prod_{d \in \CD} h(\phi_d).  
\end{equation}
From \eqref{joint1}, the marginal distribution of  $[\bx,\bc, \bd]$ factorizes as follows 
\begin{equation}\label{joint3}
   p\left(\bx, \bc, \bd\right) =  p(\bc, \bd) p(\bx| \bc,\bd) =
 p(\bc, \bd) \prod_{d \in \CD} \int \prod_{(i,j): d_{i,c_{i,j}}=d } f(Y_{i,j}|\phi)h(\phi) d\phi.
\end{equation}

Recalling that $d^*_{i,j}=d_{i,c_{i,j}}$, we can write 
\begin{equation}\label{joint1bis}
p(\bx, \bphi, \bc, \bd^*)=p(\bx| \bphi,\bd^*) p(\bphi | \CD)  p(\bc, \bd^*),
\end{equation}
where 
\begin{equation}\label{joint2bis}
  p(\bx| \bphi, \bd^*) =\prod_{i \in \CJ} \prod_{j=1}^{n_{i\cdot\cdot}} f (Y_{i,j}|\phi_{d^*_{i,j}})
\end{equation}
and  $\CD=\{ d^*_{ij} : i \in \CJ, j =1,\dots, n_{i\cdot\cdot} \}$. 
From \eqref{joint1bis},  
$\bx$ and $\bc$ are conditionally independent given $[\bd^*,\bphi]$ and since 
\[
\{(i,j): d^*_{i,j}=d \} = \{(i,j): d_{i,c_{i,j}}=d \},
\]
we obtain that
\begin{equation}\label{joint4}
   p(\bx, \bc, \bd^*) = p(\bc, \bd^*) p(\bx| \bd^*) =
p(\bc, \bd^*) \prod_{d \in \CD} \int \prod_{(i,j): d^*_{i,j}=d } f(Y_{i,j}|\phi)h(\phi) d\phi.
\end{equation}

\begin{proof}[Proof of \eqref{fullcd-final}.]

From \eqref{joint1bis},\eqref{joint2bis} and \eqref{joint4}, we have
\[
 p\left( \bx^{ \urcorner ij},\bc,\bd^*\right)=p(\bc,\bd^*) p( \bx^{ \urcorner ij}|{\bd^*}^{ \urcorner ij}),
\]
which shows that 
$ p( \bx^{ \urcorner ij}|\bc,\bd^*)= p( \bx^{ \urcorner ij}|{\bd^*}^{ \urcorner ij})$.
Hence we can write 
\[
\begin{split}
p(\bx, c_{i,j},d^*_{i,j},  \bc^{ \urcorner ij}, & {\bd^*}^{ \urcorner ij})  =
  p(\bc^{ \urcorner ij}, {\bd^*}^{ \urcorner ij})
p(c_{i,j},d^*_{i,j}|\bc^{ \urcorner ij}, {\bd^*}^{ \urcorner ij})  
 p( \bx^{ \urcorner ij}|\bc,{\bd})
 p(Y_{i,j}| \bx^{ \urcorner ij},\bc, \bd^* ) 
\\
&
=p(\bc^{ \urcorner ij}, {\bd^*}^{ \urcorner ij})
 p(c_{i,j},d^*_{i,j}|\bc^{ \urcorner ij}, {\bd^*}^{ \urcorner ij})    p( \bx^{ \urcorner ij}|\bc^{ \urcorner ij},{\bd^*}^{ \urcorner ij}) 
 p(Y_{i,j}| \bx^{ \urcorner ij},
\bc, \bd^* ).  \\
\end{split}
\]
This shows that 
\begin{equation}\label{fulltd_1}
  p(c_{i,j},d^*_{i,j}| \bx, \bc^{ \urcorner ij}, {\bd^*}^{ \urcorner ij}) \propto  
p(c_{i,j},d^*_{i,j}|\bc^{ \urcorner ij}, {\bd^*}^{ \urcorner ij}) p(Y_{i,j}| \bx^{ \urcorner ij},{\bd^*}^{ \urcorner ij},d^*_{i,j}).
\end{equation}
From \eqref{joint4} and  \eqref{Ml}, we have that 
\begin{equation}\label{marginalcondx}
   p(Y_{i,j}| \bx^{ \urcorner ij}, {\bd^*}^{ \urcorner ij},d^*_{i,j} )
   =f_{d^*_{i,j}}(\{Y_{i,j}\}).
\end{equation}
Moreover, by exchangeability and using the predictive distribution  given in 
Section \ref{Sec:1-1}, it follows that 
\begin{equation}\label{conditionaltabledish}
\begin{split}
&
p(c_{i,j}=c^{old},d^*_{i,j}=d_{i,c^{old}}|\bc^{ \urcorner ij}, {\bd^*}^{ \urcorner ij})= \omega_{n_{i\cdot\cdot}-1,c^{old}}(\bc_i^{\urcorner ij}),
\\
&
p(c_{i,j}=c^{new},d^*_{i,j}=d^{old}|\bc^{ \urcorner ij}, {\bd^*}^{ \urcorner ij})= \nu_{n_{i\cdot\cdot}-1}(\bc_i^{\urcorner ij}) \tilde
 \omega_{m_{\cdot \cdot}^{ \urcorner ij},d^{old}}(\bd^{\urcorner ij}),
\\
&
p(c_{i,j}=c^{new},d^*_{i,j}=d^{new}|\bc^{ \urcorner ij}, {\bd^*}^{ \urcorner ij})= \nu_{n_{i\cdot\cdot}-1}(\bc_i^{\urcorner ij}) \tilde
 \nu_{m_{\cdot \cdot}^{ \urcorner ij}}(\bd^{\urcorner ij}).
 \end{split}
\end{equation}

Combining \eqref{fulltd_1}, \eqref{marginalcondx} and
 \eqref{conditionaltabledish}   we obtain \eqref{fullcd-final}.

\end{proof}

\begin{proof}[Proof of \eqref{full-dishes}]
Using \eqref{joint3}, we have
\begin{equation}\label{fulltd_12}
  p(d_{i,c}=d| \bx, \bc, {\bd}^{ \urcorner ic}) \propto  
p(d_{i,c}=d |\bc, {\bd}^{ \urcorner ic}) p(\{ Y_{i,j}: \,\ ij \in \CS_{ic} \} | \{Y_{i',j'}: \,\, i'j' \in \CS_d \setminus \CS_{ic}
 \}, \bc,\bd^{\urcorner ic},d),
\end{equation}
where
\[
\CS_{ic}=\{ (i,j) : c_{i,j}=c\} \,\, \text{and}  \,\, \CS_d=\{ (i',j'): d^*_{i',j'}=d \}.\]
Note that $i$ in $\CS_{ic}$ is fixed and $j$ is such that $c_{i,j}=c$.  From  \eqref{joint3}, we get 
\[
 p(\{ Y_{i,j}: \,\ (i,j) \in \CS_{ic} \} | \{Y_{i',j'}: \,\, (i',j') \in \CS_d \setminus \CS_{ic}
 \}, \bc,\bd^{\urcorner ic},d)= f_d( \{Y_{i,j}: (i,j) \in \CS_{ic} \}  ).
\]
Moreover
\[
\begin{split}
& p(d_{i,c}=d^{new} |\bc, {\bd}^{ \urcorner ic})=\tilde \nu_{m_{\cdot\cdot}^{ \urcorner ic} }(\bd^{ \urcorner ic}), \\
&   p(d_{i,c}=d^{old} |\bc, {\bd}^{ \urcorner ic})=\tilde \omega_{m_{\cdot\cdot}^{ \urcorner ic}, d^{old} }(\bd^{ \urcorner ic}). \\
\end{split}
\]

\end{proof}

\begin{proof}[Proof of \eqref{phifull_out}]
From 
 \eqref{joint1}-\eqref{joint2} one gets
\begin{equation}\label{phifull_out2}
p(\bphi,|\bx,\bc,\bd) \propto \prod_{d \in \CD} h(\phi_d) \prod_{(i,j): d^*_{i,j}=d} f(Y_{i,j}|\phi_d).
\end{equation}
\end{proof}

\begin{proof}[Proof of \eqref{cdpredictive00}-\eqref{cdpredictive}]
Arguing as in the proof of \eqref{joint4},  one gets
\[
   p(\bx, \bc, \bd^*,c_{i,n_i+1},d_{i,n_i+1}^*) = p(\bx| \bd^*) p\left(\bc, \bd^*,c_{i,n_i+1},d_{i,n_i+1}^*\right),
\]
and then 
\[
 p\left(c_{i,n_i+1},d_{i,n_i+1}^*|\bx, \bc, \bd^*\right)= p\left(c_{i,n_i+1},d_{i,n_i+1}^*| \bc, \bd^*\right).
\]
The explicit expression for 
$p\left(c_{i,n_i+1},d_{i,n_i+1}^*| \bc, \bd^*\right)$ follows 
arguing as in the proof of \eqref{conditionaltabledish}
 by replacing $c_{i,j}$, $\bc_i^{ \urcorner ij} , {\bd^*_i}^{ \urcorner ij}$, $n_{i\cdot\cdot}-1$ and 
 $m_{\cdot\cdot}^{ \urcorner ij}$,
  with $c_{i,n_i+1}$, $\bc_i$, $\bd^*_i$, $n_{i \cdot \cdot}$ and $m_{\cdot\cdot}$. 
\end{proof}

\newpage

\section{Prior sensitivity analysis}\label{App:NumRes}
\begin{figure}[h!]
\caption{Prior distribution, for $i=1,2$, of the marginal, $\P\{ D_{i,t}=k\}$, for $k=1,\ldots,50$ (left); and total, $\P\{ D_{t}=k\}$, for $k=1,\ldots,100$ (right); number of clusters, for the following processes: i) $HDP(\theta_0,\theta_1,H_0)$ with $\theta_0=\theta_1 = 43.3$ (black solid); ii) $HPYP(\sigma_0,\theta_0,\sigma_1,\theta_1,H_0)$ with $(\sigma_0,\theta_0) = (\sigma_1,\theta_1) = (0.25,29.9)$ (blue dashed) and $(\sigma_0,\theta_0) = (\sigma_1,\theta_1) = (0.67,8.53)$ (blue dotted); iii) $HGP(\gamma_0,\zeta_0,\gamma_1,\zeta_1,H_0)$ with $(\gamma_0,\zeta_0) = (\gamma_1,\zeta_1) = (15,1450)$ (red dashed), and $(\gamma_0,\zeta_0) = (\gamma_1,\zeta_1) = (3.2,290)$ (red dotted). The values of the parameters are chosen in such a way that $E[D_{i,t}]= 25$, with $n_i=50$ and $n=n_1+n_2=100$.}\label{all}
\begin{center}
  \begin{tabular}{cc}
   \includegraphics[width=5cm]{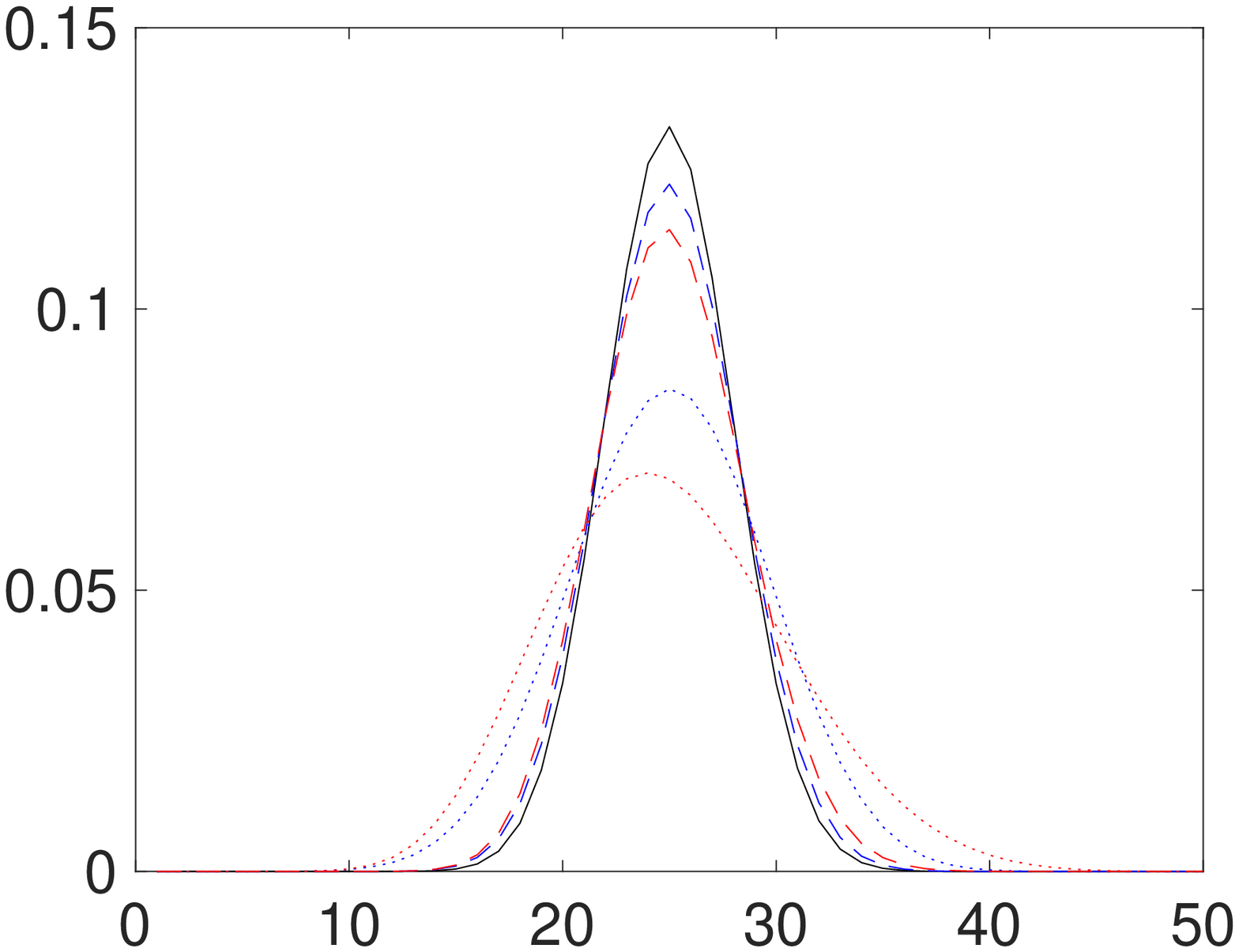}\hspace{-15pt} &
   \includegraphics[width=5cm]{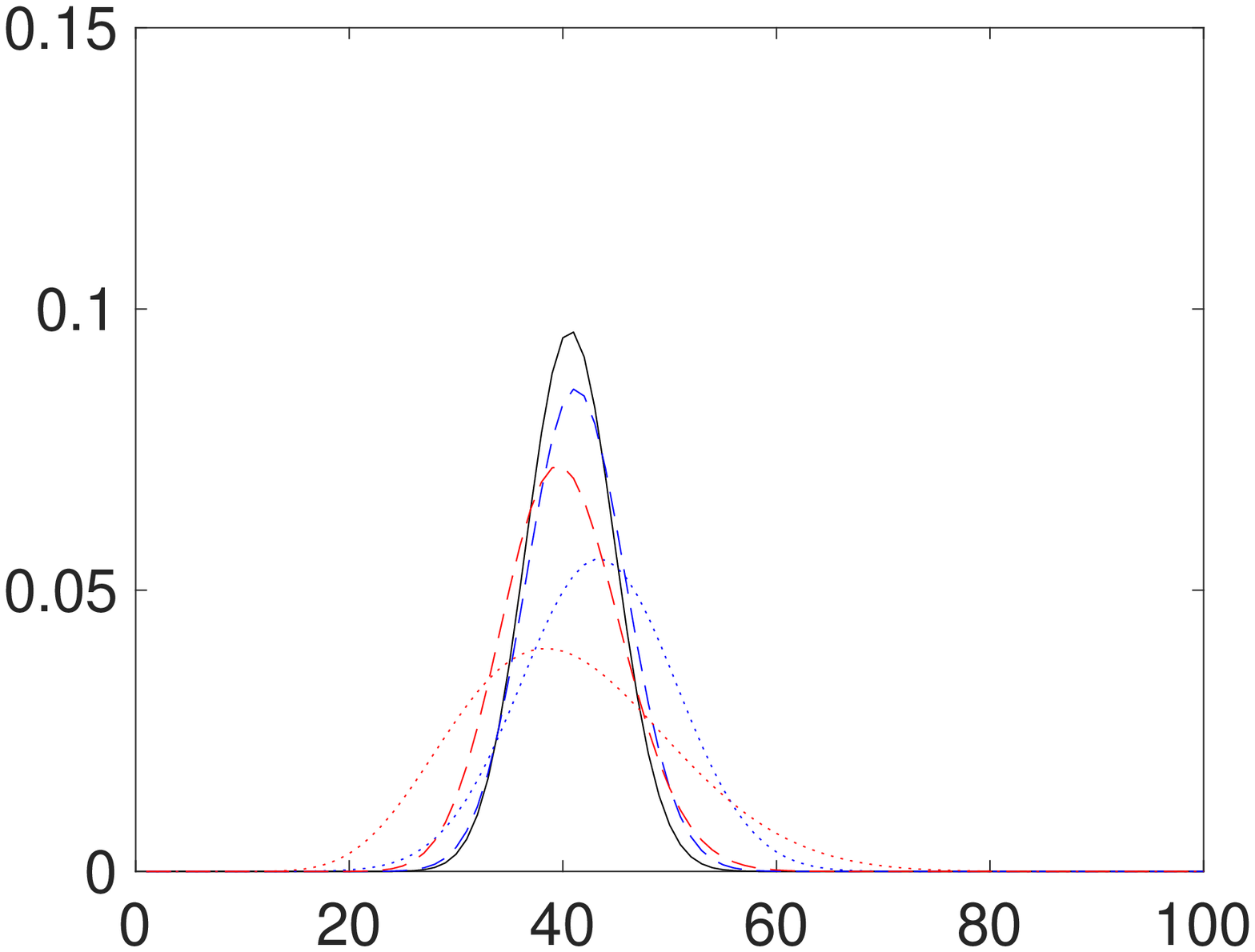} \\
  \end{tabular}
\end{center}
\end{figure}

\begin{table}[h]
\caption{Marginal and total expected number of clusters ($E(D_{i,t})$ and $E(D_{t})$, respectively) and marginal and total variance of the number of clusters ($V(D_{i,t})$ and $V(D_{t})$, respectively), when $I=2$, $n_i=50$, $i=1,2$, for three HSSMs (first column) and five parameter settings (columns $\sigma_i$, $\theta_i$, $\gamma_i$ and $\zeta_i$).}\label{Tab:priorEV}
\begin{tabular}{l|cccc|cccc}
\hline
HSSM & $\sigma_i$ & $\theta_i$ &$\gamma_i$& $\zeta_i$  & $E[D_{i,t}]$ & $V[D_{i,t}]$ & $E[D_t]$ &  $V(D_t)$\\
\hline
$HDP(\theta_0,\theta_1,H_0)$ &     & 43.3 &     &    &25.0 & 9.1  & 40.8 & 17.2\\
$HPYP(\sigma_0,\theta_0,\sigma_1,\theta_1,H_0)$& 0.25& 29.9 &     &    &25.0 & 10.6 & 41.3 & 21.5\\
$HGP(\gamma_0,\zeta_0,\gamma_1,\zeta_1,H_0)$ &     &      & 15  &1450 &25.0 & 12.1 & 40.1 & 30.5\\
\hline
$HPYP(\sigma_0,\theta_0,\sigma_1,\theta_1,H_0)$& 0.67& 8.53 &     &    &25.0 & 21.1 & 43.3 & 50.9\\
$HGP(\gamma_0,\zeta_0,\gamma_1,\zeta_1,H_0)$ &     &      & 3.2 & 290&25.0 & 30.8 & 40.6 & 99.1\\
\hline
\end{tabular}
\end{table}

\clearpage

\begin{figure}[h!]
  \caption{Prior marginal $\P\{D_{it} = k\}$ (left column) and global $\P\{D_t = k\}$ (right column) number of clusters for different processes.}
  \label{EffectsSymHierarchy}
 \begin{center}
  \begin{tabular}{cc}
   \includegraphics[width=5cm]{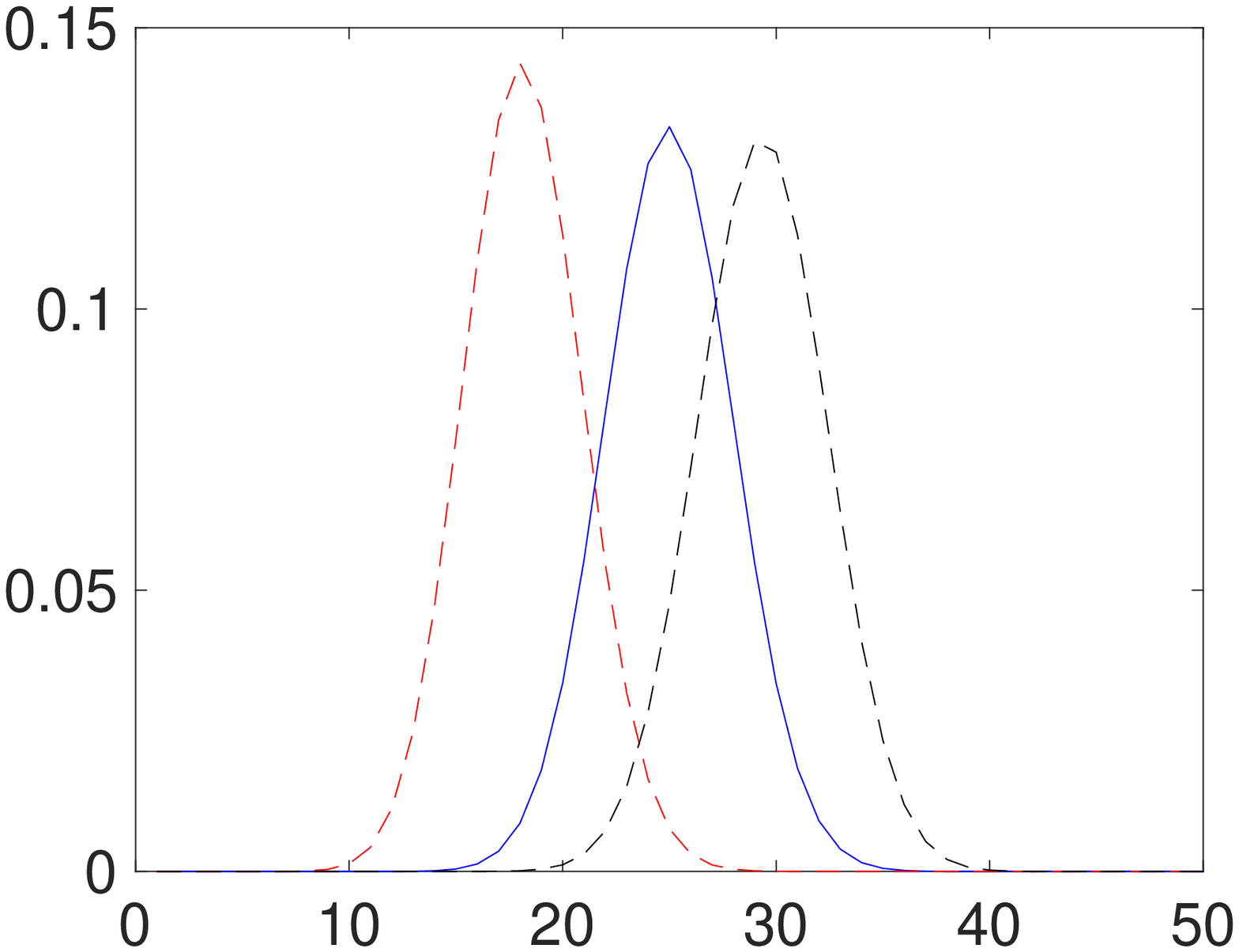}\hspace{-15pt} &
   \includegraphics[width=5cm]{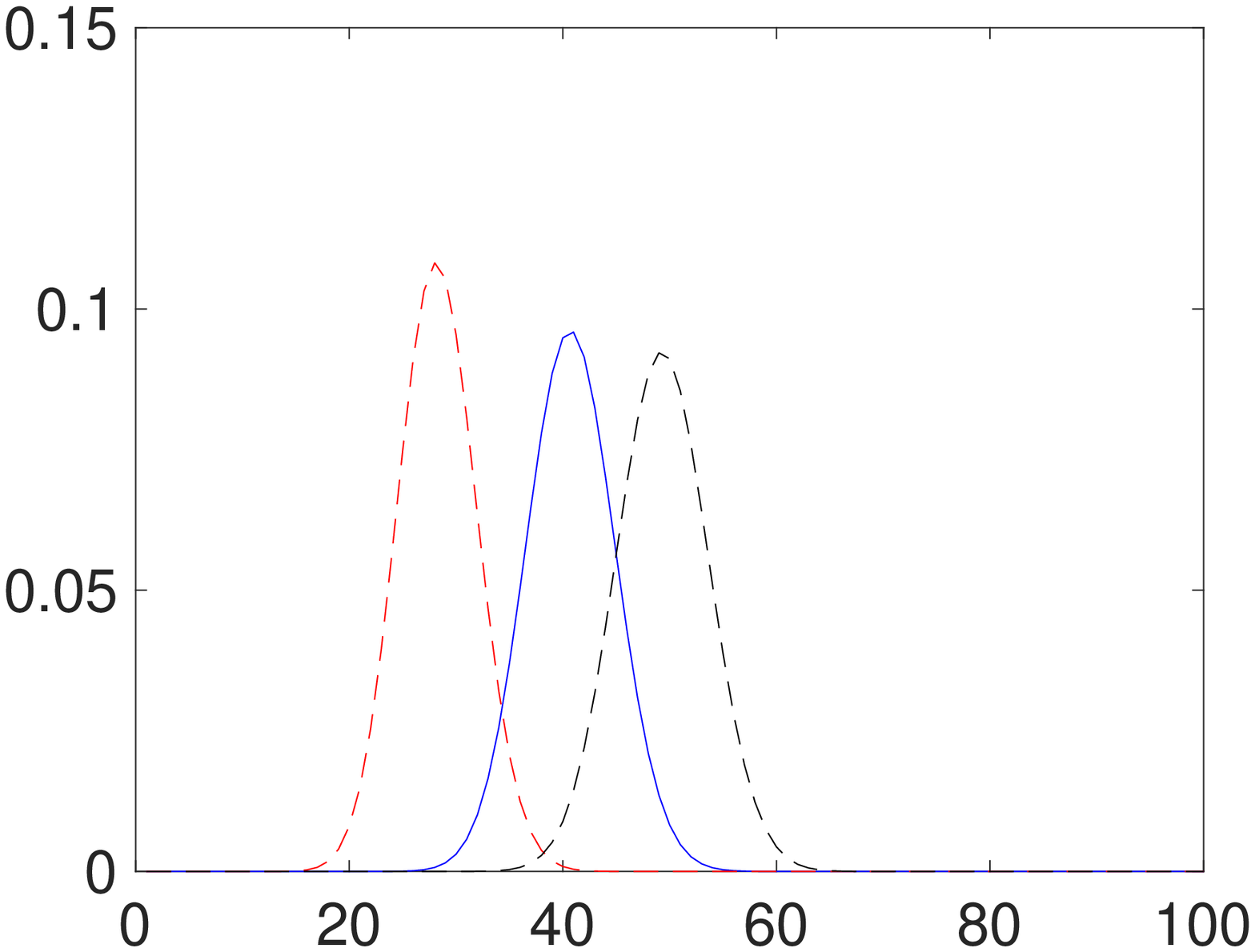}     \\
  \end{tabular}
 \end{center}
\begin{minipage}[]{0.84\textwidth}
\begin{small}
(i) HDP with $\theta_0 =\theta_i = 23.3$ (red dashed), $\theta_0 = \theta_i = 63.3$ (black dashed) and $\theta_0 = \theta_i = 43.3$ (homogeneous case, solid blue).
\end{small}
\end{minipage}
 \begin{center}
  \begin{tabular}{cc}
   \includegraphics[width=5cm]{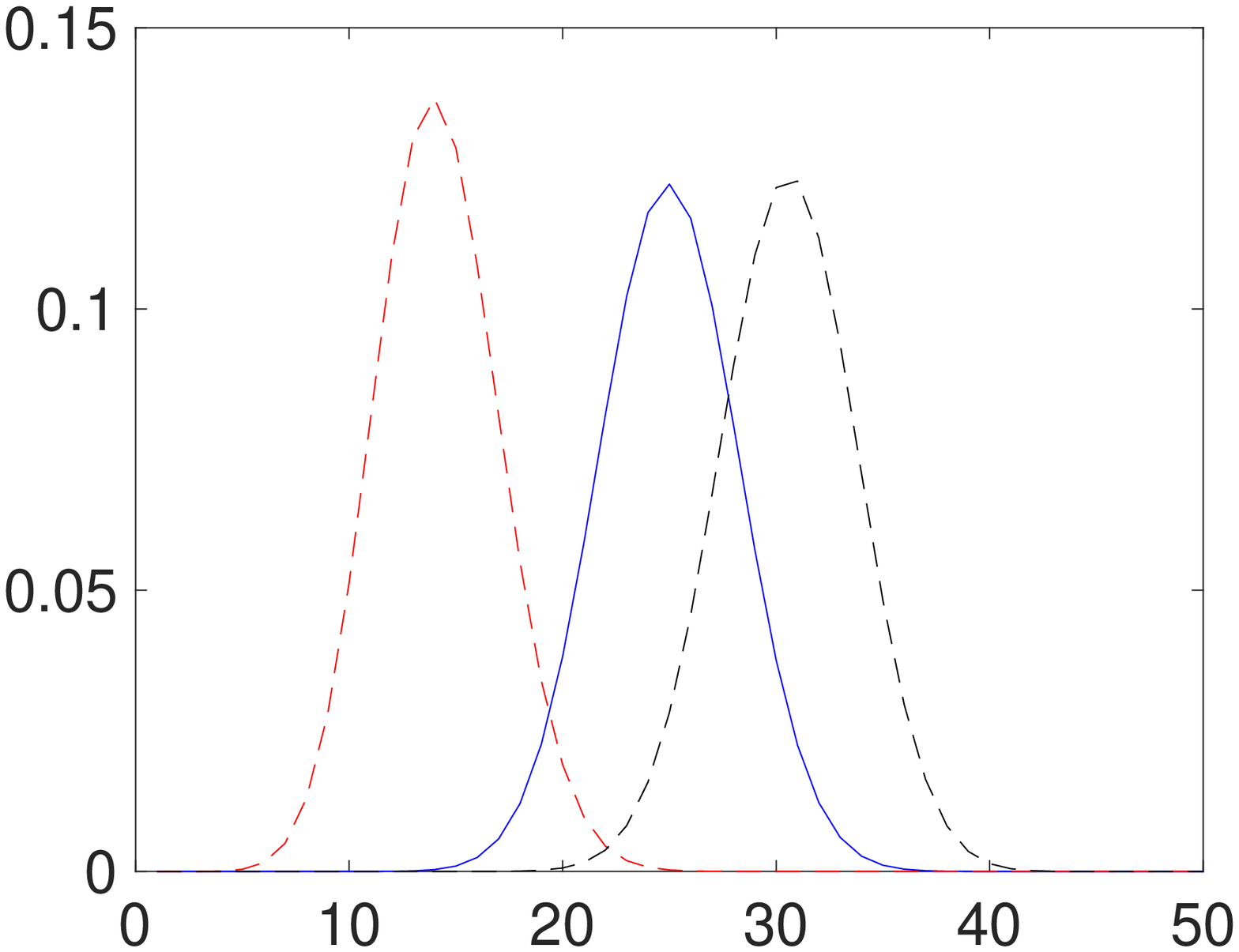}\hspace{-15pt} &
      \includegraphics[width=5cm]{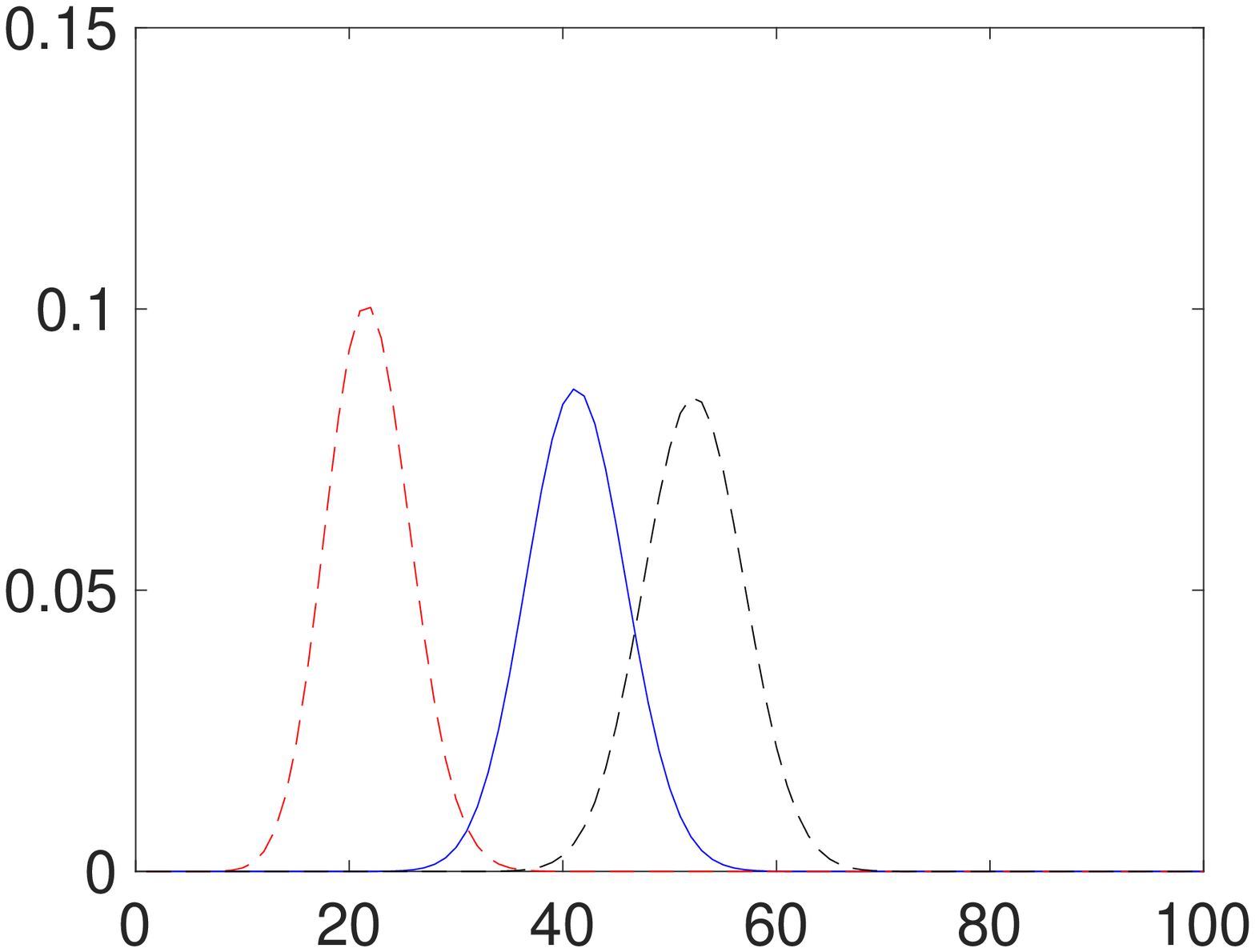}     \\
  \end{tabular}
 \end{center}
\begin{minipage}[]{0.84\textwidth}
\begin{small}
(ii) HPYP with $(\theta_0,\sigma_0) = (\theta_i,\sigma_i) = (9.9,0.25)$ (red dashed), $(\theta_0,\sigma_0) = (\theta_i,\sigma_i) =(49.9,0.25)$ (black dashed) and $(\theta_0,\sigma_0) = (\theta_i,\sigma_i) = (29.9,0.25)$ (homogeneous case, solid blue).
\end{small}
\end{minipage}
 \begin{center}
  \begin{tabular}{cc}
   \includegraphics[width=5cm]{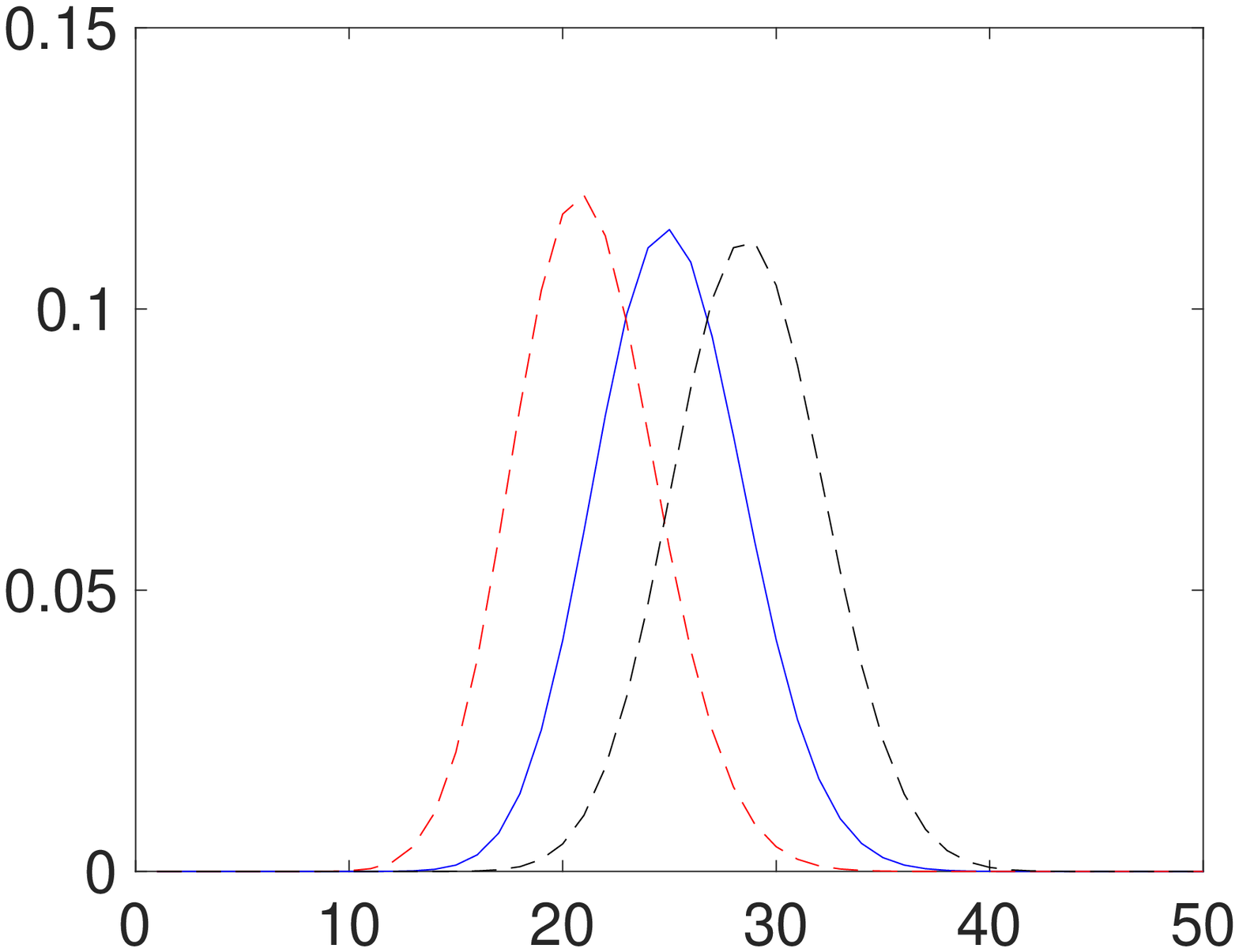} \hspace{-15pt} &
      \includegraphics[width=5cm]{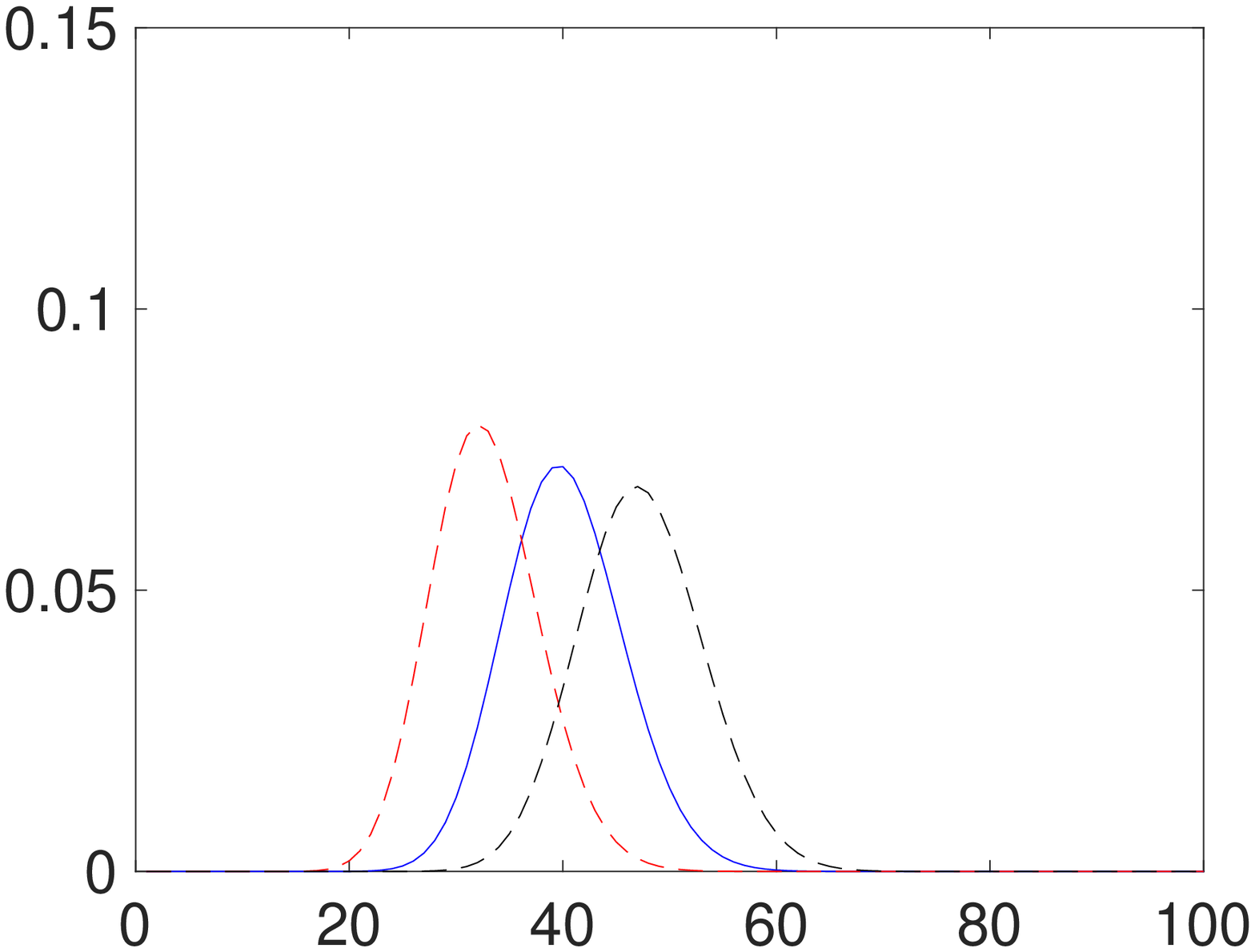}     \\
  \end{tabular}
 \end{center}
\begin{minipage}[]{0.84\textwidth}
\begin{small}
(iii) HGP with $(\gamma_0,\zeta_0) = (\gamma_i,\zeta_i) = (15,1050)$  (red dashed), $(\gamma_0,\zeta_0) = (\gamma_i,\zeta_i) = (15,1950)$ (black dashed) and $(\gamma_0,\zeta_0) = (\gamma_i,\zeta_i) = (15,1450)$ (homogeneous case, solid blue).
\end{small}
\end{minipage}
\end{figure}

\clearpage

\begin{figure}[h!]
  \caption{Prior marginal $\P\{D_{it} = k\}$ (left column) and global $\P\{D_t = k\}$ (right column) number of clusters for different processes.}
  \label{EffectsSymHierarchy_scale}
 \begin{center}
  \begin{tabular}{cc}
   \includegraphics[width=5cm]{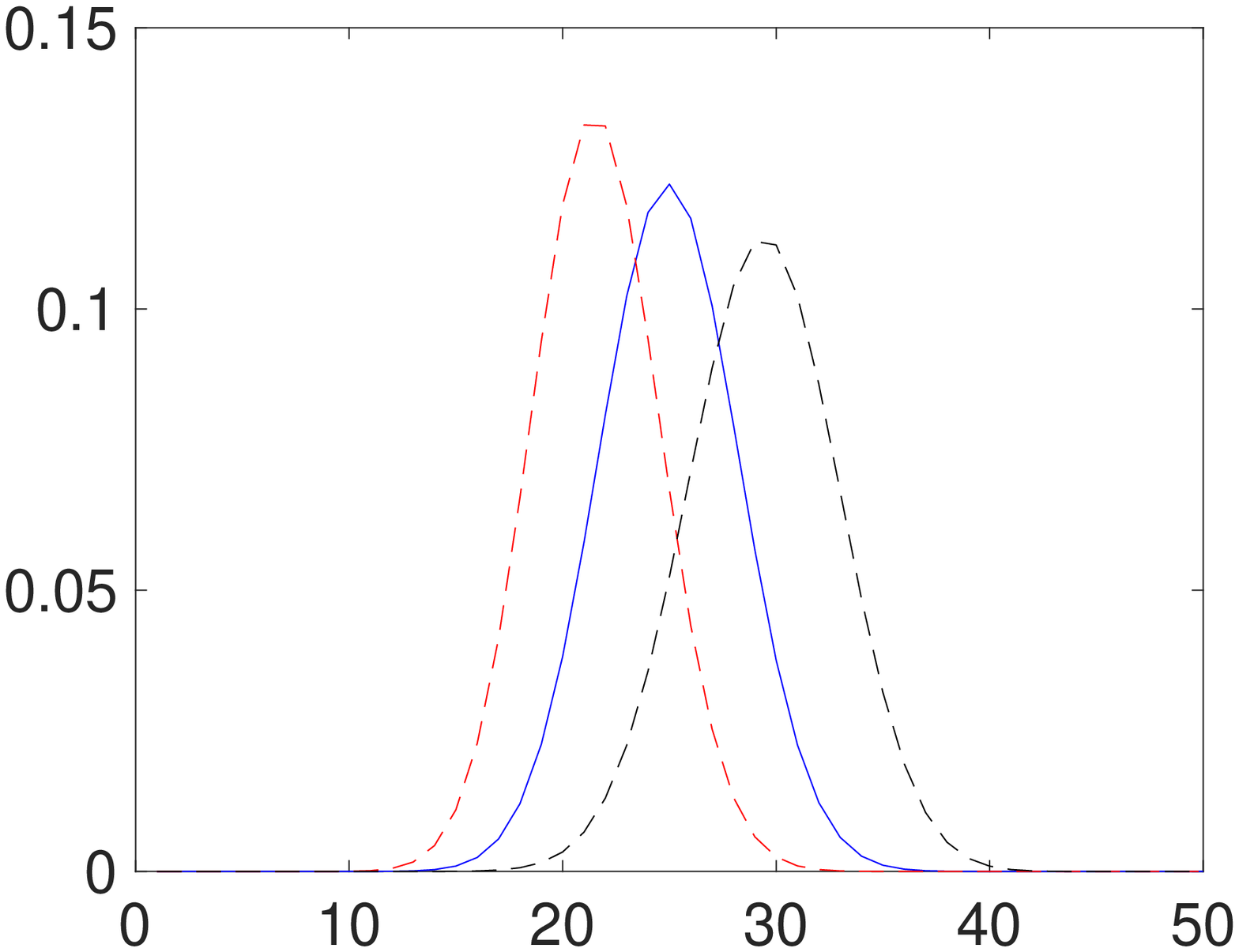}\hspace{-15pt} &
      \includegraphics[width=5cm]{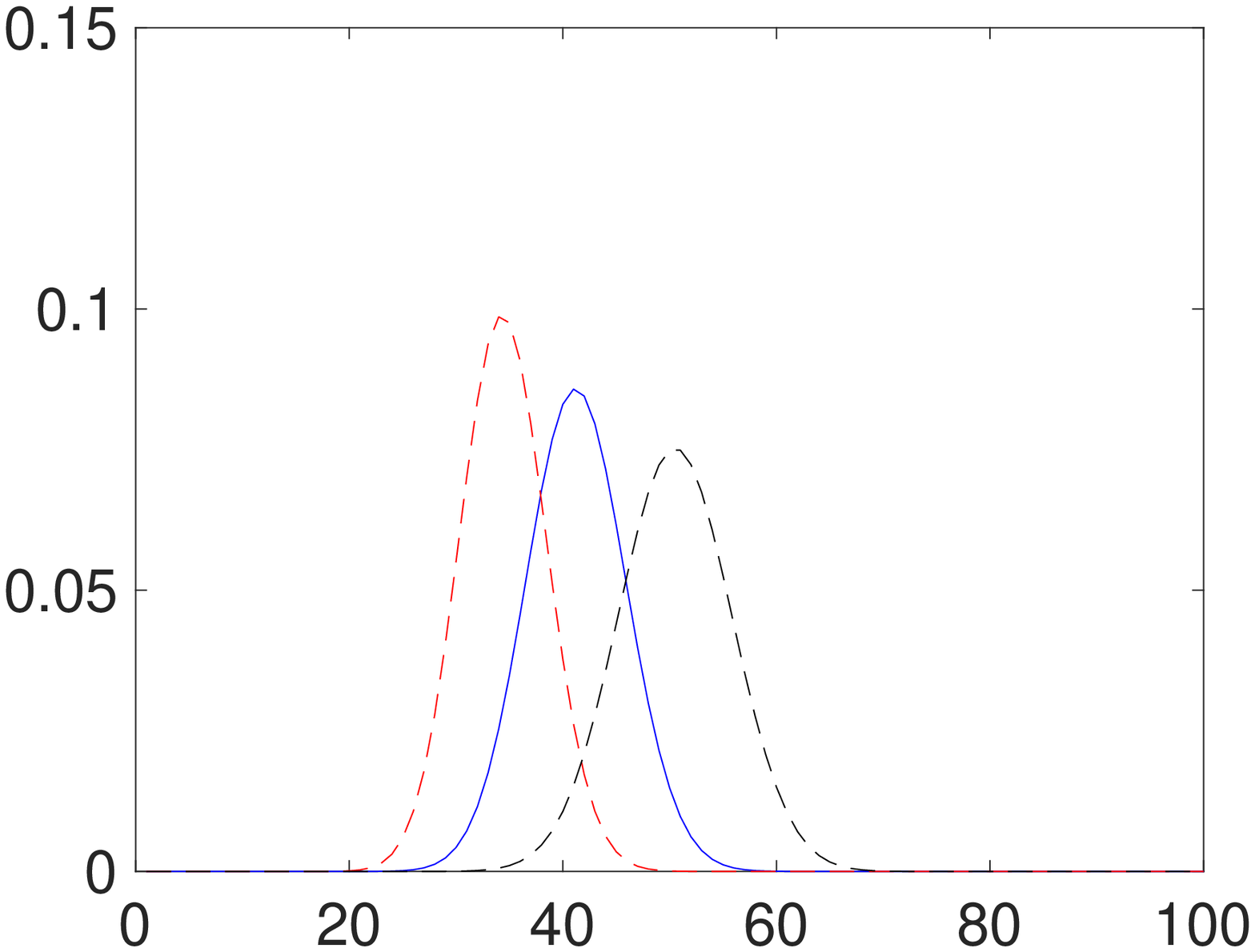}     \\
  \end{tabular}
 \end{center}
\begin{minipage}[]{0.84\textwidth}
\begin{small}
(i) HPYP with $(\theta_0,\sigma_0) = (\theta_i,\sigma_i) = (29.9,0.05)$ (red dashed), $(\theta_0,\sigma_0) = (\theta_i,\sigma_i) =(29.9,0.45)$ (black dashed) and $(\theta_0,\sigma_0) = (\theta_i,\sigma_i) = (29.9,0.25)$ (homogeneous case, solid blue).
\end{small}
\end{minipage}
 \begin{center}
  \begin{tabular}{cc}
   \includegraphics[width=5cm]{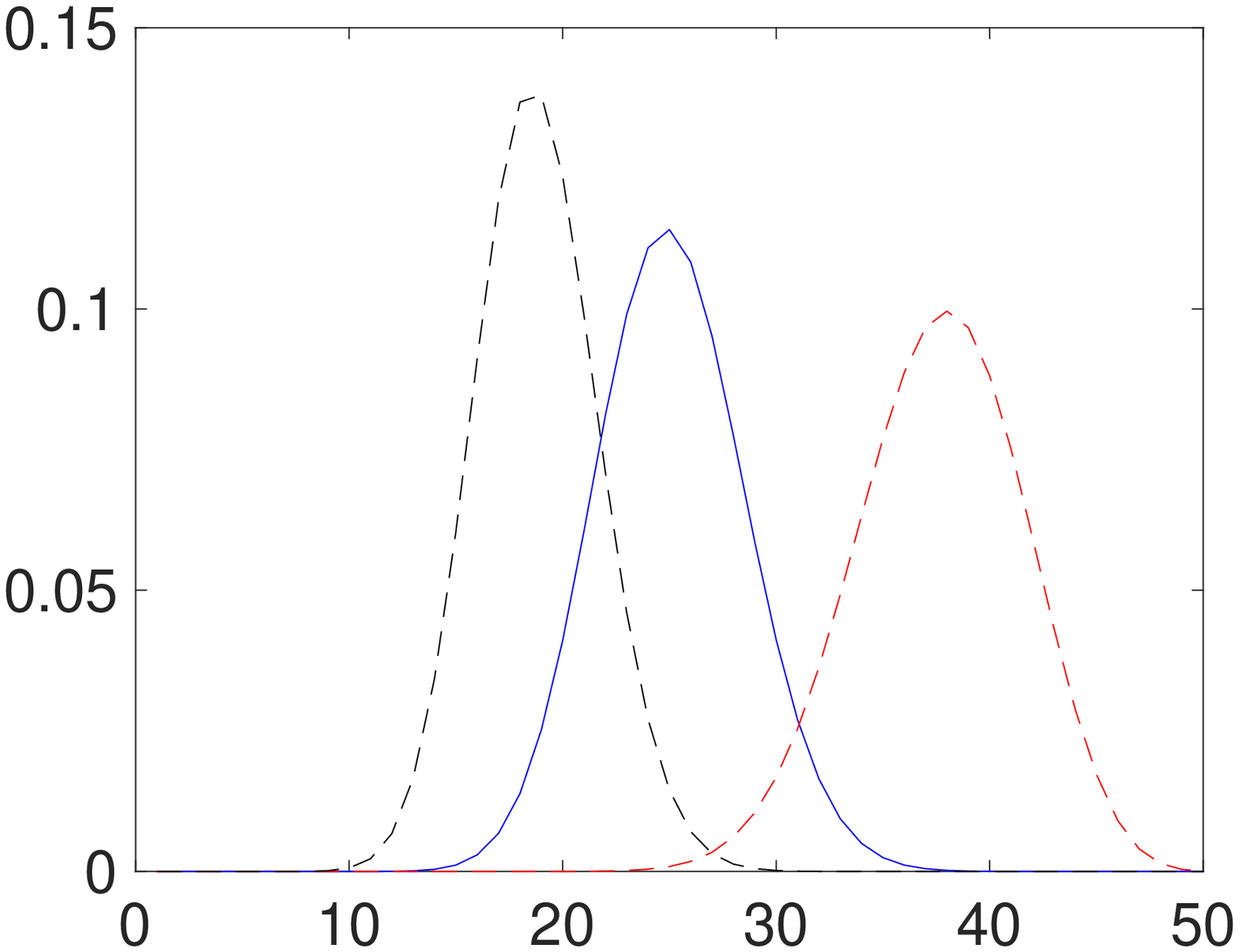} \hspace{-15pt} &
      \includegraphics[width=5cm]{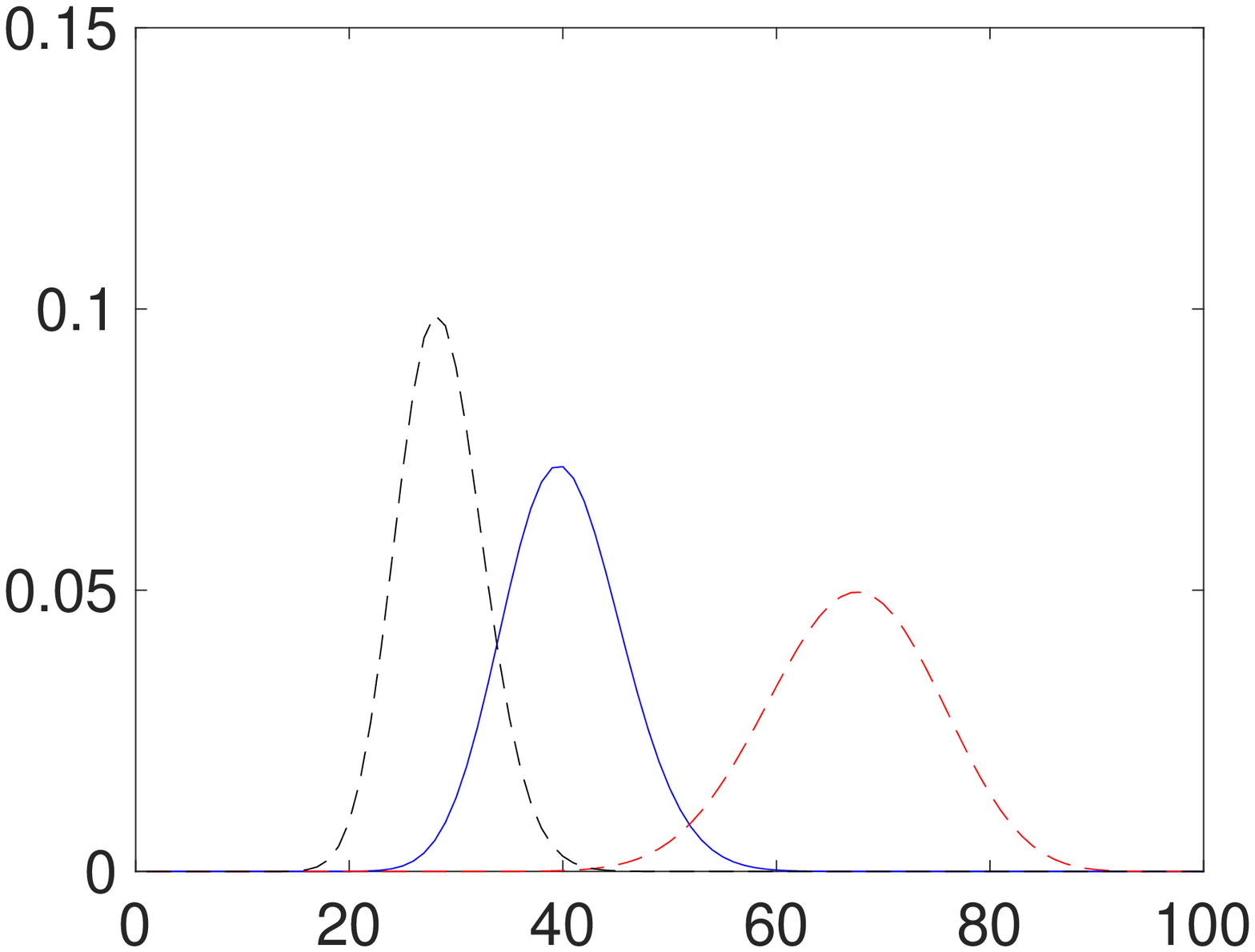}     \\      
  \end{tabular}
 \end{center}
\begin{minipage}[]{0.84\textwidth}
\begin{small}
(ii) HGP with $(\gamma_0,\zeta_0) = (\gamma_i,\zeta_i) = (5,1450)$ (red dashed), $(\gamma_0,\zeta_0) = (\gamma_i,\zeta_i) = (25,1450)$ (black dashed) and symmetric case $(\gamma_0,\zeta_0) = (\gamma_i,\zeta_i) = (15,1450)$ (homogeneous case, solid blue).
\end{small}
\end{minipage}
\end{figure}

\clearpage

\begin{figure}[h!]
  \caption{Prior marginal $\P\{D_{it} = k\}$ (left column) and global $\P\{D_t = k\}$ (right column) number of clusters for different processes.}\label{EffectsTopHierarchy}
 \begin{center}
  \begin{tabular}{cc}
   \includegraphics[width=5cm]{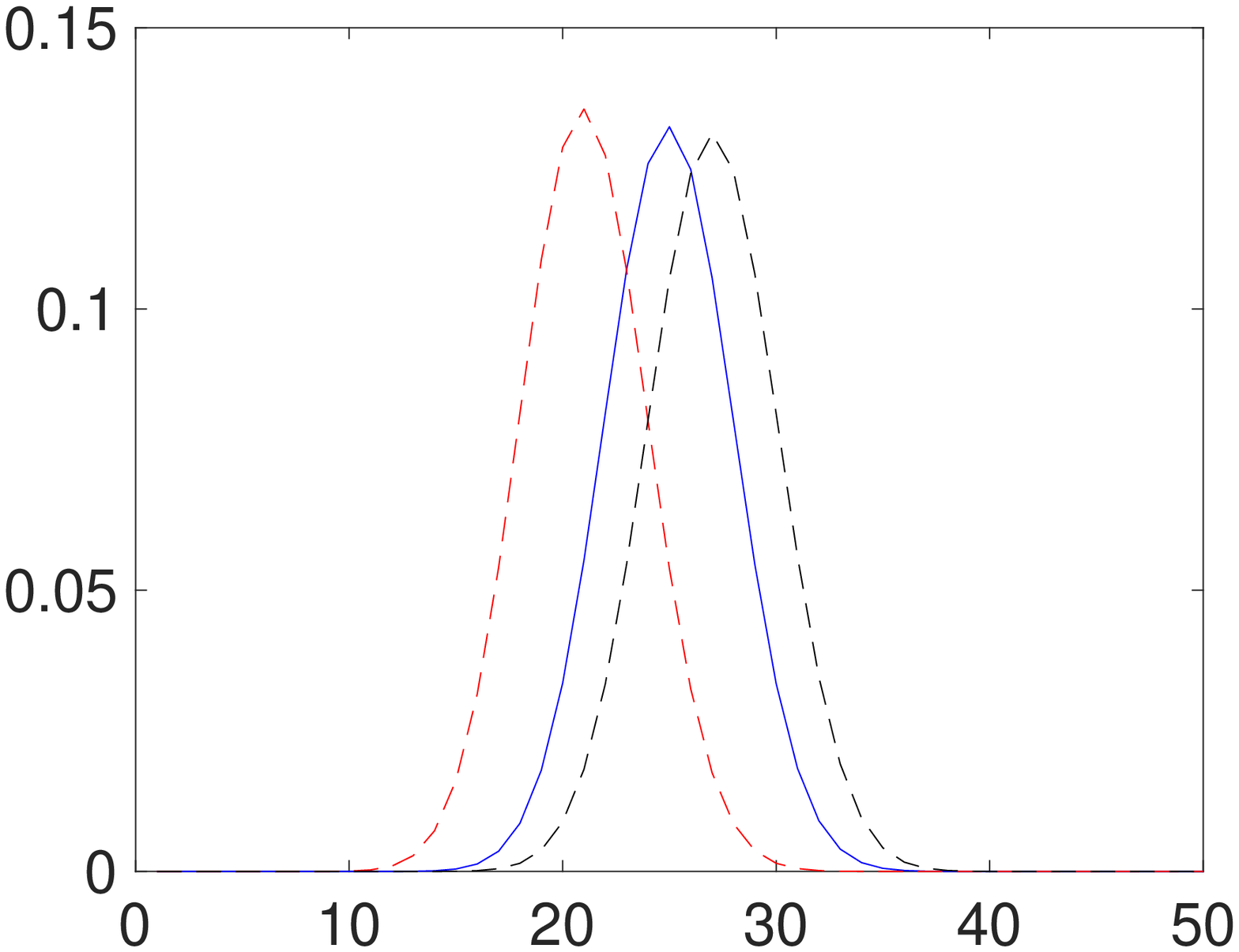}\hspace{-15pt} &
   \includegraphics[width=5cm]{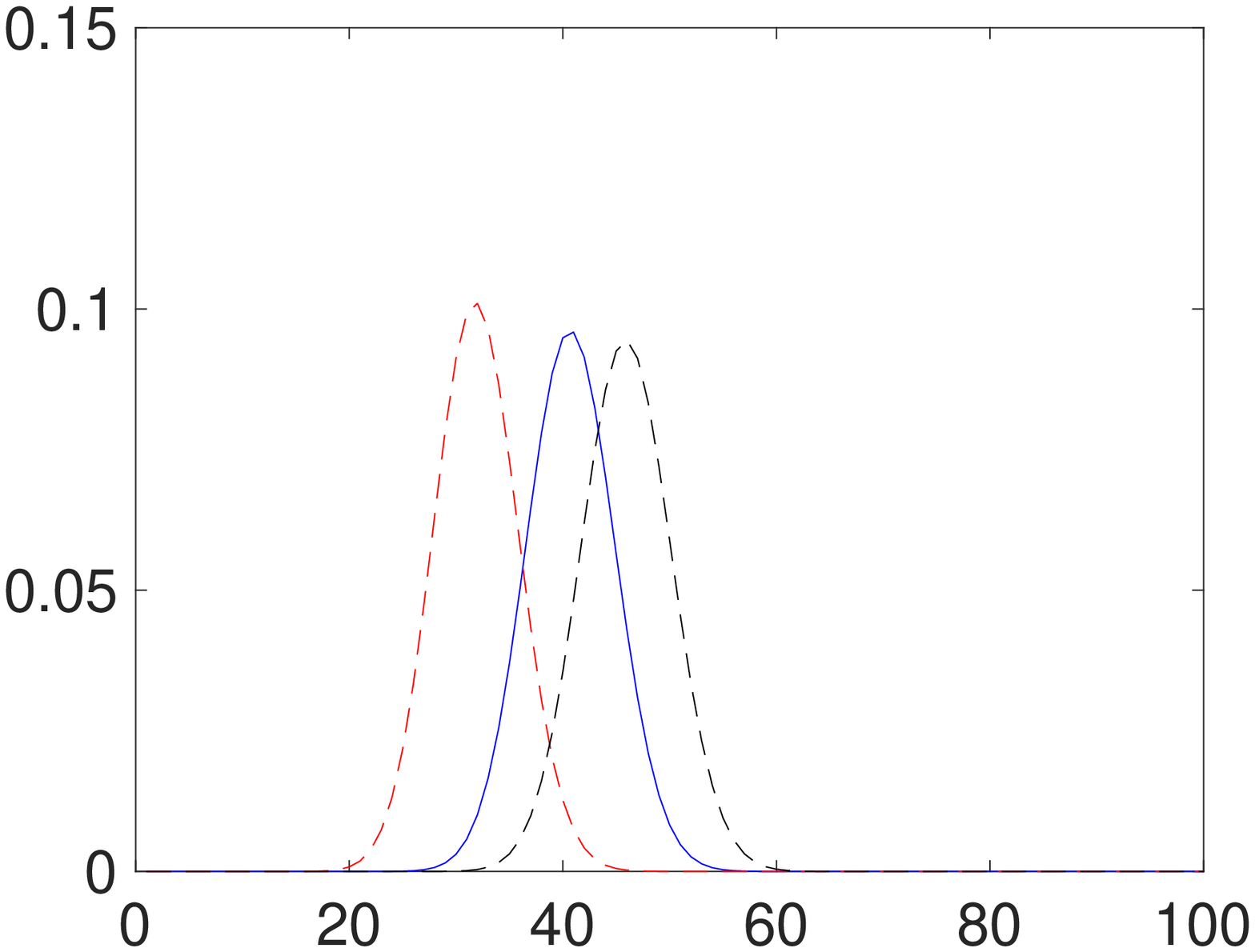}     \\
  \end{tabular}
 \end{center}
\begin{minipage}[]{0.84\textwidth}
\begin{small}
(i) HDP with $\theta_0 = 23.3$, $\theta_i = 43.3$ (red dashed), $\theta_0 = 63.3$, $\theta_i = 43.3$ (black dashed) and $\theta_0 = \theta_i = 43.3$ (homogeneous case, solid blue).
\end{small}
\end{minipage}
 \begin{center}
  \begin{tabular}{cc}
  \includegraphics[width=5cm]{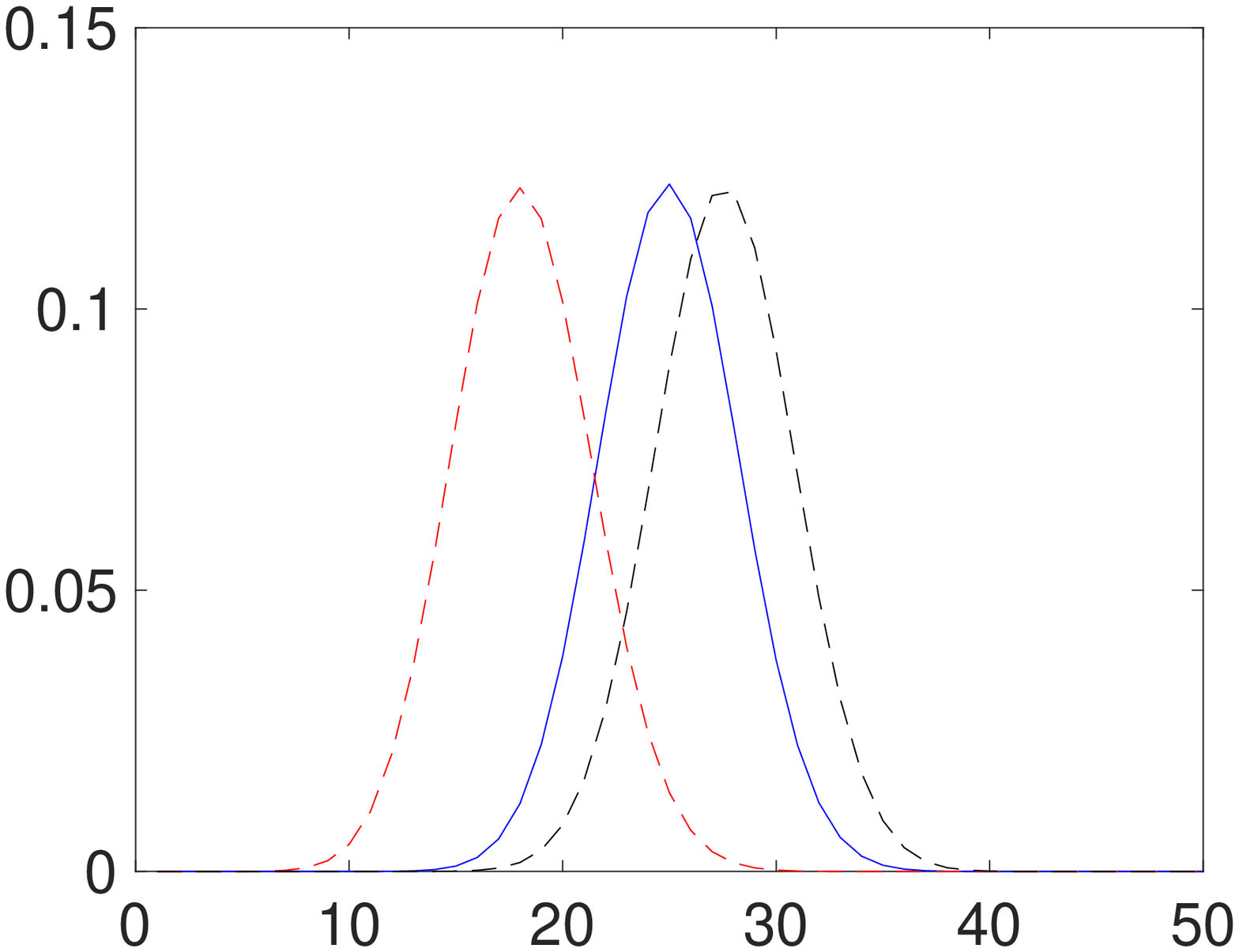}\hspace{-15pt} &
      \includegraphics[width=5cm]{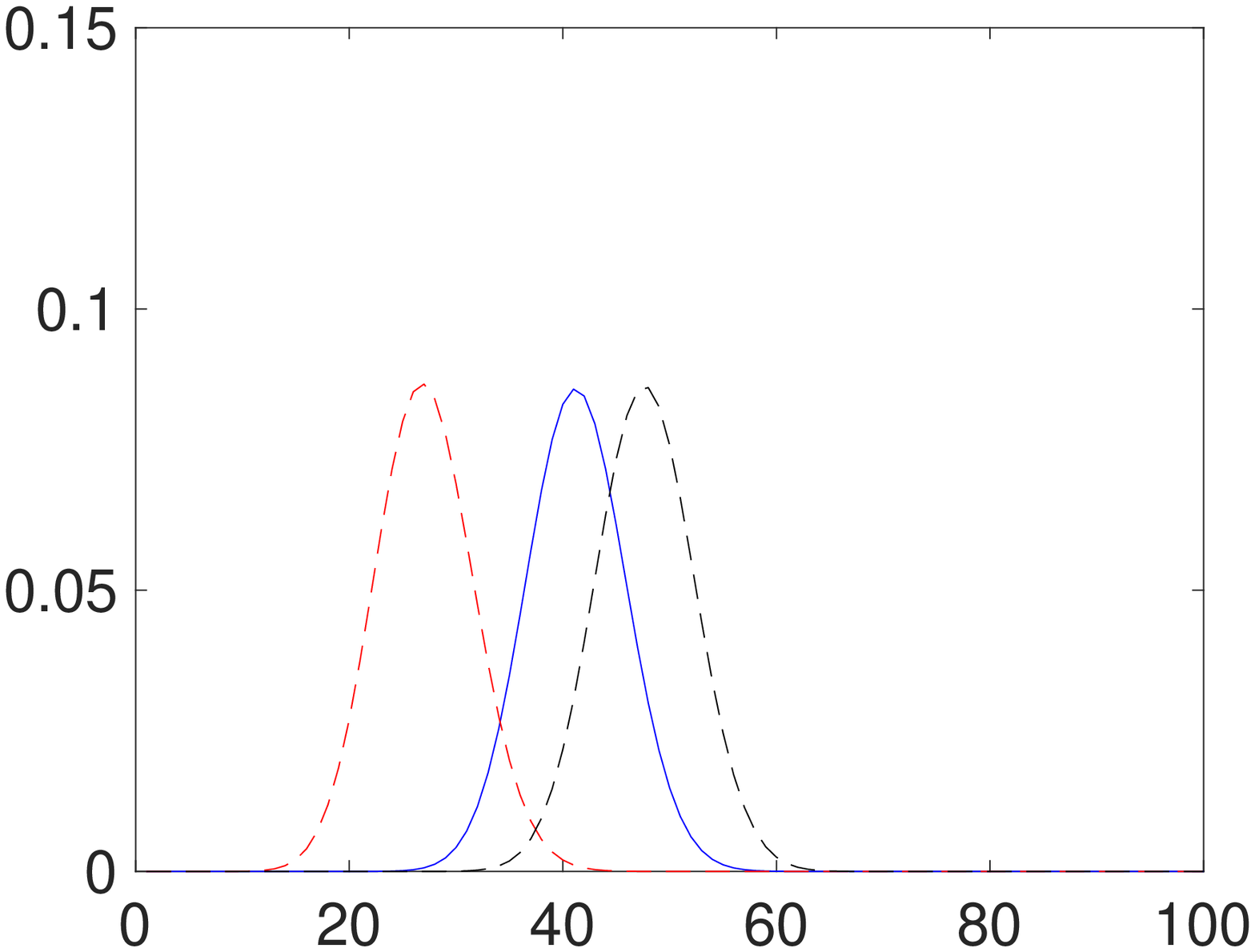}     \\
  \end{tabular}
 \end{center}
\begin{minipage}[]{0.84\textwidth}
\begin{small}
(ii) HPYP with $(\theta_0,\sigma_0) = (9.9,0.25)$, $(\theta_i,\sigma_i) = (29.9,0.25)$ (red dashed), $(\theta_0,\sigma_0) = (49.9,0.25)$, $(\theta_i,\sigma_i) = (29.9,0.25)$ (black dashed) and $(\theta_0,\sigma_0) = (\theta_i,\sigma_i) = (29.9,0.25)$ (homogeneous case, solid blue).
\end{small}
\end{minipage}
 \begin{center}
  \begin{tabular}{cc}
   \includegraphics[width=5cm]{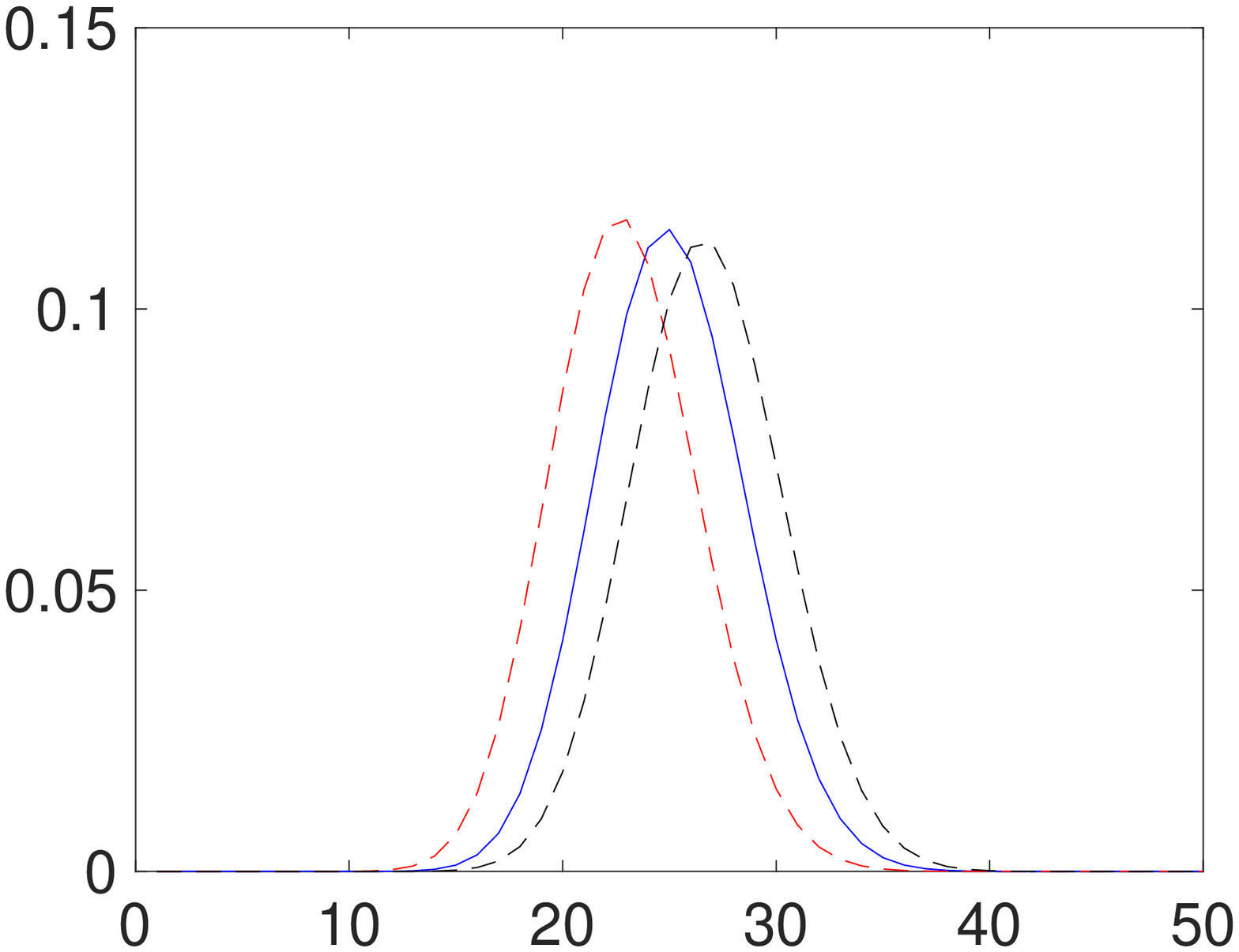} \hspace{-15pt} &
      \includegraphics[width=5cm]{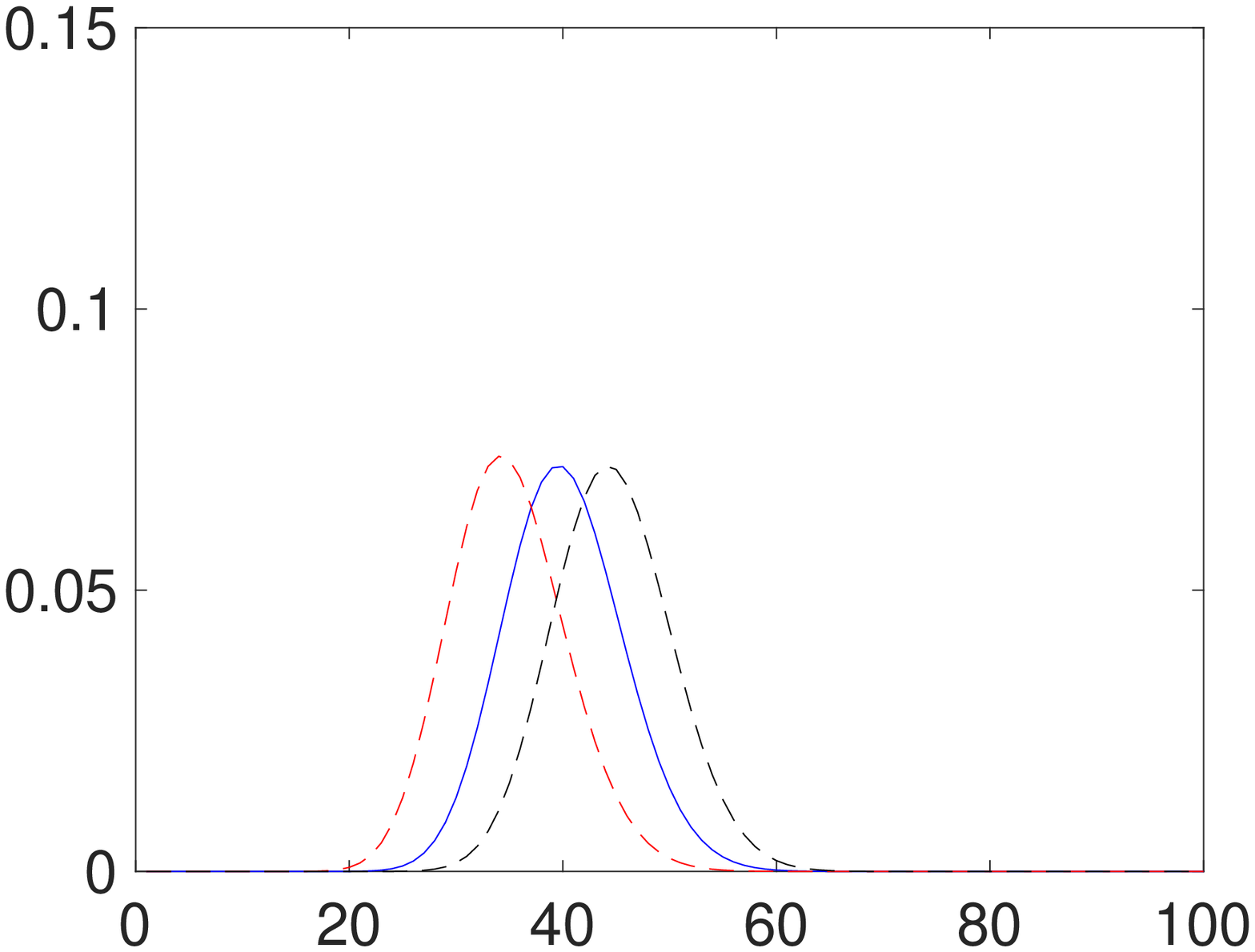}     \\    
  \end{tabular}
 \end{center}
\begin{minipage}[]{0.84\textwidth}
\begin{small}
(iii) HGP with $(\gamma_0,\zeta_0) = (15, 1050)$, $(\gamma_i,\zeta_i) = (15,1450)$ (red dashed), $(\gamma_0,\zeta_0) = (15,1950)$, $(\gamma_i,\zeta_i) = (15, 1450)$ (black dashed) and $(\gamma_0,\zeta_0) = (\gamma_i,\zeta_i) = (15,1450)$ (homogeneous case, solid blue).
\end{small}
\end{minipage}
\end{figure}

\newpage

\section{Further numerical results}\label{App:EmpRes}
\begin{figure}[h!]
 \centering
  \caption[]{Comparison of HDP (black solid), HPYP (black dashed), HGP (black dotted), HDPYP (gray solid), HPYDP (gray dashed), HGDP (dotted gray) and HGPYP (dashed-dotted gray) when $\mathbb{E}[D_{i,t}] = 5$ (left panel) and when $\mathbb{E}[D_{i,t}] = 10$ (right panel) and $t=50$.}
  \label{all_Experm_1}
  \begin{tabular}{cc}
   \includegraphics[width=4cm]{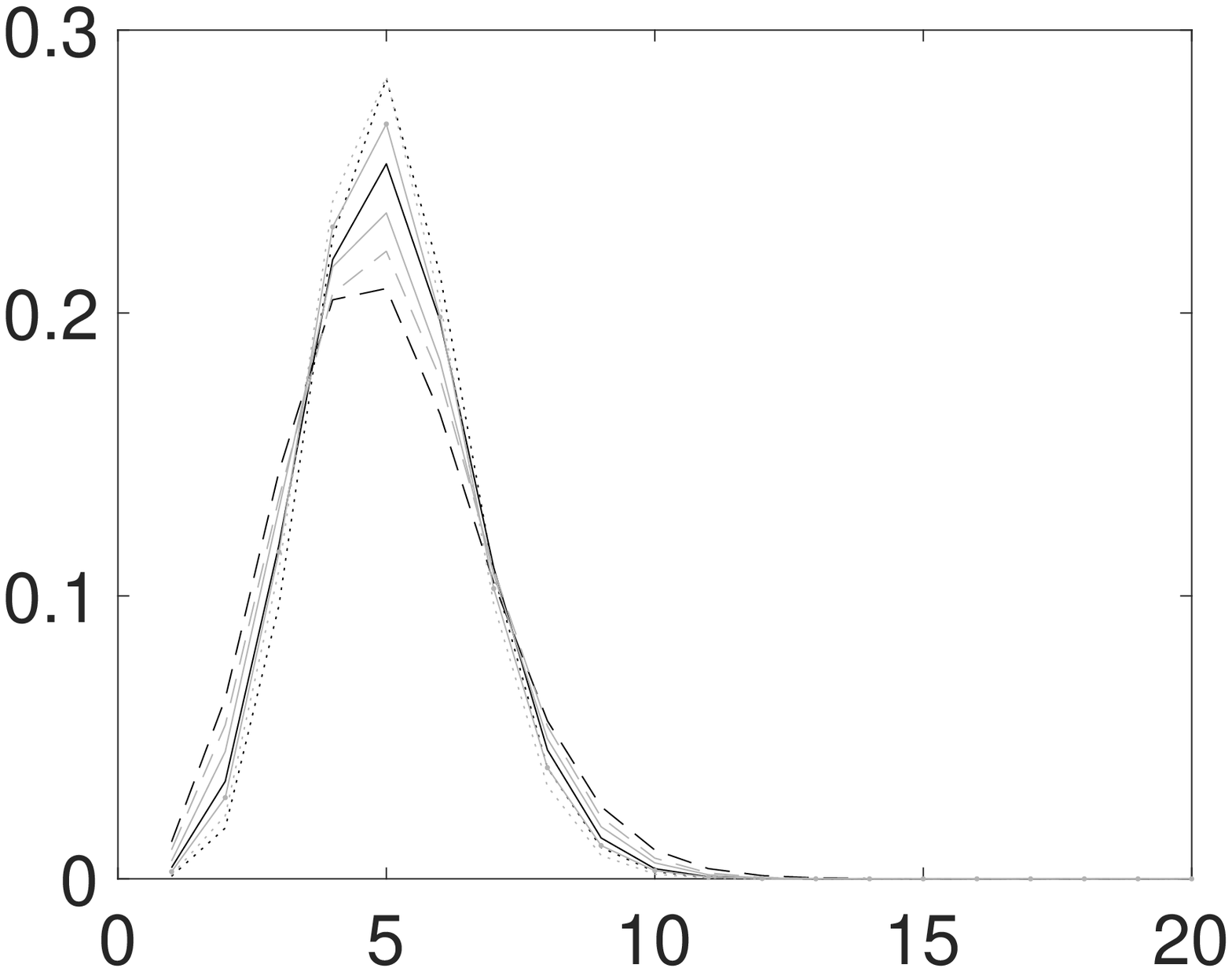} &
   \includegraphics[width=4cm]{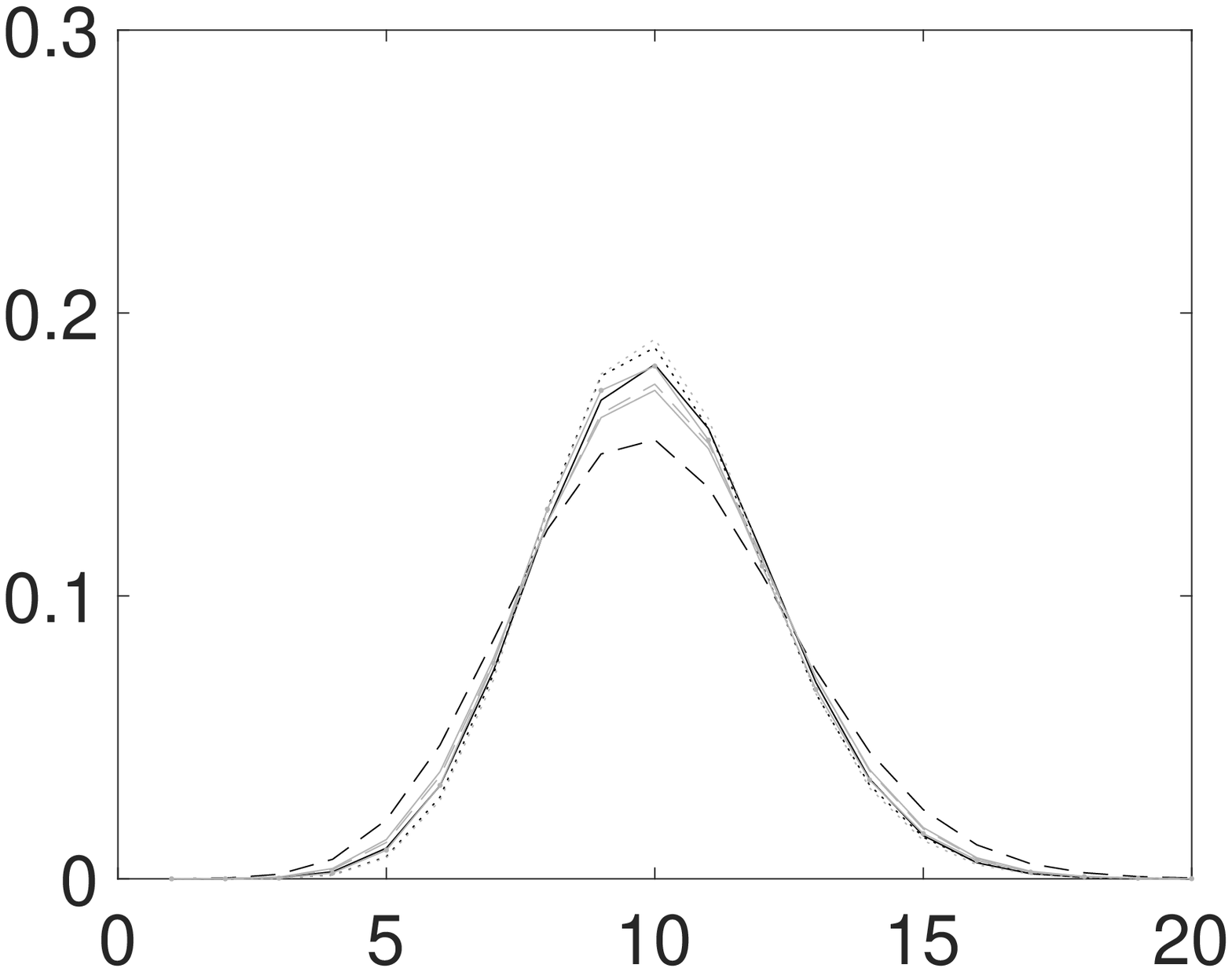}
  \end{tabular}
\end{figure}

\begin{table}[h]
\caption{Marginal and total expected number of clusters ($E(D_{i,t})$ and $E(D_{t})$, respectively), number of groups, $I$ and number of observations per group, $n_{i}$ for seven HSSMs following different parameter ($\sigma_i$, $\theta_i$, $\gamma_i$ and $\zeta_i$) and experimental (panel (a) and (b)) settings.}\label{Tab:priorSIM}
\begin{scriptsize}
\setlength{\tabcolsep}{2pt}
\begin{tabular}{l|cc|cc|cc|cc||cccc}
\hline
&\multicolumn{12}{c}{(a) Three-component normal mixtures}\\
\hline
HSSM & $\sigma_0$ & $\sigma_1$& $\theta_0$ &$\theta_1$ &$\gamma_0$&$\gamma_1$& $\zeta_0$  &$\zeta_1$  & $E[D_{i,t}]$ & $V[D_{i,t}]$ & $I$ &  $n_{i}$\\
\hline
$HDP(\theta_0,\theta_1)$                   &     &      & 3.50 &3.50  &     &      &    &    &5.00& 2.46  & 2 & 50\\
$HPYP(\sigma_0,\theta_0,\sigma_1,\theta_1)$& 0.25& 0.25 & 2.00 & 2.00 &     &      &    &    &5.00 & 3.53 & 2 & 50\\
$HGP(\gamma_0,\zeta_0,\gamma_1,\zeta_1)$  &      &      &      &      &13.50& 13.50& 140& 140&5.00 & 2.04 & 2 & 50\\
$HDPYP(\theta_0,\sigma_1,\theta_1)$       &      &  0.23& 3.30 &  2.00&     &      &    &    &5.00 & 2.81 & 2 & 50\\
$HPYDP(\sigma_0,\theta_0,\theta_1)$       &  0.22&     &  2.00 &  3.85&     &      &    &    &5.00 & 3.13 & 2 & 50\\
$HGDP(\gamma_0,\zeta_0,\theta_1)$          &      &     &       & 3.30 &14.40&      &135&    &5.00 & 1.97 & 2 & 50\\
$HGPYP(\gamma_0,\zeta_0,\sigma_1,\theta_1)$&      & 0.23&       & 2.00 &14.71&      &130&    &5.00 & 2.24 & 2 & 50\\
\hline
&\multicolumn{12}{c}{(b) Two-component normal mixtures}\\
\hline
$HDP(\theta_0,\theta_1)$                   &     &      & 9.10 &9.10  &     &      &    &    &10.00& 4.78 & 2 & 50\\
$HPYP(\sigma_0,\theta_0,\sigma_1,\theta_1)$& 0.25& 0.25 & 5.73 & 5.73 &     &      &    &    &10.00& 6.53 & 2 & 50\\
$HGP(\gamma_0,\zeta_0,\gamma_1,\zeta_1)$  &      &      &      &      &18.00& 18.00& 425& 425&10.00 & 4.54 & 2 & 50\\
$HDPYP(\theta_0,\sigma_1,\theta_1)$       &      & 0.22 & 6.50 &  6.20&     &      &    &    &10.00& 5.28 & 2 & 50\\
$HPYDP(\sigma_0,\theta_0,\theta_1)$       &  0.22&     &  10.40&  6.20&     &      &    &    &10.00& 5.16 & 2 & 50\\
$HGDP(\gamma_0,\zeta_0,\theta_1)$          &      &     &       & 9.00&18.00&      &400&     &10.00& 4.37 & 2 & 50\\
$HGPYP(\gamma_0,\zeta_0,\sigma_1,\theta_1)$&      &0.21 &       & 5.80&18.00&      &400&    & 10.00& 4.84 & 2 & 50\\
\hline
\end{tabular}
\end{scriptsize}
\end{table}

\clearpage

\begin{figure}[h!]
  \caption{Global co-clustering matrix for the three-component normal (panel (a)) and the two-component normal (panel (b)) mixture experiments.}
  \label{Fig:coclust}
 \centering
  \setlength{\tabcolsep}{-2pt} 
  \begin{tabular}{cccc}
      \multicolumn{4}{c}{(a) Three-component normal mixture experiment}\\
      \footnotesize{HDP}& \footnotesize{HPYP}& \footnotesize{HGP}&\vspace{-3pt}\\
      \includegraphics[width=3.5cm]{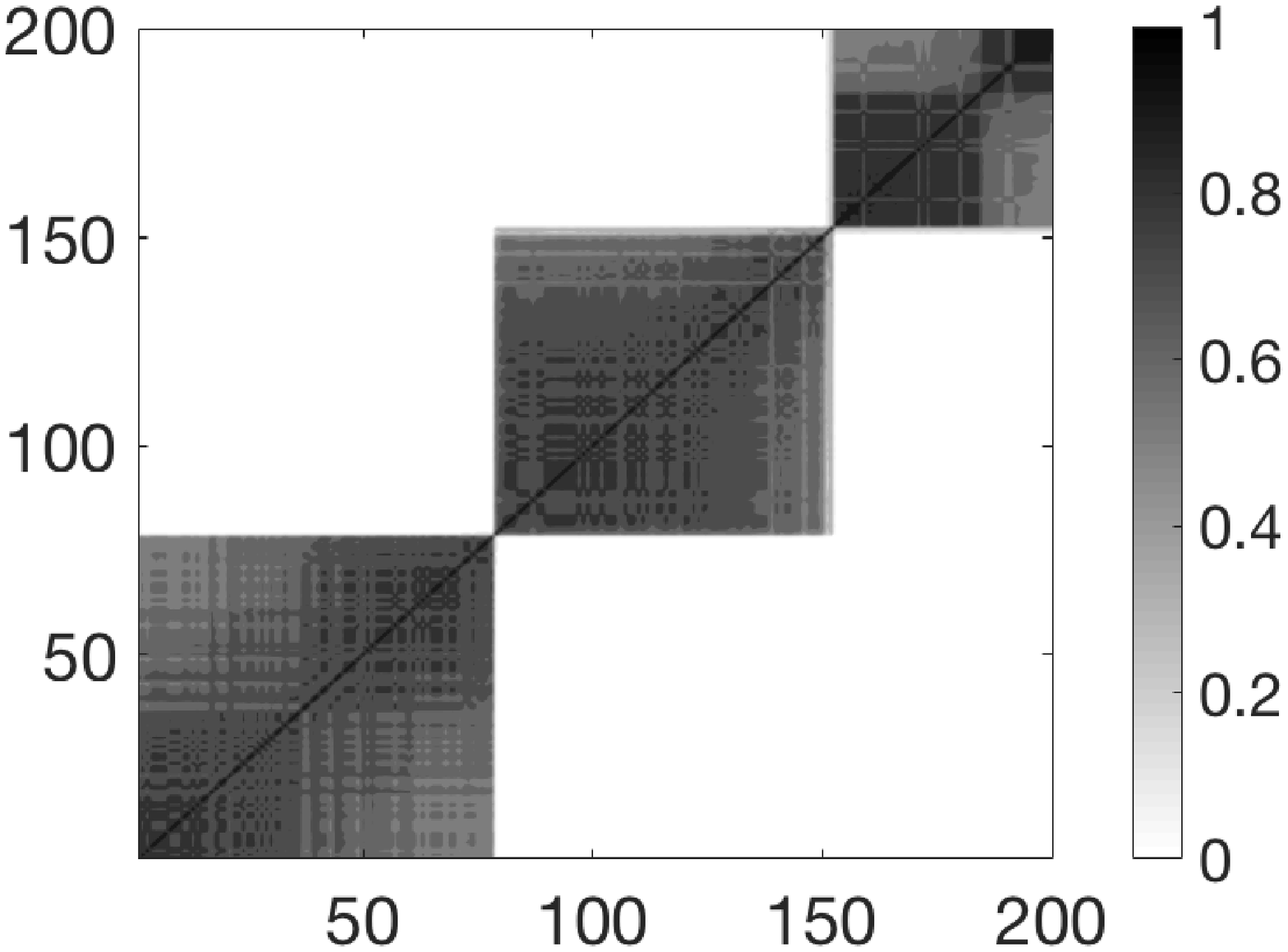}&      \includegraphics[width=3.5cm]{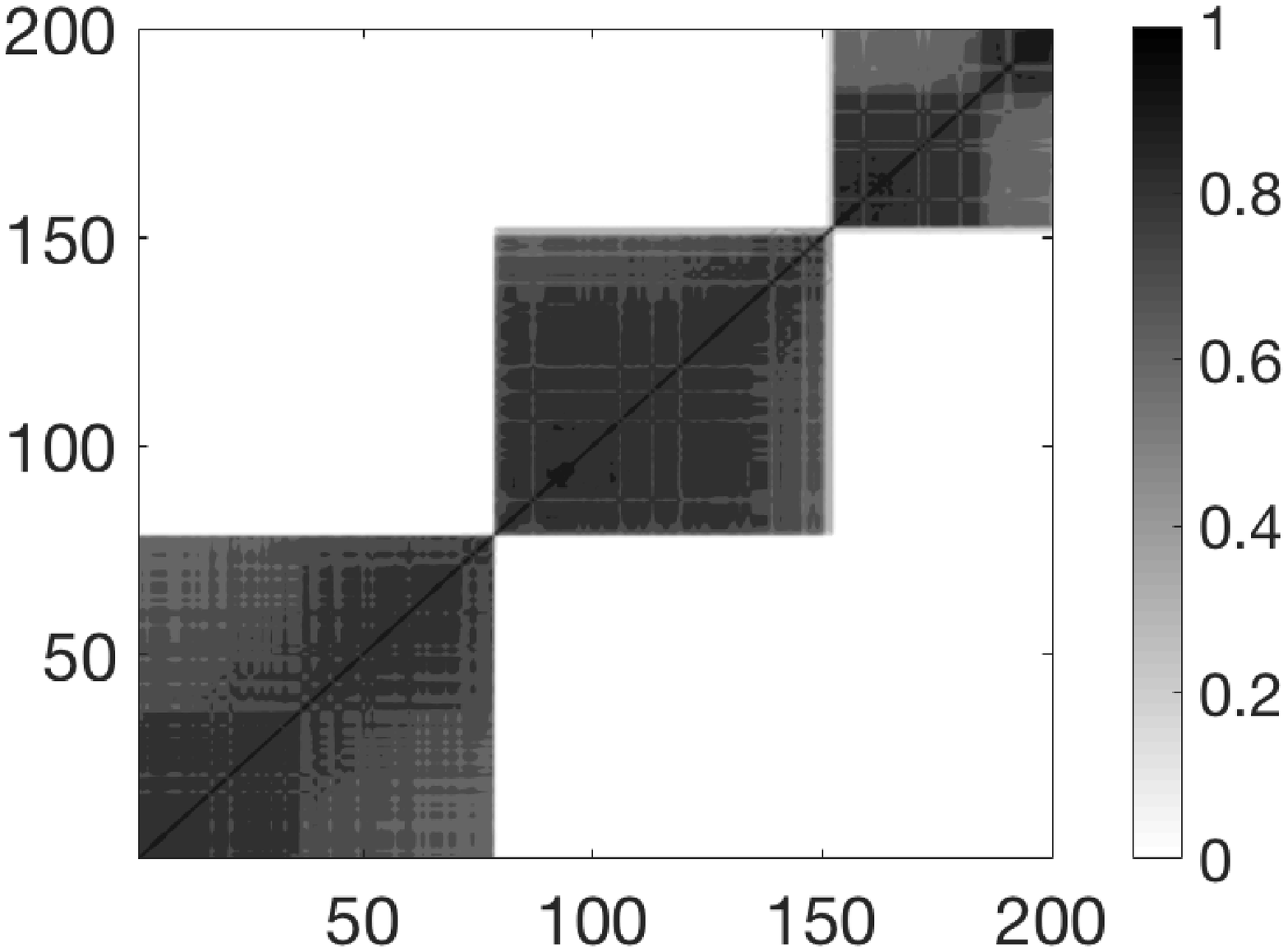}&      \includegraphics[width=3.5cm]{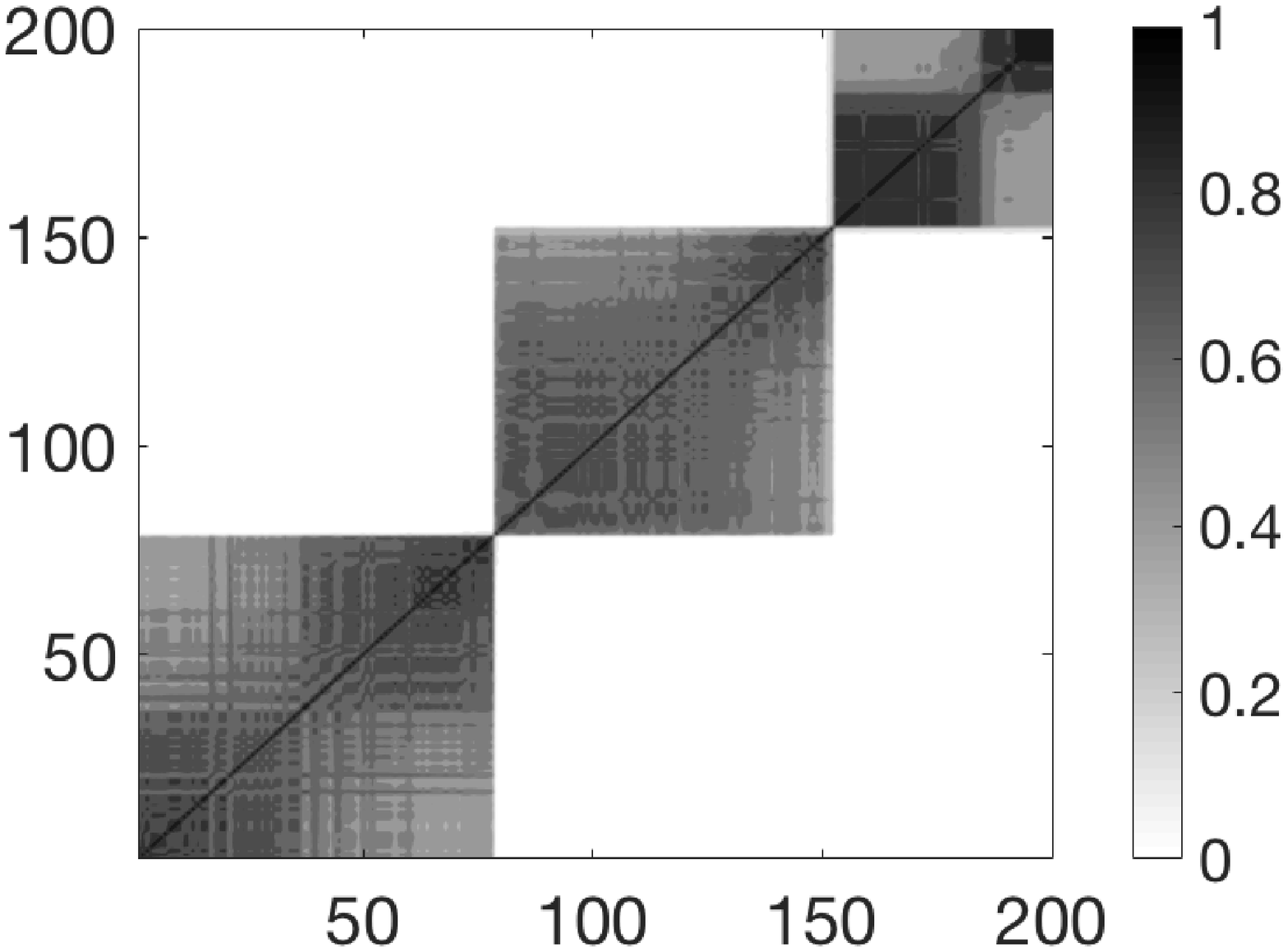}&\\      
      \footnotesize{HDPYP}& \footnotesize{HPYDP}& \footnotesize{HGDP}&\footnotesize{HGPYP}\vspace{-3pt}\\
      \includegraphics[width=3.5cm]{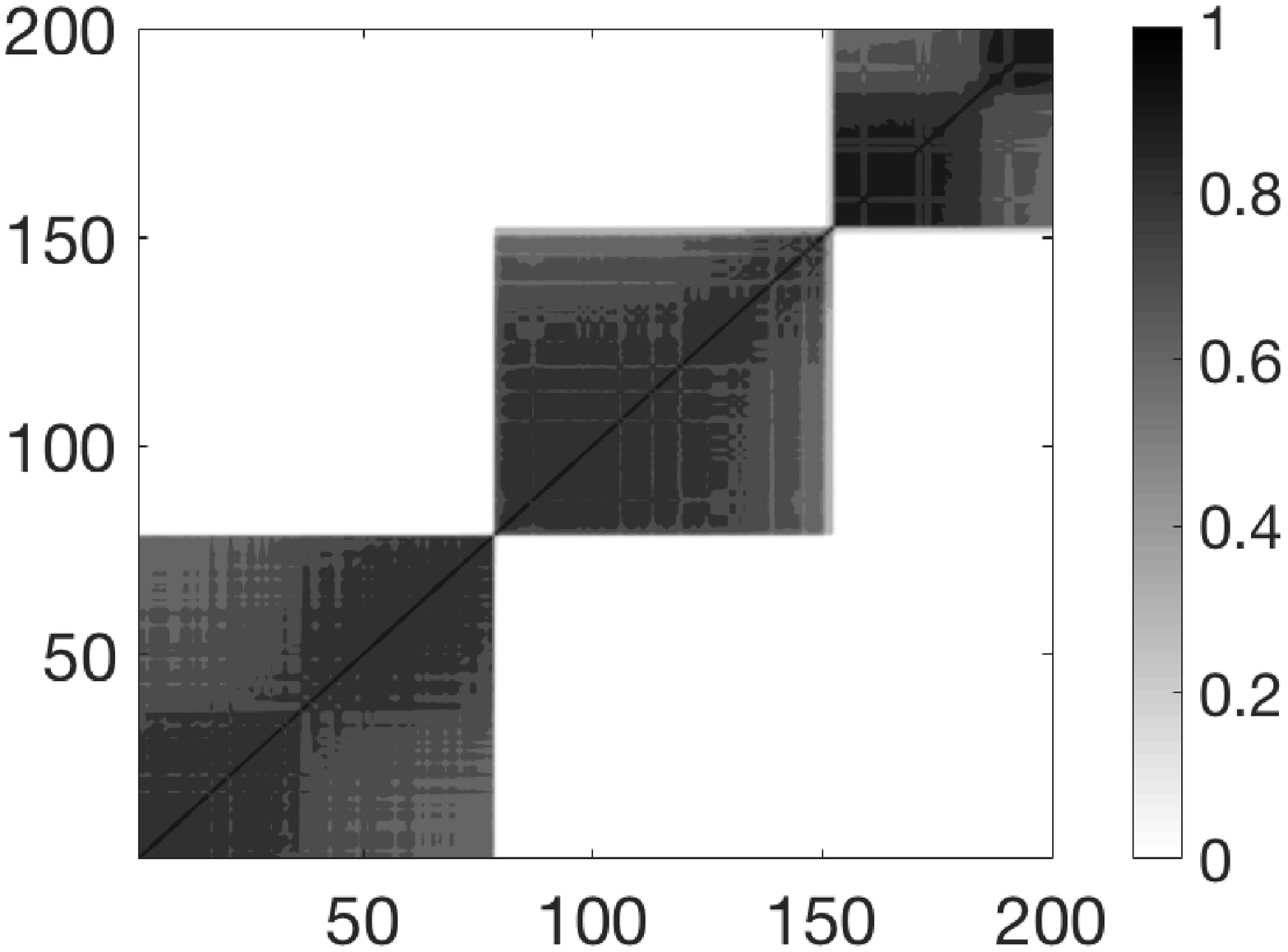}&         \includegraphics[width=3.5cm]{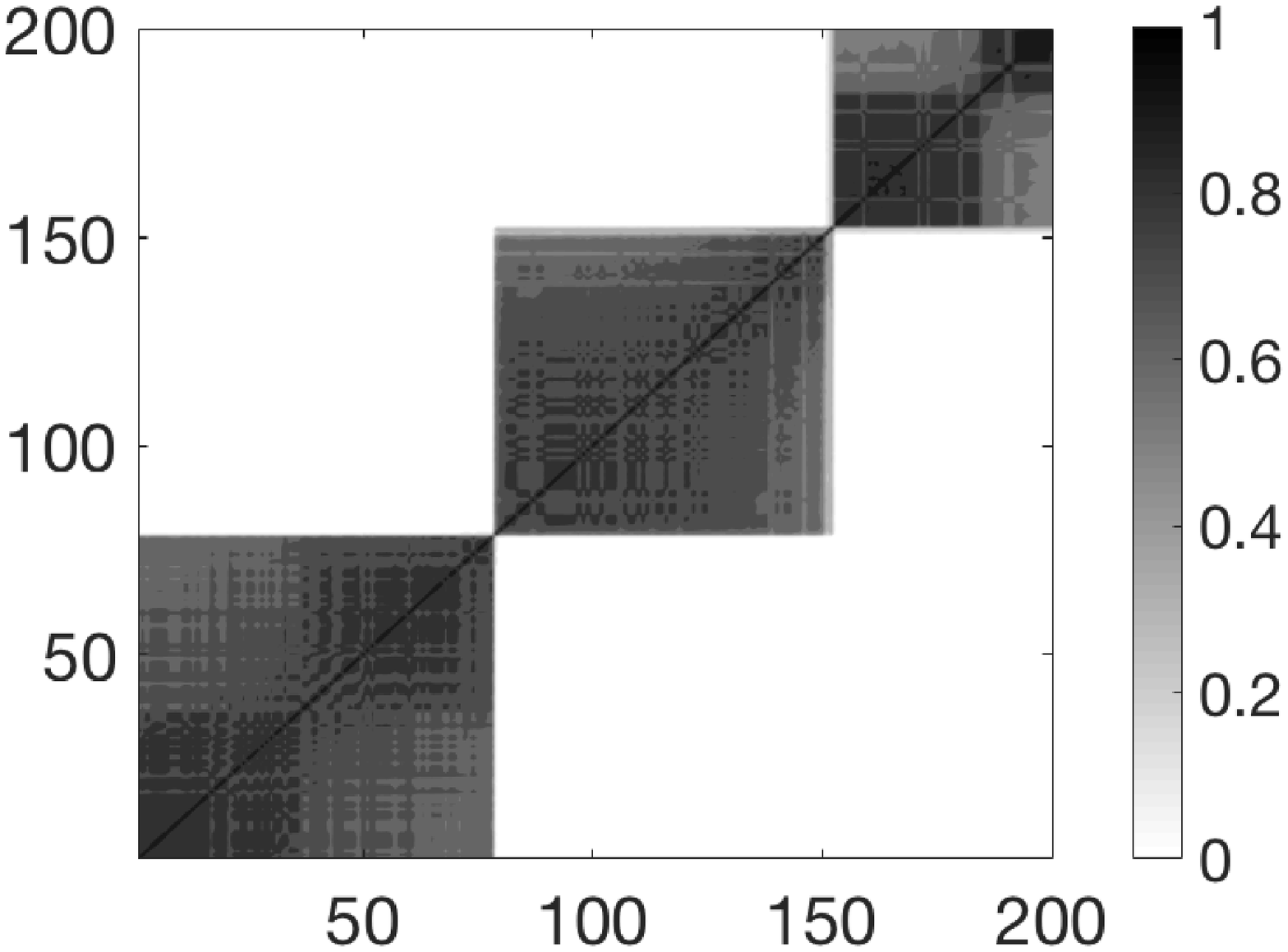}&     
      \includegraphics[width=3.5cm]{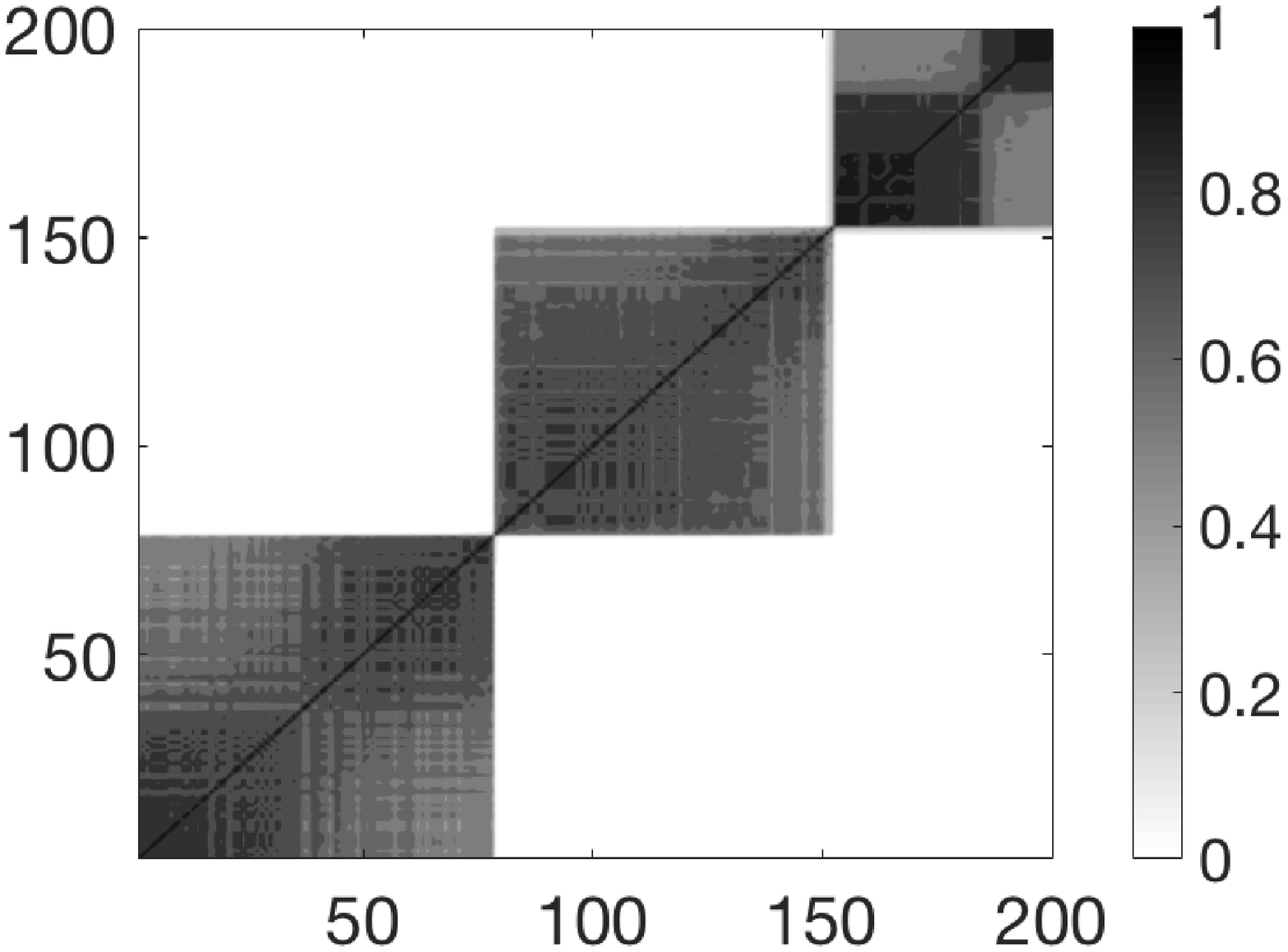}&\includegraphics[width=3.5cm]{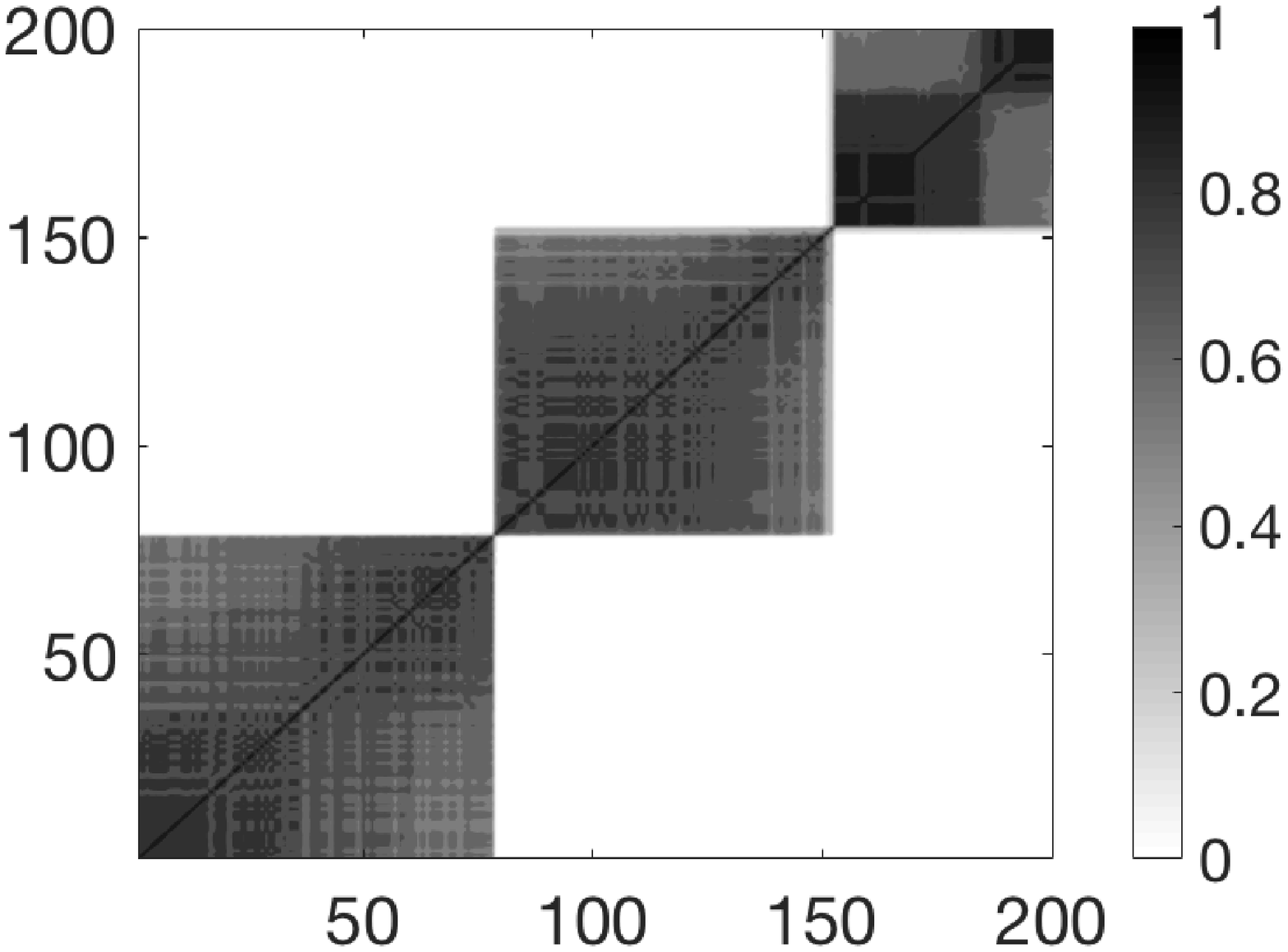}\vspace{20pt}\\
      \multicolumn{3}{c}{(b) Two-component normal mixture experiment}\\
      \footnotesize{HDP}& \footnotesize{HPYP}& \footnotesize{HGP}&\vspace{-3pt}\\
      \includegraphics[width=3.5cm]{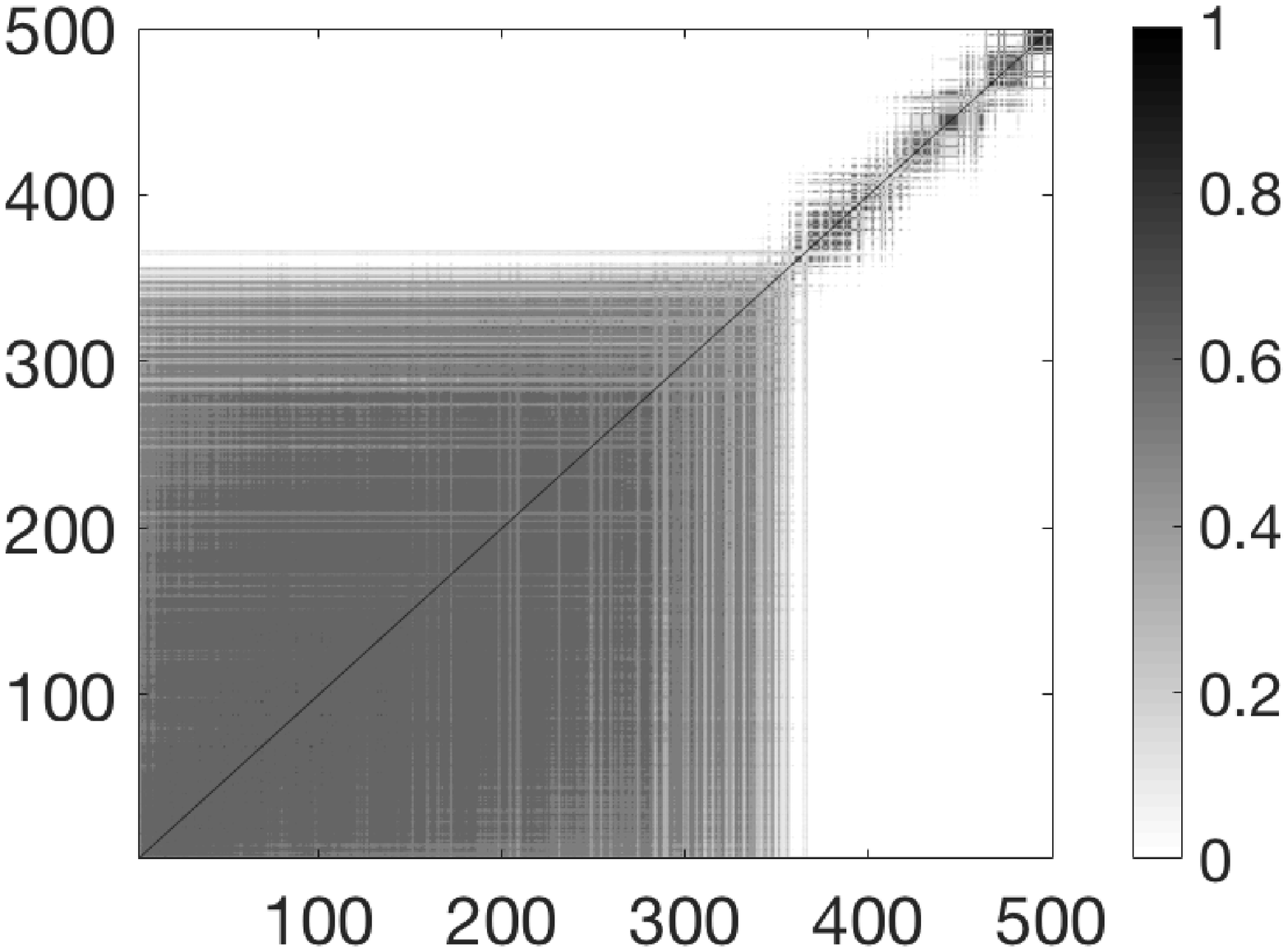}&      \includegraphics[width=3.5cm]{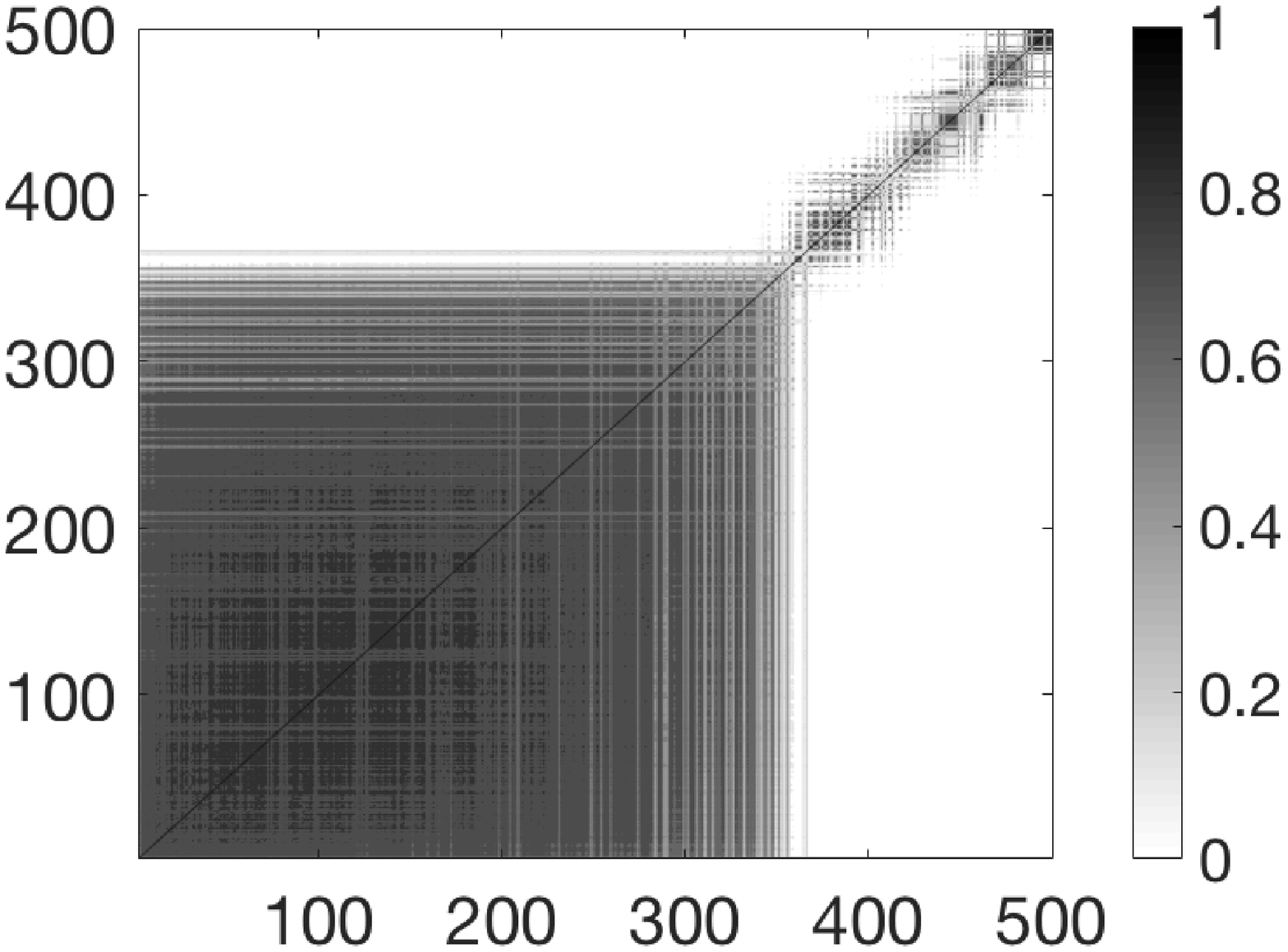}&      \includegraphics[width=3.5cm]{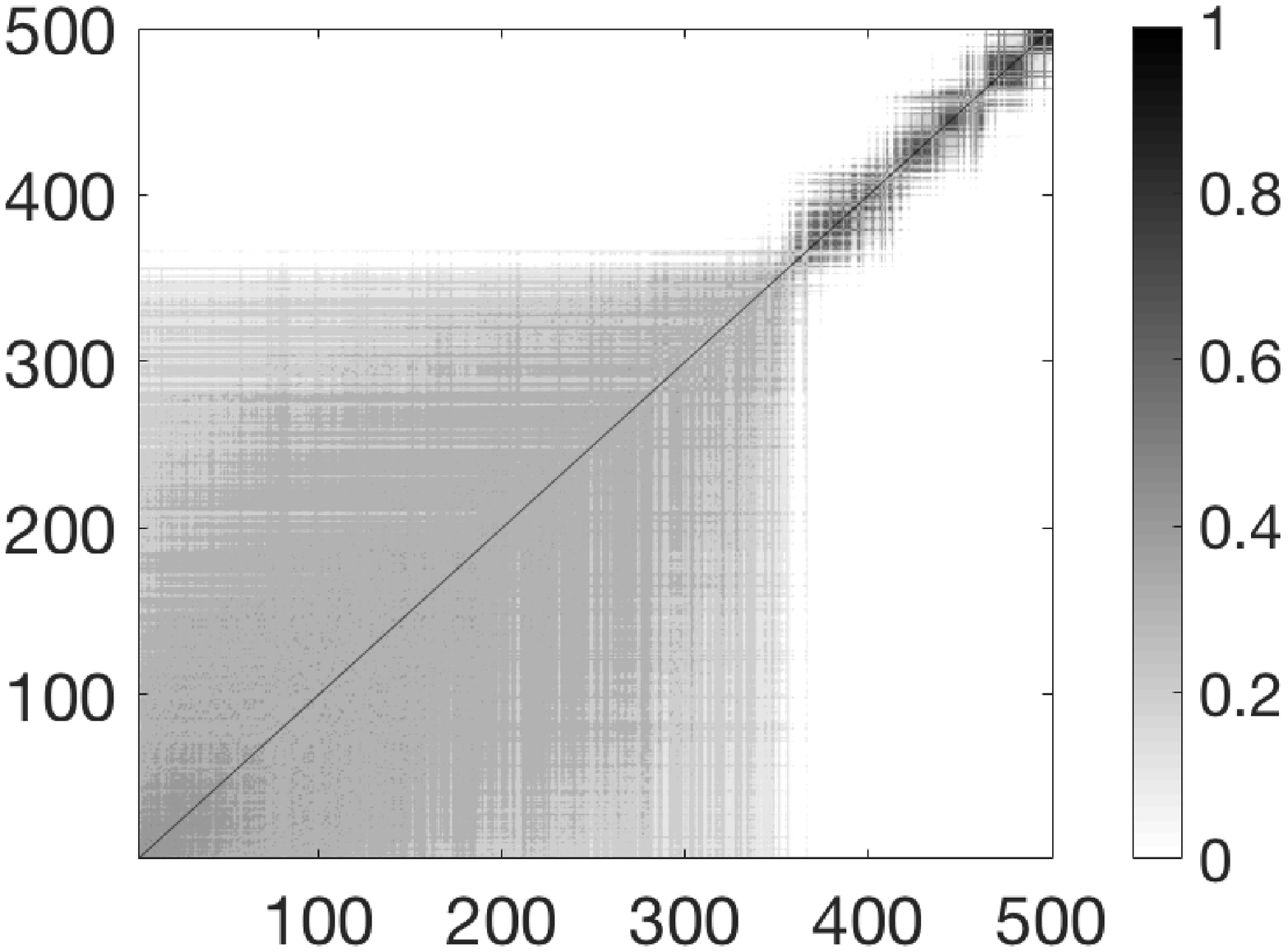}\\      
      \footnotesize{HDPYP}& \footnotesize{HPYDP}& \footnotesize{HGDP}&\footnotesize{HGPYP}\vspace{-3pt}\\
      \includegraphics[width=3.5cm]{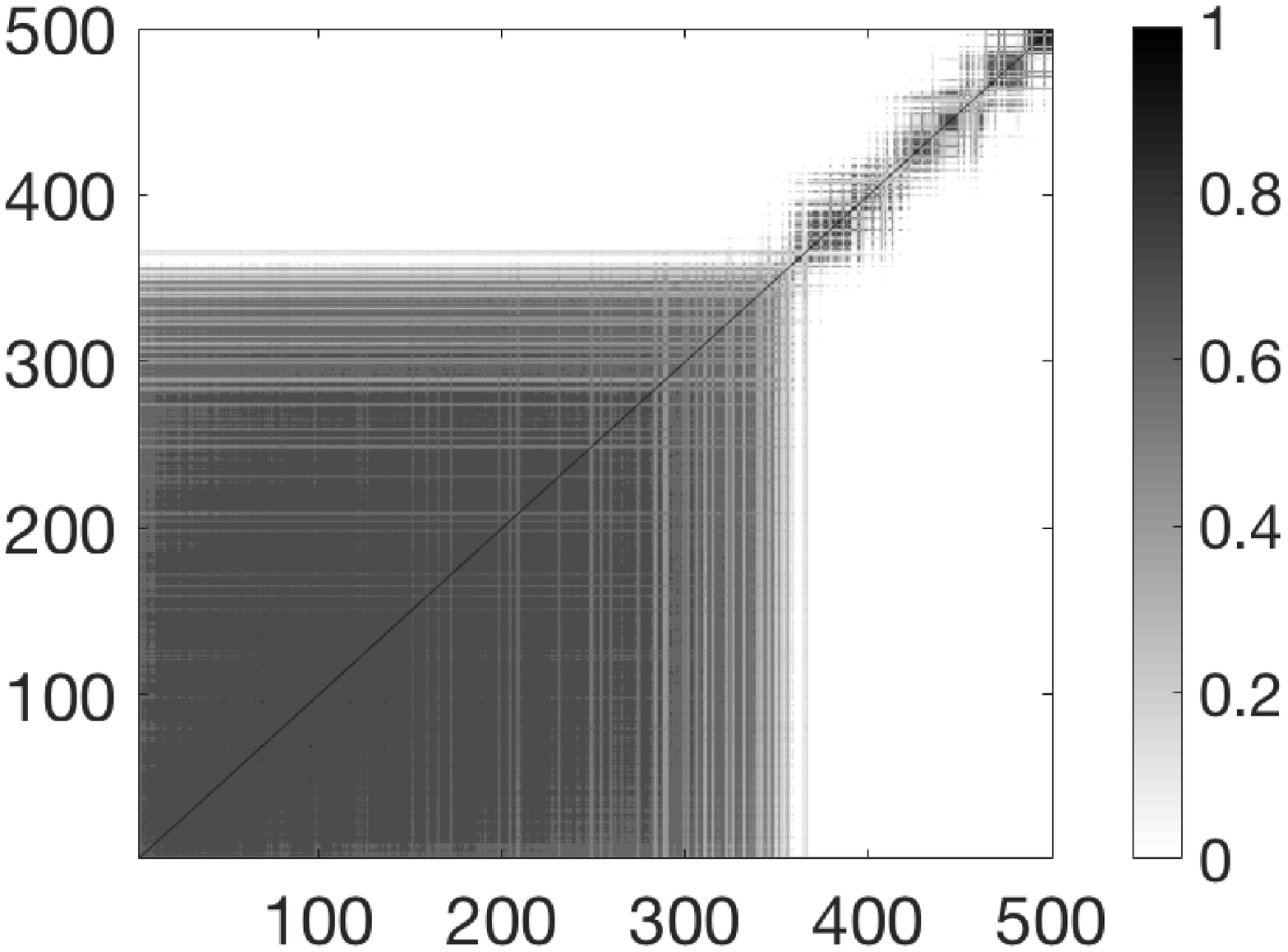}&         \includegraphics[width=3.5cm]{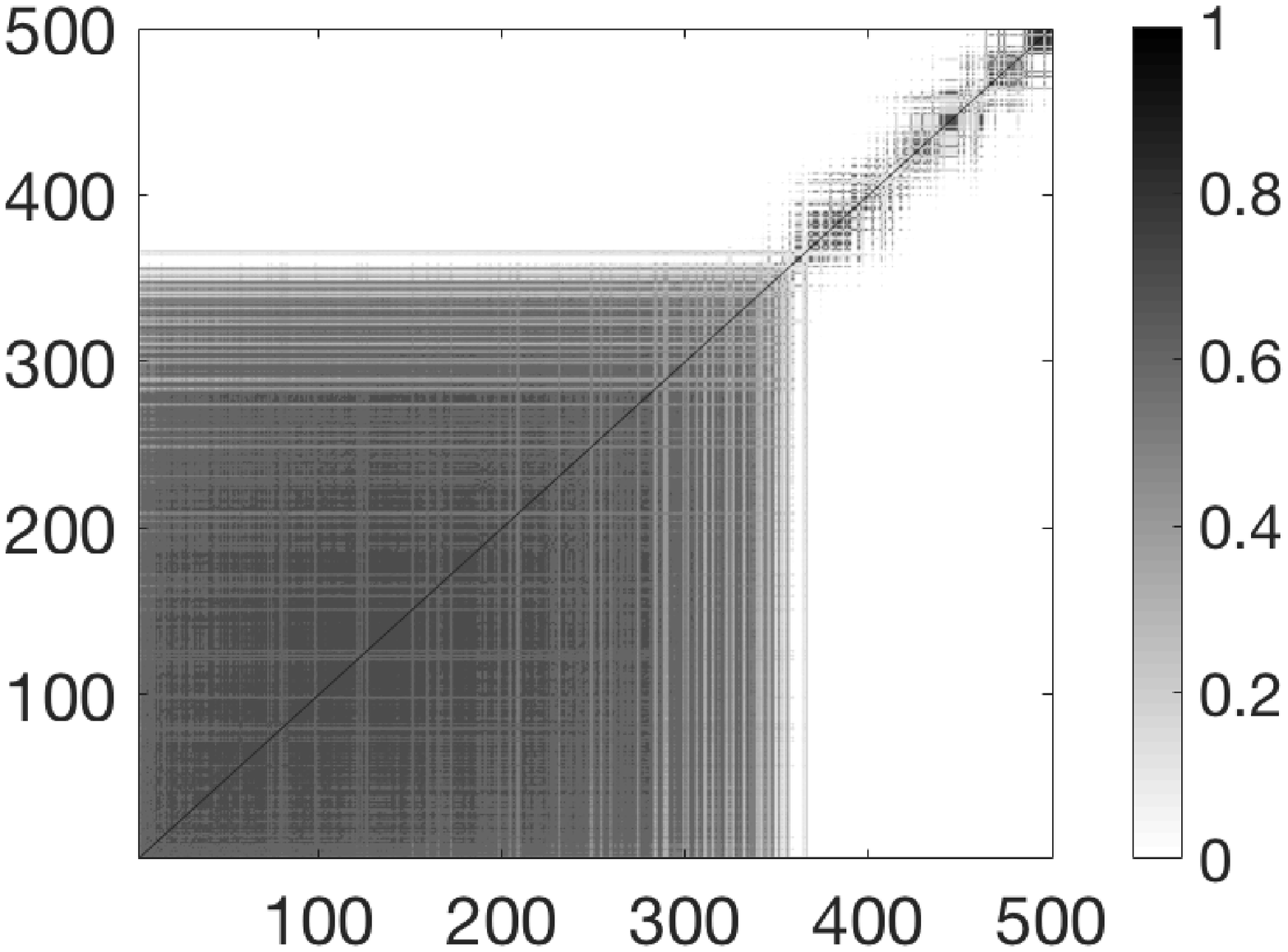}&     
      \includegraphics[width=3.5cm]{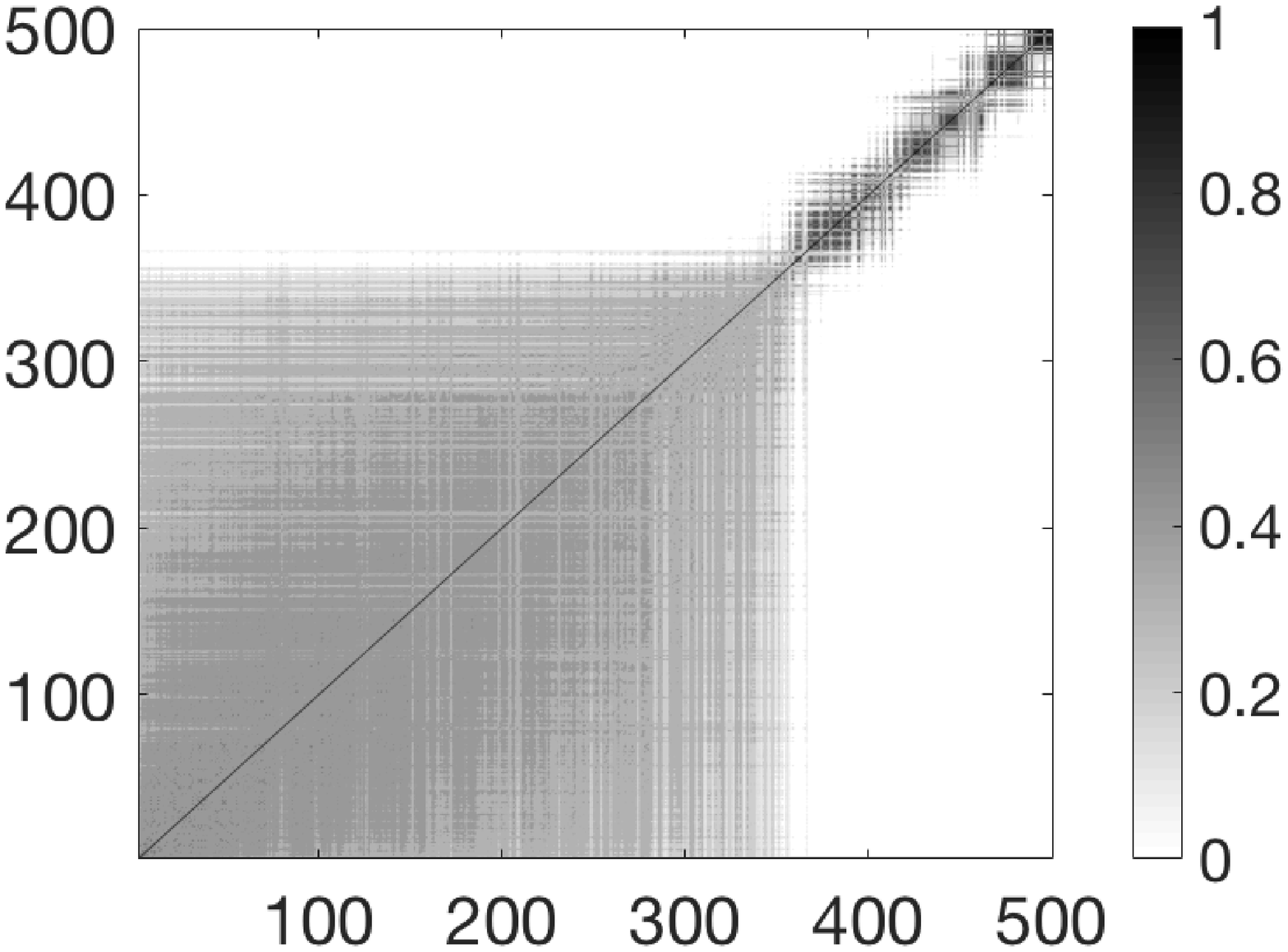}&\includegraphics[width=3.5cm]{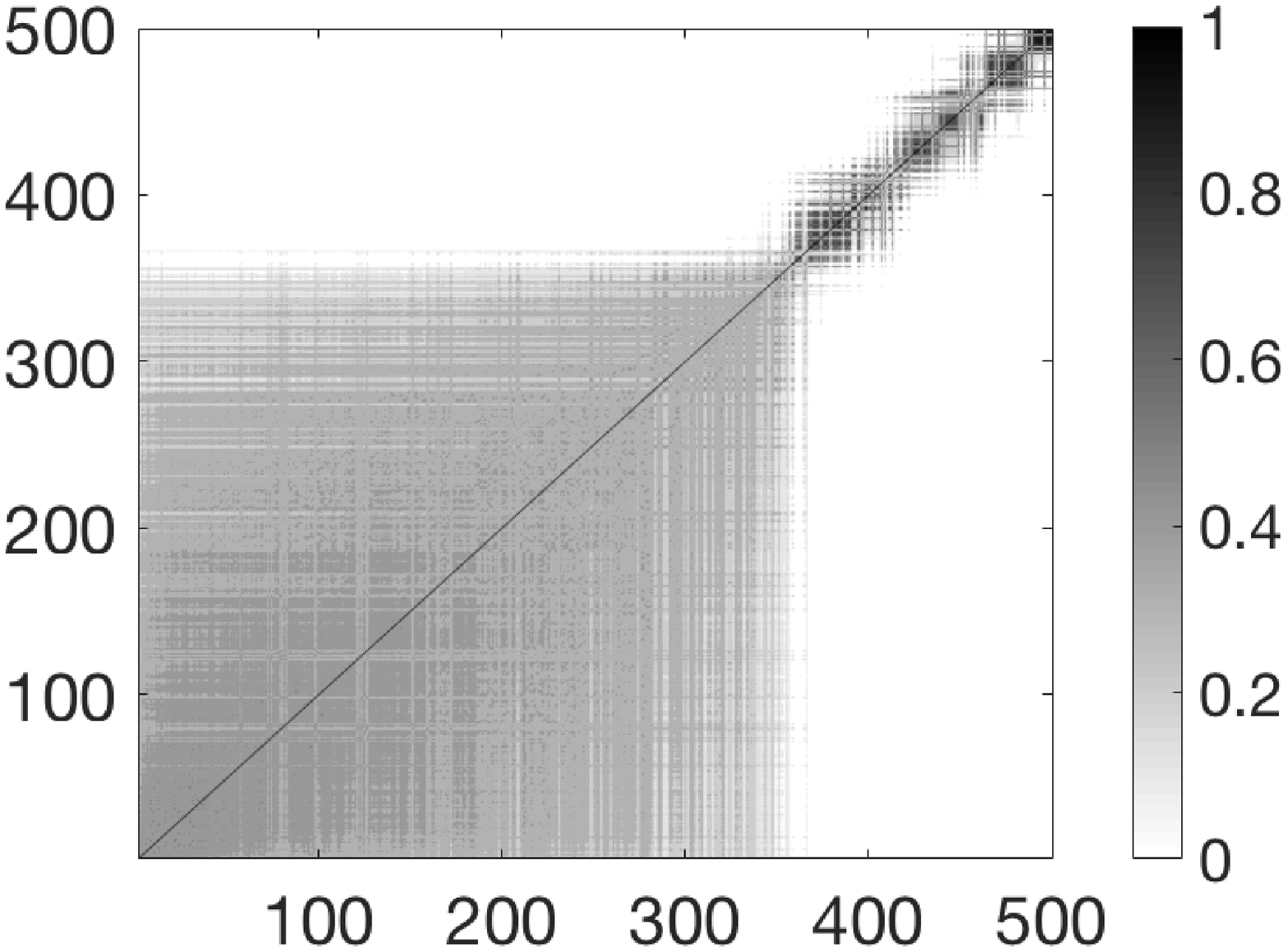}\\
  \end{tabular}
\end{figure}

\clearpage

\begin{figure}[p]
  \caption{Restaurant co-clustering matrix for the for the three-component (panel (a)) and two-component (panel (b)) normal mixture experiments. Red lines denote the co-clustering within (blocks on the minor diagonal) and between restaurants (blocks out of the minor diagonal).}
  \label{Fig:coclust_r}
 \centering
  \setlength{\tabcolsep}{-2pt} 
  \begin{tabular}{cccc}
      \multicolumn{4}{c}{(a) Three-component normal mixture experiment}\\
      \footnotesize{HDP}& \footnotesize{HPYP}& \footnotesize{HGP}&\vspace{-3pt}\\
      \includegraphics[width=3.5cm]{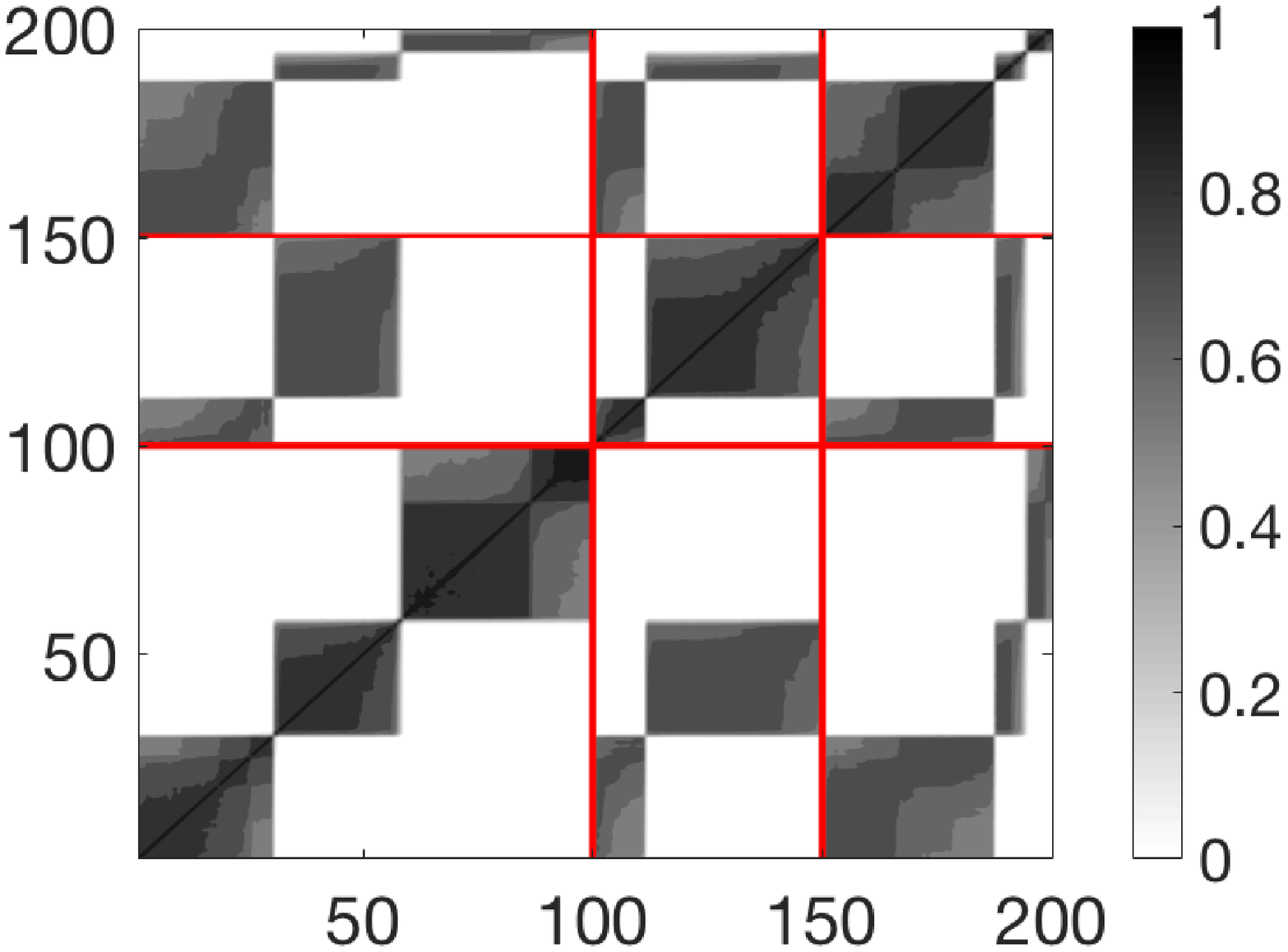}&      \includegraphics[width=3.5cm]{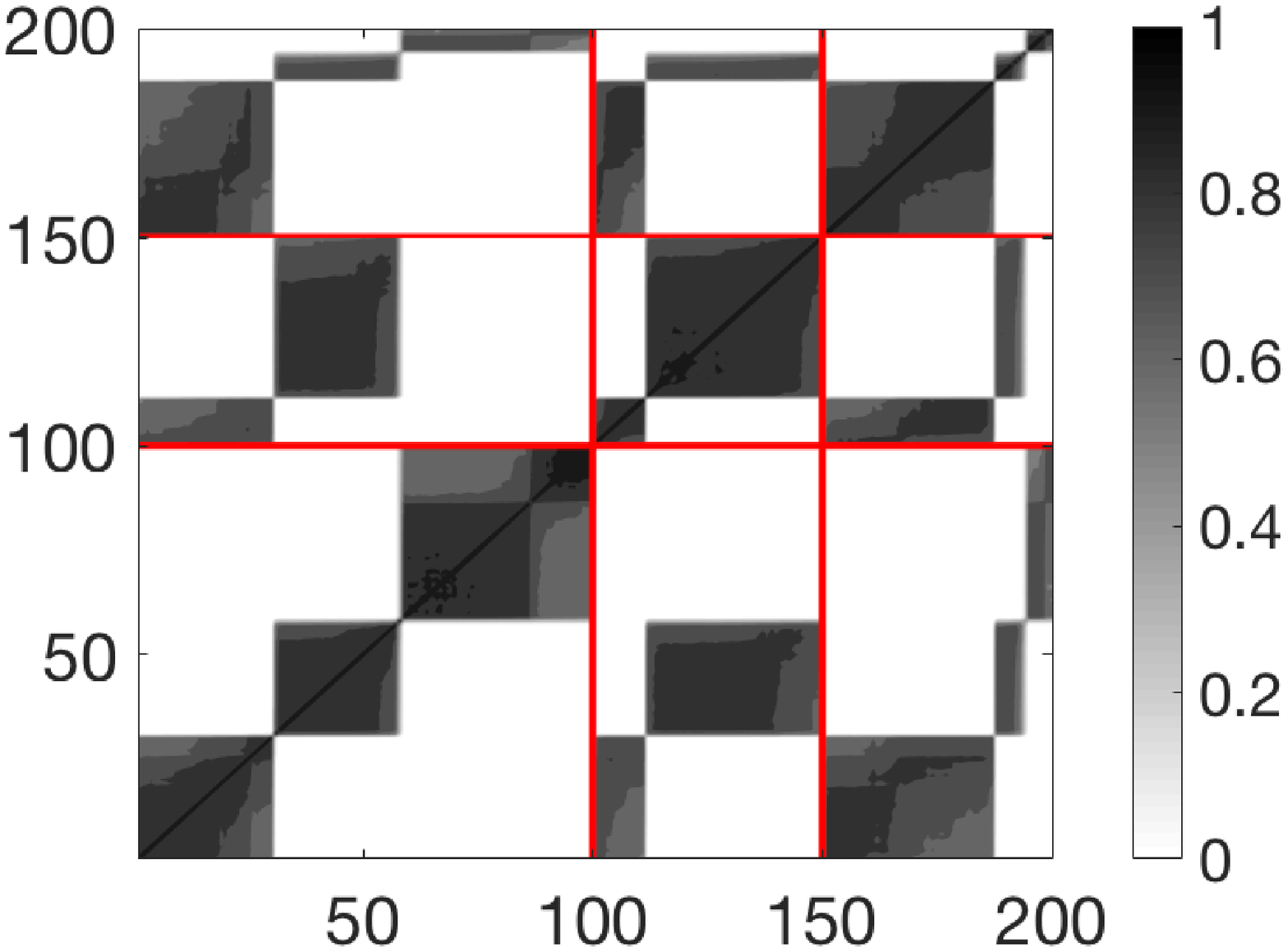}&      \includegraphics[width=3.5cm]{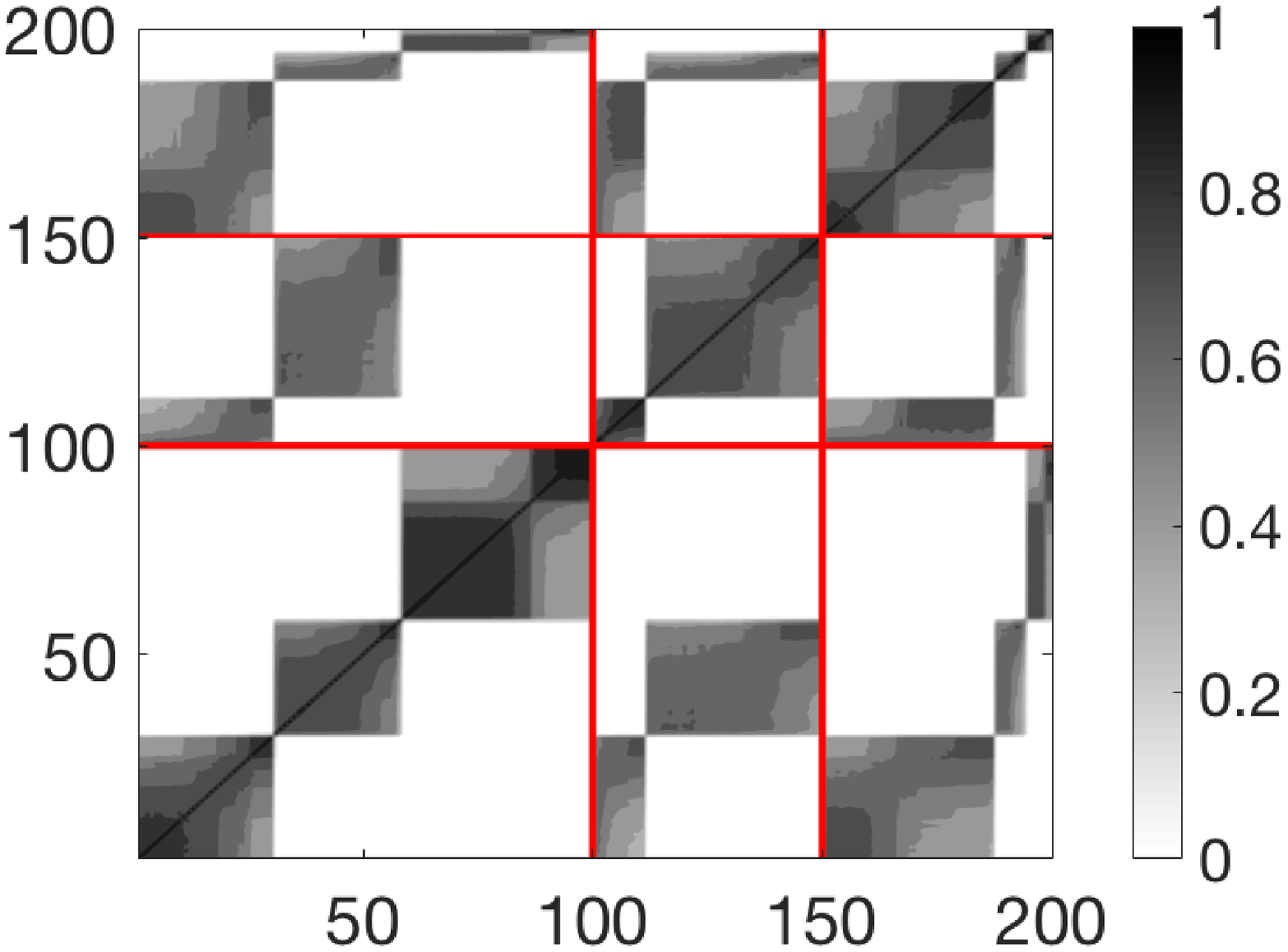}&\\      
      \footnotesize{HDPYP}& \footnotesize{HPYDP}& \footnotesize{HGDP}&\footnotesize{HGPYP}\vspace{-3pt}\\
      \includegraphics[width=3.5cm]{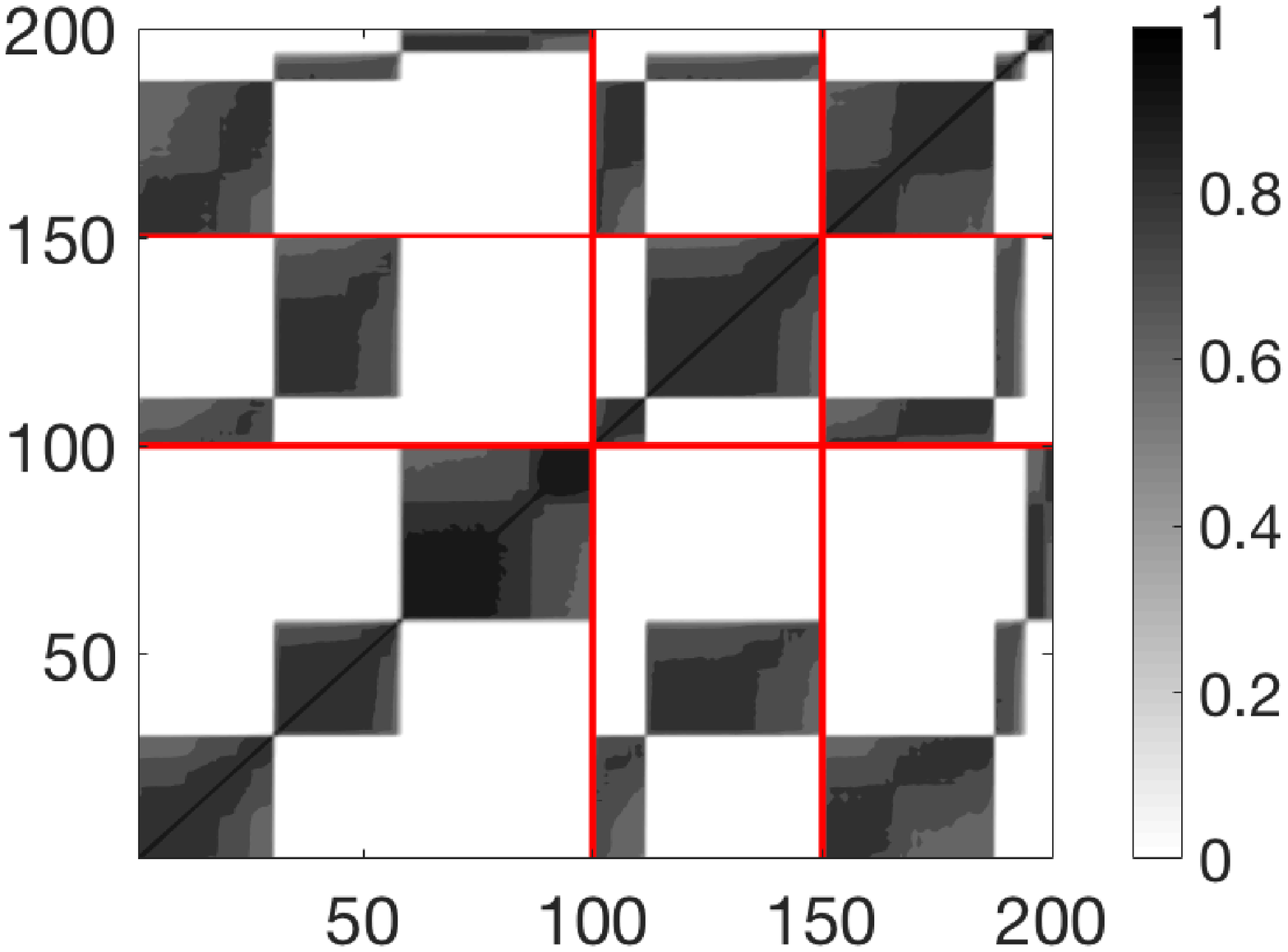}&         \includegraphics[width=3.5cm]{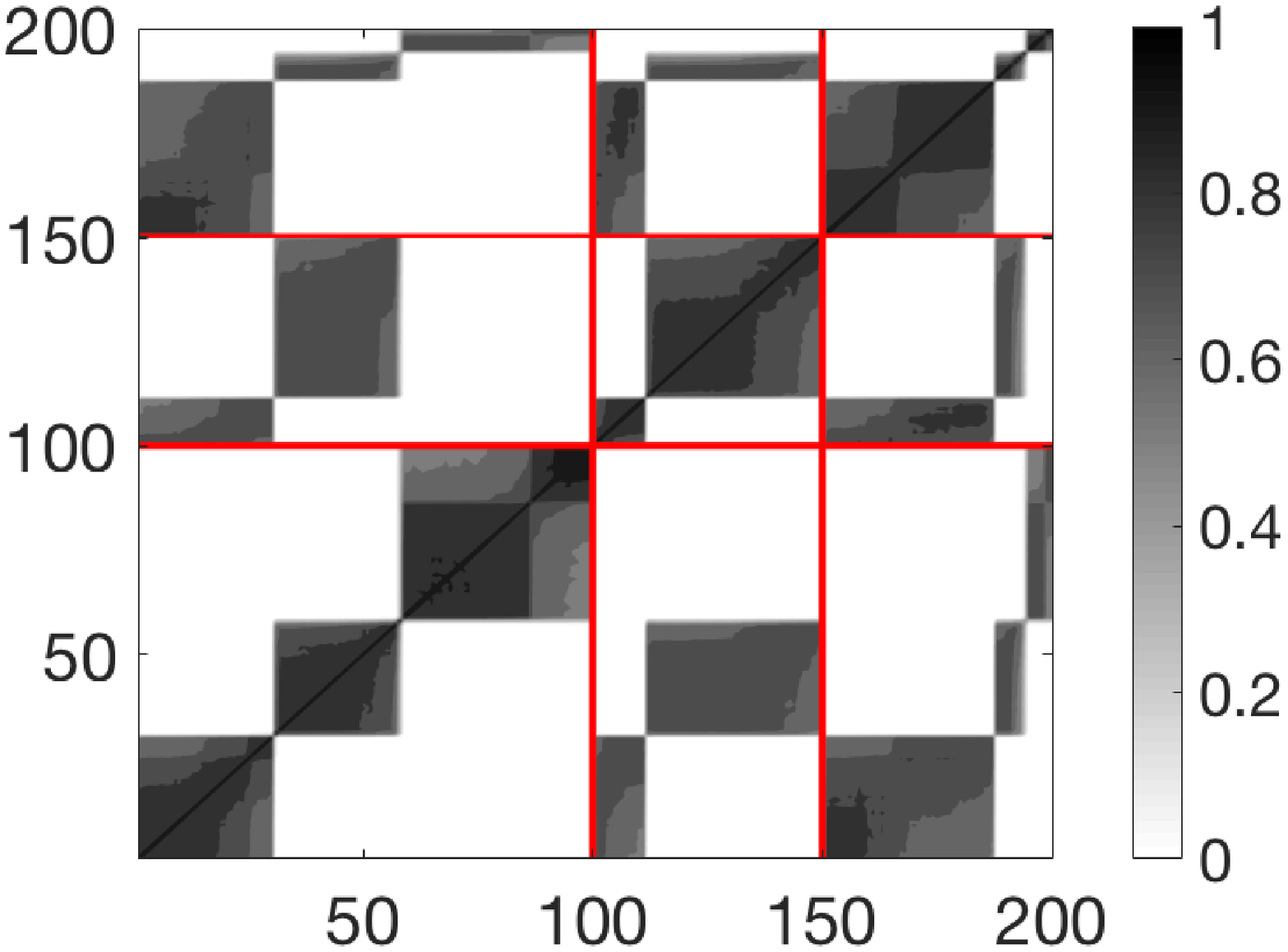}&     
      \includegraphics[width=3.5cm]{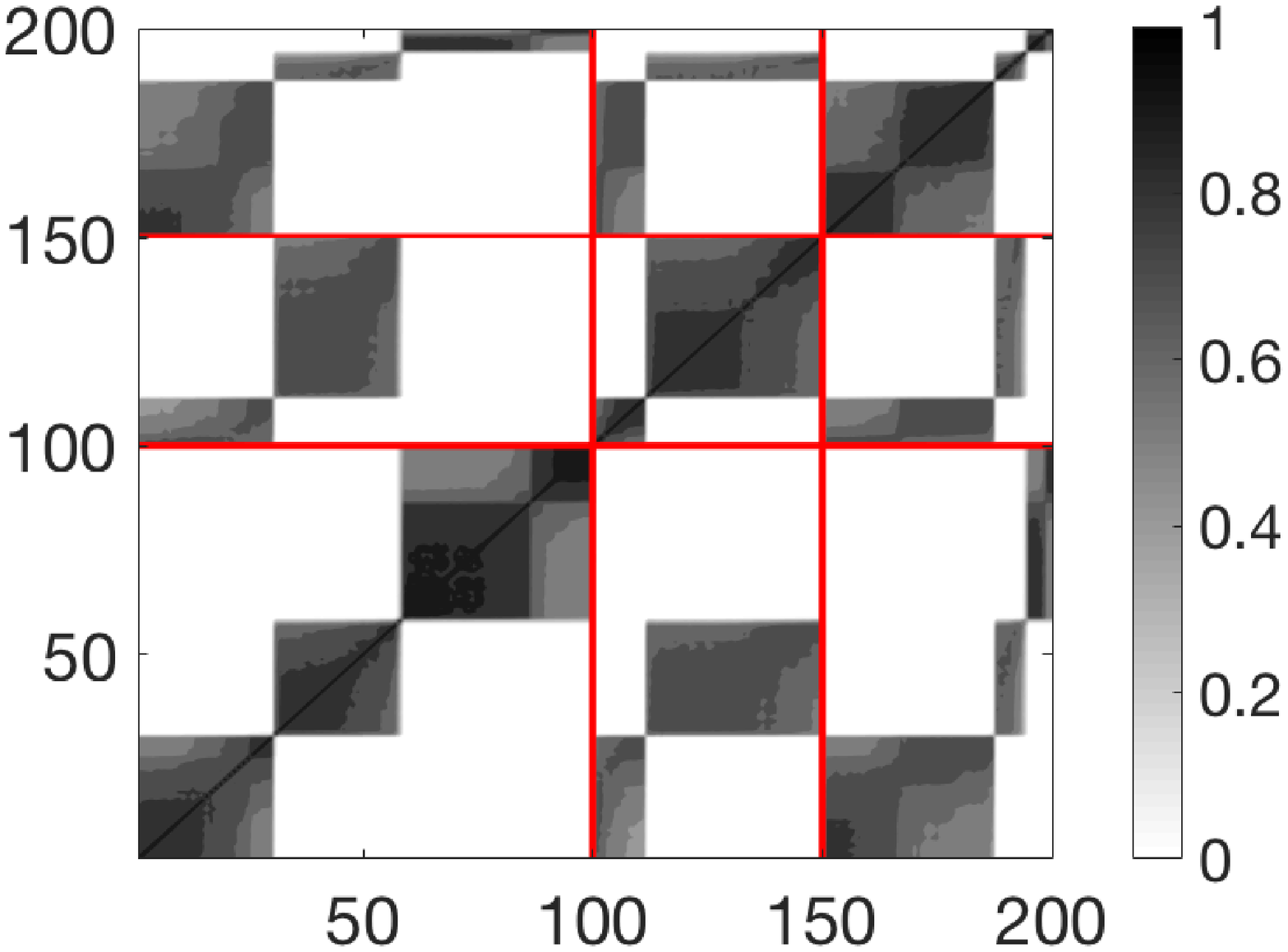}&\includegraphics[width=3.5cm]{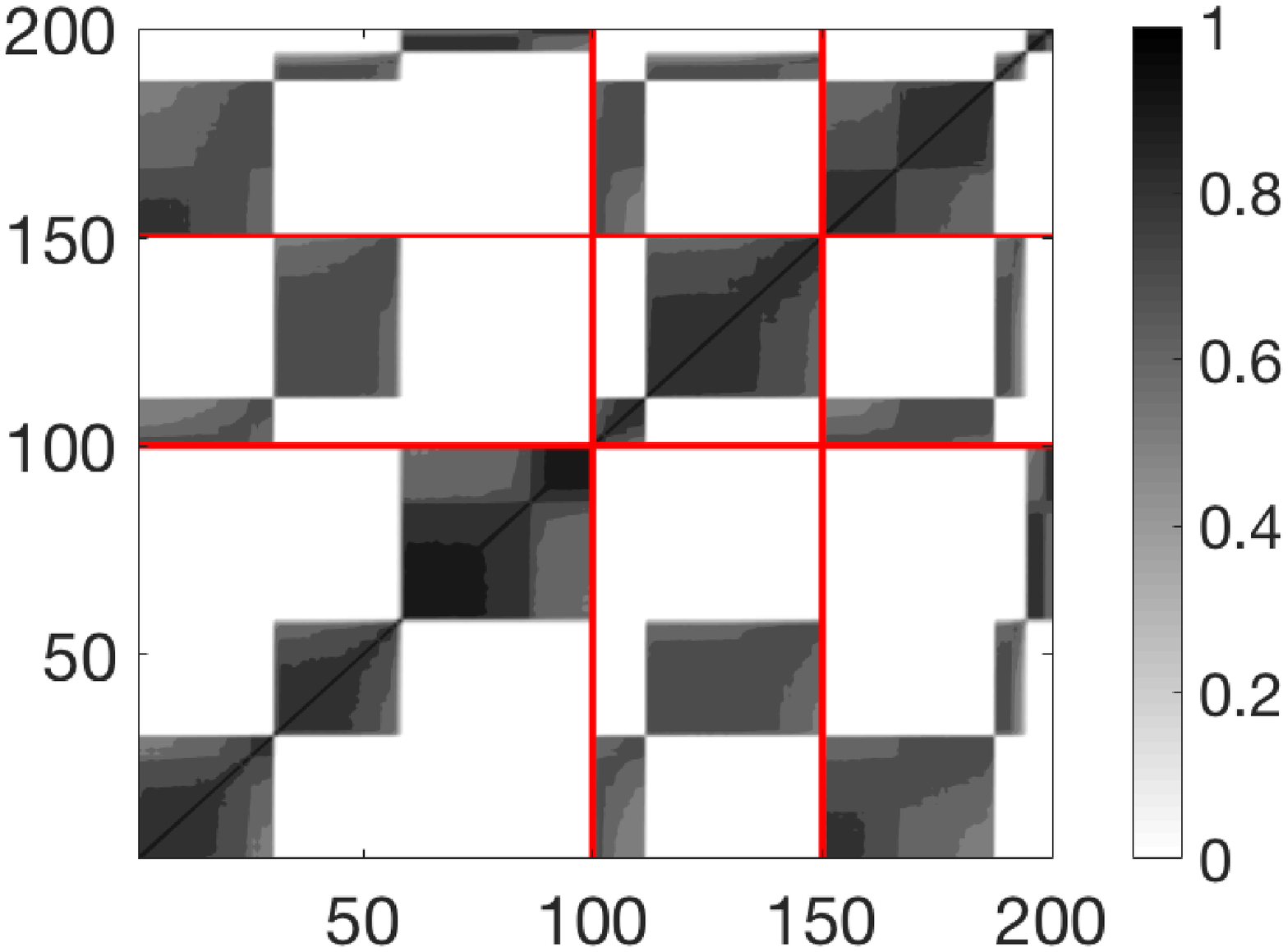}\vspace{20pt}\\
      \multicolumn{4}{c}{(b) Two-component normal mixture experiment}\\
      \footnotesize{HDP}& \footnotesize{HPYP}& \footnotesize{HGP}&\vspace{-3pt}\\
      \includegraphics[width=3.5cm]{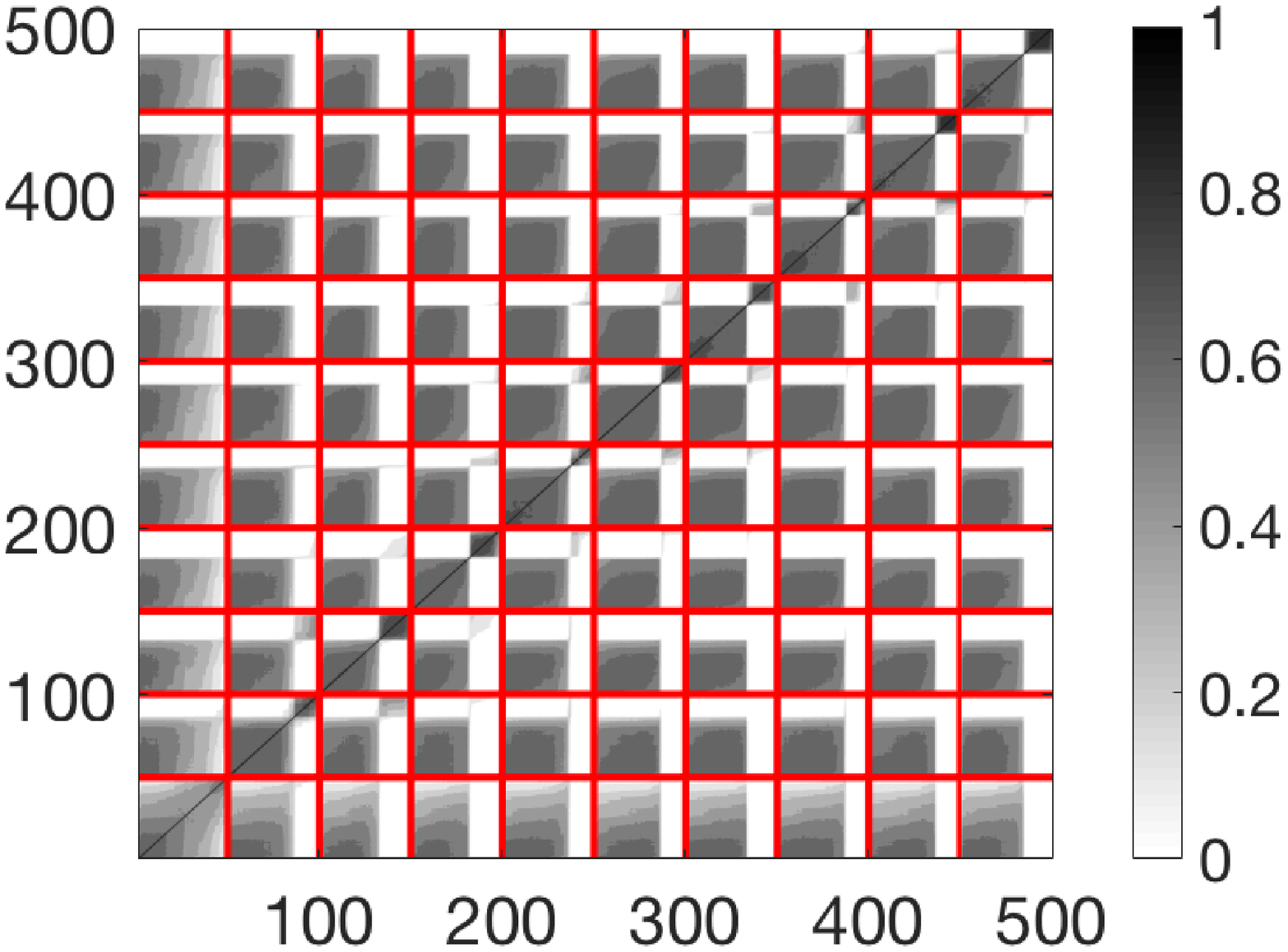}&      \includegraphics[width=3.5cm]{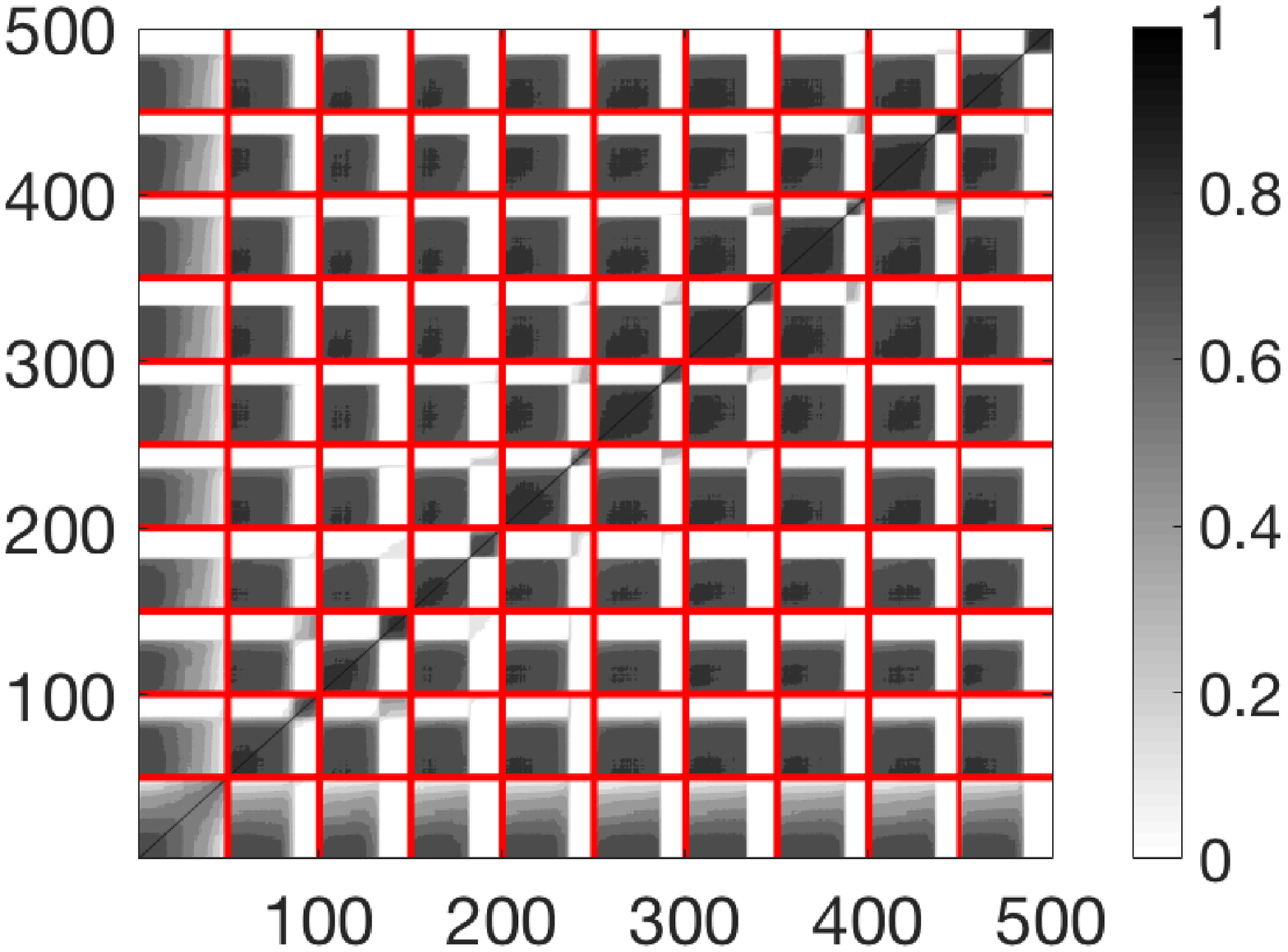}&      \includegraphics[width=3.5cm]{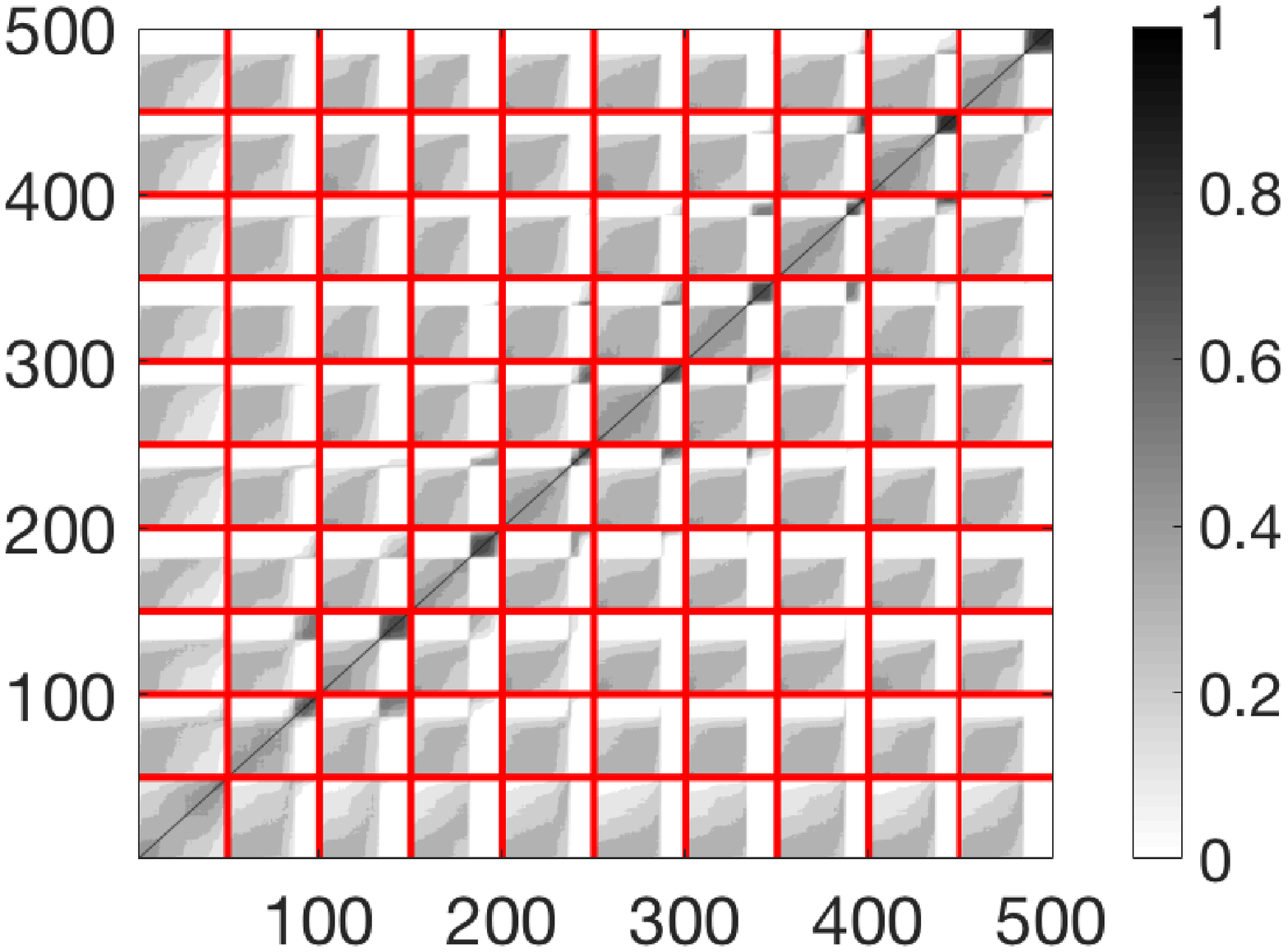}&\\      
      \footnotesize{HDPYP}& \footnotesize{HPYDP}& \footnotesize{HGDP}&\footnotesize{HGPYP}\vspace{-3pt}\\
      \includegraphics[width=3.5cm]{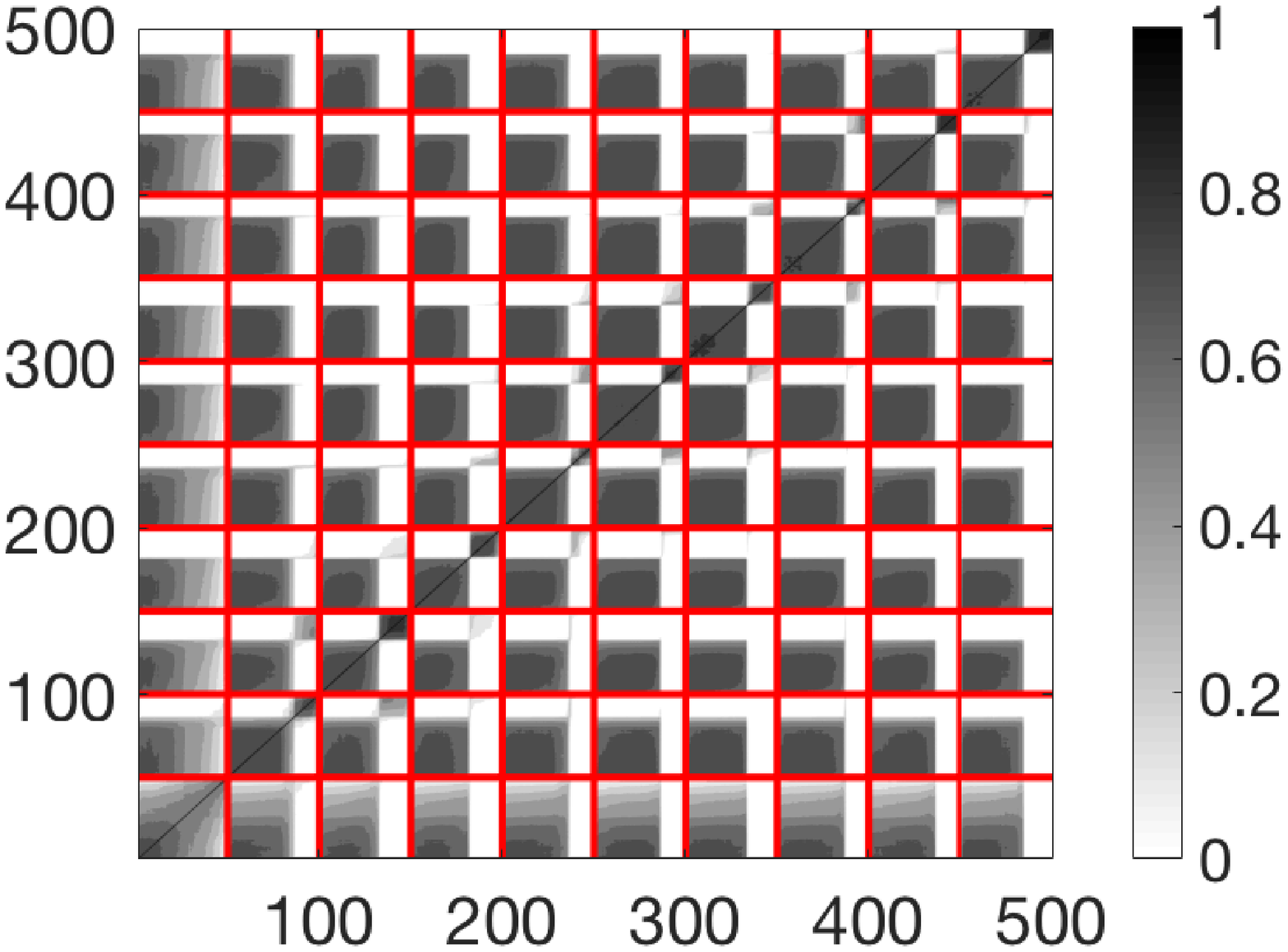}&         \includegraphics[width=3.5cm]{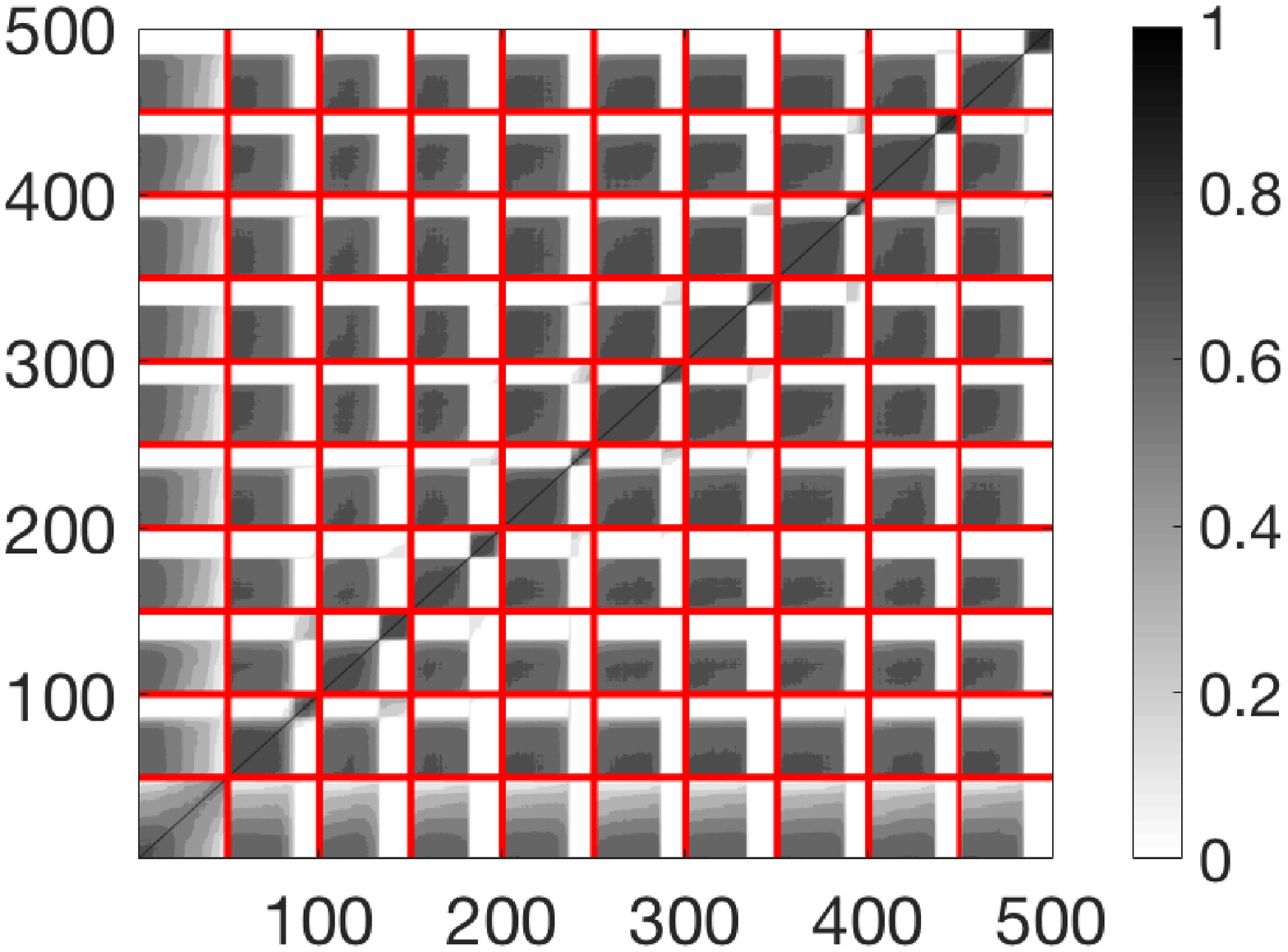}&     
      \includegraphics[width=3.5cm]{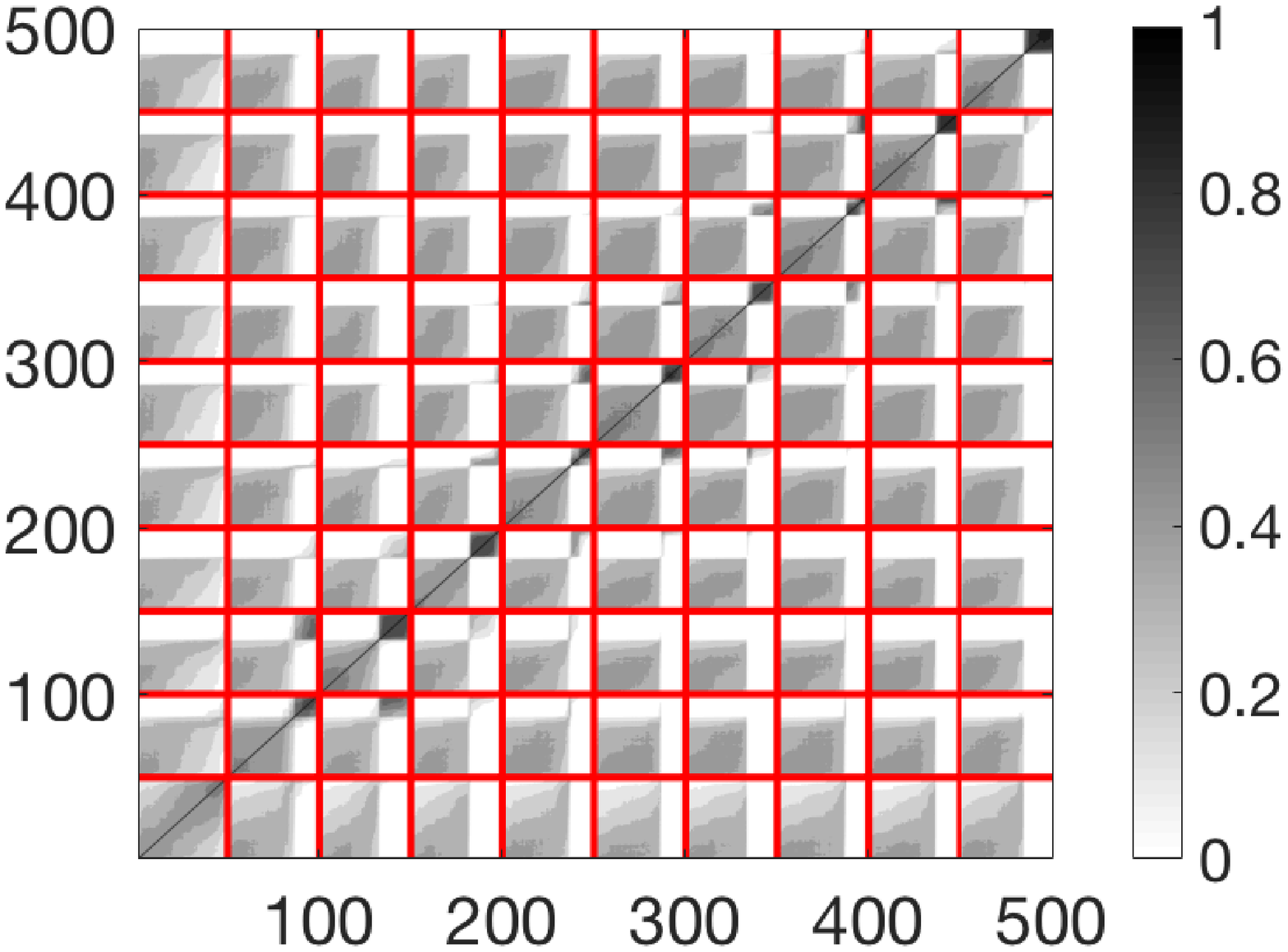}&\includegraphics[width=3.5cm]{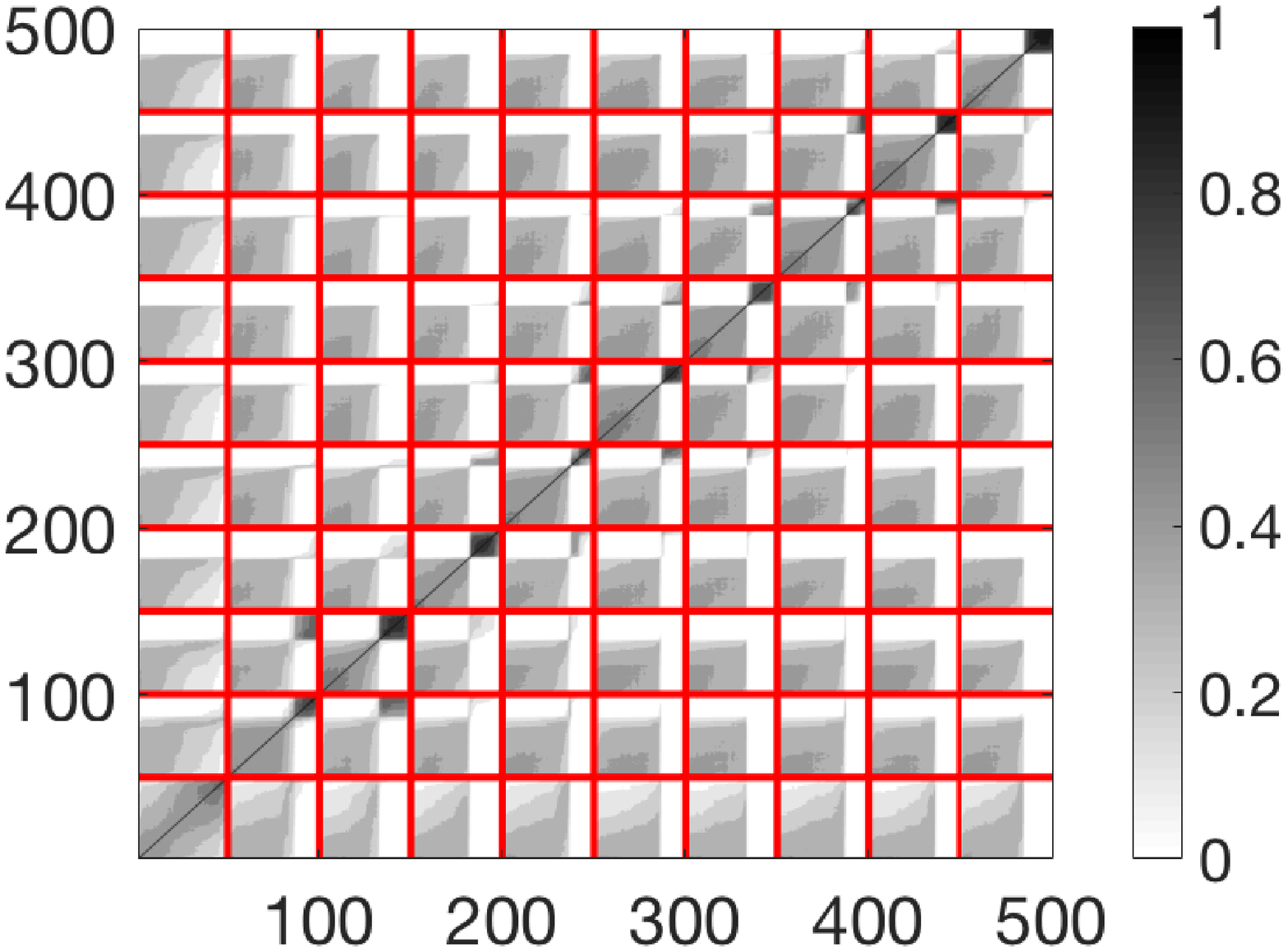}\\
  \end{tabular}
\end{figure}

\begin{figure}[p]
  \caption{Predictive density for the three-component (two-component) normal mixture experiment. Columns: the predictive for the first (first), second (fifth) ad third (tenth) restaurant. In each panel the predictive for HDP and HDPYP (first row), the HPYP and HPYDP (second row) and HGP, HGPYP and HGDP (third row).}
  \label{Fig:simpredictive}
 \centering
  \setlength{\tabcolsep}{2pt} 
  \begin{tabular}{ccc}
      \multicolumn{3}{c}{(a) Three-component normal mixture experiment}\\
      \includegraphics[width=3.6cm]{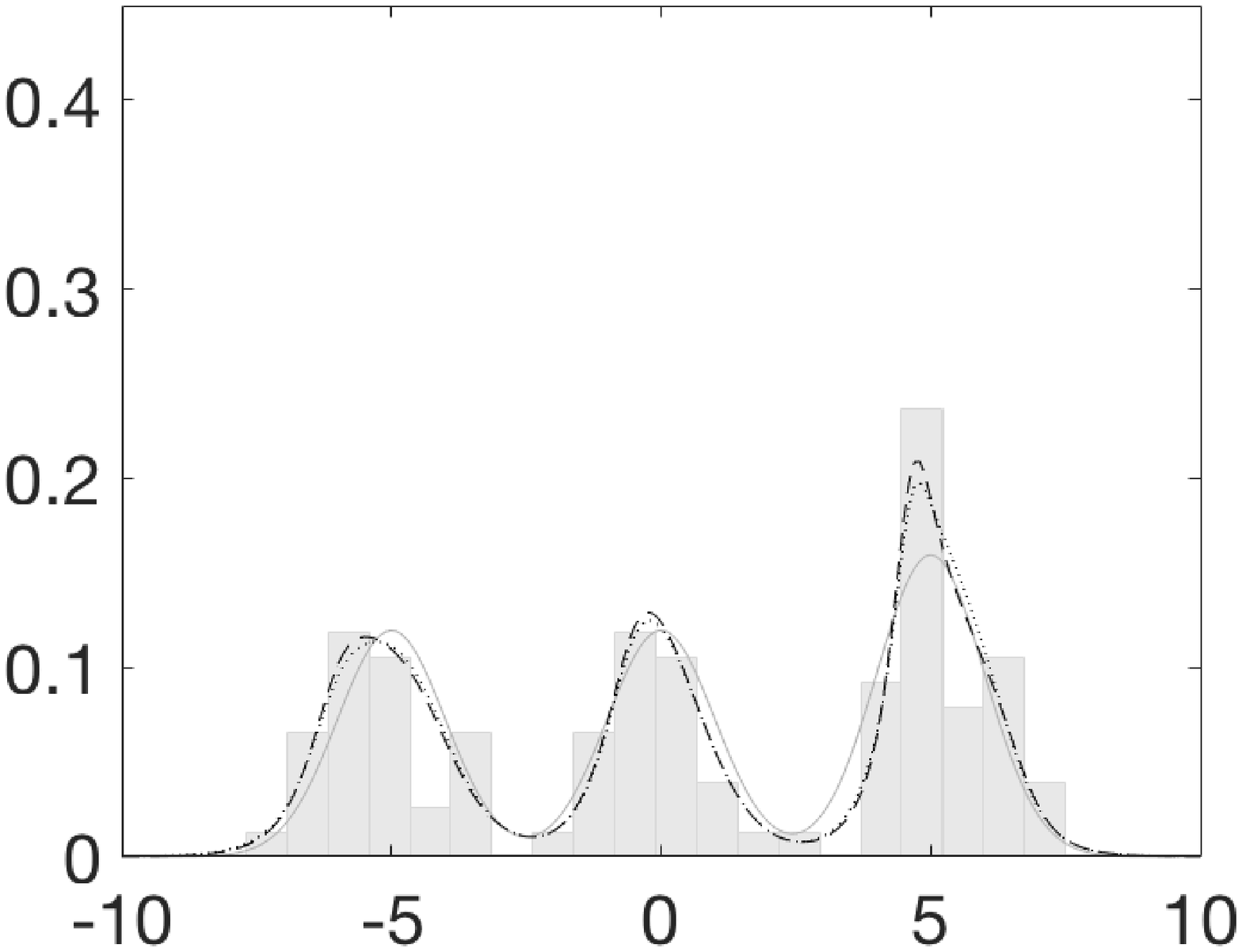}&
      \includegraphics[width=3.6cm]{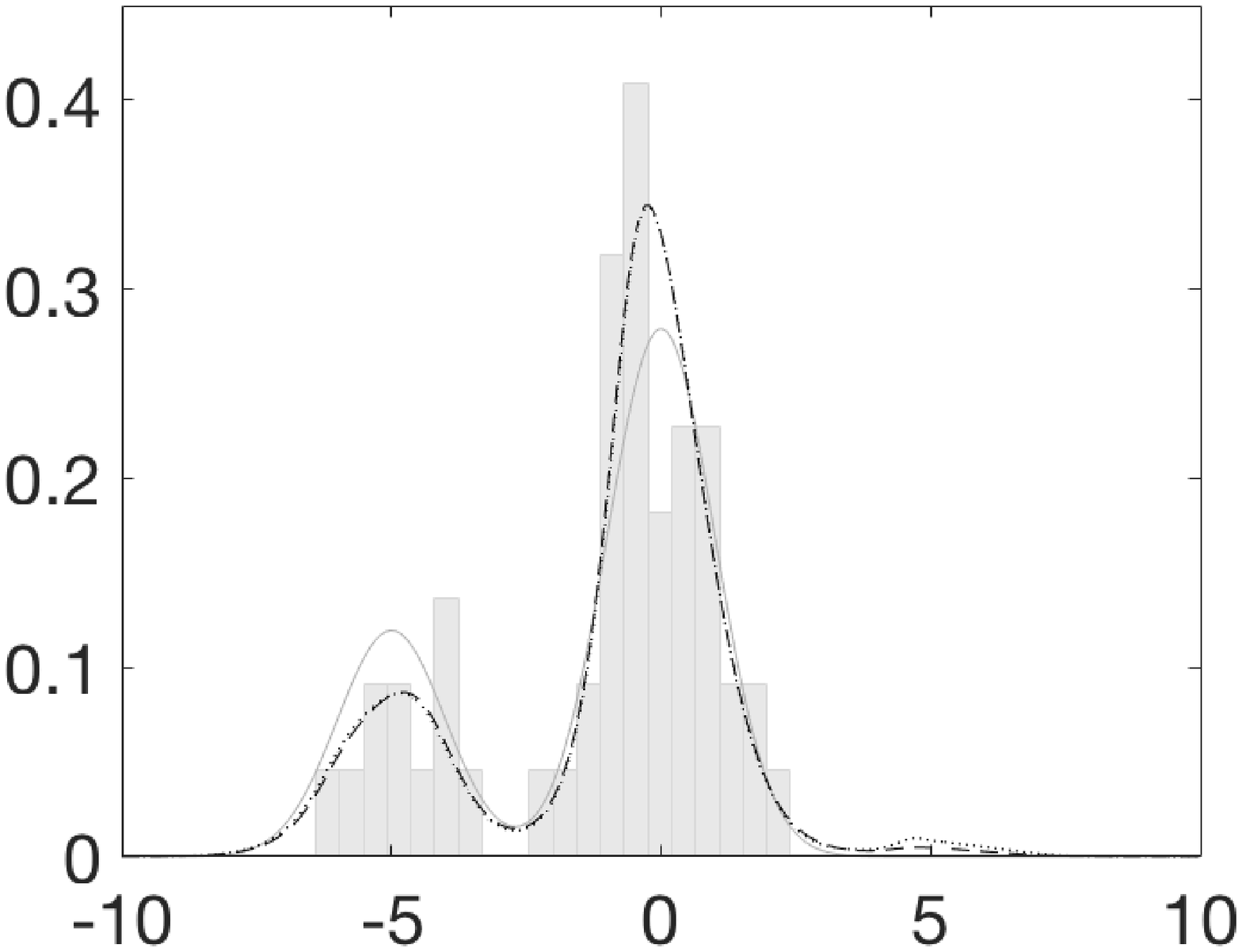}&
      \includegraphics[width=3.6cm]{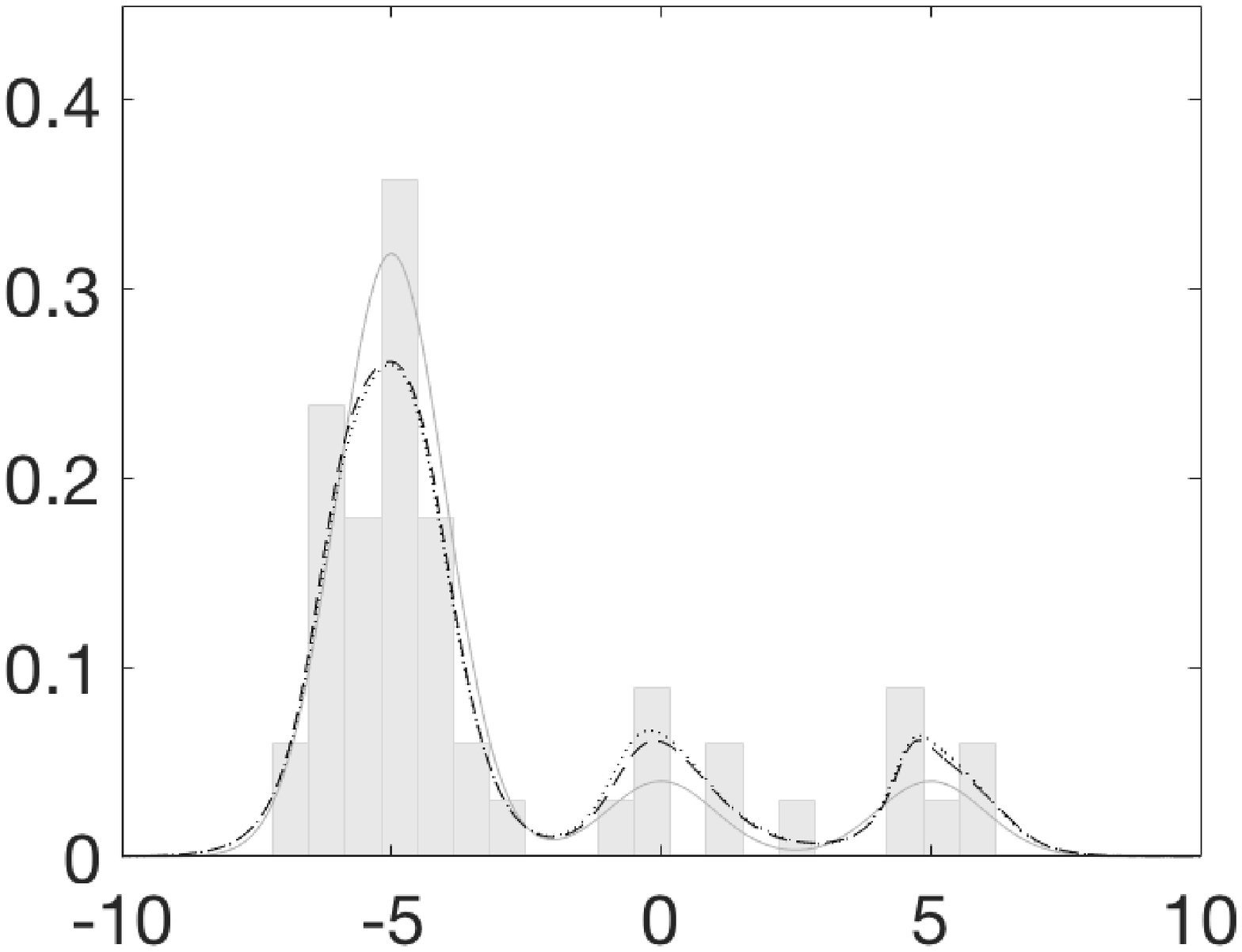}\\
      \includegraphics[width=3.6cm]{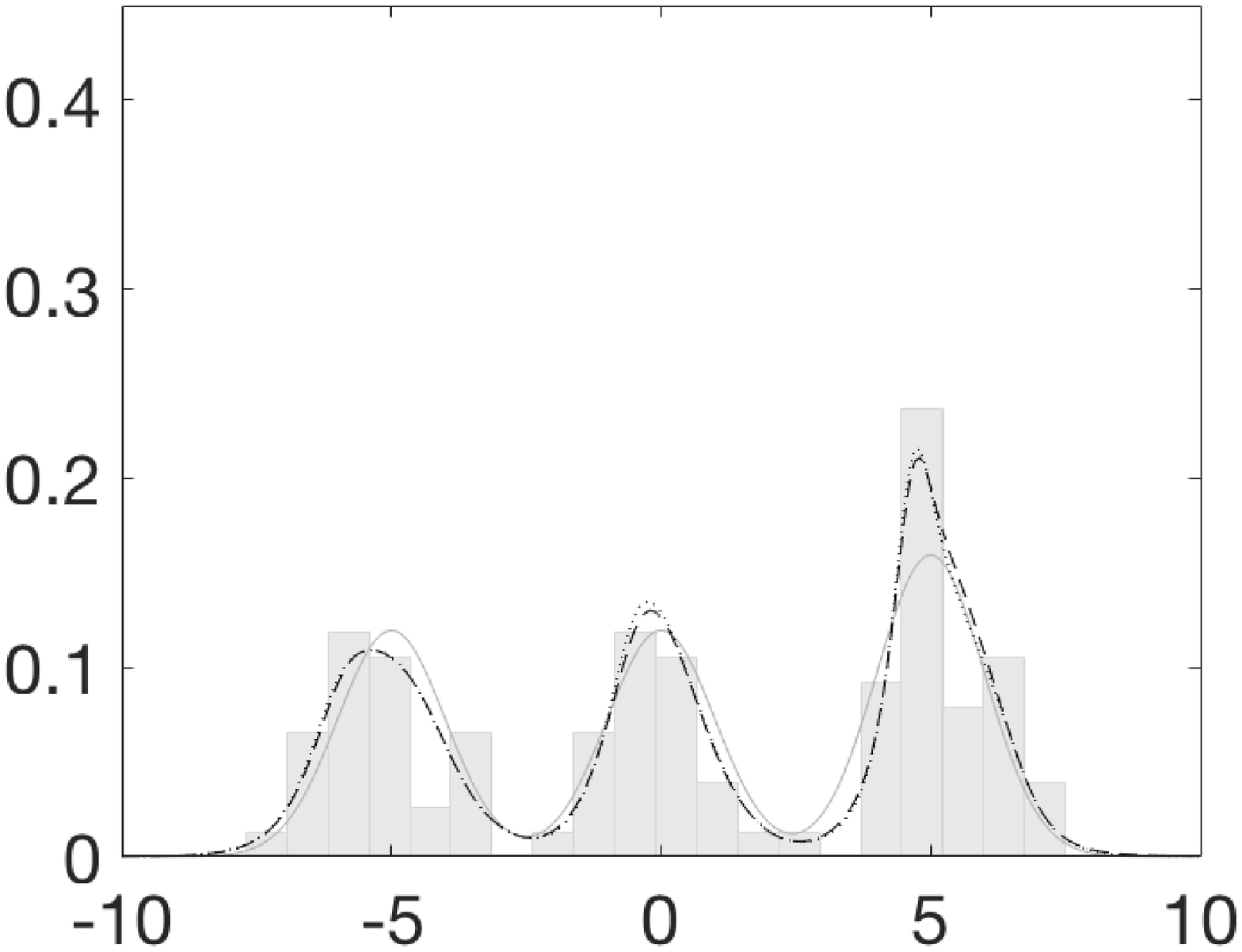}&
      \includegraphics[width=3.6cm]{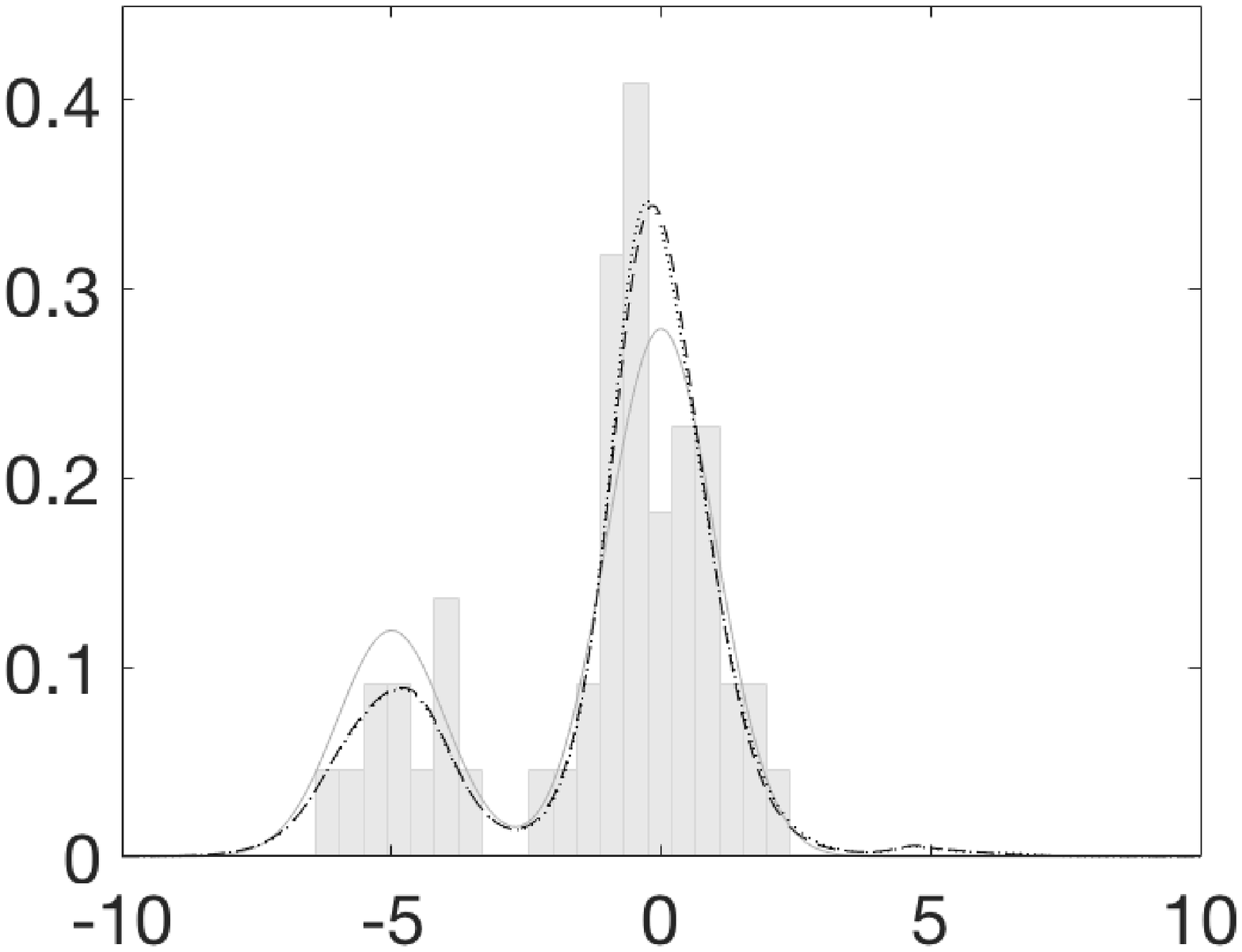}&
      \includegraphics[width=3.6cm]{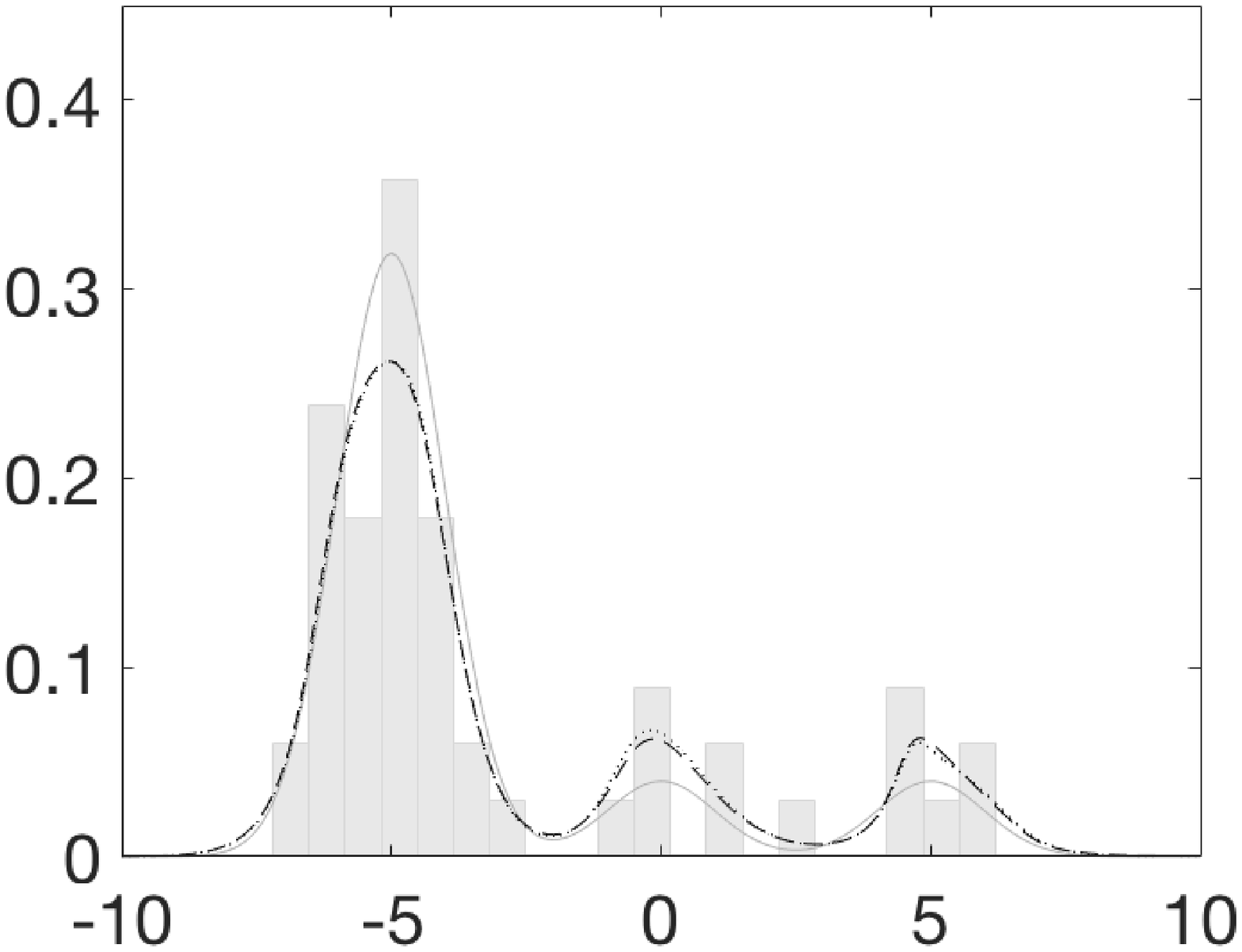}\\
      \includegraphics[width=3.6cm]{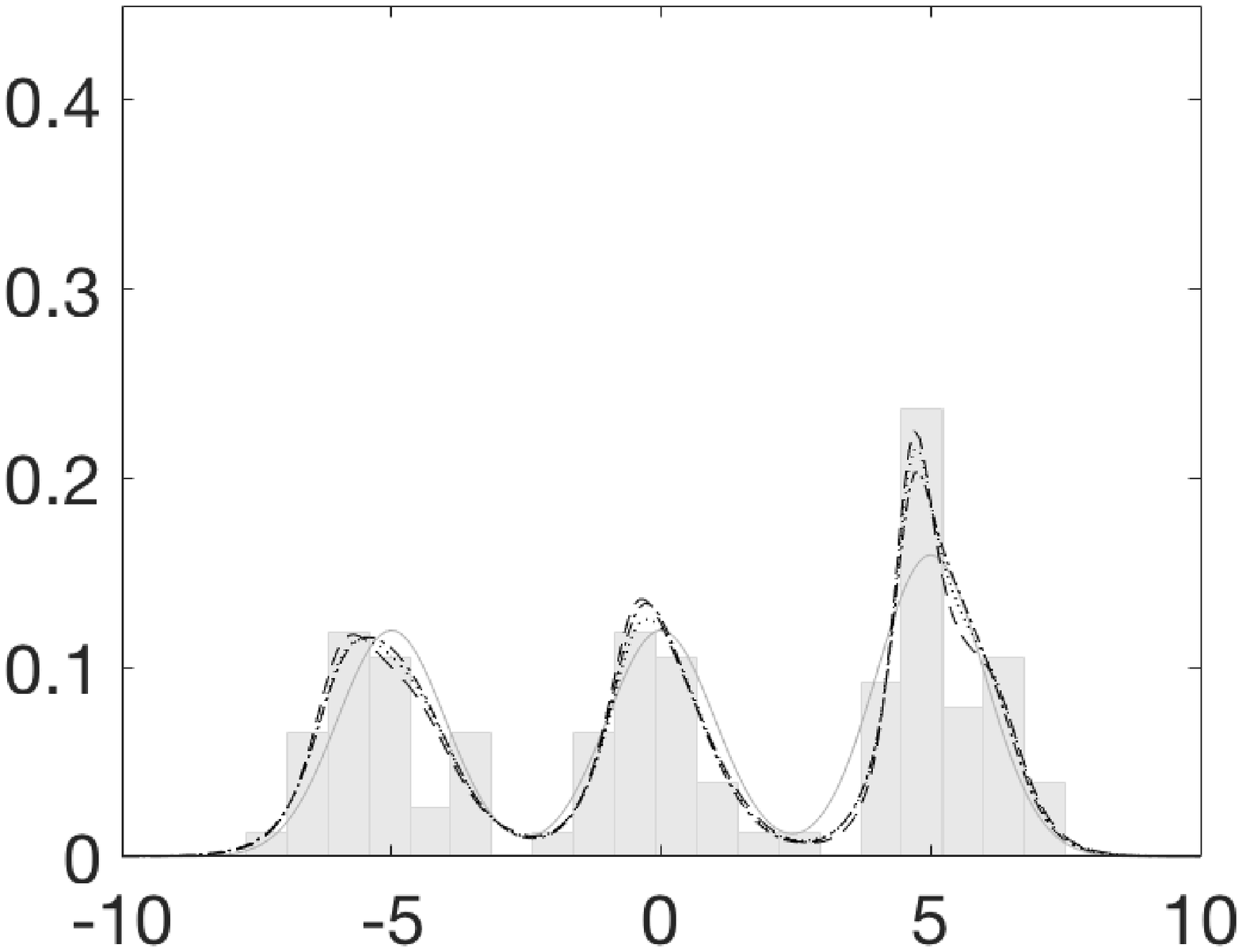}&
      \includegraphics[width=3.6cm]{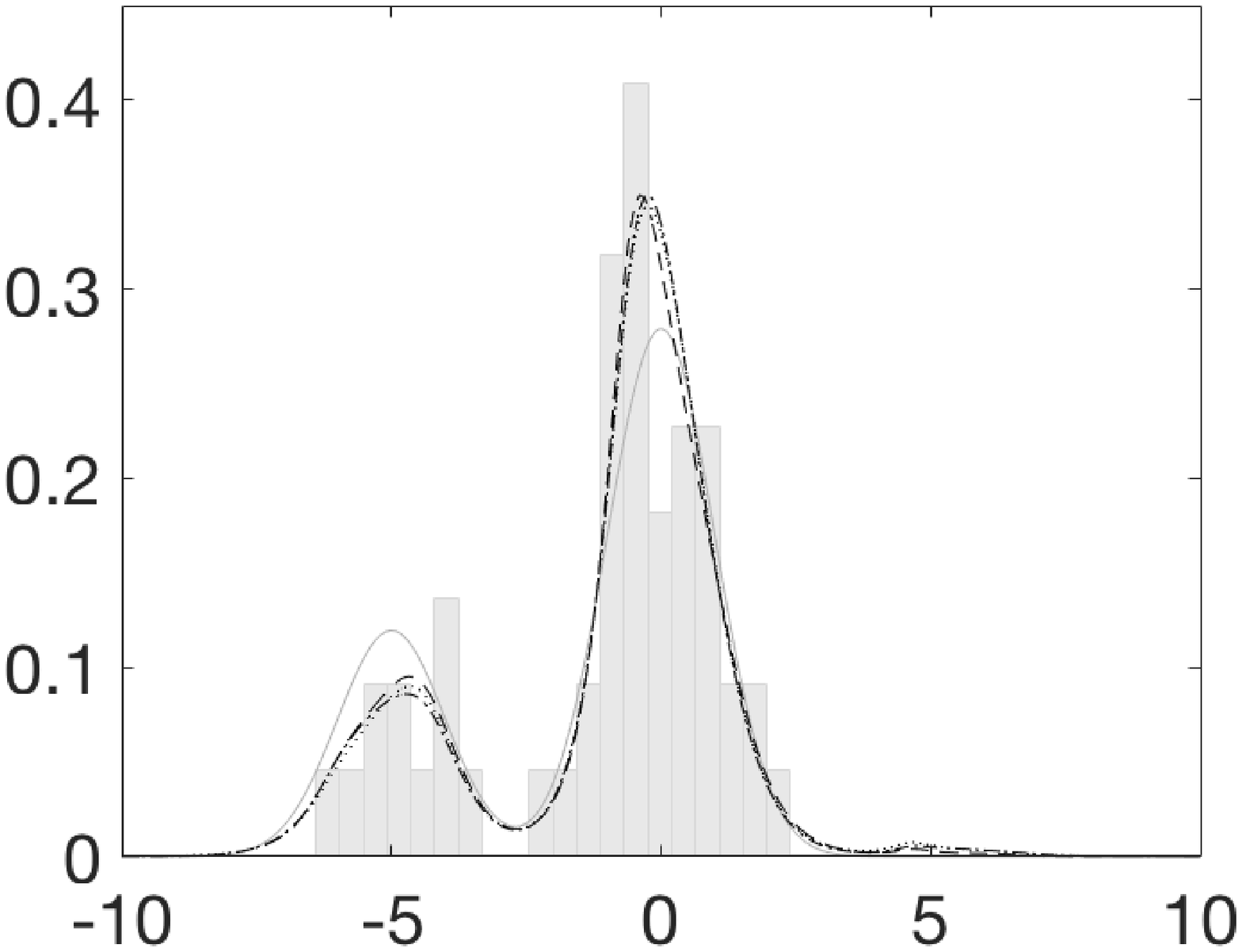}&
      \includegraphics[width=3.6cm]{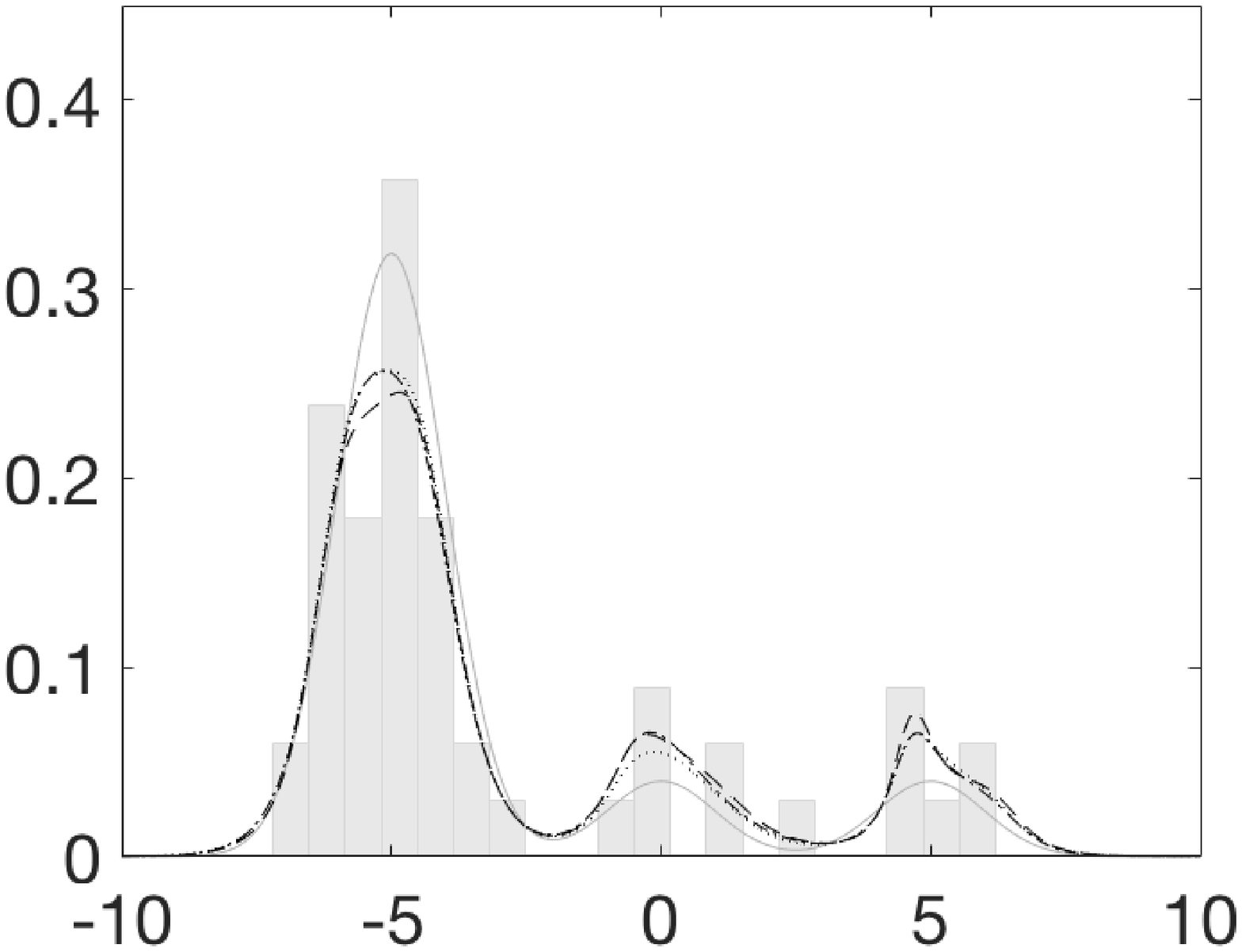}\vspace{10pt}\\
      \multicolumn{3}{c}{(b) Two-component normal mixture experiment}\\
      \includegraphics[width=3.6cm]{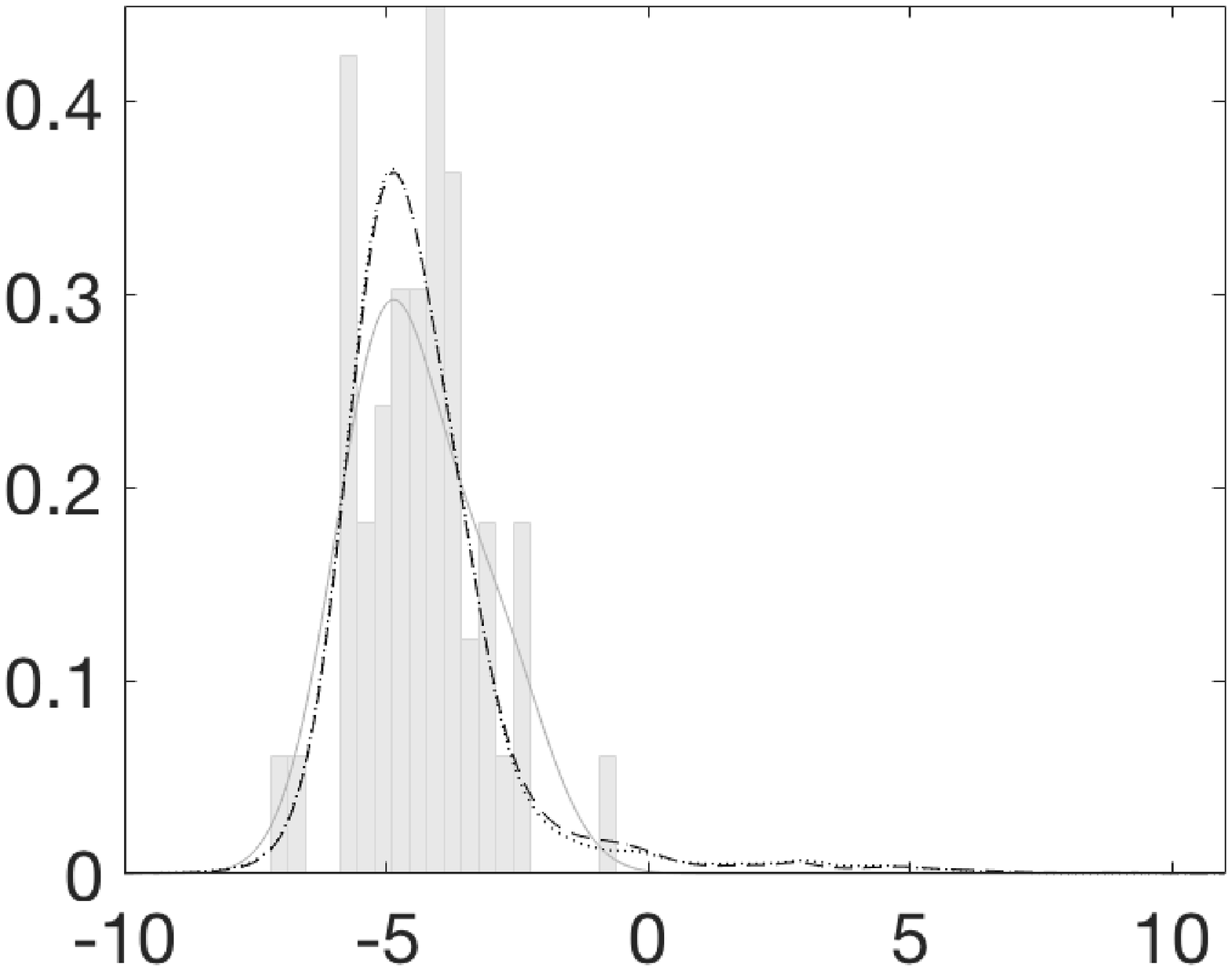}&
      \includegraphics[width=3.6cm]{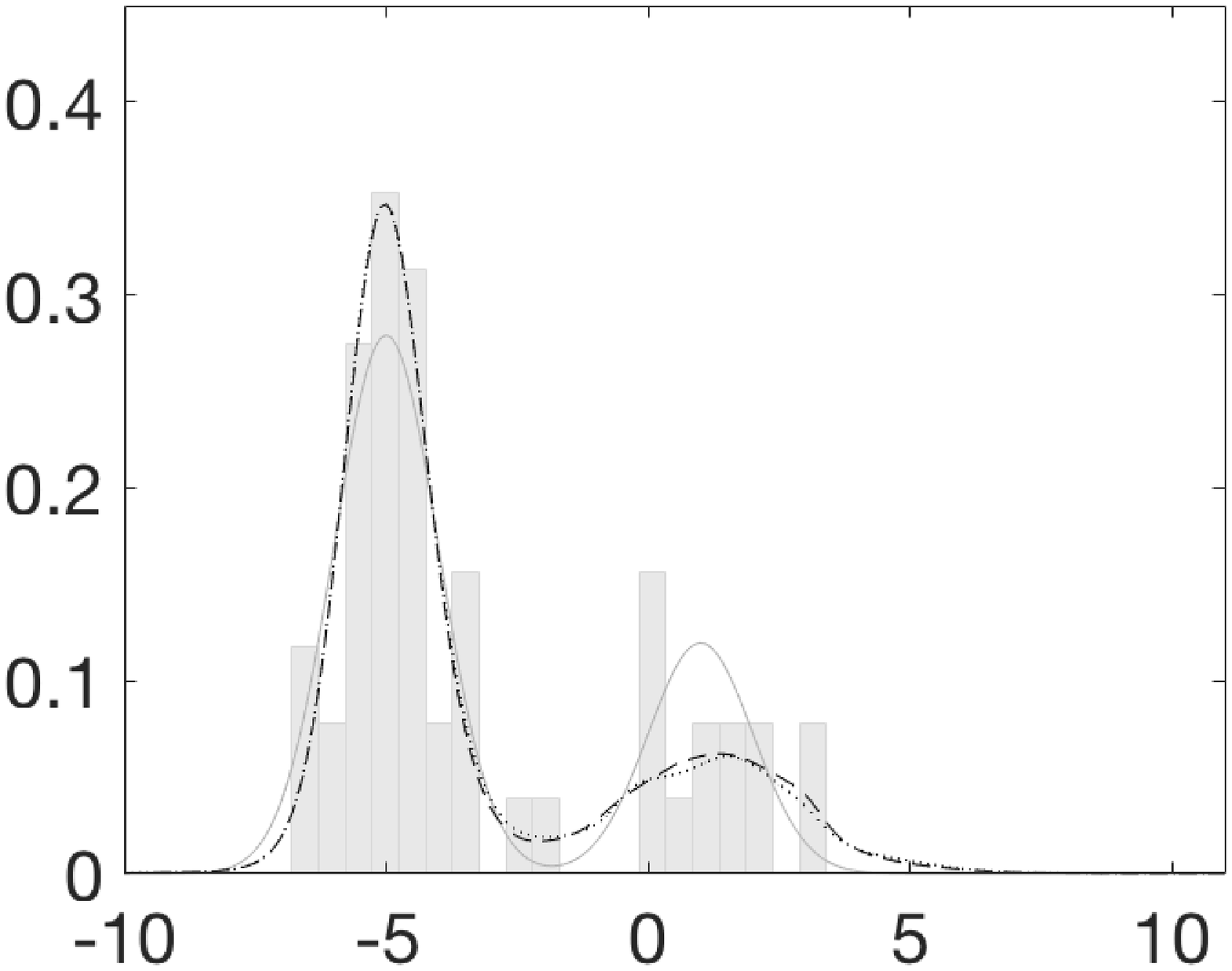}&
      \includegraphics[width=3.6cm]{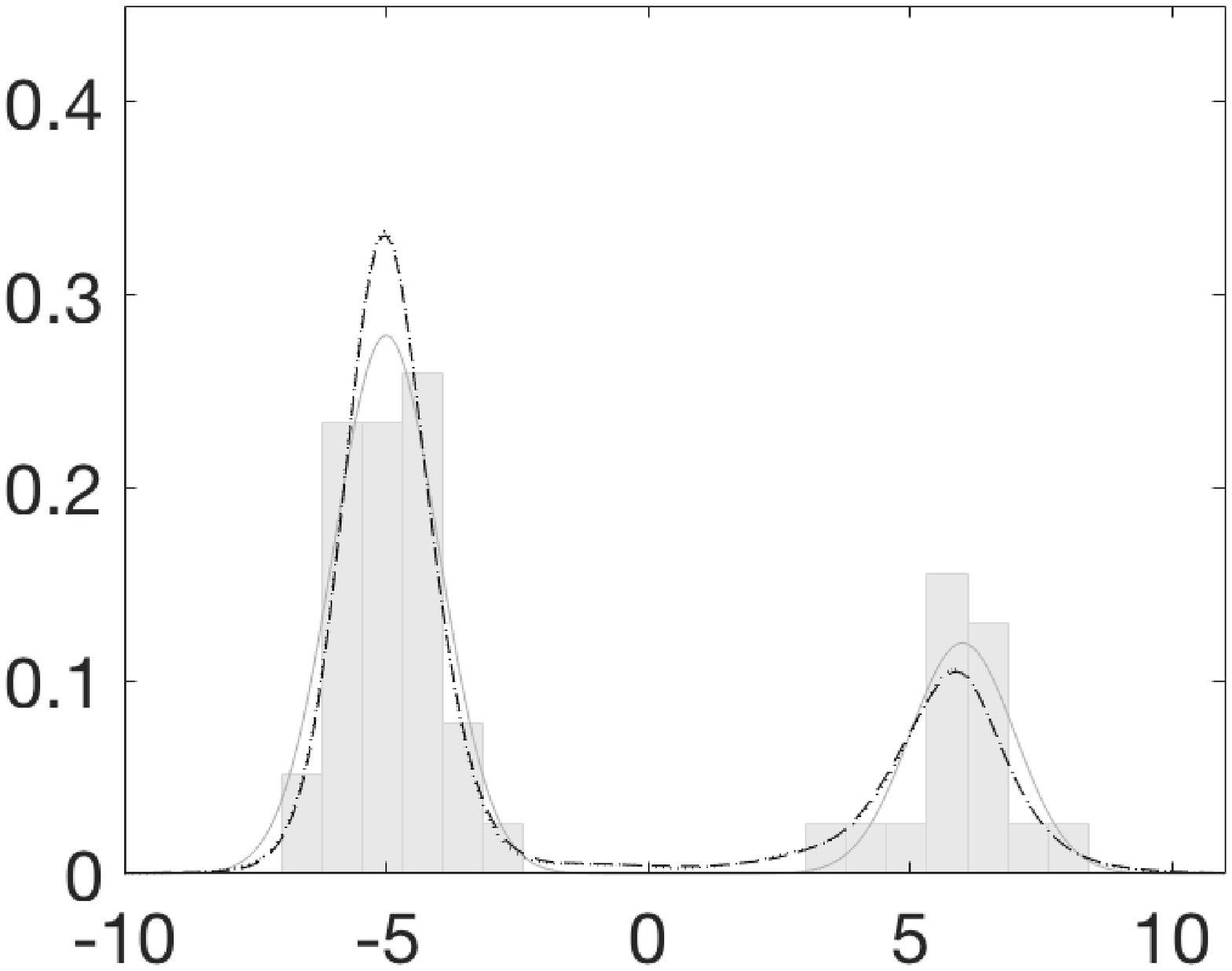}\\
      \includegraphics[width=3.6cm]{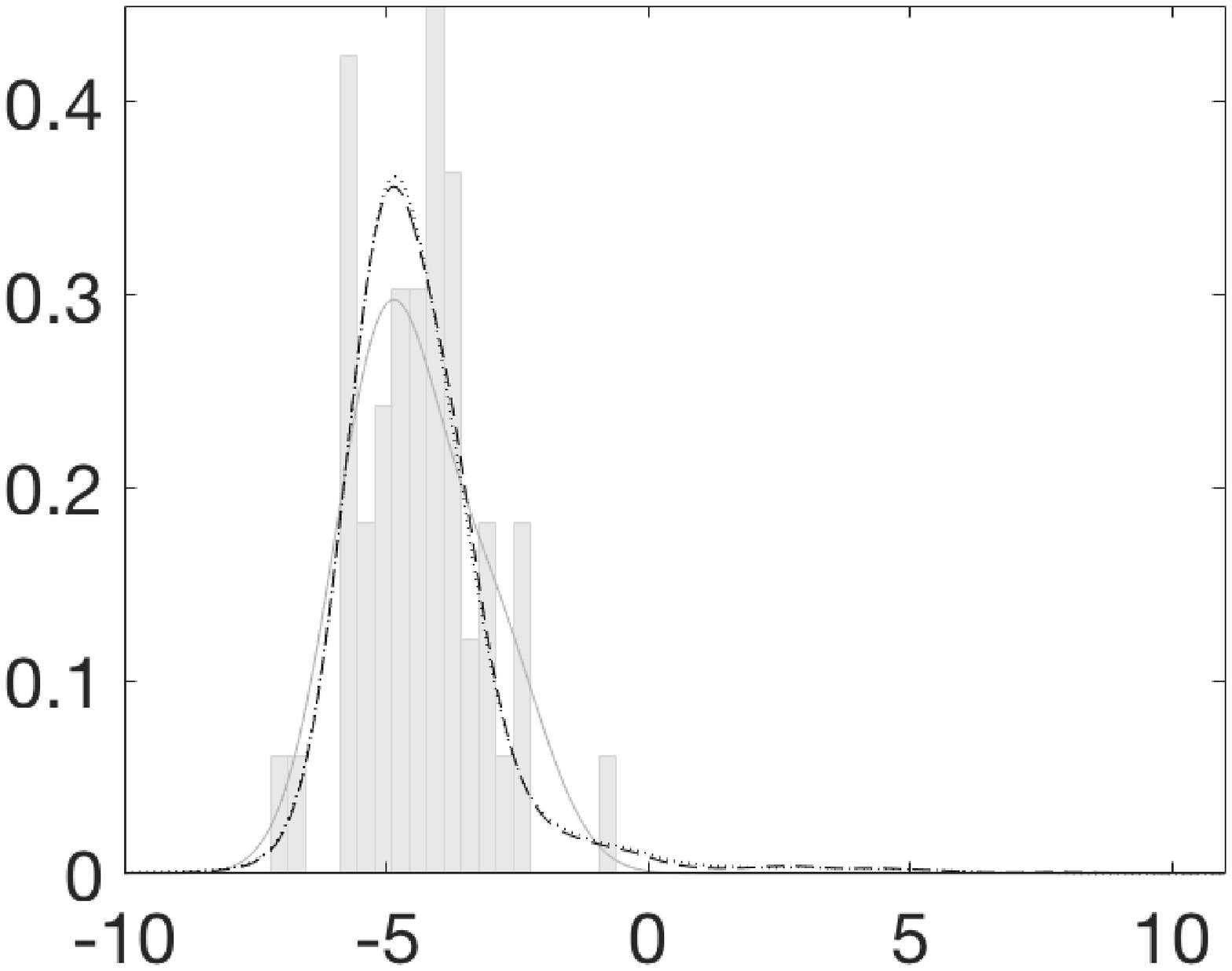}&
      \includegraphics[width=3.6cm]{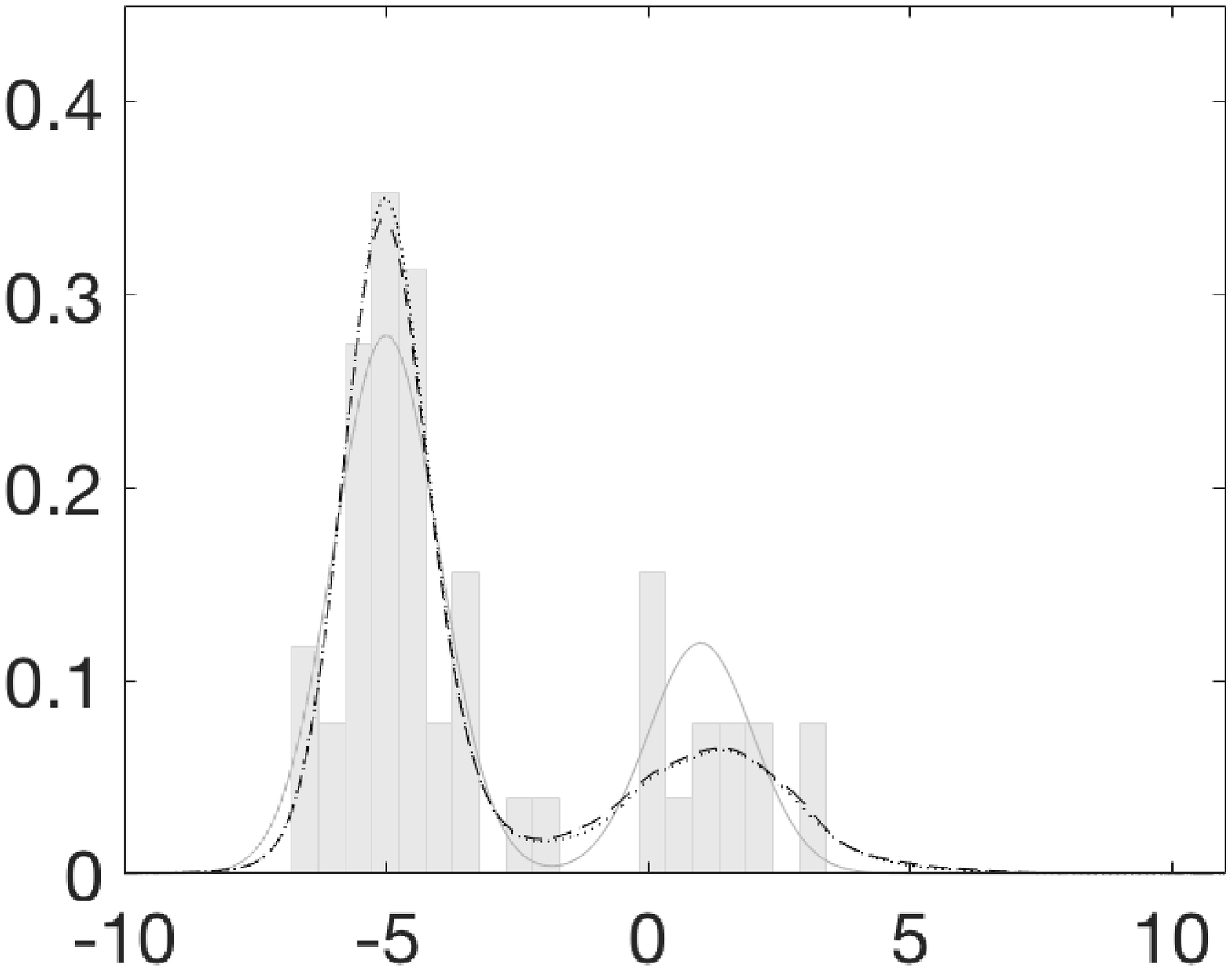}&
      \includegraphics[width=3.6cm]{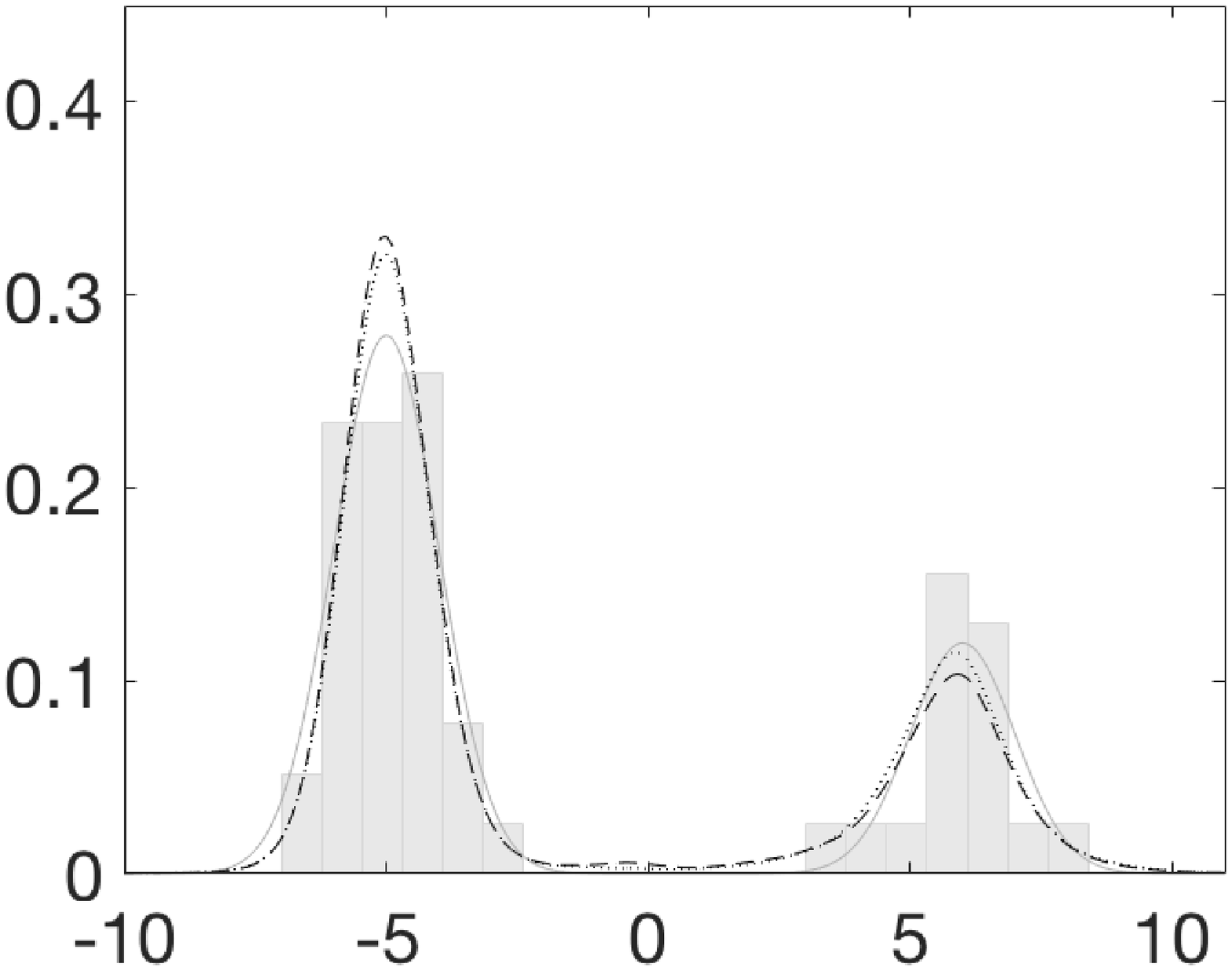}\\
      \includegraphics[width=3.6cm]{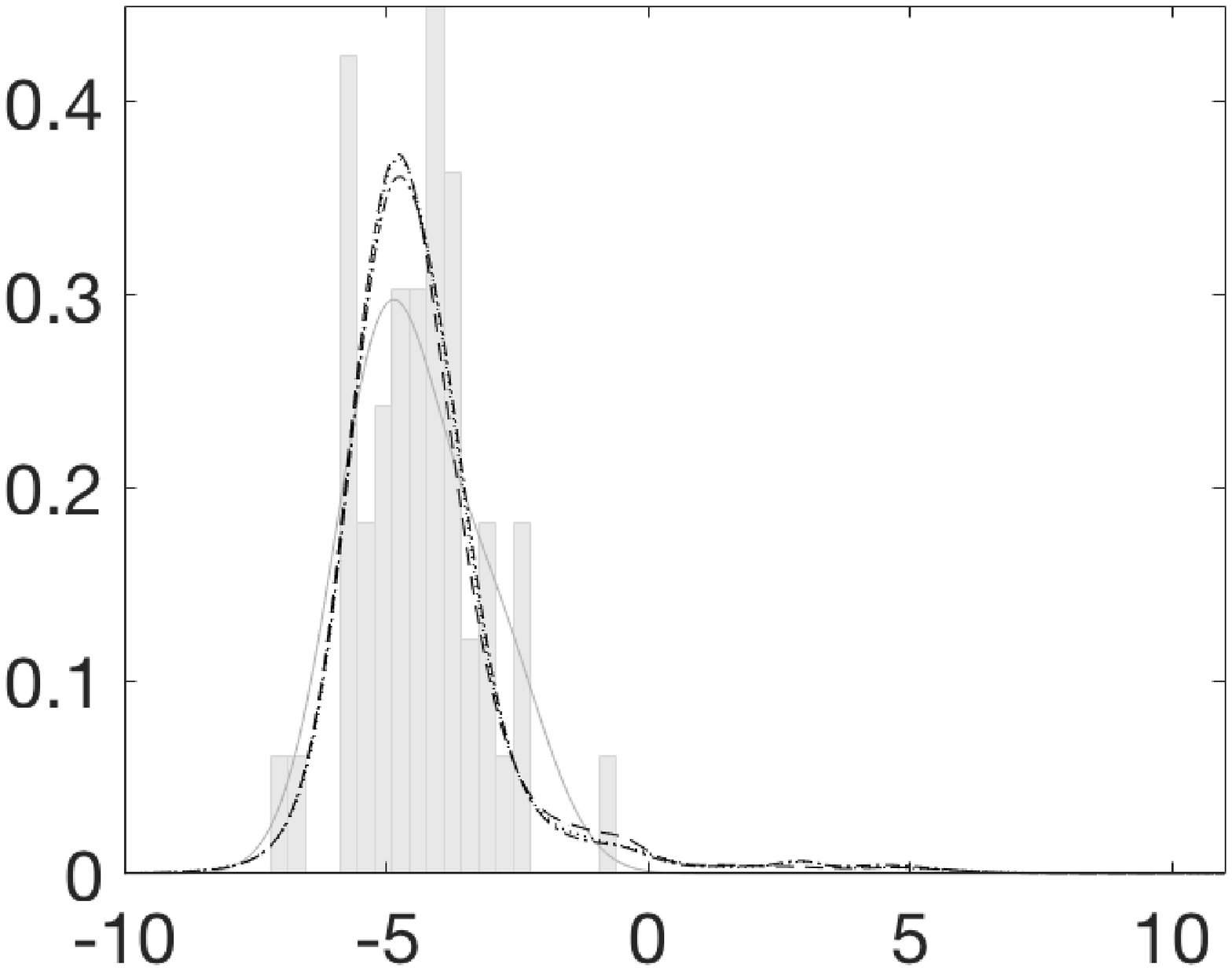}&
      \includegraphics[width=3.6cm]{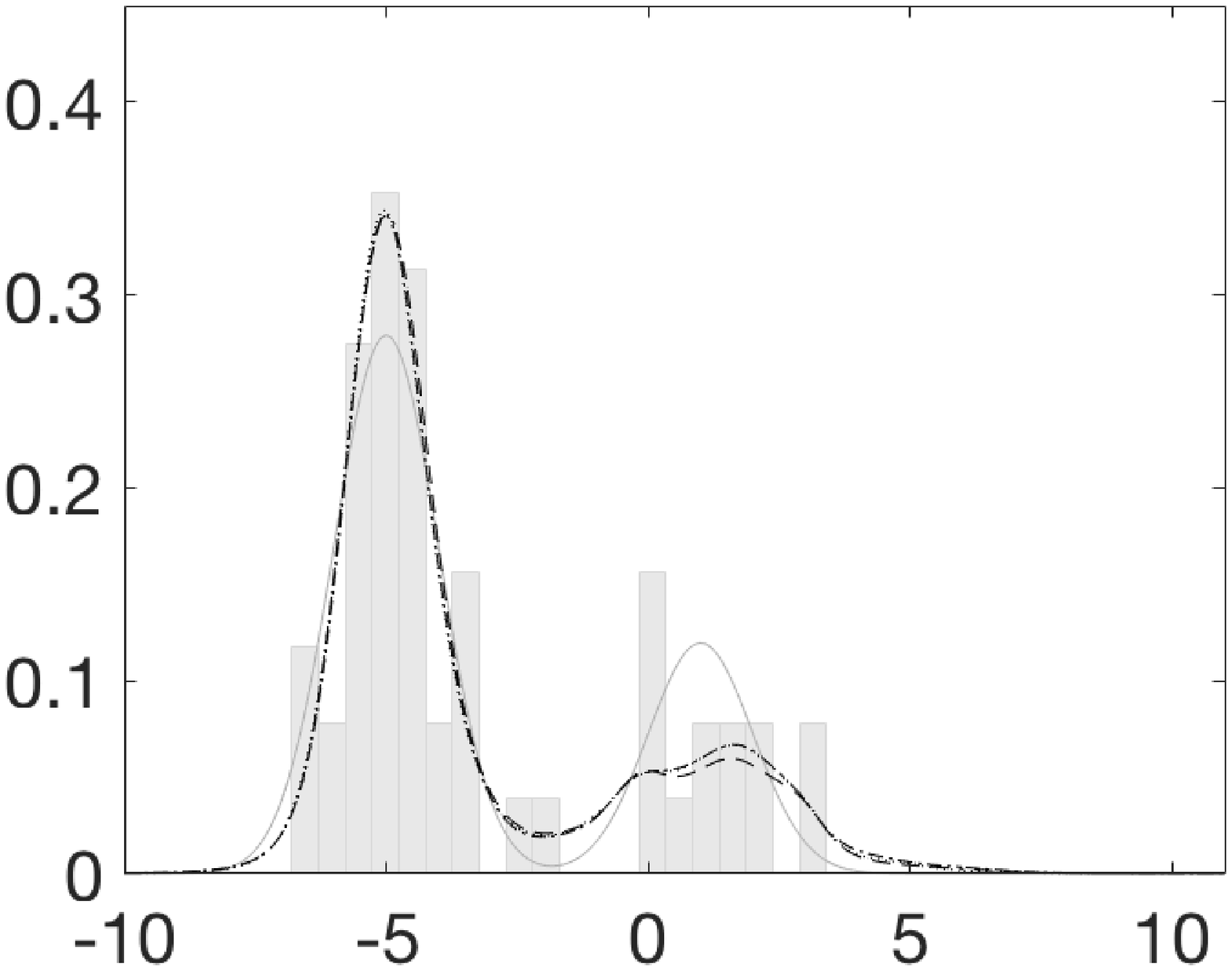}&
      \includegraphics[width=3.6cm]{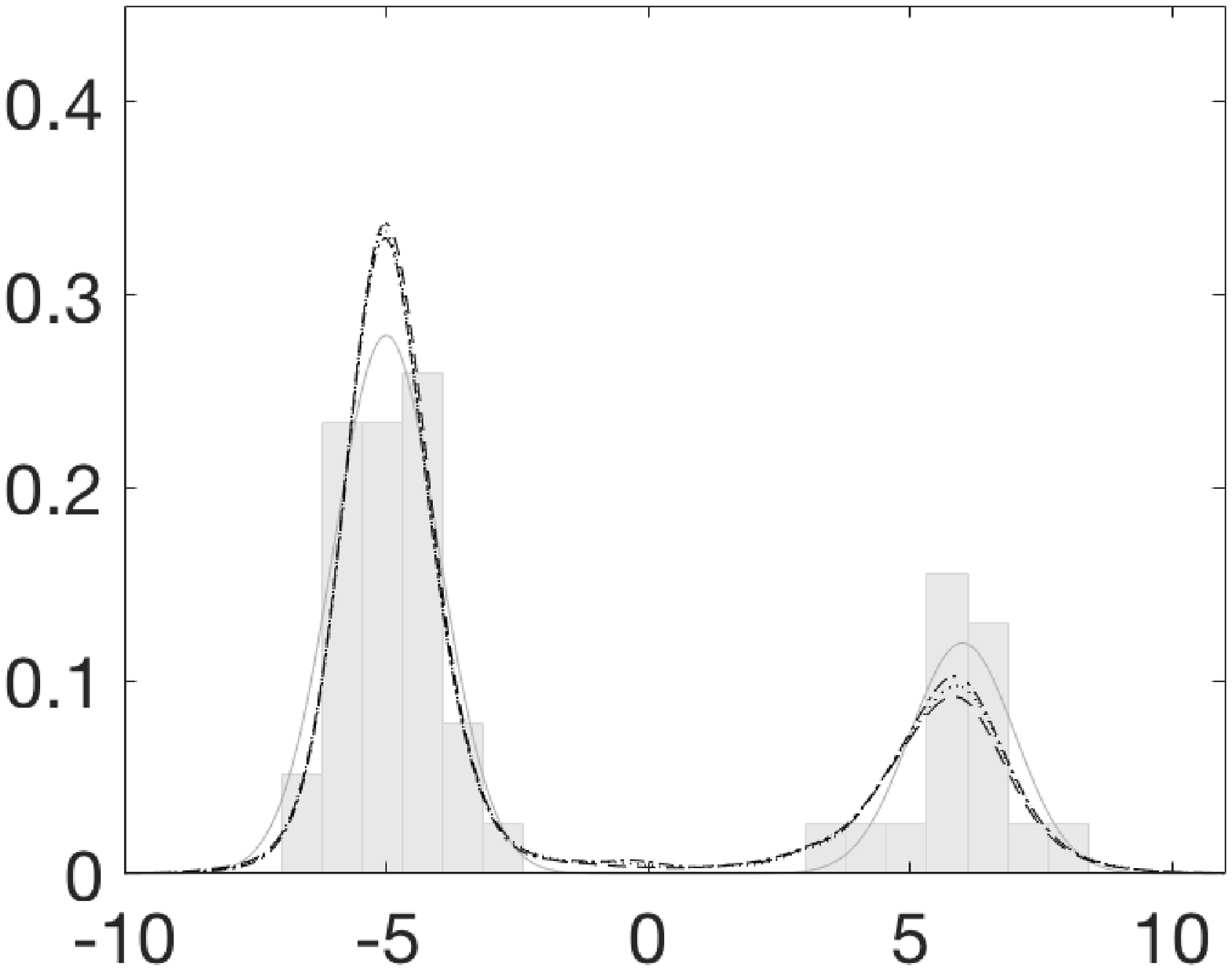}\\
       \end{tabular}
\end{figure}

\clearpage

\begin{figure}[h!]
  \caption{Posterior total number of clusters for the three-component normal (panel (a)) and the two-component normal (panel (b)) mixture experiments.}
  \label{Fig:SimNumberClust}
 \centering
  \setlength{\tabcolsep}{-2pt} 
  \begin{tabular}{cccc}
      \multicolumn{4}{c}{(a) Three-component normal mixture experiment}\\
      \footnotesize{HDP}& \footnotesize{HPYP}& \footnotesize{HGP}&\vspace{-3pt}\\
      \includegraphics[width=3.5cm]{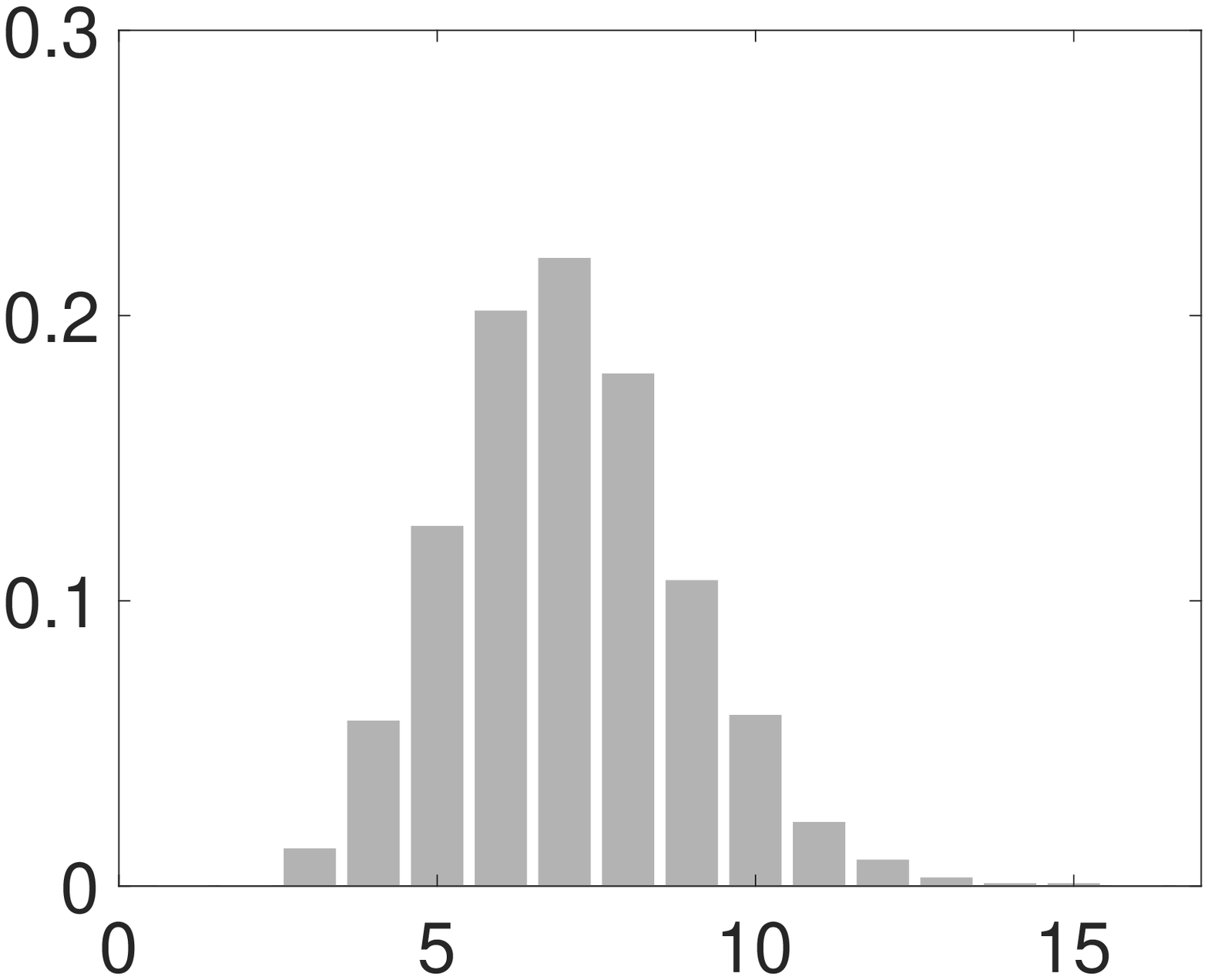}&      \includegraphics[width=3.5cm]{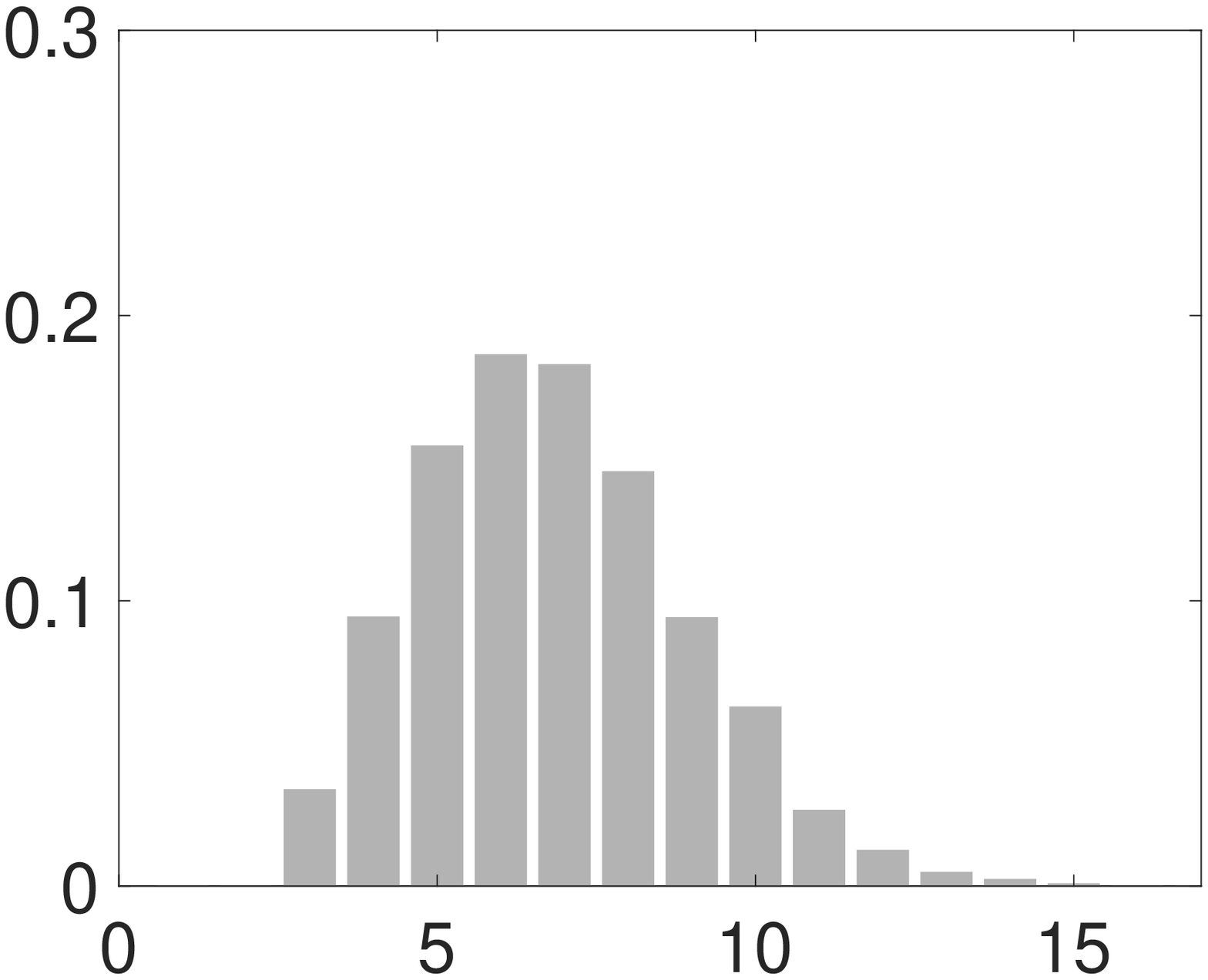}&      \includegraphics[width=3.5cm]{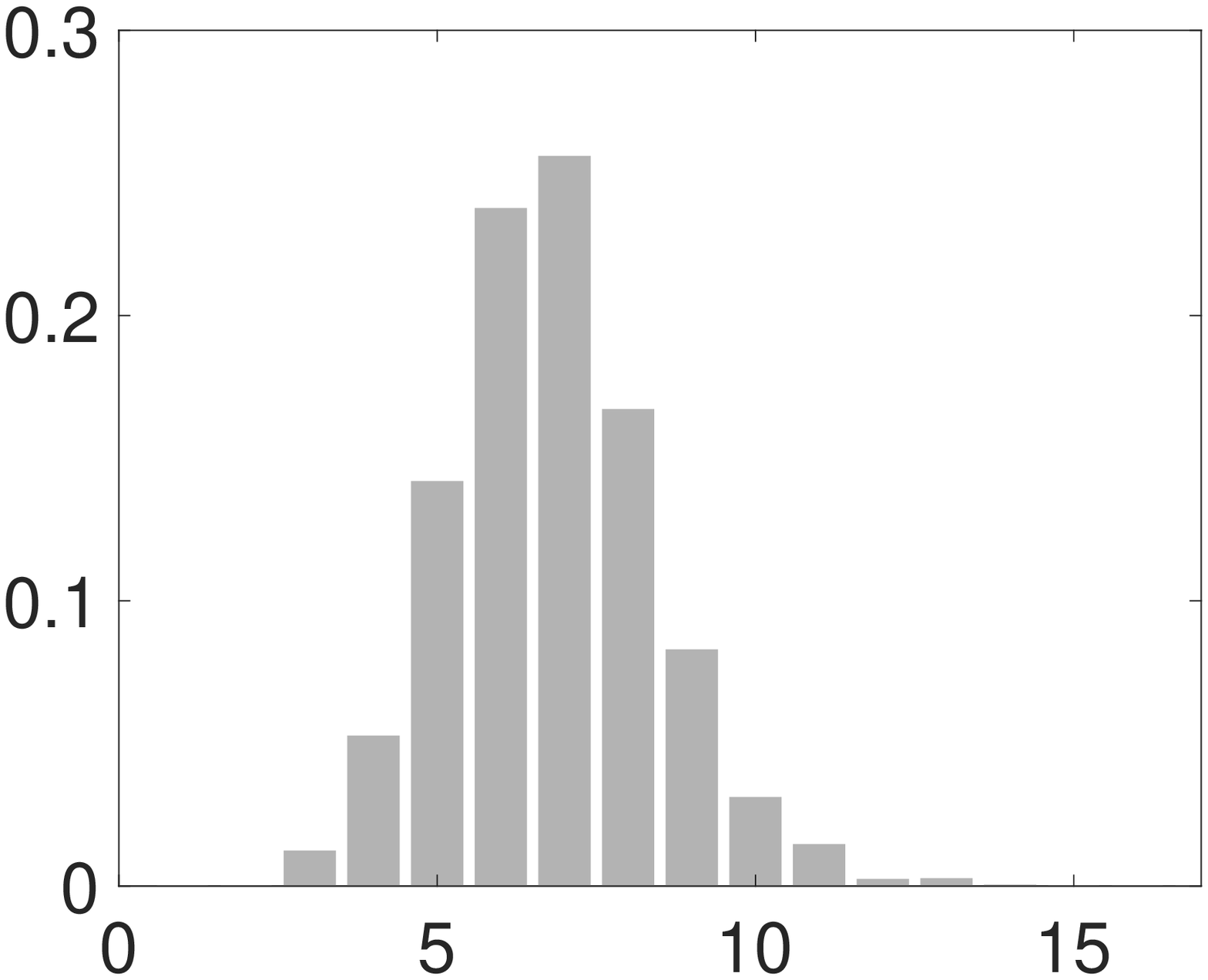}&\\      
      \footnotesize{HDPYP}& \footnotesize{HPYDP}& \footnotesize{HGDP}&\footnotesize{HGPYP}\vspace{-3pt}\\
      \includegraphics[width=3.5cm]{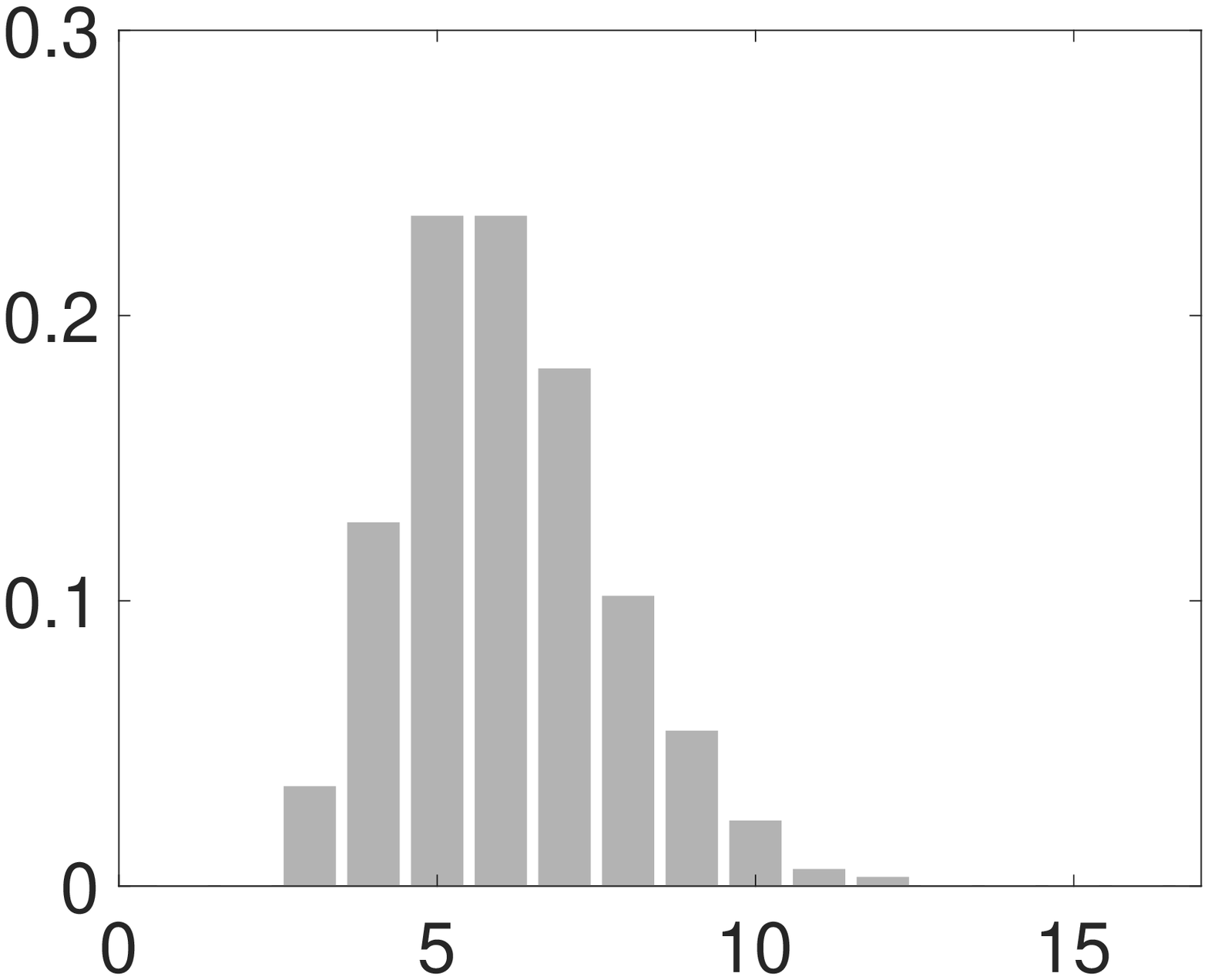}&         \includegraphics[width=3.5cm]{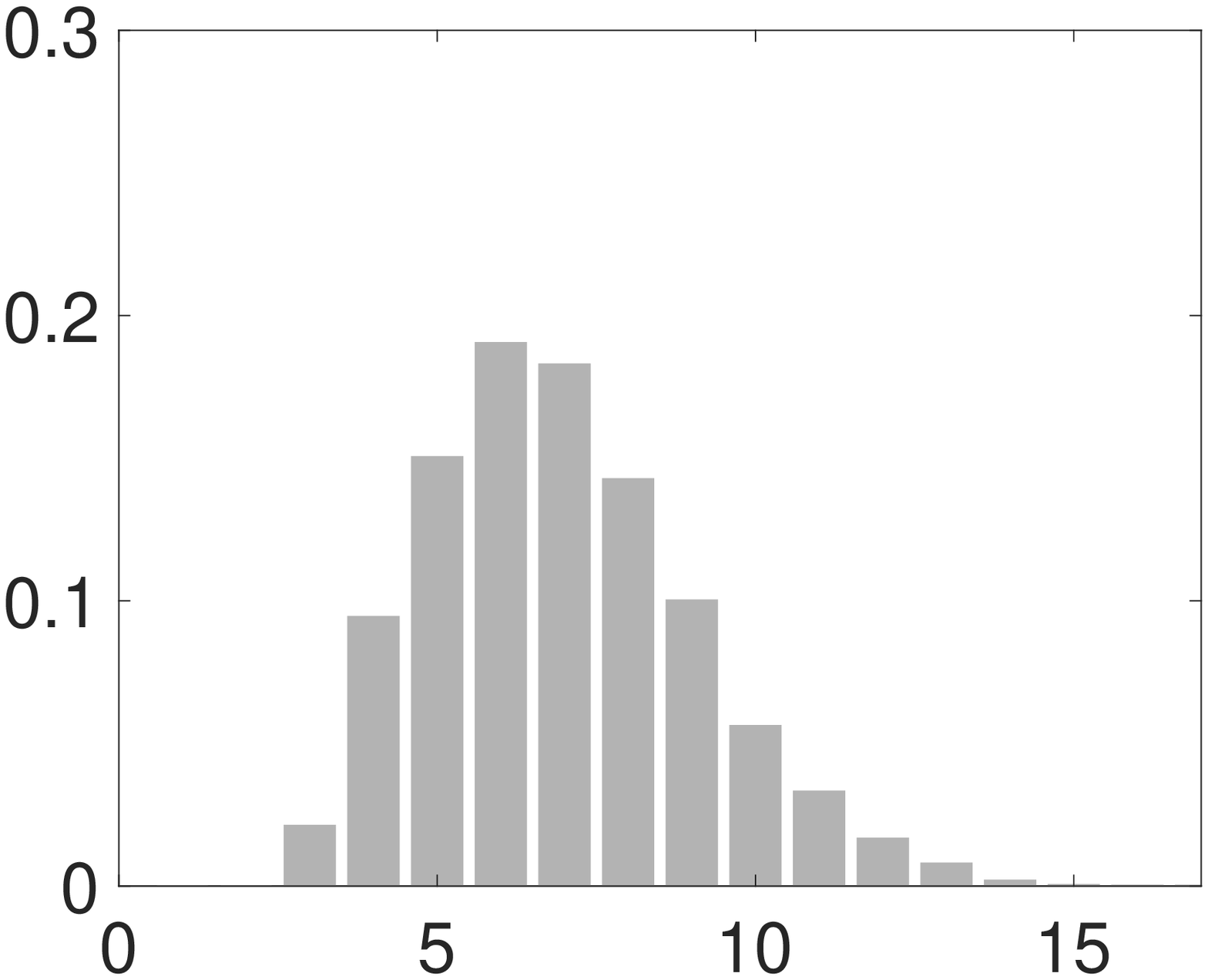}&     
      \includegraphics[width=3.5cm]{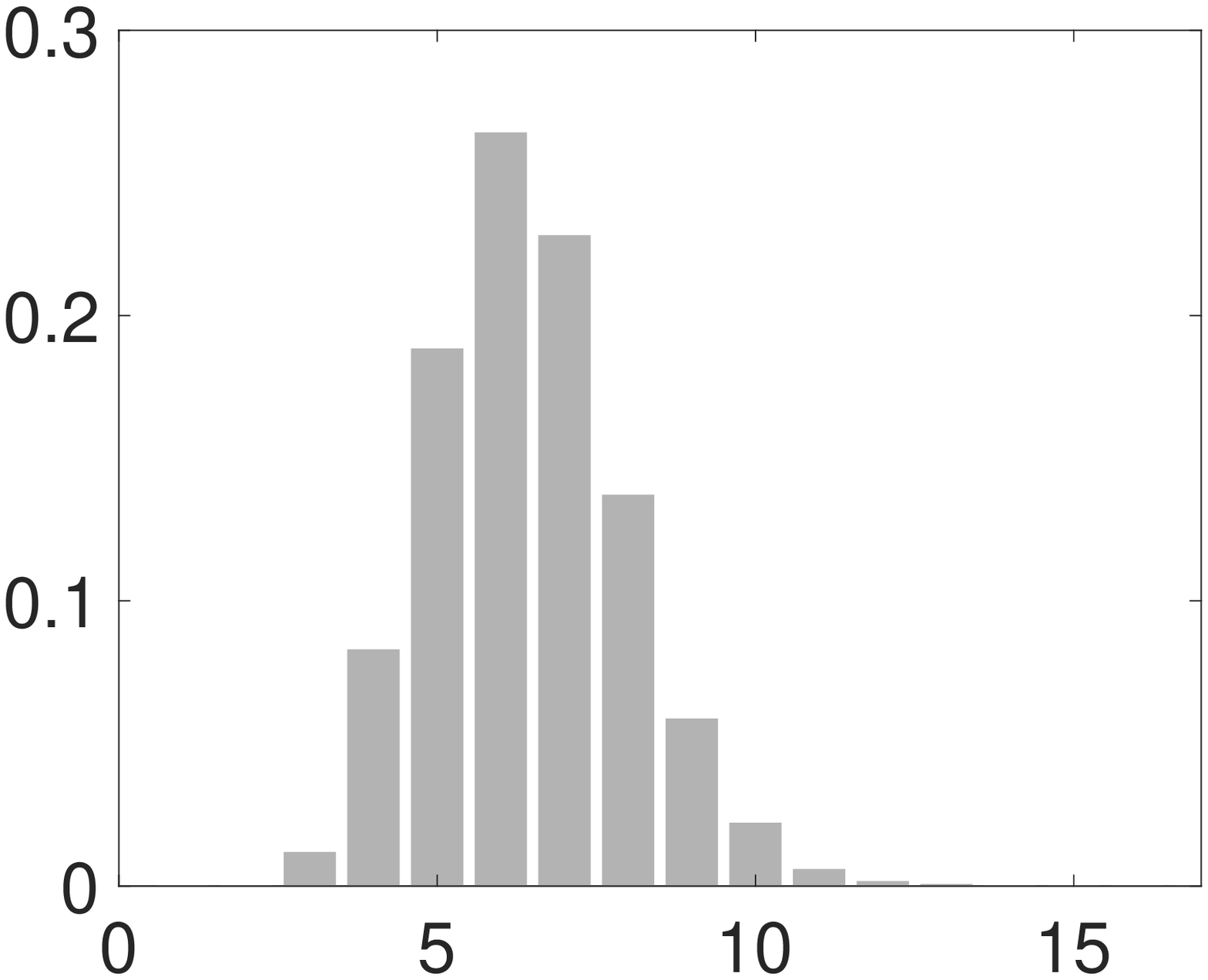}&      \includegraphics[width=3.5cm]{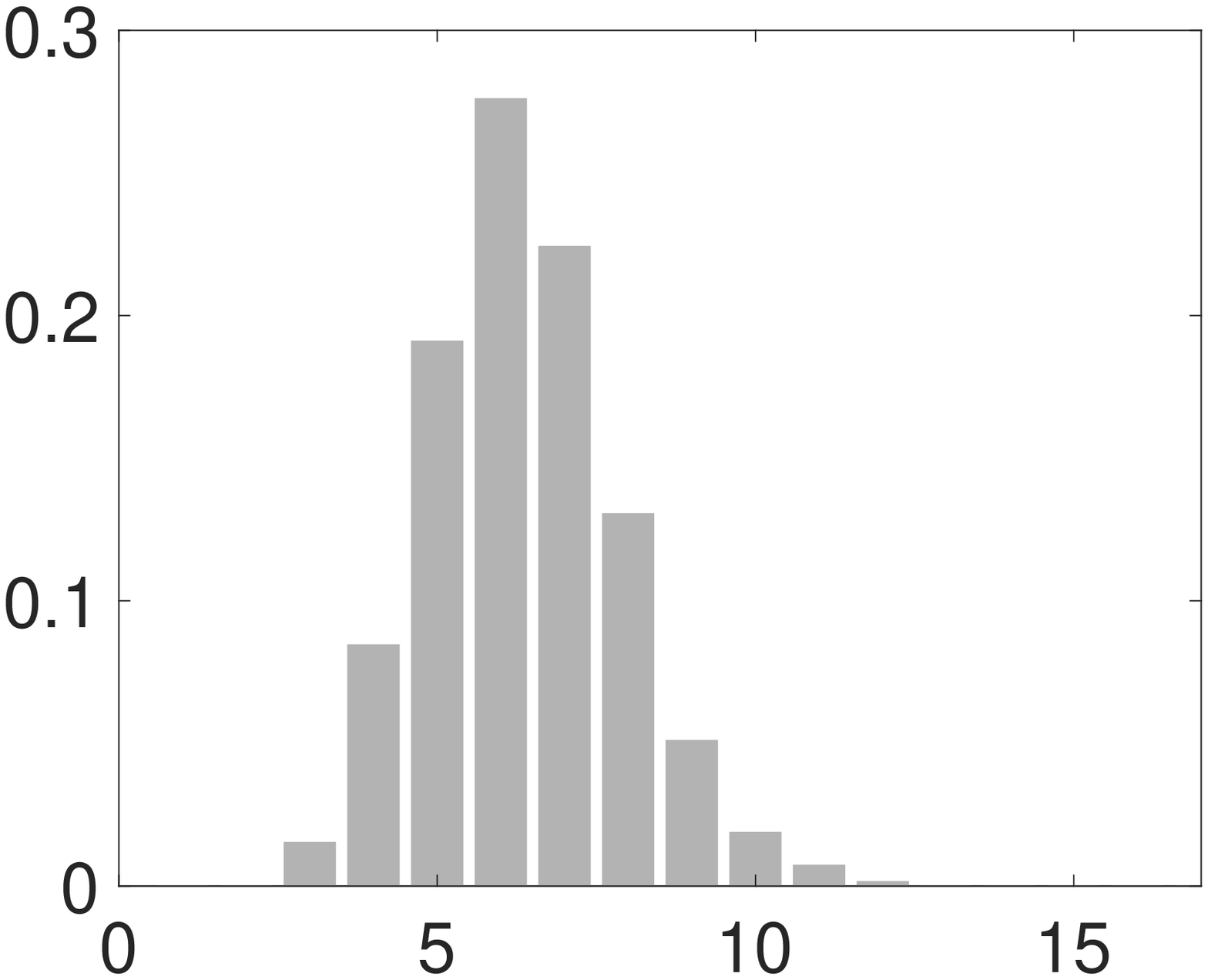}\vspace{10pt}\\
      \multicolumn{4}{c}{(b) Two-component normal mixture experiment}\\
      \footnotesize{HDP}& \footnotesize{HPYP}& \footnotesize{HGP}\vspace{-3pt}\\
      \includegraphics[width=3.5cm]{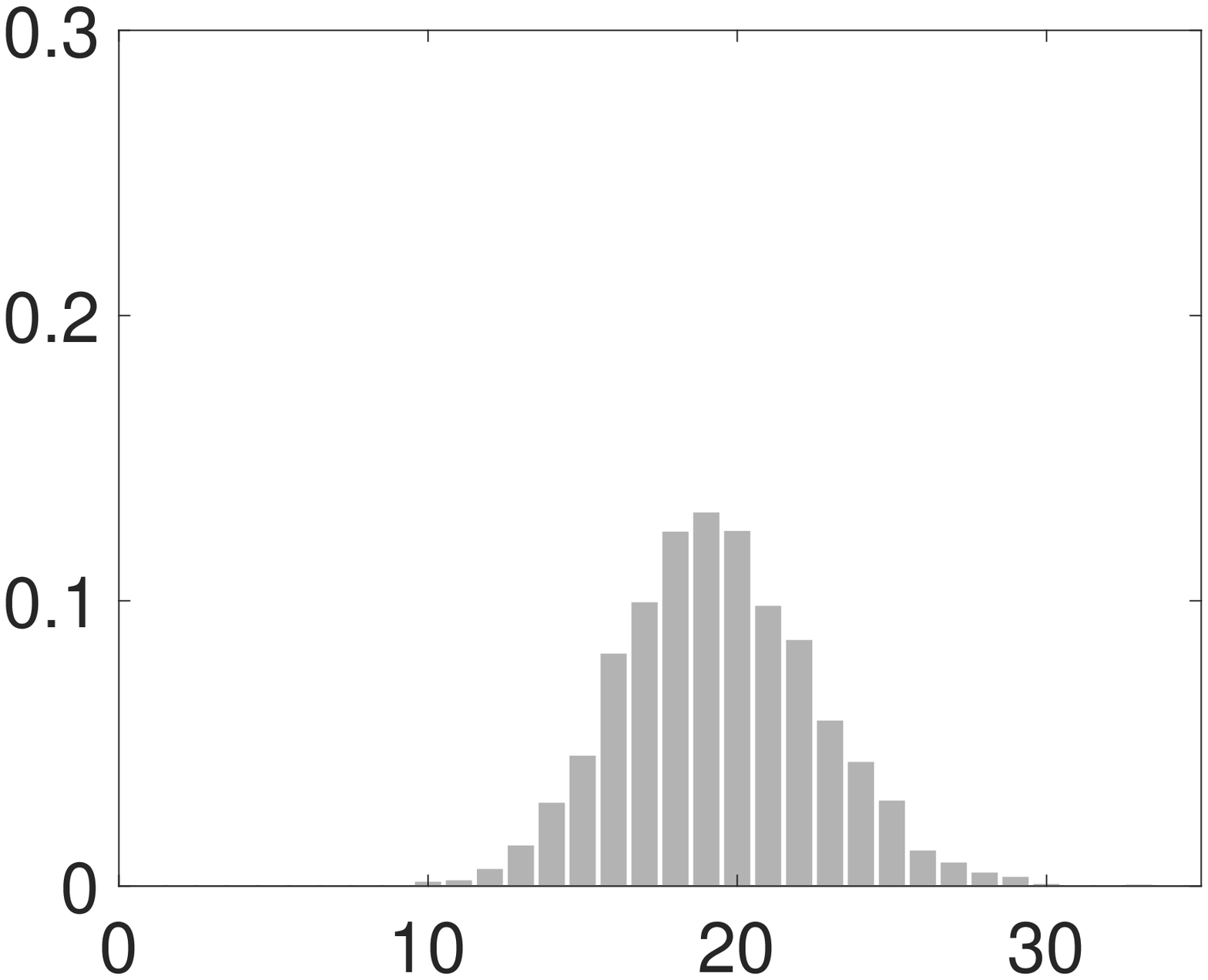}&      \includegraphics[width=3.5cm]{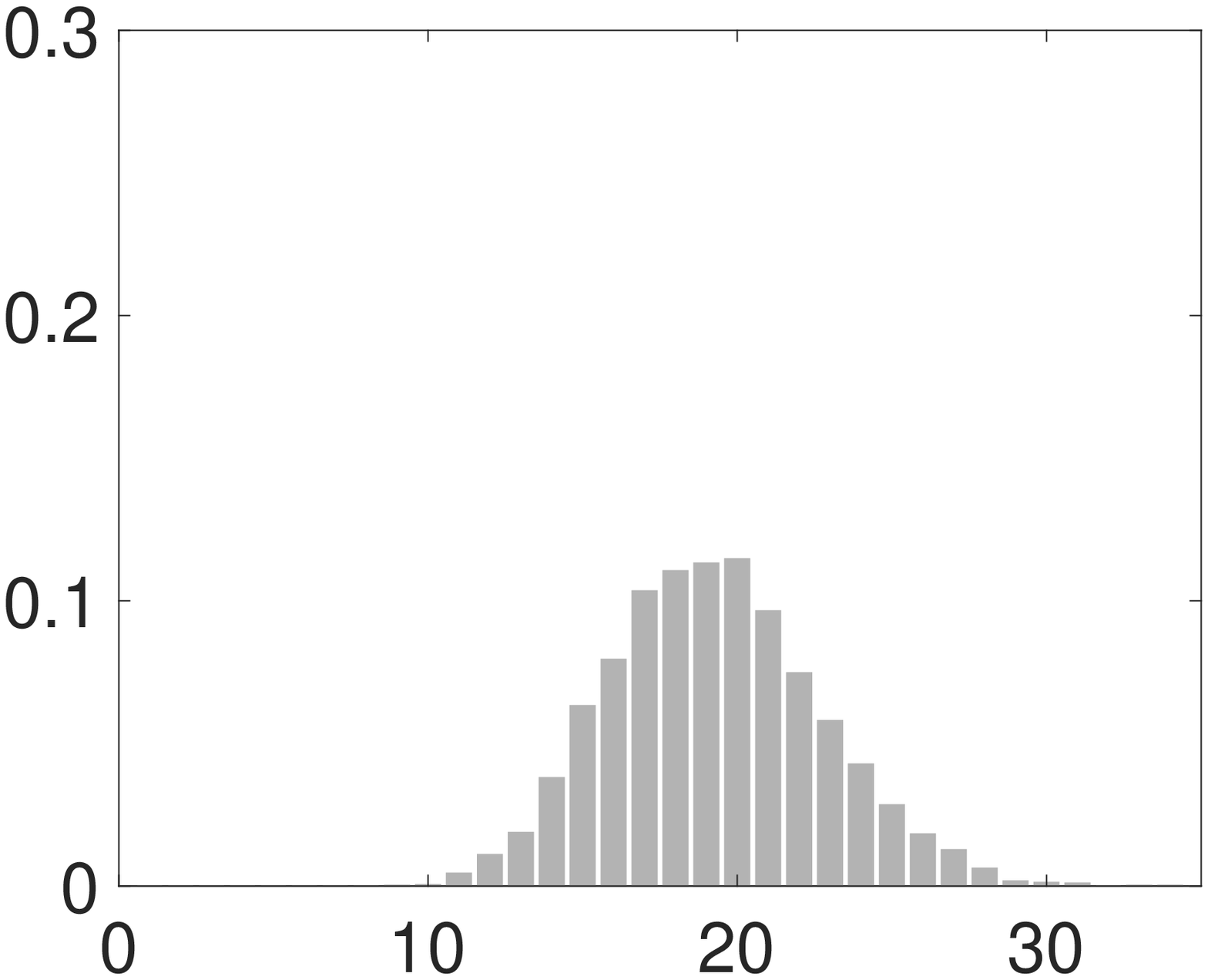}&      \includegraphics[width=3.5cm]{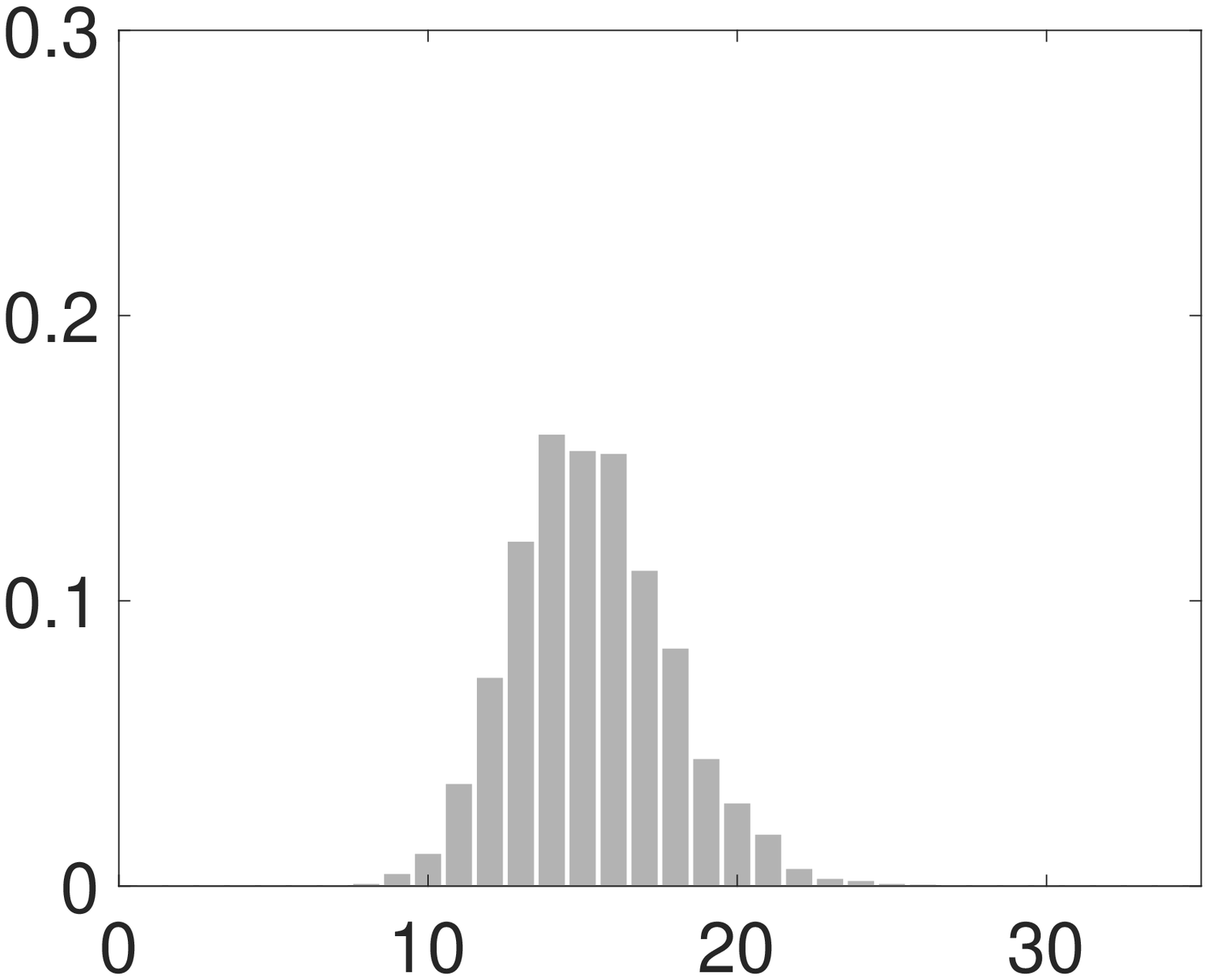}&\\      
      \footnotesize{HDPYP}& \footnotesize{HPYDP}& \footnotesize{HGDP}&\footnotesize{HGPYP}\vspace{-3pt}\\      
      \includegraphics[width=3.5cm]{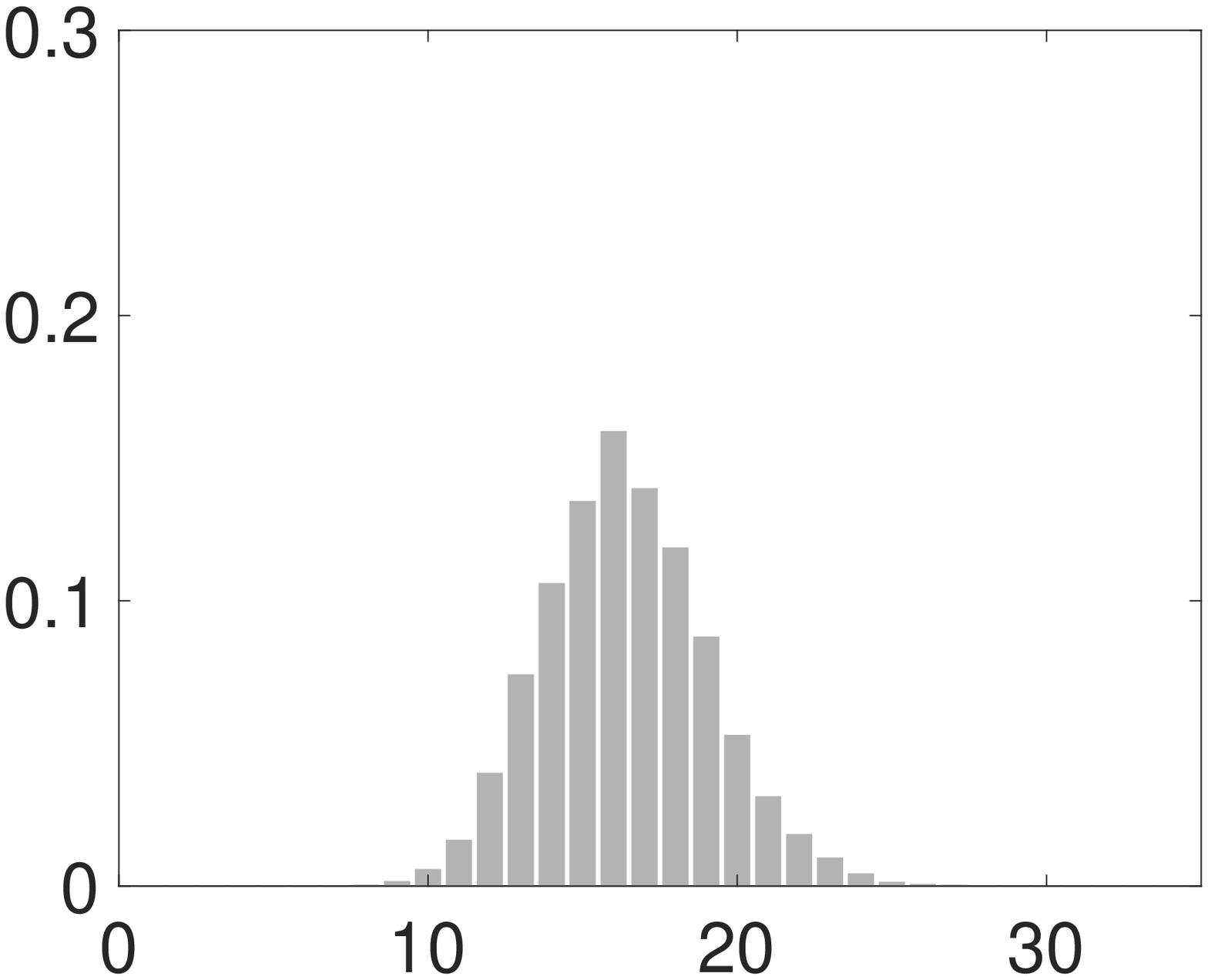}&         \includegraphics[width=3.5cm]{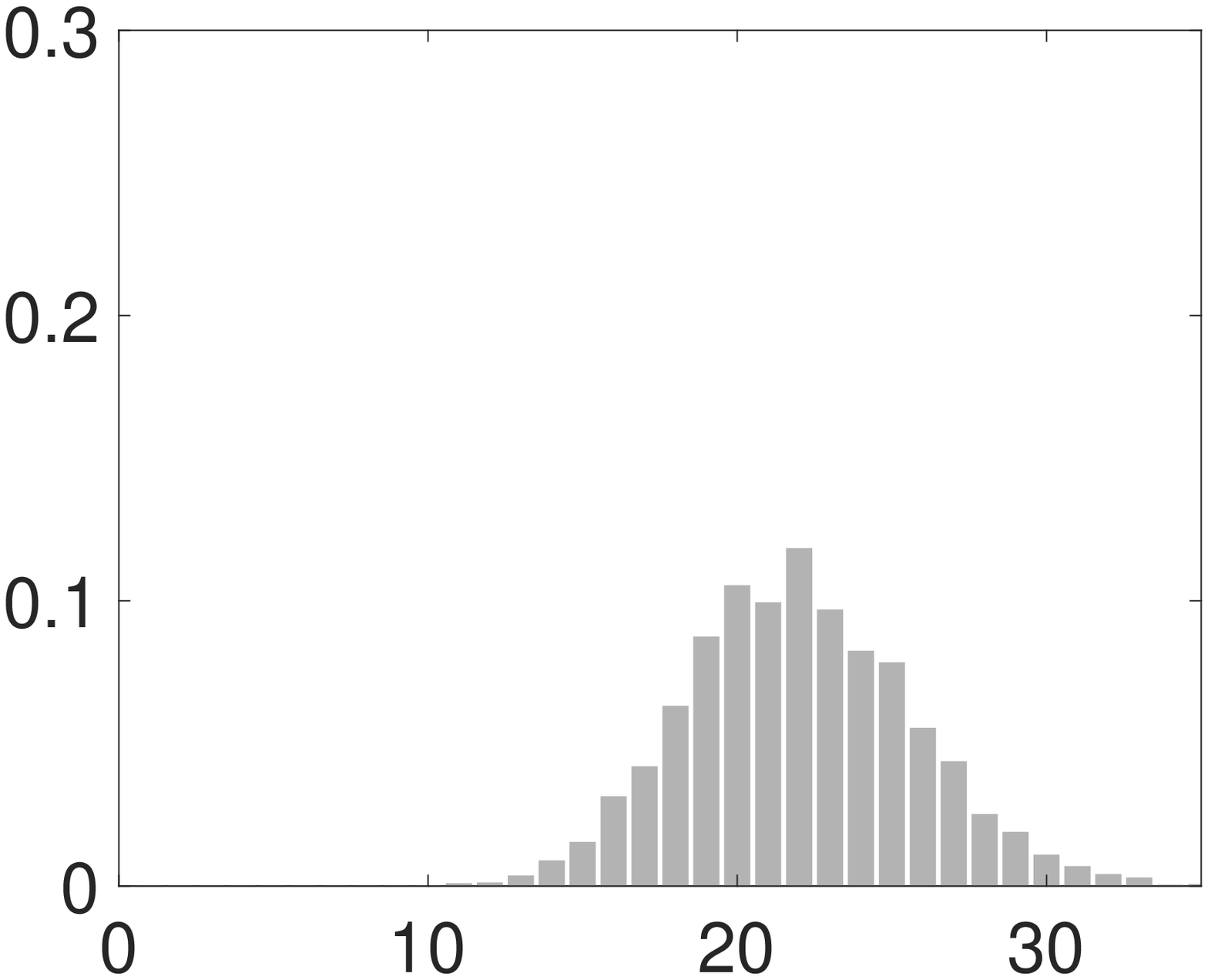}&     
      \includegraphics[width=3.5cm]{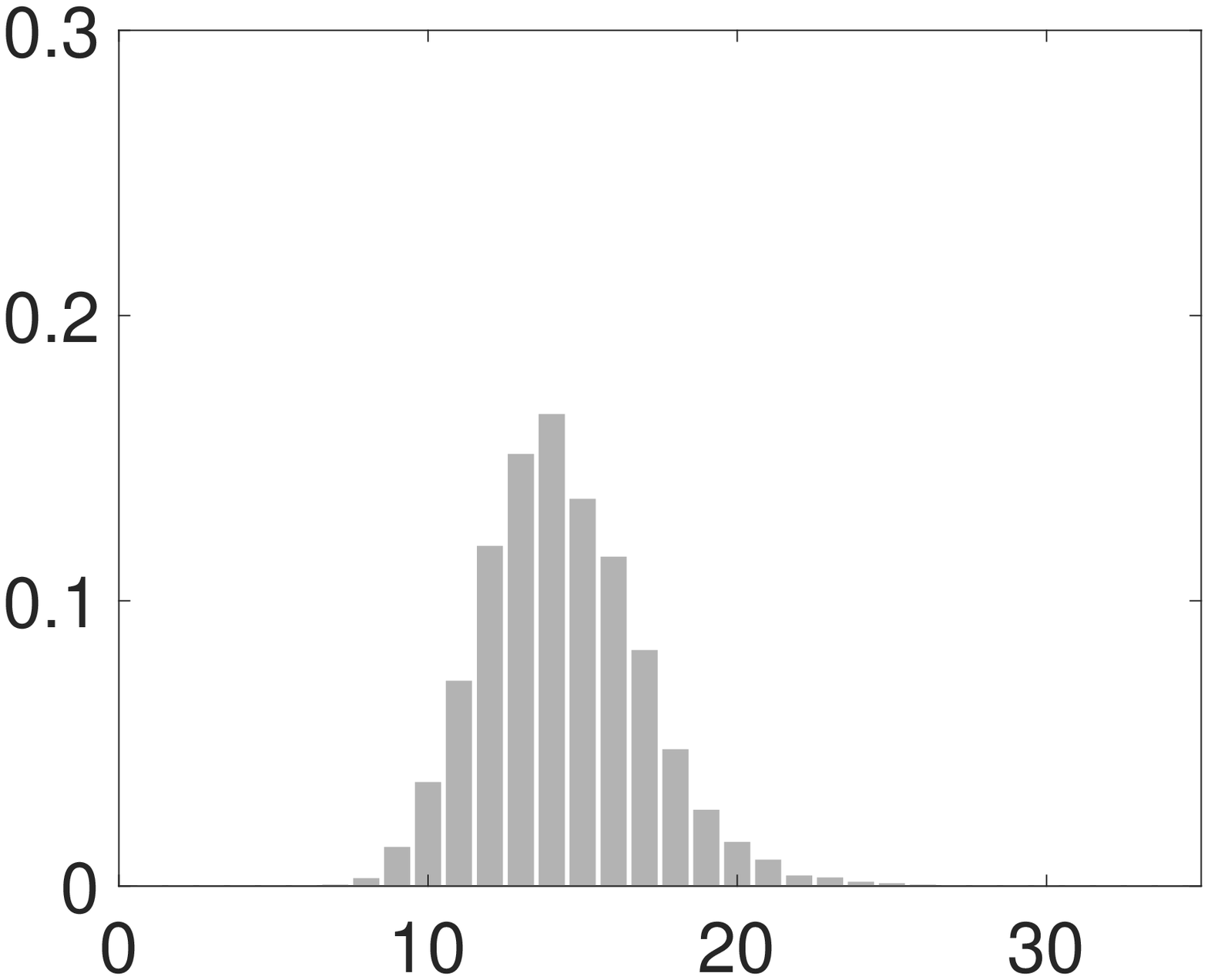}&\includegraphics[width=3.5cm]{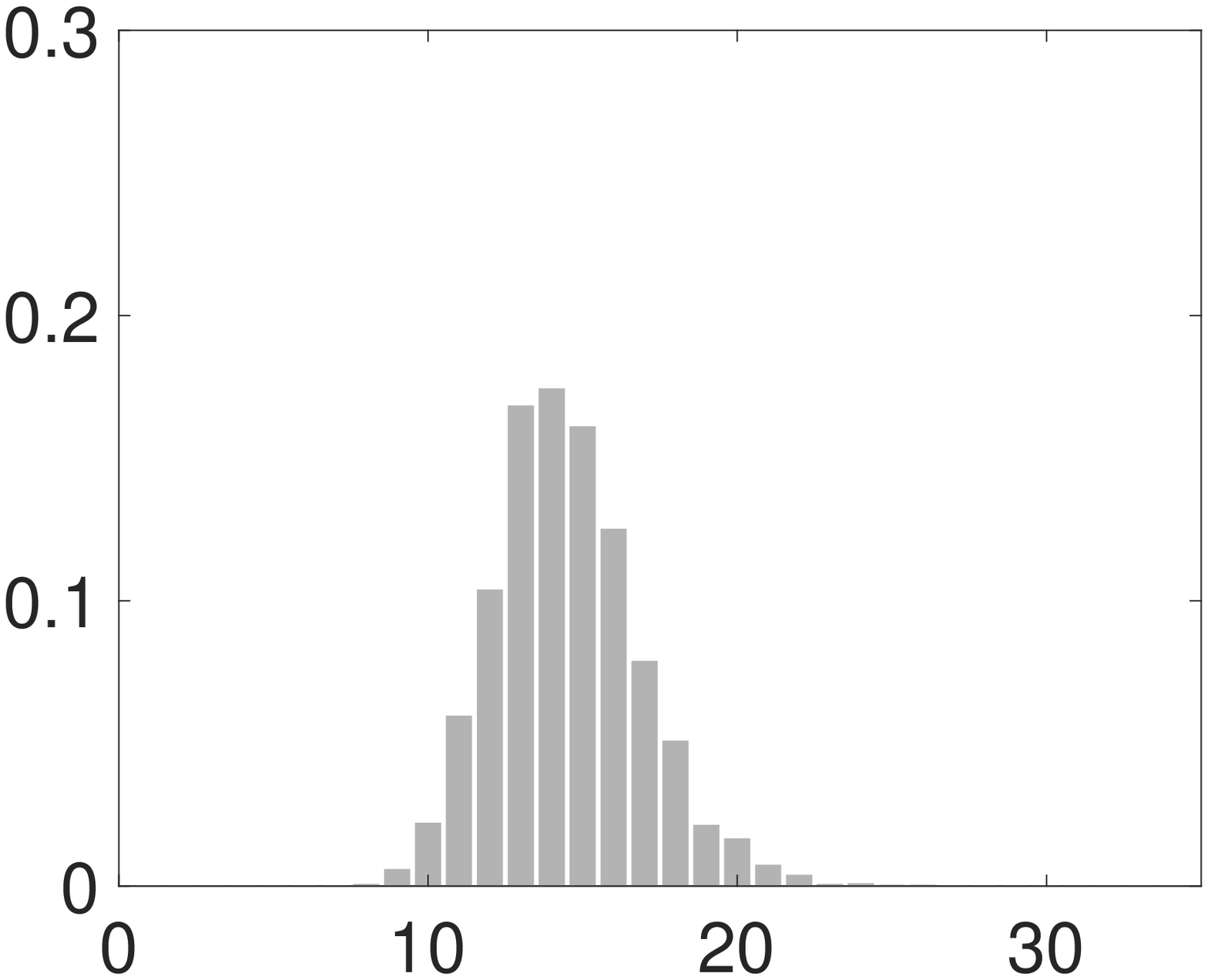}\\
  \end{tabular}
\end{figure}



\end{document}